\documentclass[]{book}
\usepackage [english] {babel}
\usepackage{amsfonts,amssymb,amsthm,epsfig,epstopdf,titling,url,array}
\usepackage [utf8]{inputenc}
\usepackage{csquotes}
\usepackage{textcomp}
\usepackage{geometry}
\usepackage{graphicx}
\usepackage{wrapfig}
\usepackage{bm}
\usepackage{booktabs}
\usepackage{floatflt}
\usepackage{amssymb}
\usepackage{amsmath}
\usepackage{amsthm}
\usepackage[margin=10pt, font=scriptsize, labelfont=bf]{caption}
\usepackage{subfigure}
\usepackage{multicol}
\usepackage{braket}
\usepackage{hyperref}
\usepackage{amsmath}
\usepackage{tikz}
\usetikzlibrary{trees}
\usetikzlibrary{calc}
\usetikzlibrary{decorations.pathmorphing}
\usetikzlibrary{decorations.markings}
\usetikzlibrary{shapes}
\usepackage{verbatim}
\theoremstyle{plain}
\newtheorem{thm}{Theorem}[chapter]
\newtheorem{lem}{Lemma}[chapter]
\newtheorem{rem}{Remark}[chapter]
\newtheorem{corollary}{Corollary}
\newtheorem{prop}{Proposition}

\theoremstyle{definition}
\newtheorem{defn}{Definition}

\hypersetup{
    colorlinks,
    citecolor=red,
    filecolor=red,
    linkcolor=blue,
    urlcolor=yellow}
\newcommand{\etalchar}[1]{$^{#1}$}

\begin{document}

\begin{titlepage}
\begin{center}
\begin{figure}
\centering
\includegraphics[scale=1]{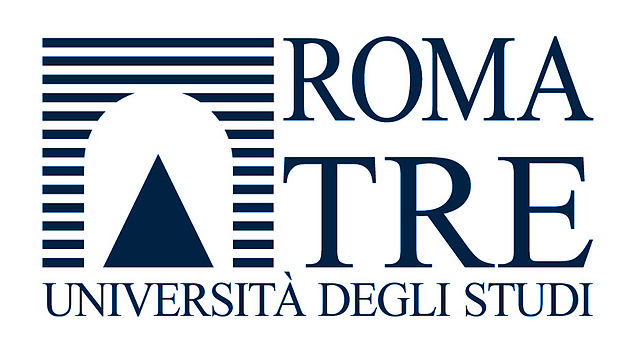}
\end{figure}
{{\Large{\textsc{Universit\`a degli Studi Roma Tre}}}} \rule[0.1cm]{15.8cm}{0.1mm}
\rule[0.5cm]{15.8cm}{0.6mm}
{\small{\bf DOTTORATO DI RICERCA IN MATEMATICA
}}\\
\vspace{1cm} {\large{XXIX CICLO}}
\end{center}
\vspace{3cm}
\begin{center}
{\LARGE{\bf Interacting fermions on the half-line:}}\\
\vspace{3mm}
{\LARGE{\bf boundary counterterms and boundary corrections}}
\end{center}
\vspace{25mm}
\par
\noindent
\begin{minipage}[l]{1\textwidth}
{\large{\bf Candidate: Giovanni Antinucci}}\\

{\large{\bf Supervisor: Prof. Alessandro Giuliani}\\

\large{\bf Coordinator: Prof. Luigi Chierchia}}
\end{minipage}
\vspace{15mm}
\begin{center}
{\large{\bf DEC, 2016}}
\end{center}
\end{titlepage}

\newpage
\null
\thispagestyle{empty}
\newpage
\begin{flushright}
\null\vspace{\stretch{1}}
{\it La spensieratezza\\
va stroncata alla nascita.}\\
-{\tiny ANTONIO REZZA - Fratto\_X}-
\vspace{\stretch{2}}\null
\end{flushright}

\chapter*{Acknowledgements}

The work behind this thesis lasted three years, one of which I spent abroad: so let me take the opportunity to thank first of all the Università degli Studi di Roma Tre, that gave me the possibility of attending different educational experiences. Besides, I wish to express my gratitude to those who hosted me: the {\it Université de Lyon 1}, where I spent three months (October - December 2015), supported by the A*MIDEX project Hypathie (n. ANR-11-IDEX-0001-02) funded by the {\it "Investissements d'Avenir"} 25 French Government program, managed by the French National Research Agency (ANR), having the opportunity to enjoy the very stimulating end exciting climate of the Probability Group of Professor Toninelli; and the  {\it Université de Geneve}, and {\it SwissMap project}, that funded me to attend, as a visiting student, the wonderful Masterclass in Statistical Physics 2015-2016, from January to June 2016.\\

Il mio più grande e sincero ringraziamento va al Prof. A. Giuliani, per la fiducia che ha avuto in me proponendomi questo problema, e per l’interesse, la disponibilità, e l’infinita pazienza con cui mi ha seguito: senza il suo supporto, {\it scientifico} e {\it umano}, la realizzazione di questo lavoro non sarebbe stata possibile.\\

Ringrazio il Prof. V. Mastropietro, per le stimolanti discussioni e per il {\it cruciale suggerimento} di inserire un {\it non-local boundary counterterm}.\\
Ringrazio il Prof. M. Porta, per tutto l’entusiasmo che mi ha trasmesso durante le mie, seppur brevi, visite a Zurigo.\\
Ringrazio il Dott. I. Jauslin, che ha reso meno traumatico e più veloce il mio inserimento nel mondo della Fisica Matematica e del RG, grazie alle lunghissime e coinvolgenti discussioni sempre affrontate con interesse e passione.\\
Ringrazio il Prof. Gallavotti: è grazie a lui, ai suoi consigli e al suo esempio che ho intrapreso la strada della Fisica Matematica e, in particolare, del Gruppo di Rinormalizzazione.\\

Infine, ringrazio i Professori M. Correggi, G. Dell’Antonio, G. Panati, A. Teta per avermi sempre dato la possibilità di partecipare alle molte conferenze da loro organizzate, e per aver sempre accolto con grande interesse qualsiasi mia domanda o semplice curiosità.

\chapter*{Abstract}
Recent years witnessed an extensive development of the theory of the critical point in two-dimensional statistical systems, which allowed to prove  {\it existence} and {\it conformal invariance} of the {\it scaling limit} for two-dimensional Ising model and dimers in planar graphs. Unfortunately, we are still far from a full understanding of the subject: so far, exact solutions at the lattice level, in particular determinant structure and exact discrete holomorphicity, play a cucial role in the rigorous control of the scaling limit. The few results about not-integrable (interacting) systems at criticality are still unable to deal with {\it finite domains} and {\it boundary corrections}, which are of course crucial for getting informations about conformal covariance. 
In this thesis, we address the question of adapting constructive Renormalization Group methods to non-integrable critical systems in $d= 1+1$ dimensions. We study a system of interacting spinless fermions on a one-dimensional semi-infinite lattice, which can be considered as a prototype of the Luttinger universality class with Dirichlet Boundary Conditions. We develop a convergent renormalized expression for the thermodynamic observables in the presence of a quadratic {\it boundary defect} counterterm, polynomially localized at the boundary. In particular, we get explicit bounds on the boundary corrections to the specific ground state energy.

\tableofcontents

\chapter{Introduction}

\section{Motivations}

\paragraph{Critical phenomena and symmetries} 
It is now well understood that the common background to {\it critical phenomena} displayed by very different systems (both classical and quantum) like liquid-vapor transition, paramagnetic-ferromagnetic transition, superfluids, superconductors, {\it etc.}, is the strong fluctuation of infinitely many coupled variables. So, once this mechanism has been identified, it is natural to introduce models that are both as realistic as possible and mathematically treatable.\\
In this framework, two dimensional $(2D)$ statistical systems play the special role of being the {\it simplest non trivial examples} of systems undergoing a phase transition: in this regard, it must be mentioned the Ising model, introduced by Ising \cite{Ising1925} and exactly solved first by Onsager \cite{Onsager:1944aa} and later by many others (with different techniques) \cite{kaufman1949crystal, kac1952combinatorial, lieb1961two,hurst1966new,samuel1980use}.\\
The importance of the Ising model is due to the fact that it has been the first model giving quantitative indications that a {\it microscopic short range interaction} can produce phase transitions. A remarkable fact is that the notion of integrability for the Ising model in zero magnetic field is really strong, meaning that the model can be {\it exactly mapped} into a system of free fermions \cite{lieb1961two,hurst1966new,samuel1980use}, so that it is not only possible to explicitly compute the {\it free energy} and the {\it magnetization}, but one can even get exact formulae (allowing an exact control of the asymptotic behaviour for large distances of some of them) for several {\it spin correlation functions}: energy-energy correlation functions, {\it spin-spin} correlation functions \cite{montroll1963correlations,wu1976spin,tracy1973neutron,barouch1973zero,mccoy2014two}, some multispin correlation functions (with some constraint on the relative positions of the spins) \cite{kadanoff1969correlations}. The impressive thing is that, thanks to these results, it is possible to caclulate the {\it critical exponents} of the model and check that they are different from those predicted by the Curie-Weiss theory of ferromagnetism: so one claims that {\it the Ising model belongs to a different universality class}. The concept of {\it universality class}, thanks to which we classify in the same {\it family} models that, even though describe very different physical systems, show the same {\it critical behaviour} (meaning that the {\it critical exponents are the same}, provided one managed to identify in some sense the corresponding thermodynamic functions for the systems under comparison) has been largely studied and understood by using the Renormalization Group (RG) tools \cite{kadanoff1966scaling,di1969microscopic,callan1970broken, symanzik1970small, wilson1971renormalization, wilson1971renormalization2,wilson1972critical}: in the language of RG one says that the correlation functions of two systems respecting the {\it same symmetries} and with {\it interactions differing only by irrelevant terms}, are characterized by {\it the same} long distance behaviour at the critical point ({\it i.e. they have the same critical exponents}).
\paragraph{Conformal invariance} If, on the one hand, the idea of  RG arises conceptually from the {\it scale invariance} of the scaling limit, which roughly tells us that under a {\it uniform change of lenght scale the correlation functions transform covariantly in a simple way}, on the other hand it is just thanks to the RG analysis that we can rigoroulsy conclude that the infrared fixed point, for many statistical systems, {\it is in fact scale invariant}, as well as invariant also with respect to the {\it usual} Euclidean symmetries.\\ 
This has been the starting point for naturally guessing that the {\it scaling limit} should be, under {\it plausible assumptions}, invariant under the {\it larger group} of the {\it conformal transformations}, which roughly speaking is a generalization of a scale transformation with a lenght-rescaling factor depending continuously on position, {\it i.e.} it is {\it conformal invariant}. The first time that the idea of the {\it conformal invariance of the scaling limit} appeared in literature was in a paper by Polyakov \cite{polyakov1970conformal}, in which he showed that the correlation functions are invariant under conformal transformations, and he used this to compute explicitly the three-point correlation functions. Nevertheless, it seems that for a while the {\it deep consequences} of the conformal invariance have not been properly understood by the community (for example, just a very short section is dedicated to this topic in the review about phase transitions and critical phenomena by Wegner \cite{wegner1976phase}). \\
The breakthrough in the field came with the seminal paper by Belavin, Polyakov and Zamolodchikov in 1984 \cite{belavin241infinite}, based on the fact that in $d=2$ the {\it conformal group} is much larger than in higher dimensions, and in particular it is isomorphic to the group of {\it analytic transformations}, whose corresponding group algebra, known as Virasoro algebra, had already been studied with different purposes in the context of particle theory \cite{kac1979lecture, jacob1974dual,mansouri1972gauge,ferrara1972conformal}. Roughly speaking they showed that, assuming the {\it conformal invariance} of the scaling limit, it is possible to get not only the {\it critical exponents} of the model, but also {\it all the multi-point correlation functions at the critical point} (the analysis is based on the correspondence of each of the {\it primary scaling operators} of a two-dimensional systems with a representation of the Virasoro algebra which allows, in some particular case, to perform explicit computations); notably they recognized that the theory is characterized by the {\it central charge} (also known as {\it conformal anomaly} since it is associated with an anomaly term in the commutation relations of the stress energy tensor).\\
\cite{belavin241infinite} paved the way for an impressive number of papers that, in the immediately following years, increased and refined the understanding of the topic \cite{dotsenko1984conformal,dotsenko1985four,dotsenko1985operator,
dotsenko1984critical,friedan1984conformal}. A special comment is deserved by the famous paper by Cardy \cite{cardy1984conformal} in which, for the first time, he realized that, using some {\it conformal mapping}, conformal invariance allows the explicit calculation of some {\it finite-size effects at the critical point}, offering the possibility of getting properties of the infinite system from some finite samples of the same system. In particular, these {\it finite-size effects} are linked to the concept of {\it central charge}: as already pointed out in \cite{belavin241infinite}, and then studied by Affleck \cite{affleck1986universal}, Bl\"ote-Cardy-Nightingale \cite{blote1986conformal} and Friedan-Qiu-Shenker \cite{friedan1984conformal}, the {\it central charge can be defined in terms of the finite size corrections to the free energy at criticality}; moreover,  it has been recognized that, for some special cases in which the critical theory is fully characterized by the value of the chentral charge, the critical exponents are all explicitly known in terms of the Kac formula.
Of course, the importance of this result has to be read taking into account that the increase of the {\it computational power} of the computers offered, at that times, a lot of convincing validation of the principle of conformal invariance at the critical point. \\
It is worth stressing again that this huge amount of impressive results have been achieved {\it regarding the conformal invariance of the scaling limit as a principle}, since there was no rigorous proof of this fact, and one big conceptual problem was that it was not even straightforward to give a {\it mathematical definition} of the scaling limit ({\it i.e. to define a precise mathematical object to study in order to check whether the scaling limit of the model is conformal invariant or not}). A milestone in this direction  has been posed by Schramm \cite{schramm2000scaling} who, in the context of percolation models (in which in some sense one can reduce the study of the critical point physics to the study of interfaces), inspired by the {\it numerical results} presented in \cite{langlands1994conformal} and by the explicit formula that Cardy proposed for the limit of percolation crossing probability \cite{cardy1992critical}, introduced the idea that interfaces of percolation models should belong to a family of {\it conformal invariant continuous non-selfcrossing curves}: the {\it Schramm-Loewner Evolutions} (SLE). The strenght of this proposal lies properly in the {\it mathematical formalization of the goal}: to prove that the interfaces of percolation models converge, at the scaling limit, to SLE processes. The revolution in the rigorous understanding of this topic came with the rigorous proof of Cardy's formula for critical site percolation on the triangular lattice by Smirnov \cite{smirnov2001critical}, whose great importance lies in an impressive consequence: {\it Cardy's formula} is equivalent to convergence of interfaces to SLE, meaning that proving the conformal invariance of a {\it well chosen observable} is enough to prove the conformal invariance of interfaces. This idea has been afterwards extended to Ising model, introducing the famous {\it fermionic observable} (a discrete holomorphic quantity) \cite{riva2006holomorphic}, which can be proved to converge to a holomorphic function in the scaling limit: this is the basic tool of the very ample literature \cite{smirnov2001critical, chelkak2012universality,chelkak2012conformal, duminil2012conformal, chelkak2014convergence, benoist2014conformal, lawler2011conformal} {\it etc.}, based on techniques of combinatorics, probability and discrete analysis (in particular discrete holomorphicity, {\it a.k.a. pre-holomorphicity}), already introduced by Kenyon in studying {\it close packed dimers} \cite{kenyon2001dominos}.\\
Fifteen years after the first step in the rigorous study of the {\it conformal invariance of the scaling limit of two-dimensional statistical systems}, the level of the understanding of these phenomena is really advanced, even though {\it mostly limited} to {\it integrable models}, since so far the {\it integrability} seems to be a {\it fundamental ingredient} to give full control of the existence and conformal invariance of the scaling limit. As a matter of fact, the two models on which results are more complete are models at the {\it free Fermi point}: Ising and dimers. Anyway, the existence and the conformal invariance of the scaling limit is believed to be independent of a {\it free fermions} description, meaning that it is believed to be true also for {\it non-integrable 2D systems} corresponding, in terms of fermions, to {\it interacting fermions} in $d=1+1$. \\
An important open problem, which motivates the study of this thesis, is the proof of {\it conformal invariance of the scaling limit of interacting non-solvable models} close, but non exactly at, the {\it free fermion point}.

\paragraph{Luttinger Liquid and its Universality Class}{In order to understand 2D critical systems outside the {\it free fermion point}, we need techniques for dealing with interacting fermionic systems in d=1+1}.\\
The starting point is to recognise that there are a few interacting models, that present a non-trivial critical behaviour, that can be exaclty solved by using special methods ({\it for instance} bosonization in the case of the Luttinger model, Bethe Ansatz in the case of the one dimensional antiferromagnetic Heisenberg model).\\
The reference model in the  framework of $1+1$ dimensional fermionic system is the {\it exactly solvable} Luttinger model, which is the simplest possible model describing many body systems consisting of two different kinds of fermions, left-movers and right-movers on a (continuous) segment {\it interacting} via a weak, short range density-density potential. The model was introduced by Luttinger \cite{luttinger1963exactly} and rigorously solved by Mattis and Lieb \cite{mattis1965exact}, using a very famous techinque now known as {\it bosonization}, (see below). The interesting feature of the Luttinger model is that the presence of the interaction really changes the physical behaviour: first of all, the ground state of the system is characterized by a density of states which does not have a discontinuity at the Fermi momentum (as the Free Fermi Gas), but its graph has an infinite slope with tangency exponent $a(\lambda)=\mathcal O(\lambda^2)$, called the anomaly of the Fermi surface; moreover, also the $n-$point functions, which can be computed exactly, show a large distance behaviour with {\it anomalous exponents continuously depending on the interaction size $\lambda$}.\\
Of course, one wonders if the Luttinger physics is in some sense robust under weak modification of the model; Luttinger model is in fact believed to give a robust description of models described in terms of spinless $1D$ fermions. By combining bosonization techniques with (formal) perturbative renormalization arguments, it has been conjectured \cite{kadanoff1977connections, haldane1981luttinger, luther1975calculation, nienhuis1984critical,den1981derivation} the existence of a {\it universality class}, called {\it $8$-vertex universality class} or {\it Luttinger liquids}, describing a variety of  two-dimensional classical systems, such as the $6$ and $8$-vertex models, the Ashkin-Teller model, and the interacting dimer models at close-packing; and one-dimensional quantum systems, such as the Heisenberg spin chains, the Luttinger model itself and the spinless Hubbard model and perturbation thereof.\\
The inspiring idea is that all the systems in the $8$-vertex universality class can be described in terms of {\it lattice fermions}, {\it i.e.} a family of Grassmann variables $\psi^{\epsilon}_{\omega,\bm x}$ indexed by lattice vertices $\bm x=(x,x_0)$ and by indices $\epsilon,\omega=\pm$. In particular, for a special choice of the model parameters  (free-fermion point), these fermions are non-interacting, so the system is analyitically diagonalizable. As soon as we change the values of these parameters, the fermions become interacting, meaning that, in the {\it action} of the system, at least a quartic term in the Grassmann variables appear, so the partition and correlation functions are given by {\it non-Gaussian Grassmann integrals}. Performing a {\it formal continuum limit}, these fermions become interacting Dirac fermions in $d=1+1$ dimensions, which are massless at criticality.\\
Let us start by considering non-interacting massless Dirac fermions $\psi^\sigma_{x,\omega}$ with propagator antidiagonal in $\sigma=\pm$, diagonal in $\omega=\pm$, and translation-invariant in $\bm x\in\mathbb R^2$:
$$\left<\psi^-_{\bm x,\omega}\psi^+_{\bm 0,\omega}\right>=\frac{1}{2\pi}\frac{1}{x_0+i\omega x}.$$
In this case, the bosonization consists in two identities \cite{itzykson1991statistical}:
\begin{itemize}
\item the multi-point correlations of the {\it fermionic density} $\psi^+_{\bm x,\omega}\psi^-_{\bm x,\omega}$ are the same as the derivative of a boson field $\phi$ (massless gaussian field):
$$\psi^+_{x,\omega}\psi^-_{\bm x,\omega}\leftrightarrow -\omega  \partial_\omega\phi(\bm x), \hspace{5mm} \partial_\omega:=\frac{1}{2}(\partial_{x_0}-i\omega\partial_{x}),$$
so that in particular correlations of {\it odd order} and truncated correlations of order larger than 2 vanish,
\item the {\it fermionic mass} $\psi^+_{\bm x,\omega}\psi^-_{\bm x,-\omega}$ has the same correlations as a normal ordered exponential of the boson field:
$$\psi^+_{\bm x,\omega}\psi^-_{\bm x,\omega}\leftrightarrow  \frac{1}{2\pi} :e^{2\pi i\omega\phi(\bm x)}:$$
\end{itemize}
The remarkable fact is that these relations, up to {\it renormalization constants}, remain valid also in the presence of suitable density-density interaction, in particular for the Thirring model  \cite{thirring1958soluble}, which can be thought of as a limit of Luttinger as the interactions tends to a local delta potential. In the case of Luttinger model, even though the correspondence between fermionic and bosonic representation is more complicated (so the formulae are more cumbersome), the consequences remain asymptotically the same.\\
There are some other models, as the {\it antiferromagnetic 1D Heisenberg model} and the {\it Hubbard model}, that are exactly solvabe by {\it Bethe ansatz}, thanks to which it is possible to compute the thermodynamic functions, the {\it critical exponents} and some of the {\it amplitudes}, unfortunately without a full control of the correlation functions.\\
Summarizing, there are some {\it very special, solvable models} as the Luttinger model, the Thirring model, the antiferromagnetic 1D Heisenberg model , the Hubbard model for which it is possible to explicitly check the conjectured properties we just mentioned. Formal perturbation/renormalization arguments suggest that the same long distance behavior should be displayed by several other models, provided that the interaction strength is suitably tuned, so that the critical exponent of (say) the Green function coincides; once this tuning is performed, all the other exponents should coincide. Even more remarkably, the resulting critical exponents and amplitudes should satisfy the same universal relations valid in the Luttinger model (Kadanoff and Haldane relations). These predictions, which are expected to hold for a very general class of models, have been first of all checked for solvable models, but checking them in absence of exact solutions or of bosonization identities is of course a hard mathematical task. Constructive quantum field theory and Renormalization Group methods are powerful tools to study these problems, and in fact they allowed to rigorously prove these prediction for several different models, as we briefly discuss in the next paragraph.
 
\paragraph{Renormalization Group in the context of many-body theories} RG methods {\it à la Wilson} \cite{wilson1971renormalization,wilson1971renormalization2} at the very beginning have been the basic tools for studying several problems in Constructive Quantum Field Theory, as the renormalization of $\phi^4_d$ theories \cite{gallavotti1985renormalization,polchinski1984renormalization,glimm2012quantum, glimm1973particle,guerra1976boundary} and the existence of the continuum limit of Quantum Theory models in $d=1+1$, as the Gross-Neveu model with $N>1$ colors \cite{gawedzki1985massless,feldman1986renormalizable}, or the massive Yukawa model \cite{lesniewski1987effective}.\\
In applying these methods to one dimensional fermionic system, one has to deal with the further difficulty given by the fact that the {\it theory is not asymptotically free}, but {\it it belongs to a class of models characterized by a vanishing beta function} (implying that a second order computation is not enough to recognize the nature of the flow of the effective coupling, but one has to exploit {\it non trivial cancellations at all orders in the renormalized expansion}).\\
Chronologically, Dzylaloshinskii and Larkin \cite{dzyaloshinskii1974correlation} first attacked the Tomonaga model \cite{tomonaga1950remarks} (not exactly solvable) performing a {\it non rigorous resummation} of the perturbative expansion after several uncontrolled estimates. Then, Metzner and Di Castro \cite{metzner1993conservation} correctly pointed out that the vanishing of the beta function, in multiplicative RG, follows from the Ward identities which, anyways, are {\it exactly true} only in the Luttinger model, not in {\it non solvable ones}.\\
Of course the natural next step is to push the understanding of these topics at a {\it rigorous level}. The Roman school gave an impressive contribution to the {\it construction of models with vanishing Beta function}: the starting point of a huge literature was the one-dimensional system of interacting non relativistic fermions in the continuum, studied in a seminal paper by Benfatto, Gallavotti, Procacci and Scoppola \cite{benfatto1993beta}, where the crucial property of vanishing beta function was proved by comparing this model with the exact solution of the Luttinger model (rigorous RG methods had already been used in attacking fermionic many-body theories in \cite{benfatto1990perturbation,feldman1992infinite}). Later, Benfatto and Mastropietro adapted the already mentioned ideas by Dzylaloshinskii, Larkin, Metzner, Di Castro to a {\it constructive RG approach}, and in doing that they had to overcome several technical problems. As a matter of fact, it is worth mentioning a series of papers in which, without any comparison with the exact solution of the Luttinger model, they proved the vanishing of the beta function \cite{benfattodensity,benfatto2001renormalization,benfatto2004ward, benfatto2005ward} overcoming well known problem due to the conflict between the Wilsonian RG and Ward Identities (basically, the Wilsonian RG breaks the local gauge invariance necessary to get Ward Identities). These techniques have been then used to study a variety of models belonging to the Luttinger universality class and, for some of these, to check the Kadanoff- Haldane predictions, \cite{giuliani2005anomalous, benfatto2009extended, benfatto2010universality,benfatto2010universal, benfatto2011drude, benfatto2014universality, benfatto2014universalityii, giuliani2015height, giuliani2016haldane}, {\it etc.}.\\
In light of these important achievements of the Renormalization Group methods, one naturally asks: what is missing to prove the {\it conformal invariance of the scaling limit of interacting non-solvable models?}
\paragraph{ Motivations of this thesis} Due to its {\it robustness} with respect to perturbations of {\it solvable models}, one is naturally tempted to use RG to {\it extend} the conformal invariance informations we have about {\it exactly-solvable systems} to {\it interacting, non solvable systems}. In this direction, a first step has been moved by Giuliani and Mastropietro \cite{giuliani2013universal}, who rigorously checked, for an {interacting Ising model on a torus} (so the system is {\it translation invariant}), the CFT prediction according to which, at the critical temperature, the finite size corrections to the free energy are universal (meaning that they are exactly independent of the interaction). Moreover, they showed that, as proposed by Affleck \cite{affleck1986universal} and Blote-Cardy-Nightingale \cite{blote1986conformal},
the central charge, defined in terms of the coefficient of the first subleading term to the free energy is constant and equal to $1/2$ for all $0<\lambda\leq \lambda_0$ where $\lambda_0$ is a small but finite convergence radius. Besides, it is worth mentioning \cite{giuliani2012scaling} where multipoint correlation functions are explicitly computed in the scaling limit in which the lattice spacing is sent to zero and the temperature at the critical one, in the case of a ferromagnetic Ising model weakly perturbed by a finite range perturbation. Anyway, if on the one hand these results confirm that the {\it energy-energy correlations} are in fact those predicted by {\it conformal field theories} and {\it bosonization}, on the other hand they are not enough to prove the {\it conformal invariance} of the scaling limit, since a control of the {\it boundary terms} is still missing.\\
Indeed, even though these papers must be considered as the starting point for a wider understanding of the conformal invariance of the interacting critical point, the rigorous contructive RG methods, which are the main tools used in those papers, built up so far are still based too heavily on the {\it translation invariance of the system}, that implies a lot of technical and conceptual simplifications. These considerations seem to identify the goal: adapting the {\it RG formalism} to the case of systems defined in non-trivial domains (hopefully a formalism independent of boundary conditions, as also Brydges suggests in \cite{brydges2007lectures}). In the context of one dimensional Fermionic systems, the simplest non trivial domain is the half-line.\\
In the last 20 years, encouraged by the possibility of realizing and performing measurements on the so called {\it quantum wires}, the theoretical physics community has been interested in trying to describe  finite one dimensional fermionic systems with open boundary conditions \cite{fabrizio1995interacting, meden2000luttinger,mattsson1997properties,grap2009renormalization}, predicting in fact that the boundary induces some {\it anomalous boundary critical exponent}. Besides, it is worth mentioning that a conceptually similar question is linked to two important problems: the one, that we will briefly comment in the conclusive chapter, is the Kondo effect, as pointed out in \cite{affleck1995conformal}.
The other one is the {\it Casimir effect} that, starting from the $1908$'s when a seminal paper by Symanzyk appeared \cite{symanzik1981schrodinger}, motivated a series of papers about $\phi^4_{4-\epsilon}$ theories in non trivial domains (properly in a {\it semispace}, meaning that the simplest possible non-trivial boundary is introduced in the theory) \cite{diehl1981field1, diehl1981field2, diehl1983universality,diehl1986field,diehl1994surface,diehl1997theory,diehl1998massive,dietrich1981critical,mattsson1997properties,cordery1981surface}, in which the basic strategy is to show that the boundary corrections are localized at the boundary and absorbed into a {\it boundary potential}. \\
We stress that when we say that the half-line is the {\it simplest non-trivial domain}, we mean that even being {\it simple to define} it already shows {\it non-trivial complications}: indeed, due to the presence of a boundary, the relevant and marginal terms that {\it naturally} are generated in the construction of the effective theory, respectively related to the {\it density} of the system and the {\it dressed density-density interaction}, are no more {\it running coupling constants}, but more in general they are {\it running coupling functions}.\\
Driven by the fact that, {\it well inside the bulk}, one expects to recover the predictions of the {\it translation invariant theory} (meaning that one expects to lose, at some point, memory of the boundary), an intuitive way to look at the contributions we are interested in, {\it i.e.} the quadratic and quartic terms of the effective theories we define in the RG procedure (being respectively {\it relevant} and {\it marginal} in a RG sense) is to split them into a  {\it bulk} and {\it boundary contributions}. One expects that the first ones are related to the {\it usual running coupling constants} appearing in the {\it analogous translation invariant theory}, while the {\it boundary contributions}, by construction, {\it keep memory} of the boundary. The main technical result is that, in fact, the boundary terms have a {\it dimensional gain}, in the sense of $L_1$ norms, with respect to the {\it bulk} contributions. This dimensional gain is enough to conclude that the {\it boundary correction to the quartic terms} are in fact {\it irrelevant}; unfortunately, on the other hand it is not enough to {\it renormalize the} quadratic contributions, that deeply modify the {\it effective theory}. \\
Of course all these intuitions have to be {\it quantified in a mathematically meaningful way}, so the question we ask is: are we able to make {\it quantitative} the {\it intuitive notions} of {\it nearby the boundary} and {\it well inside the bulk}? In order to do that, it is necessary to control the {\it quadratic terms} that, as just commented, give rise to {\it running coupling functions} instead of running coupling constants. In this thesis we show that it is possible to find a convergent expansion for termodynamic functions provided we choose a suitable {\it quadratic counterterm algebraically localized at the boundary}, whose decay law seems to be compatible with a {\it space-dependent correction to the critical exponents}.

\section{The model and the main result}

We are interested in constructing the ground state of interacting spinless fermions living in a discrete one-dimensional box of mesh size $a=1$ and volume $L\gg 1$ with {\it open boundary conditions}, meaning that the system is defined on a segment instead of on a torus. \\ 
Let $\mathcal F=\oplus_{n=0}^\infty H^{\wedge n}$ be the standard antisymmetric (fermionic) Fock space, where $\wedge$ denotes the antisymmetric tensor product, and let  $\psi^\pm_x$ be the {\it fermionic creation and annihilation} operators defined on $\mathcal F$, where $x$ is the space coordinate and $\Lambda:=\left\{x\in\mathbb Z: 1\leq x\leq L\right\}$,  $L\in \mathbb N$. We introduce the Hamiltonian
\begin{equation}
H=H_0+\lambda V+ \varpi \mathcal N,
\end{equation}
where
\begin{equation}
\begin{split}
H_0&=T_0-\mu_0 N_0,\\
T_0&=\sum_{x\in\Lambda}\psi^+_x\left(-\Delta^d  \psi^-_x\right)=\sum_{x\in\Lambda}\frac{1}{2}\left(-\psi^+_{x+1}\psi^-_x-\psi^+_{x-1}\psi^-_x+2 \psi^+_x\psi^-_x\right),\\
N_0&=\sum_{x\in \Lambda}\psi^+_x\psi^-_x,
\end{split}
\end{equation}
where, in the formula of $T_0$, we have to interpret $\psi^{\pm}_0=\psi^{\pm}_{L+1}=0$, where $\mu_0$ is the chemical potential, choosen in such a way that, if we call $\sigma(T_0):=[e_-, e_+]$ the spectral band of the kinetic operator, $\mu_0\in [e_-+\kappa, e_+-\kappa]$ for some $\kappa>0$ fixed once for all.\\
Morover the interaction of {\it strenght} $\lambda$ is 
\begin{equation}
V=\sum_{x,y \in\Lambda}\psi^+_x\psi^-_xv(x,y)\psi^+_y\psi^-_y
\end{equation}
where $v(x,y)=v(y,x)$ is a real, compactly supported function, and satisfies what we call {\it Dirichlet property}, {\it i.e.} it can be written as
\begin{equation}
v(x,y)=\frac{2}{L+1}\sum_{k \in \mathcal{D}^d_{\Lambda}}\sin(kx)\sin(ky)\hat v(k),
\end{equation}
where $\mathcal D_\Lambda^d=:\left\{k=\frac{n\pi}{L+1}, n=1,\dots,L\right\}$. We stress that the {\it Dirichlet property} of $v( x, y)$ is not crucial at all but it simplifies some technical aspects of the proof.\\
 Finally, $\mathcal N$ is a {\it boundary counterterm} of size $\varpi=\mathcal O(\lambda)$ of the form
\begin{equation}
\mathcal N =\sum_{x,y\in \Lambda}\psi^+_x\psi^-_y \pi(x,y),
\end{equation}
where $\pi(x,y)$ is a Hermitian matrix such that $\sup_{x\in \Lambda}\int dy |\pi(x,y)|=1$.\\
We present here the main result: let $\beta\geq 0$ be the {\it inverse temperature} defining the {\it finite volume specific free energy}
\begin{equation}
f_{\Lambda,\beta}=-\frac{1}{|\Lambda|\beta}\log \left(Tr \left(e^{-\beta H}\right)\right),
\end{equation}
and respectively
\begin{equation}
f_{\Lambda}=-\frac{1}{|\Lambda|}\lim_{\beta\nearrow \infty}\frac{1}{\beta}\log \left(Tr \left(e^{-\beta H}\right)\right),\hspace{3mm} f_{\infty}=-\lim_{|\Lambda|\nearrow \infty}\frac{1}{|\Lambda|}\lim_{\beta\nearrow \infty} \frac{1}{\beta}\log \left(Tr \left( e^{-\beta H}\right)\right),
\label{definition_free_energies_finite_infinite_volume}
\end{equation}
we can state the main result.
\begin{thm}
\label{theorem_main_introduction}
In this framework, there exists a radius $\lambda_0>0$ such that, for any $|\lambda|\leq \lambda_0$ it is possible to fix the {\it boundary defect} $\pi(x,y)$ and its strenght $\varpi=\varpi(\lambda)$ in such a wat that, for any $\theta\in (0,1)$, there exists a constant $C_\theta$ such that 
\begin{equation}
\sum_{y\in\Lambda} \left|\pi(x,y)\right| \leq C_\theta \left(\frac{1}{\left(1+|x|\right)^\theta}+\frac{1}{\left(1+|L-x|\right)^\theta}\right),
\end{equation}
and in such a way that the finite volume specific ground state free energy $f_\Lambda$ admits a convergent expansion in $\lambda$ and $\varpi$, uniformly in $\Lambda$.\\ 
Moreover
\begin{equation}
\left| f_\Lambda-f_\infty \right|\leq |\lambda|\frac{C_\theta}{L^\theta}.
\end{equation}
\end{thm}
Even though it is not explicitly investigated in this thesis, we stress that a straightforward extension of the proof of this theorem would allow one to control the boundary corrections at finite volume also for the correlation functions.

\section{The outline of the proof}

 \paragraph{Multiscale decomposition} The proof relies on a multiscale analysis of the model, in which the free energy and Schwinger functions are expressed as successive integrations over individual scales. To define a multiscale decomposition, we refer to momentum space, in which each scale is defined as a set of momenta $\bm k$’s contained inside an annulus at a distance of $2^h$ for $h\in\mathbb Z$ around the singularities located at the Fermi points. The positive scales correspond to the ultraviolet regime, that we do not study in detail, referring to \cite{benfatto1993beta}. The negative scales contain the essential difficulties of the problem, whose nature is intrinsically infrared.
 \paragraph{Presence of a non trivial boundary} Physically, the presence of a non trivial boundary induces, obviously, the breaking of translation invariance (so of the momentum conservation): one expects that, very far from the boundary, the bulk {\it i.e. correlation functions} tend to the translation invariant one while, going closer and closer to the boundary, one expects some non trivial boundary effect. Despite the conceptual immediacy of this difference between the physics in presence (or in absence) of a boundary, it is an hard problem to deal with from a technical point of view. Indeed, an important symmetry which most RG methods are based on is the {\it invariance of boundary conditions under RG iterative step}: starting with periodic boundary conditions, the integration of {\it a single scale degrees of freedom} gives back an effective theory having exactly the same boundary conditions as the original one, so it is {\it immediately true} that we are dealing with a {\it selfsimilar theory}; as it will be clear later, in the Dirichlet boundary condition case (and it would be the same for any {\it non translation invariant boundary conditions}) the very first integration is enough to give us an effective theory whose quadratic term is no longer diagonal in the {\it Dirichlet basis}, so it is not sufficient to iterate the {\it rescaling and dressing} process, as one {\it usualy would do} to renormalize a theory whose boundary conditions are {\it invariant under Renormalization Group procedure}.
\paragraph{The main idea} So far, we cannot renormalize the theory without the counterterm $\mathcal{N}$ we introduced in the definition of the model. Indeed, the idea will be to keep as a reference a theory with DBC. Technically, the first step is to recognize that the propagator of the model defined on a box with Dirichlet boundary conditions can be written as a linear combination of propagators of a model on a suitably defined box with periodic boundary conditions, computed respectively in the difference of the arguments (translation invariant part) and in the sum of them (remainder). So, in evaluating the Feynman diagrams coming from the fermionic Wick rule, we will follow the following steps:
\begin{itemize}
\item {\bf Dimensional analysis} Being the bulk contribution the dominant one, a naive dimensional analysis would have the same result of the translation invariant case, so the only {\it problematic terms} will be the quartic (marginal) and quadratic (relevant) operators. After a deeper analysis, one can recognize that the presence of a {\it remainder propagator} improve by {\it one scaling dimension} (this terminology will be clear later) the $L_1$  norm of the values of the graphs; so first of all the flow of quartic terms is reduced to the flow of the {\it translation invariant quartic terms} ({\it i.e.} it is the same flow of the bulk theory). On the other hand, this dimensional gain is not enough to renormalize the quadratic term, so we must do something more.
\item {\bf Dirichlet part extraction and dressing of the propagator} The idea is to redefine the {\it localization operator}, in order to, first of all, extract a bulk quadratic term diagonal in the Dirichlet basis to dress the propagator with, bringing the theory back to the well known formalism of the {\it translation invariant case}, and then to extract the relevant and marginal parts.
\item {\bf Tuning the counterterm} In addition to the bulk relevant and marginal terms, our procedure identifies a marginal, boundary quadratic term, whose divergent part is controlled by the counterterm $\varpi \mathcal N$ that we introduced in the Hamiltonian. The counterterm $\varpi\mathcal N$, that physically reflects the breaking of translation invariance of the theory, will be fixed by studying the flow of {\it coupling functions} (no more constants), whose presence is due to the boundary, {\it via a fixed point argument}.
\end{itemize}

The thesis is organized as follows: since conceptually we will refer to the {\it usual way} to perform RG on translation invariant models, first of all we will give a review about how to deal with a one dimensional system of interacting spinless fermions on a periodic lattice; then, we will be able to explain the new ideas arising in the presence of the boundary.\\
In particular,
\begin{itemize}
\item in Chapter (\ref{chapter_fermions_PBC}) we review the RG approach to translationally invariant spinless 1D systems. More precisely:
\begin{itemize}
\item in Section (\ref{section_PBC_the_model}) we define the model, the observables we are interested in and we state the main result of Chapter (\ref{chapter_fermions_PBC}),
\item in Section (\ref{section_1_pert_theory}) we first show the failure of the {\it naive perturbation theory} in computing the {\it specific free energy}, due to two different problems:
\begin{itemize}
\item the sum over all the perturbative orders diverges because of a {\it too big number} of Feynman diagrams involved in the expansion,
\item the infinite volume limit does not exists, since the rough bounds we obtain by naive perturbative estimates are not uniform in the cut-offs.
\end{itemize}
In Subsection (\ref{subsection_determinant_expansion}) we solve one of the two problems, the combinatorial one, by showing the so called {\it determinant expansion}. To solve the other problem it is necessary a multiscale analysis
\item in Section (\ref{section_multiscale_analysis}) we show the {\it multiscale analysis} of the theory, stressing in particular its hierarchical structure that allows us to represent the observables we are interested in in terms of the so called {\it Gallavotti-Nicolò} trees. Besides, we use the multiscale expansion to identify, in RG language, the {\it sources of the divergences}.
\item in Section (\ref{subsection_renormalization_group_PBC}) we explain how to prove, using RG methods, that we can express the specific free energy as a convergent series in the size of the interaction, if $\lambda$ is small enough.
\end{itemize}
\item Chapter (\ref{chapter_Interacting_fermions_on_the_half_line}) contains the new results of this thesis, in particular we prove the main Theorem (\ref{theorem_main_introduction}):
\begin{itemize}
\item in Section (\ref{the_model_DBB}) we present the model and we recall the main result,
\item in Section (\ref{section_the_interacting_case_DBC}) we perform a multiscale expansion of the thermodynamic observables of the system,
\item in Section (\ref{section_Non-renormalized expansion and properties of kernels}) we identify the source of the divergences by a non-renormalized analysis, and we extract the bulk contributions from the quadratic and the quartic terms of the effective potential,
\item in Section (\ref{section_renormalization_group_DBC}), in order to prove the main theorem, we show in a series of technical Lemmata how the presence of non-translation invariant elements improves the dimensional bound of the kernels; finally, we prove the main theorem.
\end{itemize}
\item in Chapter (\ref{chapter_conlcusion}) we draw the conclusions of this thesis:
\begin{itemize}
\item in Section (\ref{section_summary}) we summarize the result of this work, and we comment some possible and simple improvement of the bounds that can easily be reached,
\item in Section (\ref{section_outlook}) we present some very general ideas we would like to explore in more detail in order to approach the main goal of studying the theory without boundary counterterms.
\end{itemize}
\end{itemize}

\chapter{Interacting fermions on the line}
\label{chapter_fermions_PBC}

In this chapter, the main goal is to introduce the reader to the study,  via rigorous constructive Renormalization Group techniques, of one dimensional interacting Fermi systems. It is important to stress that nothing new will be shown (we will present in detail the new result in the following chapter) but, especially for a reader not familiar with RG, it will be explained how to {\it construct the ground state} of a model describing spinless fermions living on a one dimensional lattice, where the only perturbation to the free {\it hopping} Hamiltonian is a {\it weak} density-density interaction. 
For a more detailed review of RG applied to 1D fermionic systems, we refer to \cite{gentile2001renormalization}. In this chapter we will give a self-consistent presentation of the main ideas of the construction of 1D fermions in the translationally invariant case, since this will serve as reference theory for the construction of the theory on the half-space, discussed in Chapter 3.\\
Before starting, it worths doing two comments on how the technical assumption we will do reflect on the physics we are interested in:
\begin{itemize}
\item {\bf Fermions on a lattice} Thanks to this assumptions we have a natural ultraviolet cut-off (which is the mesh size of the lattice), by which we get rid of the ultraviolet divergences. Physically, assuming that the electrons can move only on a lattice corresponds to thinking the electrons as localized on atomic sites of a crystal, and by the hopping Hamiltonian we let them move to the nearest neighbor atoms. 
\item {\bf Periodic boundary conditions} As we already mentioned in the previous introducting chapter, after some decades of impressive work, nowadays the theory of RG is well developed, and a lot of important and fundamental results have been proven under the assumption of {\it translation invariance}. On the one hand, it is true that a lot of technical simplifications come from this assumption (as it will be clear by comparing this chapter with the next one), but on the other hand it is important to underline that this assumption is quite satisfactory as long as one is interested in the bulk properties of the model (which, in the case of condensed matter, translates into asking what happens very far from the boundaries of the crystal we have in the lab, driven by the idea that, being the {\it size} of particles much smaller than the distance from the boundary, a model without boundaries is a good model for the bulk behavior for system).
\end{itemize}
In the following we introduce all the necessary technical and theoretical tools, whose definition will be extended in the following to the case of Dirichlet boundary conditions. 

\section{The model}
\label{section_PBC_the_model}
\subsection{Definition and main result}

\paragraph{The Hamiltonian}

We are interested in constructing the ground state of interacting spinless fermions living in a discrete one-dimensional box of step $a=1$ and size $L\gg 1$. In particular, we perturb by a {\it weak} density- density interaction an integrable Hamiltonian describing non interacting fermions hopping to the nearest neighbouring sites in a box $\Lambda$ with periodic boundary condition (PBC), imposed by identifying the two extremal sites.\\
Let $\mathcal{F}=\oplus_{n=0}^{\infty}H^{\wedge n}$ be the standard antisymmetric fermion Fock space, where $\wedge$ denotes the antisymmetric tensor product and let $\psi_{x}^{\pm}$ be the {\it fermionic creation or annihilation} operators defined on $\mathcal{F}$, where $x$ is the spatial coordinate. Let us consider the discrete box $\Lambda:=\{x\in\mathbb{Z}: -\lfloor L/2\rfloor \leq x\leq \lfloor (L-1)/2 \rfloor\}$, and the gran-canonical Hamiltonian

\begin{equation}
H=H_0+\lambda V,
\label{hamiltonian_PBC}
\end{equation}
where
\begin{equation}
\begin{split}
H_0&=T_0-\mu_0N_0,\\
T_0&=\sum_{x\in\Lambda}\frac{1}{2}\left(-\psi_x^+\psi_{x+1}^- -\psi_x^+\psi_{x-1}^-+2\psi_x^+\psi_x^-\right),\\
N_0&=\sum_{x\in\Lambda}\psi_x^+\psi_x^-,
\end{split}
\label{free_hamiltonian_PBC}
\end{equation}
where $\mu_0$ is the chemical potential, choosen in such a way that, if we call $\sigma(T_0):=[e_-, e_+]$ the spectral band of the kinetic operator, $\mu_0\in [e_-+\kappa, e_+-\kappa]$ for some $\kappa>0$ fixed once for all;
\begin{equation}
V=\sum_{x,y\in\Lambda}\psi_x^+\psi_x^- v(x-y) \psi^+_y\psi_y^-,
\label{interaction_PBC}
\end{equation}
$v(x-y)$ is a {\it compactly supported} function: $V$ is a so-called {\it density-density interaction}, being $\psi^+_x\psi^-_x=:n_x$ the density operator in $x$.

\paragraph{Specific free energy, Schwinger functions and the main theorem}

The main goal of this chapter is to compute the {\it specific free energy}, defined as
\begin{equation}
f_{\Lambda,\beta}:=-\frac{1}{\beta |\Lambda|}\log\left(Tr \left(e^{-\beta H}\right)\right)
\label{free_energy_specific_PBC}
\end{equation}
where $\beta$ is the inverse temperature (so in order to construct the ground state energy we are interested in the {\it zero temperature limit} $\beta\to \infty$). We are also interested in the {\it finite temperature imaginary time correlation functions, or {\it Schwinger functions,}} at finite temperature $T=\beta^{-1}$, defined as
\begin{equation}
S_{\Lambda,\beta}(\bm x_1, \epsilon_1;\dots;\bm x_m, \epsilon_m):=\left< \psi^{\epsilon_1}(\bm x_1)\dots \psi^{\epsilon_n}(\bm x_m)\right>_{\Lambda, \beta}:=\frac{Tr\left( e^{-\beta H}\bm T \left( \psi^{\epsilon_1}(\bm x_1)\dots \psi^{\epsilon_m}(\bm x_m)\right)\right)}{Tr\left( e^{-\beta H}\right)}
\label{schwinger_function_n_points_PBC}
\end{equation}
where $\epsilon_i\in \{\pm\}$ for $i=1,\dots, m$ and $\bm T$ is the {\it Fermionic time ordering operator}, and where we have introduced a collection $\left\{t_1,\dots,t_m\right\}$ of {\it time variables} such that $t_i\in \left[0,\beta\right)$ $\forall i=1,\dots,m$.\\
The main strategy to compute these quantities will be to derive {\it convergent expansions} for both $f_{\Lambda,\beta}$ and $S$, uniformly in the volume $|\Lambda|$ and in the inverse temperature $\beta$, and then to take the {\it infinite volume} and the {\it zero temperature} limits: $|\Lambda|\to \infty$ first,then $\beta\to \infty$.\\
We will describe in detail how to compute the {\it specific ground state energy}, in particular how to prove the following theorem.
\begin{thm}
\label{theorem_free_energy_analyticity_PBC}
In this framework, there exists a radius $\lambda_0>0$ such that for each $\lambda\leq |\lambda_0|$ the specific ground state energy 
\begin{equation}
f:=\lim_{\beta\nearrow \infty}\lim_{|\Lambda|\nearrow \infty}\left[-\frac{1}{|\Lambda| \beta}\log\left( Tr \left({e^{-\beta H}}\right)\right)\right],
\end{equation}
exists uniformly in $|\Lambda|$ and $\beta$, an it is an analyitic function of $\lambda$.
\end{thm}
\begin{rem}
A modification of the expresion behind the proof of Theorem (\ref{theorem_free_energy_analyticity_PBC}) allows one to compute the Schwinger functions, see \cite{gentile2001renormalization}, Section 12.
\end{rem}
\subsection{Free Hamiltonian diagonalization and free propagator}
\label{subsection_free_propagator}

It is straightforward to check that the {\it free Hamiltonian} $H_0$ can be diagonalized in Fourier space by defining
\begin{equation}
\hat \psi^{\pm}_k = \sum_{x\in\Lambda}e^{\mp ik x}\psi^{\pm}_x,
\label{fourier_transform_creation_annihilation_PBC}
\end{equation}
where $k\in\mathcal D_\Lambda$, 
\begin{equation}
\mathcal{D}_\Lambda=\left\{k=2\pi n/L, n\in\mathbb{Z}, -[L/2]\leq n \leq [(L-1)/2] \right\}.
\label{dual_space_PBC}
\end{equation}
and the operator $\hat \psi^+_k/\hat \psi^-_k$ creates/annihilates a spinless electron with momentum k, so that 
\begin{equation}
H_0=\frac{1}{|\Lambda|}\sum_{k\in\mathcal D_\Lambda}\hat\psi^+_k e(k)\hat \psi^-_k,
\label{H_0_PBC_diagonal}
\end{equation}
where $e(k)=1-\cos k-\mu_0$ is called the {\it dispersion relation} defined in $\mathcal D_\Lambda$.\\
It worths noting that when $L\to \infty$, $\mathcal{D}_L\to [-\pi, \pi]$, so in the infinite volume limit there are two points, let us call them $\pm p_F$, such that $e(\pm p_F)=0$ since $\mu_0\in [e_-+\kappa,e_+-\kappa]$.

\paragraph{Free propagator}
The non interacting model, {\it i.e.} the model described by the free Hamiltonian $H_0$, is exactly solvable and all the Schwinger functions can be computed, by the anticommutative {\it Wick rule}, starting from the {\it two point Schwinger function}, also known as {\bf propagator}, that can be explicitly computed starting from the definition, we refer to \cite{benfatto1995renormalization}. \\
Let us recall that $\psi_x^{\pm}=\frac{1}{|\Lambda|}\sum_{k\in\mathcal D_{\Lambda}} e^{\pm ikx}\hat \psi_k^{\pm}$ for any $x\in\Lambda$, and if we call $\bm x=(x,x_0)$, $\bm y=(y,y_0)$, $\bm k=(k,k_0)$ the evolution in time of the operators is $\psi^{\pm}_{\bm x}=e^{H_0 x_0}\psi^{\pm}e^{-H_0x_0}$. Recalling that $\left<\cdot\right>_{\Lambda,\beta, 0}=Tr\left(e^{-\beta H_0}\cdot\right)/Tr(e^{-\beta H_0})$, we can compute, for any $-\beta < x_0-y_0 \leq \beta$,
\begin{equation}
\begin{split}
\left<\bm T\left(\psi^-_{\bm x}\psi^+_{\bm y}\right)\right>_{\Lambda,\beta, 0}=\frac{1}{|\Lambda|}\sum_{k\in\mathcal D_{\Lambda}}e^{-ik(x-y)}\cdot\\
\cdot \left[\theta(x_0-y_0)\frac{e^{-(x_0-y_0)e(k)}}{1+e^{-\beta e(k)}}-\theta(y_0-x_0)\frac{e^{-(x_0-y_0+\beta)e(k)}}{1+e^{-\beta e(k)}}\right]
\end{split}
\end{equation}
there $\theta(\cdot )$ is the Heaviside step function. The latter formula is a priori defined only for $-\beta < x_0-y_0\leq \beta$, but we can extend it periodically over the whole real axis: the periodic extension is continuous in $x_0-y_0\notin \beta \mathbb Z$, while it has jump discontinuities at $x_0-y_0\in\beta \mathbb Z$ (the jump height is equal to $(-1)^n\delta_{x,y}$ if $x_0-y_0=\beta n$), so if we define 
\begin{equation}
\mathcal{D}_{\beta,M}:=\left\{k_0:=\frac{2(n+1/2)\pi}{\beta}, n\in\mathbb{Z}, -M\leq n \leq M-1\right\},
\label{momenta_space_time}
\end{equation}
we get
\begin{equation}
S^0_{L,\beta}(\bm x,-;\bm y,+):= g(\bm x-\bm y)=\frac{1}{\beta L}\lim_{M\to \infty}\sum_{k \in\mathcal{D}_{\Lambda}}\sum_{k_0\in\mathcal{ D}_{\beta,M}}e^{i\delta_Mk_0}e^{-i\bm k\cdot (\bm x-\bm y)}\hat g(\bm k),
\label{free_propagator_PBC}
\end{equation}
where $D_\Lambda$ has been defined in (\ref{dual_space_PBC}), while $M$ is a suitable cut-off to be removed at the very end (of course the scheme will be to get bound independent of the cut-off $M$, and finally to take the limit $M$ to infinity), and
\begin{equation}
\hat{g}(\bm k):=\frac{1}{-ik_0+e(k)},
\label{free_propagator_momenta}
\end{equation}
where $e(k)$ is the {\it dispersion relation} already defined in (\ref{H_0_PBC_diagonal}). From now on, we will use $\bm k \in \mathcal{D}_{\Lambda,\beta,M}$ to denote $(k,k_0)\in\mathcal{D}_{\Lambda}\times\mathcal{D}_{\beta,M}$.
The constant $\delta_M=\beta/\sqrt{M}$ is introduced in order to take correctly into account the discontinuity of the propagator $g(\bm x-\bm y)$ at $\bm x=\bm y$, where it has to be defined as $\lim_{x_0\to 0^-}g(0,x_0)$, in fact the latter definition guarantees that $\lim_{M\to \infty}g_M(\bm x-\bm y):=g(\bm x-\bm y)$ for $\bm x\neq\bm y$, while $\lim_{M\to \infty}g_M(\bm 0):=g(0,0^-)$ at equal points. 

\section{Perturbation theory and Grassmann integral formulation}		
\label{section_1_pert_theory}

\subsection{Perturbation theory and Trotter's formula}

Let us now consider the interacting case. Our strategy is to derive first a formal perturbation theory for the specific free energy, and properly to find rules to {\it formally compute the generic perturbative order in} $\lambda$ of $f_{\Lambda,\beta}$. Then we will explain how to give sense to this formal expression, by suitable resummations of the formal power series. It is worth stressing that the interaction could in principle move, in some {\it interaction dependent way}, the Fermi points of the theory. To take into account this fact, we rewrite
\begin{equation}
\mu_0=\mu+\nu,
\end{equation}
where $\nu$ is a {\it counterterm} that will be eventually suitably chosen in order to fix the position of the singularity to some {\it interaction independent} point.\\
So we rewrite 
$$H=H_0+U,$$
where
\begin{equation}
U=\lambda V+\nu N_0=\lambda \sum_{x,y\in\Lambda}\psi^+_x\psi^-_xv(x-y)\psi^+_y\psi^-_y+ \nu\sum_{x\in\Lambda}\psi^+_x\psi^-_x,
\end{equation}
and we use the Trotter product formula
\begin{equation}
e^{-\beta H}=\lim_{n\to \infty}\left[e^{-\beta H_0/n}\left(1-\frac{\beta}{n}U\beta\right)\right]^n,
\end{equation}
so that, if we define
\begin{equation}
U(t):=e^{tH_0}Ve^{-tH_0},
\end{equation}
we get
\begin{equation}
\begin{split}
\frac{Tr\left(e^{-\beta H}\right)}{Tr\left(e^{-\beta H_0}\right)}=\\
=1+\sum_{(-1)^N}\int_0^\beta dt_1 \int_0^{t_1} dt_2\dots\int_0^{t_N-1}dt_N \frac{Tr\left(e^{-\beta H_0}U(t_1)\dots U(t_N)\right)}{Tr\left(e^{-\beta H_0}\right)},
\end{split}
\end{equation}
which, using again the {\it fermionic time-ordering operator}, can be rewritten as 
\begin{equation}
\frac{Tr\left(e^{-\beta H}\right)}{Tr\left(e^{-\beta H_0}\right)}=1+\sum_{N\geq 1}\frac{(-1)^N}{N!}\left<\bm T \left(U(\psi)^N\right)\right>_{\Lambda, \beta, 0}
\label{Tr(cdot)/Tr}
\end{equation}
where $\left<\cdot\right>_{\Lambda, \beta, 0}=Tr\left(e^{-\beta H_0 \cdot}\right)/Tr\left(e^{-\beta H_0}\right)$ and we have defined
\begin{equation}
\begin{split}
U(\psi)=\lambda\int_{[0,\beta)}dx_0 \sum_{x\in\Lambda}\int_{[0,\beta)}dy_0 \sum_{y\in\Lambda} \psi^+_{\bm x}\psi^-_{\bm x}v(x,y)\delta_{x_0,y_0}\psi^+_{\bm y} \psi^-_{\bm y}+\nu\int_{[0,\beta)}dx_0\sum_{x\in\Lambda}\psi^+_{\bm x}\psi^-_{\bm x}.
\end{split}
\end{equation}
The $N$-th order of formula (\ref{Tr(cdot)/Tr}) can be computed using the Wick rule
\begin{equation}
\begin{split}
&\left<\bm T\left(\psi^-_{\bm x_1}\dots \psi^+_{\bm x_n}\right)\right>_{0,\Lambda,\beta}=det G,\\
&G_{ij}=\left<\bm T\left(\psi^-_{\bm x_i}\psi^+_{\bm x_j}\right)\right>_{0,\Lambda,\beta}=S_{L,\beta}^0(\bm x,-;\bm y,+).
\end{split}
\end{equation}
and the explicit {\it free propagator} (\ref{free_propagator_PBC}), where the subscript $0$ denotes that the expectations are computed with respect to the free measure. In order to use the Wick rule, it is conveniente to briefly recall the Feynman rules.

\subparagraph{Feynman rules} 
\begin{figure}
\begin{center}
 \begin{tikzpicture} 
 [thick,decoration={
    markings,
    mark=at position 0.5 with {\arrow{>}}}] 
\node  at (1,3.3) {{\bf x}};
\node at (3,3.3) {{\bf y}};
\fill (1,3) circle (0.06);
\fill (3,3) circle (0.06);
\draw [postaction={decorate}] (0,2) -- ++(1,1);
\draw [postaction={decorate}](0,2) ++ (1,1)-- ++ (-1,1);
\draw [postaction={decorate}] (0,2) ++ (1,1)++ ( 2,0) ++(1,1) ++ (-1,-1) --++ (1,-1);
\draw [postaction={decorate}] (0,2) ++ (1,1)++ ( 2,0) ++(1,1) ++ (-1,-1) ++ (1,-1) ++ (-1,1) -- ++(1,1);
\draw [-,decorate,decoration=snake] (0,2) ++ (1,1) ++ (-1,1)++ (1,-1) -- ++(2,0);
\node at (-4,3.3) {\bf x};
\fill (-4,3) circle (0.1);
\draw [postaction={decorate}] (-5,3) -- ++ (1,0);
\draw [postaction={decorate}] (-4,3) -- ++ (1,0);
\end{tikzpicture}
\end{center}
\caption{Graph elements associated with $\nu$-type endpoints (left) and $\lambda$-type endpoints (right).}
\label{figure_graph_elements_PBC}
\end{figure}
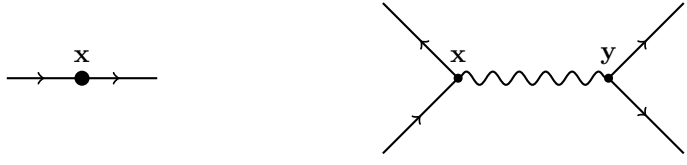

In order to compute $\left<\bm T\left(U_\beta(\psi))^N\right)\right>_{\Lambda, \beta, 0}$, it is easy to check that one can follow these steps:
\begin{itemize}
\item $\forall k,l$ such that $0 \leq k,l\leq N $ and $k+l=N$, draw $k$ graph elements consisting of {\it four legged vertices}, $l$ graph elements consisting of {\it two legged local vertices} with the vertices associated to labels $\bm x_i$, $i=1,\dots,N$, in such a way that the {\it four legged vertices} are composed by two entering and to exiting fields, while the {\it two legged vertices} are associated with one exiting and one entering leg (see Figure (\ref{figure_graph_elements_PBC}));
\item pair the fields in all possible ways, in such a way that every pair is obtained by contracting an entering and an exiting leg;
\item associate to every pairing the {\it right sign}, which is the sign of the permutation needed to bring every pair of contracted fields next to each other;
\item associate to every linked pair of fields $\left(\psi^-(\bm x_i),\psi^+(\bm x_j)\right)$ an {\it oriented} line connecting the $i-$th with the $j-$th vertex, oriented from $j$ to $i$ ({\it i.e.} from $+$ to $-$ field);
\item associate to every oriented line from $j$ to $i$ value $g(\bm x_i,\bm x_j)$ given by (\ref{free_propagator_PBC});
\item associate to every configuration of pairings, which is called {\it Feynman graph} a value, equal to the product of the sign of the pairing, times $\lambda^k\nu^l$ times the product of the values of all the oriented lines;
\item integrate over $\bm x_i$, then perform the sum over all the possible pairings, over $k, l$ and over N;
\end{itemize}
It is convenient, algebraically, to rewrite the quantities (\ref{free_energy_specific_PBC}) and (\ref{schwinger_function_n_points_PBC}) in terms of {\it Grassmann Gaussian integrals}. Even though the theory of {\it Grassmann integrals} is a very well known topic in the literature (see again, for instance, \cite{gentile2001renormalization}), in sake of self consistency we will sketch the main definitions and properties.

\subparagraph{Grassmann algebra}

Given some finite set $A$ of indices $\alpha\in A$, we define a {\it finite dimensional Grassmann algebra}, generated by a set of {\it anticommuting Grassmann variables} $\left\{\psi_{\alpha}^{\pm}\right\}_{\alpha\in A}$: we attach at each element $\alpha\in A$ a couple of variables $\psi \equiv \left\{\psi^+_{\alpha}, \psi^-_{\alpha}\right\}$ such that
\begin{equation}
\psi^{\epsilon}_{\alpha}\psi^{\epsilon'}_{\alpha'}+ \psi^{\epsilon'}_{\alpha'}\psi^{\epsilon}_{\alpha}=0 \hspace{3mm} \forall \alpha,\alpha'\in A, \forall \epsilon, \epsilon'\in\{\pm\}.
\label{anticommutation_rules_grassmann}
\end{equation}

\begin{rem}
In particular, $\forall \alpha\in A, \forall \epsilon\in\{\pm\}$ we have $\left(\psi^{\epsilon}_{\alpha}\right)^2=0.$
\label{grassmann_variables_squared_remark}
\end{rem}

\subparagraph{Grassmann integral operator} 

Let us introduce the {\it Grassmann integral operator} $\int  d\psi^{\epsilon}_{\alpha}\cdot$ acting as:
\begin{equation}
\int  d\psi^{\epsilon}_{\alpha}\psi^{\epsilon}_{\alpha}=1, \hspace{3mm} \int d\psi^{\epsilon}_{\alpha}=0
\label{grassmann_integral}
\end{equation}

A straightforward generalization in the case of many {\it Grassmann variables integral} can be obtained by iterating (\ref{grassmann_integral}):
\begin{equation}
\int \prod_{\alpha\in B}d\psi^+_\alpha d\psi^-_{\alpha}\left(\prod_{\alpha\in B}\psi^-_\alpha \psi^+_{\alpha}\right)=1,\hspace{3mm} \forall B\subset A.
\end{equation} 
so that if $F(\psi)$ is a polynomial in $\psi_\alpha^+, \psi_\alpha^-, \alpha\in A$, the operation
\begin{equation}
\int \prod_{\alpha\in A}d\psi^+_\alpha d\psi^-_\alpha F(\psi)
\end{equation}
extracts the coefficient of the linear term in $\left(\prod_{\alpha=1}^N \psi^-_\alpha \psi^+_\alpha \right)$.
\\ Using the remark (\ref{grassmann_variables_squared_remark}) and the usual Taylor series for the exponential, $e^{-\psi^+_{\alpha}C\psi^-_{\alpha}}=1-\psi^+_{\alpha}C\psi^-_{\alpha}$, so by the definition (\ref{grassmann_integral})

\begin{equation}
\frac{\int d\psi^+_{\alpha}d\psi^-_{\alpha}e^{-\psi^+_{\alpha}C\psi^-_{\alpha}}\psi^-_{\alpha}\psi^+_{\alpha}}{\int d\psi^+_{\alpha}d\psi^-_{\alpha}e^{-\psi^+_{\alpha}C\psi^-_{\alpha}}}=C^{-1}, \hspace{3mm} \forall \alpha\in A, C \in \mathbb{C}
\end{equation}

To generalize this formula in the case of $2N$ Grassmann variables, we introduce the matrix $M\in GL(N,\mathbb C)$,

\begin{equation}
\begin{split}
\int \prod_{\alpha=1}^N \left(d\psi^+_{\alpha}d\psi^-_{\alpha}\right) e^{-\sum_{\alpha,\alpha'=1}^N\psi^+_{\alpha}M_{\alpha,\alpha'}\psi^-_{\alpha'}}=\det M, \\
\int \prod_{\alpha=1}^N \left(d\psi^+_{\alpha}d\psi^-_{\alpha}\right) e^{-\sum_{\alpha,\alpha'=1}^N\psi^+_{\alpha}M_{\alpha,\alpha'}\psi^-_{\alpha'}}\psi^-_{\bar{\alpha}}\psi^+_{\tilde{\alpha}}= \bar M_{\tilde{\alpha},\bar{\alpha}}
\end{split}
\end{equation}

where $\bar M_{\bar{\alpha},\tilde{\alpha}}$ is the minor complementary to the entry $M_{\bar{\alpha},\tilde{\alpha}}$ and, if $M$ is invertible,

\begin{equation}
\frac{\int \prod_{\alpha=1}^N \left(d\psi^+_{\alpha}d\psi^-_{\alpha}\right) e^{-\sum_{\alpha,\alpha'=1}^N\psi^+_{\alpha}M_{\alpha,\alpha'}\psi^-_{\alpha'}}\psi^-_{\tilde \alpha}\psi^+_{\bar \alpha}}{\int \prod_{\alpha=1}^N \left(d\psi^+_{\bar \alpha}d\psi^-_{\tilde \alpha}\right) e^{-\sum_{\alpha,\alpha'=1}^N\psi^+_{\alpha}M_{\alpha,\alpha'}\psi^-_{\alpha'}}}=\left[M^{-1}\right]_{\bar{\alpha},\tilde{\alpha}}
\end{equation}

\begin{rem}
These properties are similar to the ones of the usual {\it Gaussian integrals}, without the constraint on $C$ to be real and on $M$ to be positive definite, but only invertible.
\label{grassmann_integrals_gaussian_remark}
\end{rem}

\paragraph{Grassmann Gaussian integration}

Inspired by the remark (\ref{grassmann_integrals_gaussian_remark}), we can build up a {\it Grassmann Gaussian integration} $P(d\psi)$ associated with the propagator $g(\bm x- \bm y)$ in order to express the specific free energy (\ref{free_energy_specific_PBC}) and the Schwinger functions (\ref{schwinger_function_n_points_PBC}) as Gaussian Grassmann integrals.
First of all, let us introduce a finite set of {\it Grassmann variables} $\{\hat \psi^{\pm}(\bm k)\}_{\bm k\in\mathcal{D}_{\Lambda,\beta,M}}$; hence we define

\begin{equation}
P(d\psi)=\left(\prod_{\bm k\in\mathcal{D}_{\Lambda,\beta,M}}\left(L\beta\hat g(\bm k)\right)\hat \psi^+(\bm k)\hat \psi^-(\bm k)\right)e^{-\sum_{\bm k \in \mathcal{D}_{\Lambda,\beta,M}}\left(L\beta\hat g(\bm k)\right)^{-1}\hat \psi^+(\bm k)\hat \psi^-(\bm k)}.
\label{grassmann_gaussian_measure_k_space_PBC}
\end{equation}
By introducing the Fourier transforms:
\begin{equation}
\psi^{\pm}(\bm x)=\frac{1}{L\beta}\sum_{\bm k\in\mathcal{D}_{L,\beta,M}}\hat \psi^{\pm}(\bm k)e^{\pm i \bm k\cdot \bm x},
\label{grassman_variables_fourier_transform}
\end{equation}
we can use the measure (\ref{grassmann_gaussian_measure_k_space_PBC}) to get
\begin{equation}
\lim_{M\to\infty}\int P(d\psi)\psi^-(\bm x)\psi^+(\bm y)=\frac{1}{L\beta}\lim_{M\to\infty}\sum_{\bm k\in \mathcal{D}_{\Lambda,\beta,M}}\hat g(\bm k)e^{-i\bm k\cdot (\bm x-\bm y)}=g(\bm x-\bm y),
\label{grassmann_gaussian_measure_x_space_PBC}
\end{equation}

where we denoted with $P$ the {\it Grassmann Gaussian integration} associated to the propagator $g$ in (\ref{free_propagator_PBC}).

\subparagraph{Expectation functional}
Calling $P(d\psi)$ {\it Gaussian fermionic integration} with covariance $g$ we mean that, for any analytic function $F$ defined on the Grassmann algebra, we can define an {\it expectation functional}

\begin{equation}
\int P(d\psi)F(d\psi)=\mathcal{E}(F).
\label{expectation}
\end{equation}

\begin{rem}
$P(d \psi)$ is not a measure in the usual sense, indeed it does not satisfy the positivity condition, so we use the terminology of expectation $\mathcal E$ by analogy. 
\end{rem}

\subparagraph{Truncated expectation functions} Given $p$ functions $X_1,\dots, X_p$ defined on the Grassmann algebra and $p$ integer numbers $n_1,\dots,n_p,$ the {\it truncated expectation} is defined as
\begin{equation}
\mathcal{E}^T\left(X_1,\dots,X_p;n_1,\dots,n_p\right)= \frac{\partial^{n_1+\dots+n_p}}{\partial_{\lambda_1}^{n_1}\dots \partial_{\lambda_p}^{n_p}} \left .\log \int P(d\psi)e^{\lambda_1X_1(\psi)+\dots+ \lambda_pX_p(\psi)}\right|_{\lambda=0}.
\label{expectation_truncated}
\end{equation}
where $\lambda=\left(\lambda_1,\dots,\lambda_p\right)$; we will use the notation
\begin{equation}
\mathcal E^T(X_1,\dots, X_p):=\mathcal E^T\left(X_1,\dots, X_p;\underbrace{1,\dots,1}_{p \mbox{ times }}\right),
\end{equation}
\\
In particular,
\begin{equation}
\mathcal{E}^T\left(X;n\right)=\frac{\partial^n}{\partial^n\lambda}\log \left . \int P(d\psi)e^{\lambda X(\psi)}\right |_{\lambda=0},
\end{equation}
and
\begin{equation}
\log \int P(d\psi)e^{X(\psi)}=\sum_{n=0}^{\infty} \frac{1}{n!} \frac{\partial^n}{\partial \lambda^n}\log \left .\int P(d\psi)e^{\lambda X(\psi)}\right|_{\lambda=0}=\sum_{n=0}^{\infty}\frac{1}{n!}\mathcal{E}^T\left(X;n\right).
\label{log_grassmann_integral_free_energy}
\end{equation}

\paragraph{Properties of Grassmann integrals and expectation functions}

\begin{itemize}
\item {\bf Wick rule} Given two sets of labels $\left\{\alpha_1,\dots,\alpha_n\right\},\left\{\beta_1,\dots,\beta_m \right\}\subset A$, so
\begin{equation}
\int P(d\psi) \psi_{\alpha_1}^-\dots\psi_{\alpha_n}^- \psi_{\beta_1}^+\dots\psi_{\beta_m}^+=\delta_{n,m} \sum_{\Gamma}\sum_{\pi}(-1)^{p_{\pi}}\prod_{\Gamma \ni \ell=(\bm x_i,\bm x_{\pi(j)})}g(\ell).
\label{wick_rule}
\end{equation}
where the sum over $\Gamma$ is the sum over all the possible pairings (or Feynman graph configurations) and the product over $\ell$ is the product over all the possible contractions compatible with the configuration $\Gamma$. 
\item {\bf Addition principle} Given two Grassmann measures $P(d\psi_1)$ with covariance $g_1$ and $P(d\psi_2)$ with covariance $g_2$, for any analytic function $F(\psi)$ defined on the Grassmann algebra and such that $\psi=\psi_1+\psi_2$, so
\begin{equation}
\int P(d\psi_1)\int P(d\psi_2)F\left(\psi_1+\psi_2\right)= \int P(d\psi)F(\psi),
\label{addition_principle}
\end{equation}
with $P(d\psi)$ associated to a covariance $g=g_1+g_2$.
\item {\bf Invariance of exponentials} Using the definition of truncated expectation (\ref{expectation_truncated}) it follows that, if $\phi$ is an external field (meaning that $\phi$ is not involved in the integration process),
\begin{equation}
\int P(d\psi)e^{X(\psi+\phi)}=\exp \left[\sum_{n=0}^{\infty}\frac{1}{n!}\mathcal{E}^T\left(X\left(\cdot+\phi\right);n\right) \right]=:e^{X'(\phi)}.
\label{invariance_of_exponential}
\end{equation}

\item {\bf Change of integration measure} Let $P_g(d\psi)$ be the integration measure with covariance $g$. Then, for any analytic function defined on the Grassmann algebra $F(d\psi)$, it holds
\begin{equation}
\frac{1}{N_{\nu}}\int P_g(d\psi)e^{-\nu\psi^+\psi^-}F(\psi)= \int P_{\tilde g}(d\psi)F(\psi),
\label{change_of_integration_measure_property}
\end{equation}

where $\tilde{g}^{-1}=g^{-1}+\nu$ and $N_{\nu}=\frac{g^{-1}+\nu}{g^{-1}}=1+g\nu=\int P_g(d\psi) e^{-\nu\psi^+\psi^-}$.

\end{itemize}

\paragraph{Free energy}

Using these definitions and the Feynman rules described above, we can rewrite equation (\ref{Tr(cdot)/Tr}) as 
\begin{equation}
\label{Tr/Tr_as_grassmann_integral}
\frac{Tr\left(e^{-\beta H}\right)}{Tr\left(e^{-\beta H_0}\right)}=\lim_{M\to\infty}\int P_M(d\psi)e^{-\mathcal V(\psi)},
\end{equation}
where 
\begin{equation}
\mathcal{V}=\lambda\int_0^{\beta} dx_0\int_0^{\beta}dy_0\sum_{x,y\in\Lambda}\psi_{\bm x}^+\psi_{\bm x}^-v(x-y)\delta_{x_0,y_0}\psi_{\bm y}^+\psi_{\bm y}^-+\nu \int_{[0,\beta)}dx_0 \sum_{x\in\Lambda}\psi^+_{\bm x}\psi^-_{\bm x},
\label{interaction_grassmann}
\end{equation}
and $e^{-\mathcal V(\psi)}$ must be identified with its Taylor series in $\lambda$ and $\nu$, which is finite for every finite $M$ due to the anticommutation rules of the Grassmann variables and the fact that the Grassmann algebra is finite for every finite M. A priori, equation (\ref{Tr/Tr_as_grassmann_integral}) has to be read as an equality between formal power series in $\lambda$ and $\nu$, however, it can be given a {\it non-perturbative meaning} provided we can prove the convergence of the Grassmann integral in the r.h.s. under analiticity assumption in a complex disc. \\
Using (\ref{Tr/Tr_as_grassmann_integral}), we can compute the specific free energy (\ref{free_energy_specific_PBC}) provided we are able to check that the r.h.s. of (\ref{Tr/Tr_as_grassmann_integral}) is analytic in a domain that is uniform in $M,\beta,\Lambda$, and that it converges to a well defined analyric function uniformly as $M\to\infty$; in fact, this will be the main goal of this chapter. Let us start by rewriting the specific free energy as:
\begin{equation}
f_{\Lambda, \beta}:=-\frac{1}{\beta L}\sum_{N\geq 1}\frac{(-1)^N}{N!}\mathcal{E}^T\left(\mathcal{V};N\right),
\label{free_energy_as_sum_of_trunc_expec}
\end{equation}
where the {\it expectations functionals} have been already defined, and now we will discuss how to compute them. We underline that we slightly abused of the notation, indeed the function $f_{\Lambda,\beta}$ just defined is actually the {\it difference between the specific free energy of the interacting system and the specific free energy of the free system $f_{0,\Lambda,\beta}=-1/|\Lambda|\beta \log Tr(e^{-\beta H_0})$.}
\subsection{How to compute truncated expectations}
\label{subsection_How_to_compute_truncated_expectations}
\paragraph{Feynman graphs}
We have already described the most immediate way to compute truncated expectation functions when we listed the {\it Feynman rules} to compute the expectations values in (\ref{free_energy_as_sum_of_trunc_expec}), getting the result we recall here.\\
Given $s$ sets of indices $P_1,\dots,P_s$, we define for each of those
\begin{equation}
\tilde{\psi}\left(P_i\right)=\prod_{f\in P_i}\psi^{\sigma(f)}_{\bm x(f)},
\end{equation}
where $\sigma(f)\in\{\pm\}$ and $\bm x(f)\in \Lambda\times\left[0,\beta\right)$. Then,
\begin{equation}
\mathcal{E}\left(\tilde{\psi}\left(P_1\right),\dots,\tilde{\psi}\left(P_s\right)\right)=\sum_{\Gamma\in\mathcal{G}_0}Val(\Gamma),
\label{expectation_truncated_s_sets}
\end{equation}
where $\Gamma$ is a Feynman graph belonging to the family of all possible Feynman graphs $\mathcal{G}_0$, and $Val(\Gamma)$ includes the integration over the space-time labels $\bm x_i$: for instance let $\Gamma\in\mathcal G_{0,N}$, where $\mathcal G_{0,N}$ is the family of all possible Feynman graphs of order $N$,
\begin{equation}
Val(\Gamma)=\sum_{1\leq k+l\leq N}\nu^k\lambda^l \int d{\bm x_1}\dots d\bm x_n(-1)^{p_\pi}\prod_{\ell\in\Gamma}g_{\ell}
\label{value_of_a_feynman_graph}
\end{equation}
where, as explained in the list of the rules, $p_{\pi}$ is the parity of the permutation, and $\ell\in\Gamma$ is the set of all the lines belonging to the Feynman graph. As we already commented in the general discussion of expectations,
\begin{equation}
\mathcal E^T(\mathcal V; N)=\sum_{\Gamma\in\mathcal G^T_{0,N}} Val(\Gamma).
\end{equation}
where $\mathcal G^T_{0,N}\subset \mathcal G_{0,N}$ is the set of {\it connected Feynman diagrams.}\\
These considerations, and the fact that we can compute $Val(\Gamma)$ using the Feynman rules, allow us to derive a very rough {\it upper bound} on the $N$-th order contribution to $f_{\Lambda,\beta}$ that, thanks to (\ref{free_energy_as_sum_of_trunc_expec}), is
\begin{equation}
f_{\Lambda,\beta}^{(N)}:=-\frac{1}{|\Lambda|\beta}\frac{(-1)^N}{N!}\mathcal E^T(\mathcal V;N).
\label{free_energy_specific_N_th_order}
\end{equation}
\begin{lem}
\label{lemma_bounds_no_multiscale_no_determinants}
Let $\epsilon:=\max\{\lambda,\nu\}$, $|\mathcal G_{0,N}^T|$ be the number of connected Feynman diagrams of order $N$, so it holds
\begin{equation}
\begin{split}
|f_{\Lambda,\beta}^{(N)}|\leq \frac{1}{\beta|\Lambda|}\frac{1}{N!}\sum_{\Gamma\in\mathcal G_{0,N}^T}|Val(\Gamma)|\leq \frac{|\mathcal G_{0,N}^T|}{N!}\epsilon^N||g||_{\infty}^{N+1}||g||_1^{N-1}\leq \\
\leq \left(C\epsilon\right)^N N! M^{N+1}\beta^{(N-1)},
\end{split} 
\end{equation}
\end{lem}

\begin{proof}
Given $\Gamma\in\mathcal G_{0,N}^T$, select an arbitrary {\it spanning tree} in $\Gamma$ (a loopless subset of $\Gamma$ connecting all the N vertices). The integrals over the space time coordinates of the product of the propagators of the spanning tree is bounded by $\beta |\Lambda| ||g||^{N-1}_1$, while the product of the remaining propagators is bounded by $||g||^{N+1}_{\infty}$. Then, we use that for some $c>0$, $|\mathcal G_{0,N}^T|\leq c^N(N!)^2$ (see Appendix A.3.3 of \cite{gentile2001renormalization}), and the estimates $||g||_\infty\leq CM$, $||g||_1\leq C\beta$, that we prove in Appendix (\ref{appendix_propagator_decay_property}).
\end{proof}

Of course this rough Lemma has two main problems:
\begin{enumerate}
\item a combinatorial problem, associated to $N!$, that does not allow us to perform the sum over N not even for finite $M, \beta$;
\item a divergence problem, associated to $M^{N+1}\beta^{N-1}$ which is exponentially divergent as $M\to\infty$ and $\beta\to\infty$
\end{enumerate}

Problem 1) can be solved via a smarter re-organization of the perturbation theory in the form of a determinant expasion together with a systematic use of the Gram-Hadamard bound. Problem 2) can be solved by a systematic resummation of the series, based on a multiscale integration of the theory.

\subsection{The determinant expansion}
\label{subsection_determinant_expansion}
Let us show how the first problem can be solved.\\
The basic idea is that, besides the already discussed Feynman diagram representation, there is another well known way to represent the truncated expectation: et us consider the same setting described in the case of (\ref{expectation_truncated_s_sets}), so $s$ sets of indices $P_1,\dots,P_s$. Let us call $|P_i|$ the number of  elements in the set $P_i$, let us label each element with a couple of indices $P_j\ni f:=(j,i)$ where the first index is associated to the set the element belongs to, and the second one is $i=1,\dots,|P_j|$. Finally, let us call $2n=|P_1|+\dots+|P_s|$, {\it i.e.} $n$ is the number of {\it lines} in the Feynman graphs $\Gamma\in\mathcal{G}_0$. So
\begin{equation}
\mathcal{E}^T\left(\tilde \psi\left( P_1\right),\dots, \tilde \psi \left( P_s\right)\right)= \sum_T \alpha_T \left(\prod_{\ell\in T}g_{\ell}\right)\int dP_T(\bm t)\det G^T(\bm t),
\label{expectation_truncated_determinants}
\end{equation}

where

\begin{enumerate}
\item $T$ is an {\it anchored tree} between the clusters of points $P_1,\dots, P_s$: $T$ is a set of lines becoming a tree if one identifies all the points in the same cluster;
\item $\alpha_T$ is a sign, irrelevant for the subsequent bounds;
\item $\bm t$ is the set of parameters $\bm t:=\{t_{j,j'}\in[0,1], 1\leq j,j'\leq s\}$;
\item $dP_T(\bm t)$ is a {\it normalized probability measure} with support on a set $\bm t$ which can be obtained as $t_{i,i'}=\bm u_j\cdot \bm u_{j'}$ for some family of unitary-normed vectors $\bm u_j\in\mathbb R^s$;
\item $G^T(\bm t)$ is a $(n-(s-1))\times(n-(s-1))$ matrix, whose elements are
\begin{equation}
\left[G^T(\bm t)\right]_{(j,i).(j',i')}=t_{j,j'}g(\bm x(j,i),\bm x(j',i'))
\end{equation}
where $1\leq j,j'\leq s$ and $1\leq i \leq |P_j|$, $1\leq i' \leq |P_{j'}|$ in such a way that the lines $\ell=(\bm x(j,i),\bm x(j',i'))$ do not belong to the anchored tree $T$. If $s=1$, $\sum_T$ is empty, and we shall interpret (\ref{expectation_truncated_determinants}) as 
\begin{equation}
\mathcal{E}^T\left(\tilde\psi\left(P_1\right)\right)= \begin{cases}
1,\mbox{ if $P_1$ is empty},\\
\det G(\bm 1), \mbox{ otherwise },
\end{cases}
\end{equation}
where $\bm 1$ is obtained by setting $t_{j,j'}=1 \forall j,j'$.
\end{enumerate}

\begin{rem}
If we expressed the left hand side of (\ref{expectation_truncated_determinants}) as a sum over all possible Feynman graphs, we would actually expand the sum into $O(s!)^2$ terms (where $s$, as in the previous list, is the number of clusters). The latter expression (\ref{expectation_truncated_determinants}) is written in terms of a sum over the family of trees connecting the boxes. It worths noting that, fixing a tree $T$, one can expand the determinant $\det G^T(\bm t)$ in order to obtain, {\it as expected}, all the possible graphs which can be obtained by contracting the $(n-(s-1))$ half-lines not belonging to $T$, {\it i.e.} one can get the Feynman graph representation leading to (\ref{expectation_truncated_s_sets}). The big {\bf improvement} is in the number of terms we are summing up: in the case on Feynman graphs expansion, the sum runs over $O(s!)^2$ terms, while in the latter case the sum runs over the anchored trees, whose number is only $O(s!)$, which morally compensates the $\frac{1}{s!}$ coming from the perturbative expansion.
\end{rem}

We do not present in this thesis the proof of the determinant representation (see \cite{gentile2001renormalization}), which is due to a fermionic reinterpretation of the interpolation formulas by Battle, Brydges and Federbush \cite{battle1984note, brydges1978new, brydges1984short}. Using (\ref{expectation_truncated_determinants}), we get that the $N-$ order of the specific free energy is
\begin{equation}
f_{\Lambda,\beta}^{(N)}=-\frac{1}{\beta|\Lambda|}\frac{(-1)^N}{N!}\epsilon^N \sum_{T\in \mathcal T_N}\alpha_T\int d\bm x_1\dots d\bm x_N \prod_{\ell\in T} g_\ell\int dP_T(\bm t)\det G^T(\bm t),
\label{free_energy_determinant_expansion_gram_hadamard}
\end{equation}
that definitely improves the rough bound in previous Lemma. Indeed, using the fact that the number of anchored trees in $\bm T_N$ is bounded by $C^NN!$ for some $C>0$ (see \cite{gentile2001renormalization}, A.3.3), we get
\begin{equation}
|f_{\Lambda,\beta}^{(N)}|\leq c^N\epsilon^N||g||_1^{N-1}||\det G^T(\cdot)||_\infty.
\end{equation}
Then, in order to bound $||\det G^T||_\infty$, we use the {\it Gram-Hadamard inequality}, 
\begin{lem}[Gram-Hadamard inequality]
\label{lemma_gram_hadamard_inequality}
If $M$ is a square matrix with elements $M_{ij}$ of the form $M_{ij}=\left<A_i,B_j\right>$, where $A_i$ and $B_j$ are vectors in a Hilbert space with scalar product $\left<\cdot,\cdot\right>$, then
\begin{equation}
\left|\det M\right|\leq \prod_{i}||A_i|| ||B_j||
\end{equation}
where $||\cdot||$ is the norm induced by the scalar product. 
\end{lem}
We do not prove this result, and we refer {\it e.g.} to \cite{gentile2001renormalization}, Theorem A.1, but we use it to state the following Lemma.
\begin{lem}
\label{lemma_gram_hadamard_for_G}
Provided we are able to prove that $t_{j,j'}g(\bm x(j,i),\bm x(j',i'))$ can be obtained as a scalar product in a suitable Hilbert space, we can use the Gram-Hadamard inequality to bound:
\begin{equation}
||\det G^T||_\infty\leq c^N ||g||_\infty^{N+1}.
\end{equation}
 \end{lem}
 Recalling that, as we already mentioned, $||g||_\infty\leq c M$
\begin{equation}
|f_{\Lambda,\beta}^{(N)}|\leq c^N\epsilon^N M^{N+1}\beta^{N-1}.
\end{equation}
The proof of the assumption is a subcase of Appendix (\ref{appendix_gram_representation}).

\begin{rem}
Now, the r.h.s. of the latter bound is summable over $N$ for $\epsilon$ small enough, even though non uniformly in $M$ and $\beta$.\\
Proving that the right side of (\ref{free_energy_determinant_expansion_gram_hadamard}) is well defined is a non trivial topic that requires a {\bf multiscale analysis} we are going to explain in the next chapter.
\label{remark_necessity_multiscale_analysis}
\end{rem}

\section{Interacting case: the multiscale analysis}
\label{section_multiscale_analysis}

In this section we explain how to set up a multiscale procedure to perform iterative resummations in order to re-express the specific free energy in terms of a modified expansion, whose $N-$th order term is summable in $N$ and uniformly convergent when the cut-offs are removed.

\subsection{Ultraviolet and infrared regimes, effective potential}

We wish to compute the {\it partition function} defined as $f_{\Lambda,\beta, M}=-\left(|\Lambda|,\beta\right)^{-1}\log \Xi_{\Lambda,\beta, M}$,
\begin{equation}
\Xi_{\Lambda,\beta,M}:=\int P_M(d\psi)e^{-\mathcal V(\psi)}.
\label{partition_function}
\end{equation}
First of all, let us fix the chemical potential: let $p_F=2\pi n_F/L$, $n_F\in\mathbb{N}$ and $\mu_0=1-\cos p_F$.\\
Dealing with fermions, what we are interested in are the excitations near the Fermi surface (which, in dimension one is the pair of points $\pm p_F$), so it is useful to look at the relative momenta with respect two $\pm p_F$:  $k =k' \pm p_F$. So we can rewrite the dispersion $e(k)$ as $$\cos p_F - \cos(k'\pm p_F)= \cos p_F- \left( \cos k'\cos p_F \mp \sin k' \sin p_F\right)$$ so that, near the singularities ({\it i.e.} for $k'\sim 0$) we can consider the linear approximation of the free propagator (\ref{free_propagator_PBC})
\begin{equation}
 \hat g(\pm p_F+k',k_0)\sim\frac{1}{-ik_0\pm k' \sin p_F}.
 \label{free_propagator_PBC_linear_approx}
\end{equation}
\begin{rem}
\label{remark_linear_part_is_luttinger_propagator}
This approximation, besides carrying the {\it physical information} of the theory being the dominant part, corresponds to the propagator of an infrared Luttinger liquid model (i.e. the Luttinger model with an ultraviolet cut-off, that we will comment a bit more when we will study the flow of the running coupling constants). Despite the fact that the infrared Luttinger model, differently from the original Luttinger model (without an ultraviolet cut-off), is not exactly solvable by bosonization, we will use it as a reference model to study the flow of the running coupling constants.
\end{rem}
In order to split the whole momentum space into the union of annuli, first of all we define
\begin{equation}
|\bm k'|=\sqrt{k_0^2+v_0||k'||_{\mathbb T}^2},
\end{equation}
where $||k'||_{\mathbb T}^2=\min_{n\in\mathbb{Z}}|k'-2\pi n|,$ and  $v_0=\sin p_F=\left .\frac{d }{dk}e(k)\right|_{k=p_F}$. So, we introduce a smooth $C^{\infty}$ function  $\chi:\mathcal{D}_{\Lambda}\times \mathcal{D}_{\beta, M}\to C^{\infty}([0,1])$ defined in such a way that
\begin{equation}
\chi(\bm k')=
\begin{cases}
1, \mbox{ if } |\bm k'|\leq \gamma^{-1} p_F/2 ,\\
0, \mbox{ if } |\bm k'|\geq p_F/2,
\end{cases}
\label{cut_off_chi_definition}
\end{equation}
where $\gamma >1$, and $|\bm k|=\sqrt{k_0^2+k^2}$. So, using $$1=1-\chi (k+p_F,k_0)-\chi(k-p_F,k_0)+\chi (k+p_F,k_0)+\chi(k-p_F,k_0)$$ we define the {\it ultraviolet and infrared propagators as follows}:
\begin{equation}
\hat{g}(\bm k)=\underbrace{\frac{1-\chi (k+p_F,k_0)-\chi(k-p_F,k_0)}{ik_0+\cos p_F-\cos k}}_{\hat g^{(u.v.)}(\bm k)}+\underbrace{\frac{\chi(k+p_F,k_0)+\chi(k-p_F,k_0)}{-ik_0+\cos p_F-\cos k}}_{\hat g^{(i.r.)}(\bm k)}.
\end{equation}

Now, using the {\it addition principle} (\ref{addition_principle}), 
we can introduce for any $\bm k\in\mathcal{D}_{\Lambda,\beta,M}$ a couple of Grassmann variables $\left(\psi^{(u.v.)}_{\bm k},\psi^{(i.r.)}_{\bm k}\right)$ with propagators respectively $\hat g^{(u.v.)}(\bm k)$ and $\hat g^{(i.r.)}(\bm k )$ so, given the potential $\mathcal{V}(\psi)$, we can split the integration as 
\begin{equation}
\int P(d\psi)e^{-\mathcal{V}\left(\psi\right)}=\int P(d\psi^{(i.r.)})\int P(d\psi^{(u.v.)})e^{-\mathcal{V}\left(\psi^{(u.v.)}+\psi^{(i.r.)}\right)}
\end{equation}
Finally, we can use the {\it invariance of exponentials} (\ref{invariance_of_exponential}) and define the {\it effective potential at scale $0$}:

\begin{equation}
\int P(d\psi) e^{-\mathcal V(\psi)}=\int P(d\psi^{(i.r.)})\int P(d\psi^{(u.v)}) e^{-\mathcal V\left(\psi^{(i.r.)}+\psi^{(u.v.)}\right)},
\end{equation}
so that
\begin{equation}
\begin{split}
e^{-\beta |\Lambda| f^{(M)}_{\Lambda,\beta}}=\int P(d\psi^{(i.r.)}) \exp \left( \sum{n\geq 1}\frac{1}{n!}\mathcal E_{u.v.}^T\left(-\mathcal V\left(\psi^{(i.r.)}+\cdot\right);n\right)\right):=\\ := e^{-\beta |\lambda| e_{M,0}}\int P(d\psi^{(i.r.)})e^{-\mathcal V^{(0)}(\psi^{(i.r.)})}.
\end{split}
\end{equation}
 
where with $\mathcal{E}_{u.v.}^T\left(\mathcal{V}\left(\cdot+\psi^{(i.r.)}\right);n\right)$ means that we are computing the truncated expectation functions with respect to the Gaussian Grassmann measure $P_{(u.v.)}$ associated to the propagator $\hat g^{(u.v.)}$ keeping the Grassmann variable $\psi^{(i.r.)}$ as an external field, and the effective potential $\mathcal V_0(\psi)$ can be written as
\begin{equation}
\mathcal V_0(\psi)=\sum_{n=1}^{\infty}\sum_{\substack{ \bm x_1,\dots,\bm x_{2n}\\ \in\\ \Lambda\times [0,\beta)}} \left(\prod_{j=1}^{n} \psi^{(i.r.)+}_{\bm x_{2j-1}}\psi^{(i.r.)-}_{\bm x_{2j}} \right) W_{M,2n} (\bm x_1,\dots,\bm x_{2n}).
\label{effective_potential_scale_0}
\end{equation}

\begin{lem}[Ultraviolet integration]
\label{lemma_ultraviolet_integration}
The kernels $W_{M,2n}(\bm x_1,\dots, \bm x_{2n})$ in the previous expansion are given by power series in $\lambda$ convergent in the complex disc $|\lambda|\leq \lambda_0$ for $\lambda_0$ small enough and independent of $M,\Lambda, \beta$, and satisfy the following bound
\begin{equation}
\frac{1}{\beta |\Lambda|}\int d\bm x_1 \dots d\bm x_{2n} \left| W_{M,2n}(\bm x_1,\dots,\bm x_{2n}) \right|\leq C^n |\lambda|^{\max\{1,n-1\}}.
\end{equation}
Moreover, the limits $e_0=\lim_{M\to \infty} e_{M,0}$ and $W_{2n}=\lim_{M\to \infty}(\bm x_1,\dots,\bm x_{2n})$ exist and are reached uniformly in M.
\end{lem}

We do not prove this Lemma because, even if the proof is not trivial, it is simpler than what we will do in studying the infrared regime, and uses the same techniques: we refer to \cite{benfatto1993beta} or to \cite{giuliani2009rigorous,giuliani2011ground}, in which the ultraviolet regime is studied by a multiscale analysis in order to deal with the (very mild) singularity of the free propagator at equal imaginary times. Anyway, the multiscale analysis for the ultraviolet regime is not stricly necessary, and it may be possible to avoid it following the ideas of \cite{pedra2008determinant}. 

\begin{rem}
The fact that the limits are reached {\it uniformely} in $M$ tells us that the infrared problem is essentially independent of M. Since in the infrared region $M$ does not play any role, from now on we drop the label $M$.
\end{rem}
What we have just explained technically, is the first step of Wilson's idea: ideed, we have integrated out the physical information coming from the high energies degree of freedom (far away from the singularities of the infinite volume free propagators), and we are left with an effective theory, described by the effective potential $\mathcal{V}^{(0)}$, describing fermions with momenta a bit closer to singularities. Of course, the information coming from higher energies degrees of freedom are averaged in the effective potential. It will be clear soon in this section how we keep this information in the effective potential by changing the so called coupling constants.

\subsection{Quasi-particles and multiscale expansion}

\paragraph{Quasi-particles in momentum space} As we have already noticed, there are two points in which $\hat g(\bm k)$ is singular and of course, having integrated a slice of momenta far away from singularity (ultraviolet integration), the infrared propagator is still singular in the same two points $\pm p_F$. So, driven by the idea of using again the {\it addition principle}, it is worth defining
\begin{equation}
\hat g^{(i.r.)}(\bm k)=\sum_{\omega=\pm 1}\frac{\chi\left(k-\omega p_F,k_0\right)}{-ik_0+\cos p_F-\cos k}=:\sum_{\omega=\pm}\hat g^{(i.r.)}_{\omega}(\bm k)
\end{equation}
allowing us to write 
\begin{equation}
\int P(\psi^{(i.r.)})e^{\mathcal{V}^{(0)}\left(\psi^{(i.r.)}\right)}=\prod_{\omega=\pm 1}\int P(d\psi_{\omega}^{(i.r.)})e^{\mathcal{V}^{(0)}\left(\psi_+^{(i.r.)}+\psi_-^{(i.r.)}\right)}
\label{quasi_partice_PBC_definition}
\end{equation}
which is the definition of the {\it quasi-particles} Grassmann fields, and the {\it label} $\omega$ is sometimes called the {\it branch label}: for the readers familiar with the Luttinger liquids theory, nearby the singularities we consider the linear approximation of the free propagator (\ref{free_propagator_PBC_linear_approx}), and $\omega=\pm$ labels the {\it right-moving and left-moving fermions}.

\paragraph{Quasi-particles in real space-time} We introduced the cutoff in momentum-space because we want to get closer and closer to the singularities in $(\pm p_F,0)$. It is worth keeping in mind that we are going to plug in the strategy we introduced in subsection (\ref{subsection_How_to_compute_truncated_expectations}), in particular formula (\ref{expectation_truncated_determinants}), so we will need to build up the matrix $G_T$ we introduced in (\ref{expectation_truncated_determinants}), and it is well known how to do it in real space. So it is convenient to define the quasi-particles in real space-time starting from the Fourier transform of the propagator $\hat g^{(i.r.)}$:

\begin{equation}
\begin{split}
g^{(i.r.)}(\bm x-\bm y)=\frac{1}{L\beta}\sum_{\omega=\pm 1} \sum_{\bm k\in \mathcal{D}_{\Lambda,\beta}}\frac{e^{-i k_0(x_0-y_0)}e^{-ik(x-y)}}{-ik_0+e(k)}\chi(k-\omega p_F,k_0)=\\
=\frac{1}{L\beta}\sum_{\omega=\pm 1}\sum_{k'\in\mathcal{D}^\omega_{\Lambda,\beta}} \frac{e^{-i k_0(x_0-y_0)}e^{-i\omega p_F(x-y)}e^{-ik'(x-y)}}{-ik_0+e(k'+\omega p_F)}\chi(k',k_0)=\\
=: \sum_{\omega = \pm 1}e^{-i\omega p_F(x-y) }g^{(i.r.)}_{\omega}(\bm x-\bm y)
\end{split}
\label{free_propagator_infrared_quasi_particles}
\end{equation}

where $\mathcal{D}^\omega_{\Lambda,\beta}=\mathcal{D}_{\Lambda,\beta}-(\omega p_F, 0)$ and 
\begin{equation}
g^{(i.r.)}_{\omega}(\bm x-\bm y)=\frac{1}{L\beta}\sum_{k'\in\mathcal{D}^\omega_{\Lambda,\beta}} \frac{e^{-i\bm k'\cdot (\bm x-\bm y)}}{-ik_0+e(k'+\omega p_F)}\chi(k',k_0).
\end{equation}

\paragraph{Multiscale expansion}

The idea is to approach the singuarities in infinitely many steps. So, we introduce the telescopic identity $$\chi (\bm k)= \sum_{h=-\infty}^0 \left(\chi(\gamma^{-h}\bm k)-\chi(\gamma^{-h+1}\bm k)\right):=\sum_{h=-\infty}^0f_h(\bm k), $$
which implies the obvious definition of {\it propagators on single scale}

\begin{equation}
\hat g^{(i.r.)}_\omega(\bm k)=\sum_{h=-\infty}^0\frac{f_h(k-\omega p_F,k_0)}{-ik_0+\cos p_F-\cos k}=:\sum_{h=-\infty}^0 \hat g^{(h)}_\omega(\bm k)=:\hat g^{(\leq h)}.
\label{propagators_splitted_on_all_scales}
\end{equation}

\begin{rem}
In fact, as far as $L$ and $\beta$ are finite, the sum on $h$ is a sum over a finitely many terms ({\it i.e.} we have a {\it natural} cut-off): indeed by the very definition of $\mathcal D_{\beta}$ (\ref{momenta_space_time}), $|k_0|\geq 2\pi/\beta$, so that $f_h(\bm k)=0$ for any $h<h_\beta$ where $$h_\beta =\min \left\{h: \gamma^{h+1}>\pi/\beta \right\},$$ {\it i.e. } $h_\beta=O\left(\log \beta\right)$ so, as already pointed out, we perform our computations keeping $L$ and $\beta$ finite and then, having obtained bounds independent of $L$ and $\beta$, we take the thermodinamic and the zero temperature limits.
\end{rem}
Again, by combining the {\it addition principle} (\ref{addition_principle}) and the {\it invariance of the exponential} (\ref{invariance_of_exponential}) we can split first of all the Grassmann field $\psi_{\omega}^{(\leq 0)}$ into to Grassmann fields $\psi_{\omega}^{(0)}$ and $\psi_{\omega}^{(\leq -1)}$ with propagators respectively $\hat g^{(0)}_\omega$ and 
\begin{equation}
\hat g_\omega^{(\leq -1)}(\bm k)=\sum_{h\leq -1}\hat g^{(h)}_{\omega}(\bm k),
\end{equation}
or, in a wider generality, $\psi^{h+1}$ and $\psi^{(\leq h)}_\omega$ with propagators respectively $\hat g^{(h+1)}$ and
\begin{equation}
\hat g_\omega^{(\leq h)}(\bm k)=\sum_{j\leq h}\hat g^{(j)}_\omega(\bm k),
\end{equation}
by which we can compute the effective potential on scale $-1$ by
\begin{eqnarray}
\begin{aligned}
\int P(d\psi^{(\leq 0)})e^{\mathcal{V}^{(0)}(\psi^{(\leq 0)})}=\int P(d\psi^{(\leq -1)})\int P(d\psi^{(0)})e^{\mathcal{V}^{(0)}(\psi^{(\leq 0)})}=\\
=: e^{|\Lambda|\beta e_{0}}\int P(d\psi^{(\leq -1)})e^{\mathcal{V}^{(-1)}(\psi^{(\leq -1)})},\\
\end{aligned}\\
\begin{aligned}
|\Lambda|\beta e_0+\mathcal{V}^{(-1)}\left(\psi^{(\leq -1)}\right)=\sum_{n=0}^{\infty}\frac{1}{n!} \mathcal{E}_0^T\left(\mathcal{V}^{(0)}\left(\cdot+\psi^{(\leq -1)}\right);n\right)=\\
=\sum_{n=0}^{\infty}\frac{1}{n!} \mathcal{E}_0^T\left(\sum_{m=0}^{\infty}\frac{1}{m!} \mathcal{E}_{(u.v.)}^T\left(\mathcal{V}\left(\cdot+\psi^{(\leq 0)}\right);m\right);n\right).
\end{aligned}
\end{eqnarray}

where $\mathcal V^{(-1)}\left(0\right)=0$ and, iteratively, for any scale $h$ we can define an effective potential by

\begin{eqnarray}
\begin{aligned}
e^{\mathcal{V}^{(h)}\left(\psi^{(\leq h)}\right)}e^{+|\Lambda|\beta e_{h+1}}=\\=\int P(d\psi^{(h+1)})\dots \int P(d\psi^{(0)}) \int P(d\psi ^{(u.v.)})e^{\mathcal{V}\left(\psi^{(\leq h)}+\psi^{(h+1)}+\dots+(\psi^{(u.v.)}\right)} ,
\end{aligned}\\
|\Lambda|\beta e_{h+1}+\mathcal{V}^{(h)}\left(\psi^{(\leq h)}\right)=\sum_{n=0}^{\infty} \frac{1}{n!}\mathcal{E}^T_{h+1}\left(\mathcal{V}^{(h+1)}(\cdot+\psi^{(\leq h)});n\right).
\label{effective_potential_scale_h_recursive}
\end{eqnarray}
\begin{figure}[htbp]
\centering
\begin{tikzpicture}
[scale=1, transform shape]
\node at (1,3) {$\mathcal V^{(-1)}$ =};
\node at (1,0) {$\mathcal V^{(0)}$ =};
\node at (4.5, 3) {+};
\node at (7.5,3) {+};
\node at (10.5,3) {...};
\node at (4,0) {=};
\node at (7.5, 0) {+};
\node at (10.5,0) {+};
\node at (13.5,0) {...};
\draw [very thick] (2,3) -- ++ (2,0) ++ (1,0) -- ++ (1,0) -- ++ (1,1) ++ (-1,-1) -- ++ (1,-1) ++ (1,1)  -- ++ (1,0) -- ++ (1,1) ++ (-1,-1) -- ++ (1,-1) ++ (-1,1) -- ++ (1,0);
\fill (2,3) circle (0.1);
\fill (2,3) ++ (1,0) circle (0.1);
\fill (2,3) ++ (1,0) ++ (1,0) circle (0.2);
\fill (2,3) ++ (1,0) ++ (1,0) ++ (1,0) circle (0.1);
\fill (2,3) ++ (1,0) ++ (1,0) ++ (1,0) ++ (1,0) circle (0.1);
\fill (2,3) ++ (1,0) ++ (1,0) ++ (1,0) ++ (1,0) ++ (1,1) circle (0.2);
\fill (2,3) ++ (1,0) ++ (1,0) ++ (1,0) ++ (1,0) ++ (1,-1) circle (0.2);
\fill (2,3) ++ (1,0) ++ (1,0) ++ (1,0) ++ (1,0) ++ (2,0) circle (0.1);
\fill (2,3) ++ (1,0) ++ (1,0) ++ (1,0) ++ (1,0) ++ (3,0) circle (0.1);
\fill (2,3) ++ (1,0) ++ (1,0) ++ (1,0) ++ (1,0) ++ (3,0) ++ (1,1) circle (0.2);
\fill (2,3) ++ (1,0) ++ (1,0) ++ (1,0) ++ (1,0) ++ (3,0)  ++ (1,-1) circle (0.2);
\fill (2,3) ++ (1,0) ++ (1,0) ++ (1,0) ++ (1,0) ++ (3,0) ++ (1,0) circle (0.2);
\draw [very thick] (2,0) -- ++ (1,0);
\fill  (2,0) circle (0.1);
\fill (2,0) ++ (1,0) circle (0.2);
\draw [very thick] (5,0) -- ++ (2,0) ++ (1,0) -- ++ (1,0) -- ++ (1,1) ++ (-1,-1) -- ++ (1,-1) ++ (1,1)  -- ++ (1,0) -- ++ (1,1) ++ (-1,-1) -- ++ (1,-1) ++ (-1,1) -- ++ (1,0);
\fill (5,0) circle (0.1);
\fill (5,0) ++ (1,0) circle (0.1);
\fill (5,0) ++ (1,0) ++ (1,0) circle (0.1);
\fill (5,0) ++ (1,0) ++ (1,0) ++ (1,0) circle (0.1);
\fill (5,0) ++ (1,0) ++ (1,0) ++ (1,0) ++ (1,0) circle (0.1);
\fill (5,0) ++ (1,0) ++ (1,0) ++ (1,0) ++ (1,0) ++ (1,1) circle (0.1);
\fill (5,0) ++ (1,0) ++ (1,0) ++ (1,0) ++ (1,0) ++ (1,-1) circle (0.1);
\fill (5,0) ++ (1,0) ++ (1,0) ++ (1,0) ++ (1,0) ++ (2,0) circle (0.1);
\fill (5,0) ++ (1,0) ++ (1,0) ++ (1,0) ++ (1,0) ++ (3,0) circle (0.1);
\fill (5,0) ++ (1,0) ++ (1,0) ++ (1,0) ++ (1,0) ++ (3,0) ++ (1,1) circle (0.1);
\fill (5,0) ++ (1,0) ++ (1,0) ++ (1,0) ++ (1,0) ++ (3,0)  ++ (1,-1) circle (0.1);
\fill (5,0) ++ (1,0) ++ (1,0) ++ (1,0) ++ (1,0) ++ (3,0) ++ (1,0) circle (0.1);
\end{tikzpicture}
\caption{Graphic representation of the first step of the iteration: the graphic expression of $\mathcal V^{(-1)}$ in the first line is the same as the graphical expression of $\mathcal V^{(0)}$ in the second line, where points have been replaced by big black docks, and the meaning of a big black dot attached to the lines is clear in the second line, {\it i.e.} it is a shortcut to write $\mathcal V^{0}$.}
\label{figure_effective_potentiale_scale_0}
\end{figure}
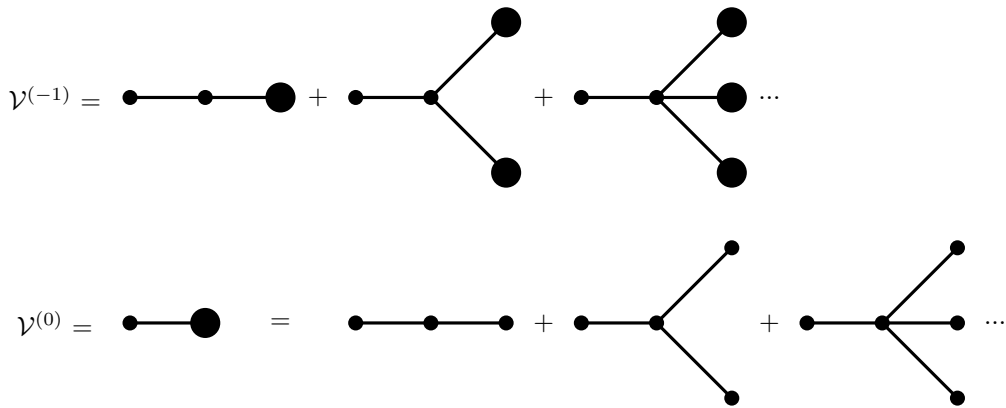
where the truncated expectation $\mathcal E^T_{h}$ (we can think that $h$ can assume also the values $h=1\equiv u.v.$, $(\leq h-1)=(\leq 0) \equiv (i.r.)$ to have a general definition), given a polynomial $F\left(\psi^{(h)}\right)$ with coefficients depending on $\psi^{(\leq h-1)}$, is defined as
\begin{equation}
\mathcal E^T_h \left(F(\cdot);n\right)=\frac{\partial^n}{\partial\lambda^n}  \int P(d\psi^{(h)}) e^{\lambda F(\psi^{(h)})}\Bigl|_{\lambda =0}.
\end{equation}
and, in the argument of the sum, we could express $\mathcal{V}^{(h+1)}$ in terms of $\mathcal{V}^{(h+2)}$ and so on until the only potential involved in the computation is the very first. The recursive structure of these formulae suggests their {\it diagrammatic representation}, known as Gallavotti-Nicolò trees \cite{gallavotti1985renormalization}. In order to understand how to draw these trees before a systematic explanation, it worths looking at the effective potential on scale $-1$ (the first non trivial one): in figure (\ref{figure_effective_potentiale_scale_0}) it is expressed in terms of the {\it previous} effective potential $\mathcal{V}^{(0)}$.\\
In general, the effective potential at scale $h$ can be written as
\begin{equation}
\mathcal V^{(h)}(\psi)=\sum_{n=1}^{\infty}\sum_{\substack{ \bm x_1,\dots,\bm x_{2n}\\ \in\\ \Lambda\times [0,\beta)}} \left(\prod_{j=1}^{n} \psi^{(\leq h)+}_{\bm x_{2j-1}}\psi^{(\leq h)-}_{\bm x_{2j}} \right) W^{(h)}_{2n} (\bm x_1,\dots,\bm x_{2n}).
\label{effective_potential_scale_h}
\end{equation}
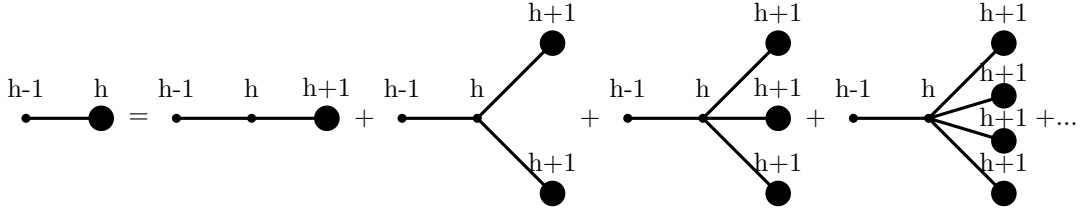
\begin{figure}[htbp]
\begin{center}
\begin{tikzpicture}
\node at (-1.5,2) {=};
\node at (1.5,2) {+};
\node at (4.5,2) {+};
\node at (7.5,2) {+};
\node at (10.7,2) {+...};
\fill (-3,2) circle (0.06);
\node at (-3,2.4) {h-1};
\fill (-2,2) circle (0.17);
\node at (-2,2.4) {h};
\fill (-1,2) circle (0.06);
\node at (-1,2.4) {h-1};
\fill (0,2) circle (0.06);
\node at (0,2.4) {h};
\fill (1,2) circle (0.17);
\node at (1,2.4) {h+1};
\fill (2,2) circle (0.06);
\node at (2,2.4) {h-1};
\fill (3,2) circle (0.06);
\node at (3,2.4) {h};
\fill (4,3) circle (0.17);
\node at (4,3.4) {h+1};
\fill (4,1) circle (0.17);
\node at (4,1.4) {h+1};
\fill (5,2) circle (0.06);
\node at (5,2.4) {h-1};
\fill (6,2) circle (0.06);
\node at (6,2.4) {h};
\fill (7,2) circle (0.17);
\node at (7,2.4) {h+1};
\fill (7,3) circle (0.17);
\node at (7,3.4) {h+1};
\fill (7,1) circle (0.17);
\node at (7,1.4) {h+1};
\fill (8,2) circle (0.06);
\node at (8,2.4) {h-1};
\fill (9,2) circle (0.06);
\node at (9,2.4) {h};
\fill (10,3) circle (0.17);
\node at (10,3.4) {h+1};
\fill (10,1) circle (0.17);
\node at (10,1.4) {h+1};
\fill (10,1.7) circle (0.17);
\node at (10,2.0) {h+1};
\fill (10,2.3) circle (0.17);
\node at (10,2.6) {h+1};
\draw [very thick] (-3,2) -- ++(1,0) ++ (1,0) -- ++ (1,0) -- ++ (1,0) ++ (1,0) 
-- ++ (1,0) -- ++ (1,1) ++ (-1,-1) -- ++(1,-1) ++ (0,1) ++ (1,0) -- ++ (1,0) -- ++ (1,1) ++ (-1,-1) -- ++ (1,-1) ++ (-1,1) -- ++ (1,0) ++ (1,0) -- ++ (1,0) -- ++ (1,1) ++ (-1,-1) -- ++ (1,-1) ++ (-1,1) -- ++ (1,0.3) ++ (-1,-0.3) -- ++ (1, -0.3);
\end{tikzpicture}
\end{center}
\caption{Graphical representation of $\mathcal V^{(h-1)}$, where the big black dot represents $\mathcal V^{(h)}$. It sould be thought of as the generalization at a generic scale $h$ of figure (\ref{figure_effective_potentiale_scale_0})}
\end{figure}

\subsection{Gallavotti-Nicolò trees}

So far we have rewritten the quantities we are interested in, as the {\it specific free energy} (\ref{free_energy_specific_PBC}) and the {\it Schwinger functions} (\ref{schwinger_function_n_points_PBC}) by combining the {\it Grassmann integrals representation} (which implies that we can express these quantities just in terms of {\it truncated expectation functions} of some simple object, as effective interaction (\ref{free_energy_as_sum_of_trunc_expec})) and a {\it multiscale representation} (based on a splitting of the momentum space due to the fact that the main ingredient of our analysis, {\it i.e.} the free propagator (\ref{free_propagator_PBC}), is singular in two points, and we want to approach these singularities following Wilson's idea of RG).\\ This led us to a recursive formula (\ref{effective_potential_scale_h_recursive})  for the effective potentials which in principle involves $h$ sums over infinitely many terms, and we have to deal with its convergence. As we will see, the important tool of Gallavotti-Nicolò trees \cite{gallavotti1985renormalization} allows us to exploit the {\it multiscale structure} of these formulae in order to study in a systematic way the convergence of the series we want to study.

\paragraph{Construction of the tree} Before starting: from now on {\it line} and {\it branch} have the same meaning.\\
Graphically, first of all we consider the plane $(x,y)$, we draw the vertical lines $x=h, h+1, h+2,\dots,0,1$, and we consider all possible graphs obtained as follows. We pick a point on the vertical line $x=h$, we call it $r$ (meaning the {\it root}  of the tree), and we draw an horizontal {\it line } starting from $r$ and leading to a point $v_0$ on the vertical line $x=h_{v_0}>h$, which is the {\it first non trivial vertex}, because it is the first (starting from the left) branching point of $s_{v_0}\geq 2$ lines, forming an angle $\vartheta_j\in (-\pi/2,\pi/2)$ with the x-axis, where $j=1,\dots,s_{v_0}$, and ending into point each of which is located on some vertical line $x=h_{v_0+1},h_{v_0+2},\dots$, which in turn will become branching points. We go on in such a way until $n$ points on the vertical line $x=1$ are reached, and we call them the {\it endpoints}. All the {\it branching points} between the root and the endpoints will be called the {\it nontrivial vertices}, while all the intersections of the lines connecting two nontrivial vertices with the vertical lines will be called {\it trivial vertices}. The integer $n$ we have already introduced, which is the number of endpoints, is the {\it order} of the tree; in sake of clarity, we will label them with numbers from $1$ to $n$ going from the top to the bottom.\\
Among all the trees, we associate a special name to the tree having only one line connecting the root to a vertex on the line $x=1$: it is the {\it trivial tree} $\tau_0$, and in such a case the root has scale $h=1$ (it is important to identify this tree among the others because it will be the starting point of an iterative procedure to rewrite the effective potentials (\ref{effective_potential_scale_h_recursive}) in terms of some numerical values we will associate to these graphical elements). The graph obtained is a {\it tree graph}, because it has no loops and it consists of a set of lines connecting a {\it partially ordered} set of points (that we call {\it vertices}). In particular, having a special point called {\it root}, it is a {\it rooted tree}. We will denote the partial ordering with the symbol $\prec$, meaning that if two vertices $v$ and $w$ are ordered as $v\prec w$, then $h_v<h_w$ (of course, since there is a one to one correspondence between branches and vertices, {\it i.e.} we can associate to each branch the vertex it enters, also the branches are ordered). By construction, to each vertex $v$ is associated an integer number $h_v$ that we call {\it scale label}. We call $\mathcal T_{h,n}$ the family of trees  with $n$ endpoints and the root at scale $h$, and we call a generic element of this family $\tau$ (see Figure (\ref{figure_gallavotti_nicolo_tree})).\\
Of course we will use the {\it Gallavotti-Nicolò trees} formalism in order to study a very precise problem, so we will need to introduce other labels to branches and/or vertices: we will call the set of all the vertices of the tree $\tau$ (included trivial vertices and endpoints) $V(\tau)$, and we introduce a special notation for the endpoints (because it will be useful in the notation we are going to use) $V_f(\tau)\in V(\tau)$. We remark that, by construction (of the tree) and by definition of the sets of vertices, for any $v\in V_f(\tau)$ $h_v=1$, while for any $w\in V(\tau)\setminus V_f(\tau)$ $h<h_w<1$.

\begin{figure}
\centering
 \begin{tikzpicture} 
[scale=1, transform shape]
\foreach \i in {1,2,3,4,5,6,7,8,9,10,11,12,13,14} {%
\draw [thick] (\i,2.9) -- (\i, 11.2); }
\foreach \j in {1,2,3,4,5} {%
\draw [very thick] (\j,7) -- ++ (1,0);
\fill (12,9) circle (0.1);
\fill (13,9.5) circle (0.1);
\fill (\j,7) circle (0.1);
\fill (6,7) circle (0.1);
}
\foreach \j in {0,1,2,3,4,5} {%
\draw [very thick] (6+\j, 7 -\j *0.5) -- +(1,-0.5);
\fill (6+\j,7-\j*0.5) circle (0.1);}
\fill (6+6, 7-3) circle (0.1);
\foreach \j in {0,1,2,3} {%
\draw [very thick] (6+\j, 7 +\j *0.5) -- +(1,+0.5);
\fill (6+\j,7+\j*0.5) circle (0.1);}
\fill (6+4, 7+2) circle (0.1);
\foreach \j in {0,1,2,3} {%
\draw [very thick] (10+\j, 9 +\j *0.5) -- +(1,+0.5);
\fill (10+\j,9+\j*0.5) circle (0.1);}
\fill (14, 11) circle (0.1);
\foreach \j in {0,1,2,3} {%
\draw [very thick] (10+\j, 9 -\j *0.5) -- +(1,-0.5);
\fill (10+\j,9-\j*0.5) circle (0.1);}
\fill (14, 7) circle (0.1);
\draw [very thick] (13,7.5) -- (14,8);
\fill (14,8) circle (0.1);
\foreach \j in {0,1} {%
\draw [very thick] (12+\j, 8 +\j *0.5) -- +(1,+0.5);
\fill (12+\j,8+\j*0.5) circle (0.1);}
\fill(14,9) circle (0.1);
\foreach \j in {0,1,1} {%
\draw [very thick] (12+\j, 4 +\j *0.5) -- +(1,+0.5);
\fill (12+\j,4+\j*0.5) circle (0.1);}
\foreach \j in {0,1,1} {%
\draw [very thick] (12+\j, 4 -\j *0.5) -- +(1,-0.5);
\fill (12+\j,4-\j*0.5) circle (0.1);}
\fill (14,3) circle (0.1);
\fill (14,4) circle (0.1);
\fill (14,5) circle (0.1);
\draw [very thick] (12,4) -- ++ (1,0) -- ++ (1,0);
\fill (13,4) circle (0.1);
\draw [very thick] (11,8.5) -- (14, 10);
\fill (14,10) circle (0.1);
\node at (1,2.7) {$\bm h$};
\node at (2,2.7) {$\bm h+1$};
\node at (3,2.7) {$\bm h+2$};
\foreach \i in {4,5,6,7,8} {%
\node at (\i,2.8) {...};}
\node at (9,2.7) {$\bm h_v$};
\node at (10,2.7) {$\bm h_v+1$};
\foreach \i in {11,12,13} {%
\node at (\i,2.8) {...};}
\node at (14,2.7) {$\bm 1$};
\node at (9,8.8) {$ v$};
\node at (1,7.3) {$ r$};
\node at (2,7.3) {$ v_0$};
\foreach \i in {0,1,2,3} {%
\draw [very thick] (8+\i,8-0.5*\i) -- ++ (1,-0.5);
\fill (8+\i,8-0.5*\i) circle (0.1);}
\fill (12,6) circle (0.1);
\foreach \i in {0,1} {%
\draw [very thick] (12+\i,6-0.25*\i) -- ++ (1,-0.25);
\fill (12+\i,6-0.25*\i) circle (0.1);}
\foreach \i in {0,1} {%
\draw [very thick] (12+\i,6+0.25*\i) -- ++ (1,+0.25);
\fill (12+\i,6+0.25*\i) circle (0.1);}
\fill (14, 6.5) circle (0.1);
\fill (14, 5.5) circle (0.1);
\end{tikzpicture}
\caption{Example of a tree $\tau\in \mathcal T_{h,n}$ where $n=10$.}
\label{figure_gallavotti_nicolo_tree}
\end{figure}
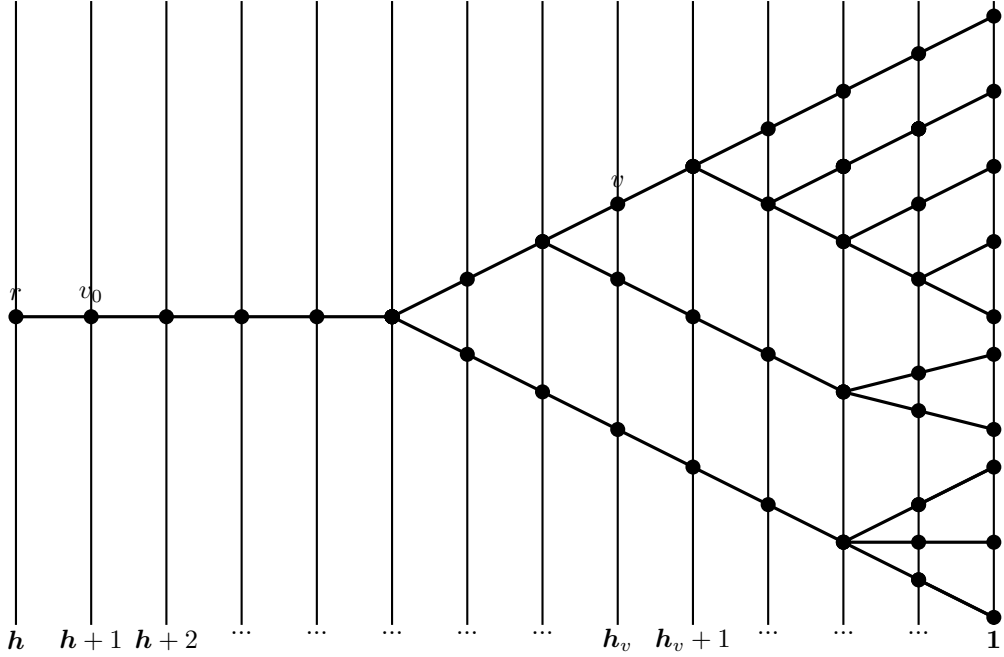

\paragraph{The "importance" of the endpoints}
Since we have built up these trees by an iterative procedure up to the endpoints, the endpoints themself have a kind of a special role: indeed they are the vertices corresponding to an interaction part of the effective potential $\mathcal V^{(0)}$ (\ref{effective_potential_scale_0}) (in fact the vertices $V(\tau)\setminus V_f(\tau)$ correspond to effective potentials).\\
Then each endpoint $v\in V_f(\tau)$ on scale $h_v=1$ is labeled by a further label $i$ which identifies in a unique way the contribution $V_i$ to the potential $\mathcal V^{(0)}$, and we will say that the endpoint $v\in V_f(\tau)$ is {\it of type} $r_i$ if $i_v=i$.\\
Besides, we assign to each endpoint $v\in V_f(\tau)$ a set of spacetime points $\{\bm x_v\}$ which are the integration variables corresponding to the particular interaction contribution $V_i$: only one integration variable if we have an endpoint of $\nu$ tipe, {\it i.e.} a counterterm $\nu \sum_{x\in\Lambda}\psi^+(\bm x)\psi^-(\bm x)$, or two integration points if we have an endpoint of type $\lambda$, {\it i.e.} a two-points interaction and so on. We extend the assignment of the index $\{\bm x_v\}$ also to the not-endpoint vertices $v\in V(\tau)\setminus V_f(\tau)$, saying that $\{\bm x_v\}$ is the family of all space-time points associated with the endpoints following $v$, {\it i.e.} with all the endpoints $w\in V_f(\tau)$ such that $v\prec w$.\\
Finally, we introduce a {\it field label} $f$ to recognize the different fields appearing in the terms associated to the endpoints, and for each endpoint $v\in V_f(\tau)$ we collect the field labels  into a further label $I_v=\{f_1^{(v)},\dots,f_s^{(v)}\}$ in the case that the endpoint $v$ is associated to $s$ fields; so the variables $\bm x(f)$, $\epsilon(f)$ and  $\omega(f)$ will indicate respectively the {\it spacetime point}, the {\it creation/annihilation index} and the {\it quasi-particle index} of the field $f$. As a concrete example, let us consider a quartic endpoint $v\in V_f(\tau)$ associated with  four different fields and four different integration points: $$\lambda \sum_{\bm x_1\dots \bm x_4} \psi^+_{\omega_1}(\bm x_1)\psi^-_{\omega_2}(\bm x_2)\psi^+_{\omega_3}(\bm x_3)\psi^-_{\omega_4}(\bm x_4)W_{4,\bm \omega}(\bm x_1,\dots,\bm x_4),$$
so $\{\bm x_v\}=\{\bm x_1, \bm x_2,\bm x_3, \bm x_4\}$, $I_v=\{f_1,f_2, f_3,f_4\}$ so that  
\begin{equation}
\begin{split}
\bm x(f_1)=\bm x_1,\hspace{3mm} \epsilon(f_1)=+,\hspace{3mm} \omega(f_1)=\omega_1,\\
\bm x(f_2)=\bm x_2,\hspace{3mm} \epsilon(f_2)=-,\hspace{3mm} \omega(f_2)=\omega_2,\\
\bm x(f_3)=\bm x_3,\hspace{3mm} \epsilon(f_3)=+,\hspace{3mm} \omega(f_3)=\omega_3, \\
\bm x(f_4)=\bm x_4,\hspace{3mm} \epsilon(f_1)=-,\hspace{3mm} \omega(f_4)=\omega_4.
\end{split}
\end{equation}
Finally, we call the family of the spacetime points $\bm x(I_v)=\{\bm x(f): f\in I_v\}$.
\paragraph{Clusters and effective potentials} Once we have introduced Gallavotti-Nicolò trees,  we can exploit this diagramatic structure to write the effective potentials on scale $h$ as
\begin{equation}
\mathcal V^{(h)}\left(\psi^{(\leq h)}\right)+L\beta e_{h+1}= \sum_{n=1}^{\infty}\sum_{\tau\in\mathcal{T}_{h,n}}\mathcal{V}^{(h)}\left(\tau, \psi^{(\leq h)}\right)
\label{effective_potential_sum_over_trees}
\end{equation}
where $e_{h+1}$ is a normalization factor for any $h\leq 1$, and of course there is a little abuse of notation in the use of the symbol $\mathcal V^{(h)}$, which has two different meanings in the left side, where $\mathcal{V}^{(h)}(\cdot)$ depends only on the fields on scale $\leq h$, and in the right hand side where the argument of the sum, which depends both on the fields on scale $\leq h$ and on a specific tree $\tau$ chosen among the trees of the family $\mathcal{T}_{h,n}$, is defined iteratively as follows:
\begin{itemize}
\item if $\tau$ is trivial, $\mathcal V^{(0)}(\tau_0,\psi^{(\leq 0)})$ is given simply by one of the contributions to $\mathcal{V}(\psi)$ of the interaction,
\item so if $\tau$ is not trivial, there is a {\it first vertex} $v_0$ the $s_{v_0}$ subtrees $\tau_1,\dots,\tau_{s_{v_0}}\subset \tau$ with root $v_0$ arise from, then
\begin{equation}
\mathcal{V}^{(h)}\left(\tau, \psi^{(\leq h)}\right)=\frac{1}{s_{v_0}!} \mathcal{E}_{h+1}^T\left(\mathcal{V}^{(h+1)}\left(\tau_1,\psi^{(\leq h+1)}\right),\dots, \mathcal{V}^{(h+1)}\left(\tau_{s_{v_0}},\psi^{(\leq h+1)}\right)\right)
\label{effective_potentials_tree_wrt_first_vertex}
\end{equation}
\item of course, each of the $s_{v_0}$ trees can be handled as the {\it original} tree $\tau$ (note that $s_v=0$ if $v\in V_f(\tau)$ is an endpoint), and we can iterate, for each argument of $\mathcal{E}_{h+1}^T$, the formula (\ref{effective_potentials_tree_wrt_first_vertex}).
We define, for any vertex $v$ of the tree, a subset $P_v$ of $I_v$, the {\it external fields} of $v$, that must satisfy some constraints:
\begin{itemize}
\item if $v$ is an endpoint, $P_v=I_v$,
\item if $v$ is not an endpoint,  $v_1,\dots, v_{s_v}$ are $s_v\geq 1$ vertices immediately following it, so $P_v\subset\cup_{i}P_{v_i}$,
\item if $v$ is not an endpoint, we define $Q_v=P_v \cap P_{v_i}$, being the set of labels of the fields associated to the external fields both of $v_i$ and of $v$, that implies $P_{v}=\cup_i^{s_v} Q_{v_i}$,
\item \begin{equation}
 \tilde \psi^{(\leq h_v)}\left(P_{v_j}\right)=\prod_{f\in P_{v_j}}\psi^{(\leq h_v)\epsilon(f)}_{\omega(f)}(\bm x(f)),
\label{tilda_psi_product_of_fields}
\end{equation}
is a product of $|P_{v_j}|$ fields on scale $\leq h_v$ (as almost all these formulae, this one can be proven by induction on the scale $h_v$).
\end{itemize}
Finally, we get the general formula
\begin{equation}
\begin{aligned}
\mathcal{V}^{(h)}\left(\tau, \psi^{(\leq h)}\right)= \left(\prod_{v\in V(\tau)}\frac{1}{s_v!}\right)\\ \mathcal{E}_{h+1}^T\left(\mathcal{E}_{h+2}^T\left(\mathcal{E}_{h+3}^T\dots \mathcal{E}_{-2}^T\left(\mathcal{E}_{-1}^T\left(\mathcal{E}_{0}^T\left(\mathcal V^{(0)}(\tau_0,\psi^{(\leq 1)}),\dots\right),\dots\right),\dots\right),\dots\right),\dots\right)
\end{aligned}
\end{equation}
where, thanks to the first step, we know $\mathcal V^{(0)}(\tau_0,\psi^{(\leq 1)})$.\\
Since the starting point of the latter iterative formula are the {\it trivial trees}, the {\it direction} to follow in order to compute the truncated expectation values is from the endpoints toward the root: once a vertex $v$ is reached, one is left with computing a quantity as
\begin{equation}
\frac{1}{s_v!}\mathcal{E}_{h_v}^T\left(\tilde \psi^{(\leq h_v)}\left(P_{v_1}\right),\dots, \tilde \psi^{(\leq h_v)}\left(P_{v_{s_v}}\right)\right).
\label{expectation_truncated_scale_h_v_gallavotti_nicolo}
\end{equation}
\end{itemize}

\begin{rem}
\label{remark_cluster_structure_gallavotti_nicolo_trees}
At this point, the intrinsic cluster structure of the Gallavotti-Nicolò trees comes out: indeed, using the {\it determinant expansion} (\ref{expectation_truncated_determinants})
$$\mathcal{E}^T\left(\tilde \psi\left( P_1\right),\dots, \tilde \psi \left( P_s\right)\right)= \sum_T \alpha_T \left(\prod_{\ell\in T}g_{\ell}\right)\int dP_T(\bm t)\det G^T(\bm t)$$
 to compute (\ref{expectation_truncated_scale_h_v_gallavotti_nicolo}), we obtain a sum over all possible Feynman diagrams obtained by contracting the half-lines coming from the sets $P_{v_1}, \dots, P_{v_{s_v}}$: when we reach, moving along the tree $\tau$ (from the endpoint towards the root), a vertex $v\in V(\tau)$, we construct a {\it diagram} formed by lines $\ell$ on scales $h_\ell\geq h_v$. Moreover, for any vertex $w\succ v$ there is a subdiagram, that we call $\Gamma_w$ such that all the lines $h_w$ form a connected set if all the further subdiagrams $\Gamma_{w_j}$, $j=1,\dots,w_{s_w}$, corresponding to the $s_w$ vertices immediately following $w$ ({\it i.e.} the roots of the subtrees arising from $w$) are seen as as single elements, which is a consequence of the very definition of truncated expectation.\\
We define a {\it cluster} on scale $h$  as a set of endpoints which are connected by lines on scale $h'\geq h$ such that there is at least one line on scale $h$. The endpoints are trivial clusters at scale $h=1$.
\end{rem}

In this way, we set up a {\it hierarchical structure} of the endpoints into {\it clusters} contained into each other following the order of scales $h\leq 1$. So
\begin{itemize}
\item We associate with each vertex $v\in V(\tau)$ the cluster $G_v$ containing all the endpoints following $v$. From this definition it follows the inclusion relation
\begin{equation}
v\prec w \Rightarrow G_v \supset G_w
\label{hierarchy_inclusion_relation_clusters}
\end{equation}
So there is a one to one map allowing us to represent a trees as a set of hierarchically organized clusters and {\it vice versa}.
\item Given the cluster structure, we can define the {\it anchored trees}: if all the maximal subclusters of $G_v$: $G_{v_1},\dots,G_{v_{s_v}}\subset G_v$ are thought as points, then the set of these points is connected. So, it is possible to select a set of $s_v-1$ lines connecting all of them. This set is, by definition, an {\it anchored tree} (which is a minimal connection between the maximal subclusters of $G_v$).
\begin{rem}
\label{independence_of_trunc_exp_functions_of_the_internal_structure_of_clusters}
Each truncated expectation function sees the clusters associated to the sets of labels $P_{v_1},\dots,P_{v_{s_v}}$ as points, so the action of truncated expectations is independent of the internal structures of the subclusters $G_{v_1},\dots,G_{v_{s_v}}$, and depends only on the external lines of these clusters.
\end{rem}
\end{itemize}

\subsection{Non-renormalized expansion, non-perturbative estimates and classification of the divergences}
\label{subsection_non-renormalized_expansion}

\paragraph{Properties of propagators}

Let us study the behaviour of the single-scale propagators.

\begin{lem}
\label{lemma_propagator_faster_any_power}
For any $N\in\mathbb{N}$ there exists a constant $C_N$ such that the {\it quasi-particle} propagator is bounded by
\begin{equation}
\left|g_{\omega}^{(h)}(\bm x-\bm y)\right|\leq \gamma^h\frac{C_N}{1+\left(\gamma^h |\bm x|\right)^N} .
\label{bound_propagator_faster_than_any_power}
\end{equation}
\end{lem}

We prove this Lemma in Appendix (\ref{appendix_propagator_decay_property}).\\
It is worth underlining that the fact $C_N$ grows with $N$ is not an issue, since we will use the previous bound for $N\leq 4$. Besides, the previous bound comes from the well known result in Fourier analysis saying that the Fourier transform of a $C^\infty$ function decays faster than any power, {\it i.e.} it is just a consequence of the cut-off $f_h$ we used. 

\begin{corollary}
As a trivial consequence of Lemma \ref{lemma_propagator_faster_any_power}, we can bound the norms $||\cdot||_\infty$ and $||\cdot||_1$ of the propagator:
\begin{equation}
||g_{\omega}^{(h)}||_{\infty}:=\sup_{\bm x,\bm y}|g_{\omega}^{(h)}(\bm x-\bm y)|\leq C\gamma^h,
\label{norm_infty_PBC_propagator}
\end{equation}
and 
\begin{equation}
||g_{\omega}^{(h)}||_1=\left|\frac{1}{L\beta}\int d\bm x d\bm y g_{\omega}^{(h)}(\bm x-\bm y)\right|\leq C\gamma^{-h},
\label{norm_1_propagator}
\end{equation}
for some $C>0$.
\end{corollary}

\paragraph{Estimates of kernels of effective potential} Summarizing the analysis of the previous section in a more compact notation, we can write that the {\it effective potential} at scale $h$ is
\begin{equation}
\begin{split}
\mathcal{V}^{(h)}\left(\psi^{(\leq h)}\right)=\sum_{n=1}^{\infty} \sum_{\tau\in\mathcal T_{h,n}}\mathcal{V}^{(h)}\left(\tau,\psi^{(\leq h)}\right),\\
\mathcal{V}^{(h)}\left(\tau,\psi^{(\leq h)}\right)=\int d\bm x(I_{v_0})\sum_{P_{v_0}\subset I_{v_0}}\tilde \psi^{(\leq h)}\left(P_{v_0}\right)\mathcal W^{(h)}\left(\tau, P_{v_0},\bm x(I_{v_0})\right)
\end{split}
\label{effective_potential_expanded_in_kernels}
\end{equation}
where, besides all the quantities we have already defined, $\mathcal{W}^{(h)}$ is defined by the latter expression itself, that we use to get the recursive relation 
\begin{equation}
\begin{split}
\mathcal W^{(h)}\left(\tau,P_{v_0}, \bm x(I_{v_0})\right)=\\
\sum_{P_{v_1},\dots,P_{s_{v_0}}}\left( \prod_{j=1}^{s_{v_0}} \mathcal W^{(h+1)}\left(\tau_j, P_{v_j},\bm x(I_{v_j})\right) \right)\\
\frac{1}{s_{v_0}!}\mathcal E^T_{h+1}\left(\tilde{\psi}^{(h+1)}\left(P_{v_1}\setminus Q_{v_1}, \right) \dots,\tilde{\psi}^{(h+1)}\left(P_{v_{s_{v_0}}}\setminus Q_{v_{s_{v_0}}} \right) \right)
\end{split}
\end{equation}
where 
\begin{equation}
Q_{v_j}=P_{v_0}\cap P_{v_j},\hspace{3mm} j=1,\dots,s_{v_0}.
\end{equation}
It worths noting that the sets $Q_v$ are {\it uniquely determined} by the sets $\{P_v\}$ because, for any $v\in V(\tau)$,
\begin{equation}
Q\in P_v, \mbox{ and } P_v=\bigcup_{j=1}^{s_v}Q_{v_j}\Longrightarrow Q_{v_j}=P_v\cap P_{v_j}, j=1,\dots, s_v
\end{equation}
To get an {\it explicit expression} for $\mathcal W^{(h)}$, we can iterate the latter formula going along the tree $\tau$ and getting
\begin{equation}
\begin{split}
\mathcal W^{(h)}\left(\tau,P_{v_0},\bm x(I_{v_0})\right)=\\
= \sum_{\{P_v\}_{v\in V(\tau )}} \left(\prod_{v\notin V_f(\tau)}\frac{1}{s_v!} \mathcal E^T_{h_v+1}\left(\tilde{\psi}^{(h_v)}\left(P_{v_1}\setminus Q_{v_1}, \right) \dots,\tilde{\psi}^{(h_v)}\left(P_{v_{s_{v}}}\setminus Q_{v_{s_{v}}} \right) \right)\right)\left(\prod_{v\in V_f(\tau)}r_v\right)
\end{split}
\label{kernels_as_truncated_expectation_values}
\end{equation}
where we repeat that $r_v\in\{\nu,\lambda\}$ and the sum over $\{P_v\}_{v\in V(\tau)}$ is a sum over all the possible choices of the sets $P_v$ corresponding to the vertices of $\tau$, except $P_{v_0}$ (which is fixed). Finally, we can define a more intuitive notation to rewrite (\ref{effective_potential_expanded_in_kernels}) as

\begin{equation}
\begin{split}
\mathcal{V}^{(h)}\left(\psi^{(\leq h)}\right)=\sum_{n=1}^{\infty} \sum_{\tau\in\mathcal T_{h,n}}\mathcal{V}^{(h)}\left(\tau,\psi^{(\leq h)}\right),\\
\mathcal{V}^{(h)}\left(\tau,\psi^{(\leq h)}\right)=\sum_{P_{v_0}\subset I_{v_0}} \int d\bm x(P_{v_0})\tilde \psi^{(\leq h)}\left(P_{v_0}\right)\mathcal W^{(h)}\left(\tau, P_{v_0},\bm x(P_{v_0})\right)
\end{split}
\label{effective_potential_expanded_in_kernels_improved_formalism}
\end{equation}

where the gain is that now the kernel $\mathcal{W}^{(h)}$ depends only on the variables $\bm x(P_{v_0})=\{\bm x(f)\}_{f\in P_{v_0}}$ being
\begin{equation}
\mathcal{W}^{(h)}\left(\tau, P_{v_0},\bm x(P_{v_0})\right)=\int d\bm x(I_{v_0}\setminus P_{v_0})\mathcal{W}^{(h)}\left(\tau, P_{v_0},\bm x(I_{v_0})\right).
\end{equation}

Now we can state the 
\begin{thm}
\label{theorem_bound_of_kernels}
In the framework described by the Hamiltonian (\ref{hamiltonian_PBC}), we can bound the kernels we just intruduced as
\begin{equation}
\begin{split}
\int d\bm x(P_{v_0})\left| \mathcal W^{(h)} (\tau, P_{v_0}, \bm x(P_{v_0})) \right| \leq \\ \leq \beta L \gamma^{-h D(P_{v_0})}\sum_{\{P_v\}}\left(\prod_{v\notin V_f(\tau)} \gamma^{-(h_v-h_{v'})D(P_v)} \right)\left(\prod_{v\in V_f(\tau)}\gamma^{-h_{v'}(\frac{|I_v|}{2}-2)}\right)\left(C\epsilon\right)^n
\end{split}
\label{kernels_bound}
\end{equation}
\end{thm}
where $D(P_v)=|P_v|/2-2$, $\epsilon=\max\{\nu,\lambda\}$ and $C=C_N$ for some fixed $N$ in (\ref{bound_propagator_faster_than_any_power}).

\begin{rem}
\label{remark_necessity_of_renormalization_procedure}
Of course the estimate (\ref{kernels_bound}) is finite, and it could not have been otherwise because, thanks to the cut-off, we are performing a power counting on finite quantities. Of course, the troubles come when, in order to reconstruct the original theory, we try to perform a sum over all the scales $h\leq 1$ and the infinite volume limit: as soon as $D(P_{v_0})\leq 0$ ({\it i.e.} when $n_{v_0}=2,4$) the infinite volume limit of the sum diverges being $h_v-h_{v'}>0$. So, despite we still have a {\it divergence problem} as in the {\it naive estimate} in Lemma (\ref{lemma_bounds_no_multiscale_no_determinants}), the advantage of using the {\it multiscale analysis} is that we can clearly identify the {\it sources} of the problem: {\bf if $n_v^e:=|P_v|\leq 4$ the above sum cannot be performed uniformly in $\beta$ and $|\Lambda|$}. So we must deal in some smart way with the clusters with $2$ or $4$ external lines. Actually, our plan consists in nothing else than a {\it slightly different} expansion of the same quantities, by which we can prove that the sum is well defined.\\
Besides the identification of the divergences, the latter estimate is useful also to understand how, thanks to the {\it Gram-Hadamard} estimate (Lemma (\ref{lemma_gram_hadamard_inequality})), we have no longer the combinatorial problem, indeed we can perform the sum in $n$ of the right hand side of (\ref{kernels_bound}) provided $\epsilon$ is small enough.
\end{rem}

The latter remark is the motivation for the following three definitions:
\begin{itemize}
\item the terms which are, after the dimensional estimate, well behaved a priori ({\it i.e.} the terms with more than six external legs) are calld {\it irrelevant};
\item the terms with $D(P_{v})=0$ are called {\it marginal};
\item the terms with $D(P_{v})=-1$ are called {\it relevant}.
\end{itemize}

The multiscale decomposition we just described involves the computation, for any vertex $v\notin V_f(\tau)$, of {\it scale $h_v$ truncated expectations $\mathcal E^T_{h_v}$}, for which the determinant expansion (\ref{expectation_truncated_determinants}) has to be rewritten as
\begin{equation}
\label{expectation_truncated_scale_h_v}
\mathcal{E}_{h_v}^{T_v}\left(\tilde \psi^{( h_v)}\left( P_1\right),\dots, \tilde \psi^{( h_v)} \left( P_s\right)\right)= \sum_{T_v} \alpha_{T_v }\left(\prod_{\ell\in T_v}g^{(h_v)}_{\ell}\right)\int dP_{T_v}(\bm t)\det G^{h_v,T_v}(\bm t),
\end{equation}
where by $g_\ell^{h_v}$ we mean that the propagators associated with lines of the spanning tree $T_v$ live at scale $h_v$, and the matrix $G^{(h_v,T_v)}$ is the analogous of the already described $G^{T}(\bm t)$, but the propagators contributing to the entries $G^{h_v,T_v}_{i,j}$ live at scale $h_V$. So, assuming that $g^{(h_v)}(\bm x,\bm y)$ can be written as a scalar product of vectors belonging to a suitable Hilbert space (as we prove in (\ref{appendix_gram_representation}), looking only at quantities $A^{(h)}_{2(L+1)}$ and $B^{(h)}_{2(L+1)}$) and using the estimate (\ref{bound_propagator_faster_than_any_power}), we can adapt Lemma (\ref{lemma_gram_hadamard_for_G}) as
\begin{equation}
\begin{split}
||\det G^{h_v,T_v}||_\infty \leq c^{\sum_{j=1}^{s_v}|P_{v_j}|-|P_v|-2(s_v-1)}||g^{(h_v)}||_\infty^{\frac{1}{2}\left(\sum_{j=1}^{s_v}|P_{v_j}|-|P_v|-2(s_v-1)\right)}\leq\\ \leq c_1^{\sum_{j=1}^{s_v}|P_{v_j}|-|P_v|-2(s_v-1)}\gamma^{\frac{h_v}{2}\left(\sum_{j=1}^{s_v}|P_{v_j}|-|P_v|-2(s_v-1)\right)},
\end{split}
\end{equation}
so that (\ref{expectation_truncated_scale_h_v}) is bounded by 
\begin{equation}
\label{expectation_truncated_scale_h_v_bound}
\begin{split}
\left|\mathcal{E}_{h_v}^{T_v}\left(\tilde \psi^{( h_v)}\left( P_1\right),\dots, \tilde \psi^{( h_v)} \left( P_s\right)\right)\right)\leq \\ \leq  \sum_{T_v} \prod_{\ell\in T_v}\left| g^{(h_v)}_{\ell}\right|c_1^{\sum_{j=1}^{s_v}|P_{v_j}|-|P_v|-2(s_v-1)}\gamma^{\frac{h_v}{2}\left(\sum_{j=1}^{s_v}|P_{v_j}|-|P_v|-2(s_v-1)\right)}.
\end{split}
\end{equation}
\begin{proof}[Proof of Theorem (\ref{theorem_bound_of_kernels})]
The integration variable $d\bm x(P_{v_0})$ means that that the integration has to be performed over all the endpoints.\\
First of all, it is convenient to decrease the number of integration points $n+n_4\rightarrow n$ by integrating, for any enpoint with $4$ external legs, the {\it finite range} potential $v(x-y)\delta(x_0-y_0)$.\\
Then, using (\ref{expectation_truncated_scale_h_v_bound})
 we get:
\begin{equation}
\begin{split}
\left| \mathcal E^T_{h_v}\left(\tilde{\psi}^{(h_v)}\left(P_{v_1}\setminus Q_{v_1}, \right) \dots,\tilde{\psi}^{(h_v)}\left(P_{v_{s_{v}}}\setminus Q_{v_{s_{v}}} \right) \right)\right|\leq\\
\leq \sum_T \left(\prod_{\ell\in T}|g_\ell|\right)(CC_n)^{\sum_{j=1}^{s_v}|P_{v_j}|-|P_v|-2(s_v-1)}\gamma^{\frac{h_v}{2}\left(\sum_{j=1}^{s_v}|P_{v_j}|-|P_v|-2(s_v-1)\right)},
\end{split}
\end{equation}
where $g_\ell$, $\ell\in T$ are propagators contracted on scale $h_v$. We used that for each anchored tree $T_v$ defined by $T=\cup_{v\notin V_f(\tau)}T_v$ contributing to the sum we can perform $s_v-1$ integrations by using the $s_v-1$ propagators $g_\ell\in T$ getting a factor (thanks to (\ref{norm_1_propagator})), a factor
\begin{equation}
\gamma^{-h_v(s_v-1)}.
\end{equation}
We rewrite the contribution $\prod_{v\notin V_f(\tau)}\gamma^{-h_v(s_v-1)}$ using the formula (that can be easily proved by induction):
\begin{equation}
\begin{split}
\sum_{v\notin V_f(\tau)}h_v(s_v-1)=h(n-1)+\sum_{v\notin V_f(\tau)}(h_v-h_{v'})(n^e_v-1),
\end{split}
\end{equation} 
where $v'$ is the vertex immediately preceeding $v$ on $\tau$, and $n^e_v$ is the number of endpoints following $v$ on $\tau$.\\
It has to be pointed out that the integral runs over $n$ variables, and that the number of variables involved in the integrals we are performing using the propagators belonging to the spanning tree is given by:
\begin{equation}
\begin{split}
\sum_{v\notin V_f(\tau)}(s_v-1)=|V_f(\tau)|-1=n-1,\\
\sum_{\bar v\in V(\tau_v)}\left[\frac{1}{2}\left(\sum_{i=1}^{s_{\bar v}}|P_{\bar v_i}|-|P_{\bar v}|\right)\right]=\frac{1}{2}\left(|I_v|-|P_v|\right).
\label{number_of_integral_variables_translation_invariant_PBC}
\end{split}
\end{equation}
This means that we exploit the {\it compact support properties} of the propagators to integrate out all the variables but one, whose integration (running over all the available space), gives a factor $\beta L$.\\
Furthermore, by definition we have that each endpoint is associated with either $\nu$ or $\lambda$, and that $|V_f(\tau)|\leq n$, so we have that $$\prod_{v\in V_f(\tau)}|r_v|\leq \epsilon^n.$$
Finally we get the bound we are interested in for the left hand side of (\ref{kernels_bound}),
\begin{equation}
\beta L \left(\prod_{v\notin V_f(\tau)}\gamma^{h_v\left[\frac{1}{2}\left(\sum_{j=1}^{s_v}|P_{v_j}|-|P_v|\right)-2(s_v-1)\right]}\right)\left(C\epsilon\right)^n.
\label{bound_proof_theorem_kernel_PBC}
\end{equation}
and using formulae (\ref{number_of_integral_variables_translation_invariant_PBC}) we get 
\begin{equation*}
\begin{split}
\int d\bm x(P_{v_0})\left| \mathcal W^{(h)} (\tau, P_{v_0}, \bm x(P_{v_0})) \right| \leq \\ \leq \beta L \gamma^{-h D(P_{v_0})}\sum_{\{P_v\}}\left(\prod_{v\notin V_f(\tau)} \gamma^{-(h_v-h_{v'})D(P_v)} \right)\left(\prod_{v\in V_f(\tau)}\gamma^{-h_{v'}\left(\frac{|I_v|}{2}-2\right)}\right)\left(C\epsilon\right)^n.
\end{split}
\end{equation*}
\end{proof}

\section{Renormalization Group}
\label{subsection_renormalization_group_PBC}
In Remark (\ref{remark_necessity_of_renormalization_procedure}) we pointed out that, after the multiscale decomposition, the {\it na\"ive cluster expansion} is not enough to conclude that the theory is well defined (meaning that the observables as the {\it specific free energy} and the {\it Schwinger functions} are expressed as finite sums), because the sum over all possible trees is not well defined due to the clusters with $2$ and $4$ external legs.\\
The key idea is that, combining {\it multiscale} and {\it cluster} expansions, we are {\it fragmenting} the quantities we are interested in into infinitely many different pieces in some way easier to control and that, once individually {\it controlled}, we would like to re-sum up in order to reconstruct the initial quantity in such a way that they are clearly well defned (in particular, as analytic functions of the perturbative parameter $\lambda$ within a radius of convergence $\lambda_0>0$). So there is not a unique way to perform this cluster expansion, and our plan is to change the cluster expansion we explained in subsection (\ref{subsection_non-renormalized_expansion}) to {\it cure} the divergences arising from $2$ and $4$ external legs diagrams.\\
Morally, there are no problems coming from the {\it harmless part} consisting of all the clusters with $6$ or more external legs; our strategy will be to identify what is the real source of divergences  in the dangerous clusters and to perform a further splitting of them into two contributions: the {\it renormalized part} which we will put once for all in the harmless part of the theory, and the {\it local part}, which is properly the dangerous part, we will {\it dress} the free theory with.

\subsection{Localization and Renormalization operator}
 \label{subsection_localization_renormalization_PBC}
It is worth by recalling that the explicit expression of the effective potential at scale $h$ (\ref{effective_potential_scale_h}). In particular, using quasi-particle decomposition of the Grassmann variables we can rewrite
\begin{equation}
\begin{split}
\mathcal V^{(h)}(\psi)=\sum_{n=1}^{\infty}\sum_{\substack{ \bm x_1,\dots,\bm x_{2n}\\ \in\\ \Lambda\times [0,\beta)}} \left(\prod_{j=1}^{n} \psi^{(\leq h)+}_{\bm x_{2j-1}}\psi^{(i\leq h)-}_{\bm x_{2j}} \right) W^{(h)}_{2n} (\bm x_1,\dots,\bm x_{2n})=\\
=\sum_{n=1}^{\infty}\sum_{\bm \omega}\sum_{\substack{ \bm x_1,\dots,\bm x_{2n}\\ \in\\ \Lambda\times [0,\beta)}} \left(\prod_{j=1}^{n} e^{-i\omega_{2j-1} x_{2j-1}p_F}\psi^{(\leq h)+}_{\omega_{2j-1}\bm x_{2j-1}}e^{i\omega_{2j}x_{2j}p_F}\psi^{(\leq h)-}_{\omega_{2j},\bm x_{2j}} \right)\cdot \\ \cdot W^{(h)}_{2n} (\bm x_1,\dots,\bm x_{2n})=\\
=:\sum_{n=1}^{\infty}\sum_{\bm \omega}\sum_{\substack{ \bm x_1,\dots,\bm x_{2n}\\ \in\\ \Lambda\times [0,\beta)}} \left(\prod_{j=1}^{n} \psi^{(\leq h)+}_{\omega_{2j-1}\bm x_{2j-1}}\psi^{(\leq h)-}_{\omega_{2j},\bm x_{2j}} \right) W^{(h)}_{2n, \bm \omega} (\bm x_1,\dots,\bm x_{2n}),
\end{split}
\label{effective_potential_quasi_particles_scale_h}
\end{equation} 
where $W^{(h)}_{2n, \bm \omega} (\bm x_1,\dots,\bm x_{2n})=\left(\prod_{j=1}^{n} e^{-i\omega_{2j-1} x_{2j-1}p_F}e^{i\omega_{2j}x_{2j}p_F} \right)W^{(h)}_{2n} (\bm x_1,\dots,\bm x_{2n})$. Analogously, we can define $\hat W^{(h)}_{2n, \bm \omega}(\bm k'_1,\dots,\bm k'_{2n-1})=\hat W^{(h)}_{2n}(\bm k'_1+\omega_1p_F,\dots,\bm k'_{2n-1}+\omega_{2n-1}p_F)$, with $\bm k'\in\mathcal D^{\omega}_{\Lambda,\beta}$, as the Fourier transform of $W^{(h)}_{2n,\bm \omega} (\bm x_1,\dots,\bm x_{2n})$, where $\hat W^{(h)}_{2n,\bm \omega}$ depends on $2n-1$ momenta because of the momentum conservation, meaning that $\sum_{i=1}^{2n}\left(\bm k_i'+\omega_ip_F\right)=0$.\\
We saw that the diagrams with $2$ and $4$ external legs are {\it dangerous} (respectively relevant and marginal terms in the RG terminology), meaning that they do not allow us to perform the sum over all the scales $h\leq 1$, so they do not allow us to conclude that the {\it specific free energy} and the {\it Schwinger functions} defined in (\ref{free_energy_specific_PBC}) and (\ref{schwinger_function_n_points_PBC}) are well defined. So we are forced to {\it manipulate} in some sense the dangerous part: the idea is to extract, first of all, the source of troubles from the $2$ and $4$ external legs diagrams. In particular, we split the quartic terms into the sum of a {\it marginal term} (that can be controlled by studying the flow of a single running coupling constant $\lambda_h$) and an {\it irrelevant one}, and the two external legs terms into three different contributions: an {\it irrelevant one}, a {\it marginal one} that we properly {\it use to dress} the theory and a {\it relevant} one, that we compensate thanks to the {\it counterterm} $\nu$ (which has to be fixed in such a way that the coupling constant $\nu_h$ vanishes when $h\to -\infty$).\\
In particular, we pointed out that the singularities of the propagator $\hat g$ are at $\bm p_F$, so the idea is to expand the kernels in Taylor series near these singularities; now we can appreciate the choice of the change of variables $\bm k=\bm k'+\omega p_F$: a Taylor expansion around the singularities is a Taylor expansion in $\bm k'\sim 0$.\\
We should keep in mind that the fact that the volume is finite gives rise to a lot of technical difficulties: in a finite volume the momenta $\bm k'\in\mathcal{D}_{\Lambda,\beta}^{\omega}$ are quantized and they are not precisely zero, so we should define a Taylor expansion at finite volume, meaning that even if we morally want to expand around the Fermi points $\pm\bm p_F$ we are forced to localize not precisely at $\pm \bm p_F$, but at {\it the closest possible point}. For pedagogical reasons, we will take care of these problems giving a precise {\it finite volume localization definition} only in Appendix (\ref{appendix_real_space-time_localization}), while here we give a more intuitive definition of localization which is correct only at {\it the limit} $\beta,|\Lambda|\nearrow \infty$ (since it neglects the finite volume corrections). \\
Let $\mathcal{L}$ be the {\it localization operator} acting on the {\it effective potentials} in the following way:
\begin{itemize}
\item the terms with more then $6$ external legs cause no problems, so we have nothing to extract:
$$\mathcal{L}W^{(h)}_{2n,\omega}(\bm k'_1,\dots,\bm k'_{2n-1})=0 \mbox{ if } 2n\geq 6,$$
\item on the terms with $4$ external legs,
\begin{equation}
\begin{aligned}
\mathcal L \left(\frac{1}{\left(\beta L\right)^4}\sum_{\bm k'_1,\bm k'_2,\bm k'_3,\bm k'_4 \in \mathcal D^{\bm \omega}_{L,\beta}}\psi^{(\leq h)+}_{\omega_1,\bm k'_1}\psi^{(\leq h)+}_{\omega_2, \bm k'_2}\psi^{(\leq h)-}_{\omega_3, \bm k'_3}\psi^{(\leq h)-}_{\omega_4, \bm k'_4}\right.\\ \left .\hat W^{(h)}_{4,\bm \omega}(\bm k'_1,\bm k'_2,\bm k'_3,\bm k'_4)\delta_{\bm k'_1+\bm k'_2,\bm k'_3+\bm k'_4}\delta_{\omega_1+\omega_2,\omega_3+\omega_4}\right)=\\
\frac{1}{\left(\beta L\right)^4}\sum_{\bm k'_1,\bm k'_2,\bm k'_3,\bm k'_4 \in \mathcal D^{\bm \omega}_{L,\beta}}\psi^{(\leq h)+}_{\omega_1,\bm k'_1}\psi^{(\leq h)+}_{\omega_2, \bm k'_2}\psi^{(\leq h)-}_{\omega_3, \bm k'_3}\psi^{(\leq h)-}_{\omega_4, \bm k'_4}\cdot \\ \cdot \hat W^{(h)}_{4,\bm \omega}(0, 0, 0, 0)\delta_{\bm k'_1+\bm k'_2,\bm k'_3+\bm k'_4}\delta_{\omega_1+\omega_2,\omega_3+\omega_4},
\end{aligned}
\label{localization_4el_finite_volume_limit}
\end{equation}
\item  on the terms with $2$ external legs
\begin{equation}
\begin{split}
\mathcal L\left( \frac{1}{L\beta}\sum_{\bm k\in\mathcal{D}^{\omega}_{L,\beta}} \psi_{\omega,\bm k'+}^{(\leq 0)+}\hat \psi_{\omega,\bm k'}^{(\leq 0)-}W_{2,\omega}(\bm k') \right)
=\frac{1}{L\beta}\sum_{\bm k\in\mathcal{D}^{\omega}_{L,\beta}} \psi_{\omega,\bm k'}^{(\leq 0)+}\hat \psi_{\omega,\bm k'}^{(\leq 0)-}\cdot \\
\cdot\left[\hat W^{(h)}_{2, \omega}(\bm 0)+k'\frac{\partial \hat{W}^{(h)}_{2,\omega}}{\partial k}(\bm 0)+k_0 \frac{\partial \hat W^{(h)}_{2,\omega}}{\partial k_0}(\bm 0)\right].
\end{split}
\label{localization_2el_infinite_volume_limit}
\end{equation}
\end{itemize}
\begin{rem}
Let us comment that we factorized $\delta(\sum_{i}\bm k_i)=\delta (\sum_i \omega_i )\delta(\sum_{i} \bm k'_i)$, which is strictly true only if the scale $h$ is small enough.
\end{rem}
Finally, we simply define the {\it renormalization operator}
\begin{equation}
\mathcal R=1-\mathcal L,
\label{renormalization_operator}
\end{equation}
where $1$ has to be read as the {\it identity operator}.\\
Summarizing, we get:p
\begin{equation}
\mathcal L \mathcal V^{(h)}\left(\psi^{(\leq h)}\right)=\gamma^hn_h F_\nu^{(\leq h)} + z_h F_{\zeta}^{(\leq h)}+ a_h F_{\alpha}^{(\leq h)}+ l_h F_{\lambda}^{(\leq h)},
\label{local_effective_potential_scale_h_PBC}
 \end{equation}
where $n_h,z_h,a_h,l_h$ are real numbers defined by the latter definition itself, and we remark that they are not yet the running coupling constants we want to study, that we will define in a while after a rescaling procedure, and 
\begin{eqnarray}
\begin{aligned}
F_\nu^{(\leq h)}&=\frac{1}{L\beta}\sum_{\omega=\pm}\sum_{\bm k'\in\mathcal{D}_{L,\beta}^{\omega}}\hat \psi_{\omega, \bm k'}^{(\leq h)+}\hat \psi_{\omega, \bm k'}^{(\leq h)-},\\
F_{\alpha}^{(\leq h)}&=\frac{1}{L\beta}\sum_{\omega=\pm} \sum_{\bm k'\in\mathcal{D}_{L,\beta}^{\omega}} \omega v_0 k'\hat \psi_{\omega, \bm k'}^{(\leq h)+}\hat \psi_{\omega, \bm k'}^{(\leq h)-},\\
F_{\zeta}^{(\leq h)}&=\frac{1}{L\beta}\sum_{\omega=\pm} \sum_{\bm k'\in\mathcal{D}_{L,\beta}^{\omega}}(-ik_0)\hat \psi_{\omega, \bm k'}^{(\leq h)+}\hat \psi_{\omega, \bm k'}^{(\leq h)-},\\
F_{\lambda}^{(\leq h)}&=\frac{1}{\left(L\beta\right)^4}\sum_{\bm k_1',\dots,\bm k'_4\in\mathcal{D}^{\omega}_{L,\beta}}\hat \psi_{+, \bm k_1'}^{(\leq h)+}\psi_{-, \bm k_2'}^{(\leq h)+}\hat \psi_{+, \bm k_3'}^{(\leq h)-}\hat \psi_{-, \bm k_4'}^{(\leq h)-}\delta(\bm k'_1+\bm k'_2-\bm k'_3-\bm k'_4),
\end{aligned}
\label{local_effective_potential_scale_h_PBC_term_by_term}
\end{eqnarray}
 \begin{rem}
\label{remark_uniqueness_of_counterterm_PBC}
We used that, in the vicinity of $(p_F,0)$, $\hat W^{(h)}_{\omega,2}(\bm k')\simeq -iz_hk_0+\omega v_0 a_h k'$ (see (\cite{benfatto1993beta} for details), and we would like to comment, especially in order to compare the system we are studying with the one we will study in the next chapter (\ref{chapter_Interacting_fermions_on_the_half_line}), the constants arising from this linearization procedure.
\begin{itemize}
\item First of all, we underline that $\gamma^h n_h$  is a constant, as expected due to the translation invariance.\\
Besides, we point out that, despite $\mathcal L$ is defined as acting on $\bm \omega$-dependent kernels (\ref{localization_4el_real_spacetime}), (\ref{localization_2el_real_spacetime}), so {\it a priori} there would be two different constants, we get a unique $\bm \omega$-independent constant. Of course, this is due to some simmetry of the kernels: indeed by the equations (\ref{local_effective_potential_scale_h_PBC}) and (\ref{local_effective_potential_scale_h_PBC_term_by_term}), we actually define
\begin{equation}
\gamma^h n_h=\hat W^{(h)}_{2,\omega}(\bm 0)=\hat W^{(h)}_{2,-\omega}(\bm 0)
\end{equation}
which is true thanks to the momentum conservation and the parity properties of the propagators (\ref{free_propagator_PBC}).
\item We defined $\partial_k\hat W^{(h)}_{\omega,2}(\bm 0)=\omega v_0 a_h$, with $v_0=\sin p_F$, using that $\partial_k\hat W^{(h)}_\omega(\bm 0)$ is odd in $\omega$.
\item In $F_{\lambda}$ we heavily exploited the anticommutation rules  of the Grassman variables and the momentum conservation to see that the only non vanishing term is associated to the choice $\bm \omega=(+,-,+,-)$. 
\end{itemize}
\end{rem}

Of course, for $h=0$ we have a quite explicit control on the constants, and one can check that:
\begin{equation}
\begin{cases}
n_0=\nu+O(\lambda),\\
a_0=O(\lambda),\\
z_0=O(\lambda),\\
l_0=\lambda\left(\hat v(0)-\hat v(2p_F)\right)+O(\lambda^2)
\end{cases}
\end{equation}

\paragraph{Dimensional gain of renormalized clusters}

\subparagraph{Momentum space} As we have seen, for marginal and relevant terms $D(P_v)$ is, respectively, $0$ and $-1$. So after the action of the operator $\mathcal R$, the goal is to gain a factor $\gamma^{-(h_v-h_{v'})z_v}$ where it is enough to have $z_v=1$ in the case of marginal clusters and $z_v=2$ in the case of relevant ones.\\
This is one of the best examples of the power of the {\it multiscale cluster expansion}, because we manage to gain the right $\gamma^{-z_v (h_v-h_{v'})}$ thanks to the hierarchical structure of the cluster which is, by definition, such that {\it the propagators belonging to a cluster on scale $h$ lives on scales $\geq h$, while the external lines are necessarily on scales $<h$, otherwise they would have been included in the cluster}. \\ 
Let us consider the case of the two external legs terms. As we pointed out, if we stay in the Fourier space we can simply represent the localization (and then the renormalization) operator as acting directly on the kernels $\hat W^{(h)}_{2n, \omega}(\bm k_1,\dots,\bm k_{2n-1})$, so in particular we are interested in the kernel $\hat W^{(h)}_{2,\omega}(\bm k')$ where, first of all, we recall that it depends on only one $\bm k'$ and $\omega$ because the kernels preserve the momentum (if $h$ is small enough, the entering momentum $\bm k'_{in}+\omega_{in} \bm p_F\bm$ must be equal to the exiting one $\bm k'_{out}+\omega_{out}\bm p_F$, {\it i.e.} $\bm k'_{in}=\bm k'_{out}=:\bm k'$ and $\omega_{in}=\omega_{out}=:\omega$ since $\mathcal D^+_{\Lambda,\beta}\cap \mathcal D_{\Lambda,\beta}^-=\emptyset$); then we recall that $\bm k'$ and $\omega$ are the momentum and the quasi-particles
 index carried by the {\it external lines}. Since we are interested in the {\it dimensional analysis} of the renormalized cluster $\mathcal R\hat W_{2,\omega}(\bm k')$, let us start with recalling that:
 \begin{equation}
 \mathcal{L}\hat W^{(h)}_{2,\omega}(\bm k')=\hat W^{(h)}_{2,\omega}(\bm  0)+\bm k' \partial_{\bm k'} \hat W^{(h)}_{2,\omega}(\bm 0),
 \end{equation}
where the notation is a bit inaccurate, but for our aim it is enough because it is based on the fact that we are expanding around $\bm k'=0$ and we used the {\it linear approximation} (\ref{free_propagator_PBC_linear_approx}). So the kernel can be rewritten as
$$\mathcal{L}\hat W^{(h)}_{2,\omega}(\bm k') + \mathcal R \hat W^{(h)}_{2,\omega}(\bm k')=\hat W^{(h)}_{2,\omega}(\bm k')=\hat W^{(h)}_{2,\omega}(\bm 0)+\bm k' \partial_{\bm k'} \hat W^{(h)}_{2,\omega}(\bm 0)+ \bm k'^2\int_0^1d t (1- t)^2\partial^2_{t} \hat W^{(h)}( t\bm k')$$
from which
\begin{equation}
\mathcal R \hat W^{(h)}_{2,\omega}(\bm k')=\bm k'^2\int_0^1d t (1- t)^2\partial^2_{t} \hat W^{(h)}( t\bm k'),
\end{equation}
where we used the subscript $t$ meaning that $\bm k'$ is considered as an external variable.\\
Looking at the cluster representation of the kernel, we see that the external momentum $\bm k'$ is associated to an external leg of the cluster $G_v$, so it is on scale $h=h_{v'}$, $|\bm k'|\sim \gamma^{h_{v'}}$. The derivative, being the kernel a {\it convolution of propagators on scales $>h_v$}, acts on a propagator with scale $\geq h_v$ and, being a derivative, in the dimensional estimate we get a (bad) scale factor $\gamma^{-h_v}$ with $h_v$. Being a second order remainder, we have a further factor with respect to the usual estimate:
\begin{equation}
\gamma^{2(h_{v'}-h_v)}
\end{equation}
so exactly $z_v=2$, as we wanted.\\ At this point of the presentation it should be clear that, besides the dimensional estimates, we have to be careful also in dealing with combinatorial problems arising from the fact that we are dealing with an infinite number of trees of order $n\to \infty$. This is the problem of the so called {\it incapsulated resonances, i.e.} a configuration such that the clusters $G_{v_1}\supset G_{v_2} \supset\dots \supset G_{v_m}$, corresponding to $v_1\prec v_2 \prec\dots\prec v_m$, have to be renormalized. So let us imagine to iteratively apply the procedure described in ({\ref{localization_4el_finite_volume_limit}) and (\ref{localization_2el_infinite_volume_limit}) starting from the most external cluster $G_{v_1}$, then $G_{v_2}$ and so on until the very last one $G_{v_m}$. The recipe we have given to renormalize the clusters says that the derivatives act on some propagator internal to the cluster, so it is possible that in renormalizing $G_{v_1}$ the derivative acts on a propagator belonging to the innermost cluster $G_{v_m}\subset G_{v_1}$, in renormalizing $G_{v_2}$ on the same and so on until the renormalization of the cluster $G_{v_m}$. After $m$ renormalization steps, all the incapsulated clusters $G_{v_1},\dots, G_{v_m}$ have been renormalized but, among all the contributions we get by the renormalization procedure, there are also terms like $\partial^m_{\bm k'}g_\ell$, $\ell\in G_{v_m}$ that, in addition to the right dimensional factor, contributes to the bound with a factor $(m!)^\alpha$, $\alpha\geq 1)$.\\
There are several ways to solve this problem, but to convince the reader that this is not a {\it real problem} we present a very simple  argument, and we refer for details to \cite{benfatto2001renormalization}. The main idea is to show that all the propagators are at most derived twice, since once a gain has been obtained corresponding to some resonance there is no need ore to  renormalize it. In the cluster configuration we have just described, let us imagine that, in renormalizing the cluster $G_{v_1}$ the derivatives acts on a propagator $g_\ell^{h_{v_n}}$ with $\ell\in G_{v_n}$ but $\ell\notin G_{v_{n+1}}$. Using the result described above, we know that we have a {\it scale jump} $\gamma^{(h_{v'_1}-h_{v_n})}$ that can be rescribed as
\begin{equation}
\gamma^{(h_{v'_1}-h_{v_n})}= \gamma^{(h_{v'_1}-h_{v_1})}\gamma^{(h_{v'_2}-h_{v_2})}\dots\gamma^{(h_{v'_{n}}-h_{v_n})}
\end{equation}
which clearly shows that, as a consequence of a single renormalization at scale $h_{v_n}$, each cluster $G_{v_n}\subset G_{v_j}\subset G_{v_1}$ has a {\it scale jump} $\gamma^{(h_{v'_j}-h_{v_j})}$, so there is no need of a further renormalization.
\subparagraph{Real space}
It is convenient (especially for the next chapter) to understand what is the corresponding gain mechanism in real space and the possible sources of problems. Let us consider the {\it first order} remainder (then we can generalize the idea to the {\it second order} remainder),that involves a difference of fields 
$$\psi^{(\leq h)\epsilon}_{\omega,\bm x_i}-\psi^{(\leq h)\epsilon}_{\omega,\bm x_4},$$
which we can formally rewrite as
\begin{equation}
\psi^{(\leq h)\epsilon}_{\omega,\bm x_i}-\psi^{(\leq h)\epsilon}_{\omega,\bm x_1}=(\bm x_i-\bm x_4)\cdot \int _0^1 ds \bm \partial\psi^{(\leq h)\epsilon}_{\omega,\bm x_1+s(\bm x_i-\bm x_4)}.
\end{equation}
where it hat to be stressed that the latter equation has to be read in the {\it weak sense}, meaning that it is properly true once we {\it contract the fields}.\\
The key idea is that the factor $\bm x_i-\bm x_4$ is associated with the kernel $W_{\bm \omega, 4}(\bm x_1-x_4,\bm x_2-\bm x_4,\bm x_3-\bm x_4)$ with $i=1,2,3$. So, we can estimate that $|\bm x_i-\bm x_4|\sim \gamma^{-h_v}$ (actually, one should first expand $|\bm x_i-\bm x_4|$ along the spanning tree, and then bound each contribution by $\gamma^{-h_v}$). Similarly the derivative $\bm \partial$, acting on $\psi_{\omega,\bm x_i}^{(\leq h)\epsilon}$, is constracted at scale $h_i\leq h_{v'}<h_v$, so using (\ref{bound_propagator_faster_than_any_power}), we get a further contribution $\gamma^{h_{v'}}$ to the usual bound. Summarizing, we get, besides the usual dimensional bound, a further term
\begin{equation}
\gamma^{h_{v'}-h_v}.
\end{equation}
This strategy can be extended also in the case of the second order remainder.\\
The analogous of what we worried about in the momentum space representation could happen: if some {\it field variable} is, at the same time, the external line of a big number $m$ of {\it dangerous} (marginal or irrelevant) clusters, it could happen that $m$ derivatives act on the same external line, giving rise (conceptually) to the same combinatorial problem as before. Exploiting the {\it freedom} that we have in choosing the localization point (it is equivalent to localize in any of the points due to the translation invariance of the theory), it is possible to define a localization procedure in such a way that at most two derivatives act on the same external line. Since the intuitive idea is similar to what we used in the momentum space representation, and the rigorous solution of this {\it problem} is well known in literature, we refer to (\cite{benfatto2001renormalization}, section $3.2-3.5$).  
\subsection{Scale h integration and dressed theory}
\label{subsection_anomalous_integration_PBC}
 The starting point to define a Gaussian Grassman integration is, of course, a quadratic operator (in our case an operator which is quadratic in the field variables, as the free initial Hamiltonian $H_0$). So far, we have identified, scale by scale, the irrelevant part of the theory (given by all the terms of degree $\geq 6$ and by the {\it renormalized part} of the $2$ and $4$ external legs terms) and the {\it local part} of the theory which is, at the same time, both problematic and the part containing the physical informations of the model. In particular, having in mind the explicit form of the local part at scale $h$ (\ref{local_effective_potential_scale_h_PBC}) and (\ref{local_effective_potential_scale_h_PBC_term_by_term}) we notice that:
\begin{itemize}
\item $l_hF^{(\leq h)}_\lambda$ reproduces, on scale $h$, the initial two points interaction with a different interaction potential encoded in $l_h\delta(\bm k'_1+\bm k'_2-\bm k'_3-\bm k'_4)$. Obviously, being true scale by scale it defines a recursive relation between the constants $\{l_h\}_{h\leq 1}$;
\item $n_hF^{(\leq h)}_{\nu}$ reproduces the counterterm operator of the initial Hamiltonian, where the constant value $\nu$ is replaced by $n_h$. As before, this explicit shape of the {\it counterterm at scale } $h$ gives us a recursive relations between $\{n_h\}_{h\leq 1}$
\item the terms $a_h F^{(\leq h)}_\alpha$ and $z_h F^{(\leq h)}_\zeta$  are at a first sight {\it new} if considered as part of the interaction, but their sum has the same shape, up to $\mathcal O(k'^2)$ terms, as 
\begin{equation*}
\begin{split}
\left(\hat g^{(h)}_\omega(\bm k')\right)^{-1}=\left(-ik_0+(1-\cos k')\cos p_F+\omega v_0\sin k'\right) f^{-1}_h(\bm k')=\\
=\left(-ik_0+\omega v_0 k'+[(1-\cos k')\cos p_F+\omega v_0(\sin k'-k')]\right) f^{-1}_h(\bm k')=:\\
=:\left(-ik_0+\omega v_0 k'+t_{0,\omega}(k')\right) f^{-1}_h(\bm k'),
\end{split}
\end{equation*} 
with constants $a_h$ and $z_h$ replacing $1$, and where we called $t_{0,\omega}(k')$ the $\mathcal O(k'^2)$ term.
\end{itemize}
The main idea is to {\it absorbe} step by step, in a sense that will be clarified during this paragraph, the quadratic terms $a_h F^{(\leq h)}_\alpha$ and $z_h F^{(\leq h)}_\zeta$ in the integration: this will have the effect to change the {\it propagator} (in RG language we will say that these terms will be used to {\it dress the propagator}) the Gaussian Grassman measure is associated with, and we will encode this {\it dressing} in a new {\it running coupling constant}, called $Z_h$ with $h\leq 0$ and $Z_0=1$, whose flow we will study again in an iterative way. In the following, we will describe the {\it generic $h-th$ step} but, to be able to handle these arguments in a technical way, we warmly raccomend to work out the very first step (from scale $h=0$ to scale $h=-1$), in which all the constants and the computations are quite explicit.\\ 
Let us introduce a sequence of constants $\left\{Z_h\right\}$, $Z_0=1$ and let us define the function $C_h(\bm k')$  by
\begin{equation}
C_h(\bm k')^{-1}=\sum_{j=h_\beta}^h f_h(\bm k').
\end{equation}
So, after the integration of the degrees of freedom on scales $>h$ we get, up to a constant, a Gaussian Grassman integral
\begin{equation}
\int P_{Z_h}(d\psi^{(\leq h)})e^{-\mathcal{V}^{(h)}\left(\sqrt{Z_h}\psi^{(\leq h)}\right)},
\label{integral_rescaled_Zh_PBC}
\end{equation}
where the Gaussian Grassman measure is defined as
\begin{equation}
\begin{split}
P_{Z_h}(d\psi^{(\leq h)})=\left(\prod_{\omega=\pm}\prod_{\bm k'\in \mathcal{D}_{L,\beta}^\omega} d\psi^{(\leq h)+}_{\omega,\bm k'}d\psi^{(\leq h)-}_{\omega,\bm k'}\right)\\
\exp\left[ -\frac{1}{L\beta}\sum_{\omega=\pm}\sum_{\bm k'\in \mathcal{D}_{L,\beta}^\omega}C_h(\bm k')Z_h \left(-ik_0+\omega v_0 k'+t_{0,\omega}(k')\right) \psi^{(\leq h)+}_{\omega,\bm k'}\psi^{(\leq h)-}_{\omega,\bm k'} \right],
\end{split}
\label{P_Zh_PBC}
\end{equation}
associated to a covariance which is the $\hat g^{(\leq h)}$ we are familiar with, except for the multiplicative factor $Z_h$ due to the {\it wave function renormalization}, as we are going to explain. As we anticipated, we want to move some terms from the interaction to the measure.\\
First of all, let us notice that the interaction is computed in $\sqrt{Z_h}\psi^{(\leq h)}$, so all the terms of the interaction are suitably multiplied by a power of $Z_h$; in particular, in (\ref{local_effective_potential_scale_h_PBC}) and (\ref{local_effective_potential_scale_h_PBC_term_by_term}) $$F_j^{(\leq h)}\left(\sqrt{Z_h}\psi^{(\leq h)}\right)= Z_h F_j^{(\leq h)}(\psi^{(\leq h)}), \hspace{4mm} F_\lambda^{(\leq h)}\left(\sqrt{Z_h}\psi^{(\leq h)}\right)= Z_h^2 F_\lambda^{(\leq h)}(\psi^{(\leq h)}),$$ for $j=\alpha, \nu, \zeta$:
\begin{itemize}
\item in order to {\it dress the propagator} ({\it i.e.} to move into the measure a part of the effective potential), we rewrite the {\it local part of the effective potential at scale $h$} (\ref{local_effective_potential_scale_h_PBC}) as
\begin{equation}
\begin{split}
  \mathcal L\mathcal V^{(\leq h)}\left(\sqrt{Z_h}\psi^{(\leq h)}\right)= \\ =\mathcal L \mathcal V^{(\leq h)}\left(\sqrt{Z_h}\psi^{(\leq h)}\right) + z_h F_{\alpha}^{(\leq h)}\left(\sqrt{Z_h}\psi^{(\leq h)}\right) - z_h F_{\alpha}^{(\leq h)}\left(\sqrt{Z_h}\psi^{(\leq h)}\right) =\\
 =\gamma^h n_h F_{\nu} ^{(\leq h)}\left(\sqrt{Z_h}\psi^{(\leq h)}\right) + z_h \left( F_{\zeta}^{(\leq h)}\left(\sqrt{Z_h}\psi^{(\leq h)}\right) +F_{\alpha}^{(\leq h)}\left(\sqrt{Z_h}\psi^{(\leq h)}\right) \right)+\\ + \left(a_h-z_h\right) F_{\alpha}^{(\leq h)} \left(\sqrt{Z_h}\psi^{(\leq h)}\right)+ l_h F^{(\leq h)}_{\lambda} \left(\sqrt{Z_h}\psi^{(\leq h)}\right)=:\\
 =: \mathcal L \tilde{ \mathcal V}^{(h)} \left(\sqrt{Z_h}\psi^{(\leq h)}\right) + z_h\left(F_\zeta^{(\leq h)}\left(\sqrt{Z_h}\psi^{(\leq h)}\right) +F_\alpha^{(\leq h)}\left(\sqrt{Z_h}\psi^{(\leq h)}\right)\right).
\end{split}
\label{local_tilde_effective_potential_definition}
\end{equation}
where it worths pointing out that 
\begin{equation*}
\begin{split}
z_h \left( F_{\zeta}^{(\leq h)}+F_{\alpha}^{(\leq h)}\right)=\\=\frac{1}{|\Lambda|\beta}\sum_{\bm k'\in\mathcal D_{\Lambda,\beta}^\omega} z_h \left(-ik_0+\omega v_0 k'\right) Z_h \psi^{(\leq h)+}_{\omega,\bm x}\psi^{(\leq h)-}_{\omega,\bm x},
\end{split}
\end{equation*}
{\it i.e.}, except for the constant $z_h$, it is the same as the exponent of the Grassmann integration.
\item Now, in the integral (\ref{integral_rescaled_Zh_PBC}), using the usual exponential properties we {\it move} the term $$z_h \left( F_{\zeta}^{(\leq h)}+F_{\alpha}^{(\leq h)}\right)$$ into the measure (\ref{P_Zh_PBC}), which becomes
\begin{equation}
\begin{split}
P_{Z_h}(d\psi^{(\leq h)})=\left(\prod_{\omega=\pm}\prod_{\bm k'\in \mathcal{D}_{L,\beta}^\omega} d\psi^{(\leq h)+}_{\omega,\bm k'}d\psi^{(\leq h)-}_{\omega,\bm k'}\right)\\
\exp\Biggl[ -\frac{1}{L\beta}\sum_{\omega=\pm}\sum_{\bm k'\in \mathcal{D}_{L,\beta}^\omega}C_h(\bm k')Z_h\left(1+ C_h(\bm k')^{-1}z_h\right)\cdot\\
\cdot \left(-ik_0+\omega v_0 k'+\vartheta_{h,\omega}(\bm k')\right) \psi^{(\leq h)+}_{\omega,\bm k'}\psi^{(\leq h)-}_{\omega,\bm k'}\Biggl].
\end{split}
\end{equation}
Since we need some recursive relations between the {\it running coupling constants}, by the latter formula we can define
\begin{equation}
Z_{h-1}(\bm k')=Z_h\left(1+C^{-1}_h(\bm k')z_h\right),
\label{Z_h-1(k')_definition}
\end{equation}
and $\vartheta_{h,\omega}(\bm k')$ is defined as
\begin{equation}
\vartheta_h(\bm k')=\begin{cases}
t_{0,\omega}(\bm k') &\mbox{ if } h=0,\\
\frac{Z_{h+1}}{Z_{h}(\bm k')}\vartheta_{h+1,\omega}(\bm k') &\mbox{ if } h<0.
\end{cases}
\end{equation}
Let us underline that we are dressing only the linear part of the covariance, so this rescaling of the $\vartheta_{h,\omega}$ terms simply allows to keep $\vartheta_{h,\omega}$ into the brackets.
\item Finally we can rewrite the integral (\ref{integral_rescaled_Zh_PBC}) as 
\begin{equation}
\int P_{Z_h}(d\psi^{(\leq h)})e^{-\mathcal{V}^{(h)}\left(\sqrt{Z_h}\psi^{(\leq h)}\right)}= \frac{1}{\mathcal N_h}\int \tilde P_{Z_{h-1}}(d\psi^{(\leq h)})e^{-\mathcal{\tilde V}^{(h)}\left(\sqrt{Z_h}\psi^{(\leq h)}\right)}
\label{integral_scale_leqh_dressed_measure_PBC}
\end{equation}
where, of course,
\begin{equation}
\begin{split}
\tilde P_{Z_{h-1}}(d\psi^{(\leq h)})=\left(\prod_{\omega=\pm}\prod_{\bm k'\in \mathcal{D}_{L,\beta}^\omega} d\psi^{(\leq h)+}_{\omega,\bm k'}d\psi^{(\leq h)-}_{\omega,\bm k'}\right)\\
\exp\Biggl[ -\frac{1}{L\beta}\sum_{\omega=\pm}\sum_{\bm k'\in \mathcal{D}_{L,\beta}^\omega}C_h(\bm k')Z_{h-1}(\bm k')\cdot \\ \cdot \left(-ik_0+\omega v_0 k'+\vartheta_{h,\omega}(\bm k')\right) \psi^{(\leq h)+}_{\omega,\bm k'}\psi^{(\leq h)-}_{\omega,\bm k'}\Biggl],
\end{split}
\label{tilde_P_z_h-1}
\end{equation}
and 
$$\tilde{\mathcal V}^{(\leq h)}\left(\sqrt {Z_h}\psi^{(\leq h)}\right)=\mathcal L \tilde{\mathcal V}^{(h)}\left(\sqrt {Z_h}\psi^{(\leq h)}\right) + \left(1-\mathcal L\right)\mathcal V^{(h)}\left(\sqrt {Z_h}\psi^{(\leq h)}\right).$$
\item It is worth noticing that, by definition, if $|\bm k'|<\gamma^{h-1}$, $Z_{h-1}(\bm k')$ assumes a constant value, properly $Z_{h-1}(\bm k')=Z_h(1+z_h)$ (this comment will become useful in performing the {\it usual} scale by scale integration, using the {\it addition principle} (\ref{addition_principle})).
\end{itemize}
As it should be clear, the power of this machinary is the possibility to integrate (\ref{integral_rescaled_Zh_PBC}) scale by scale. So we split the measure in the right hand side of (\ref{integral_scale_leqh_dressed_measure_PBC}) as
\begin{equation}
\frac{1}{\mathcal N_h}\int P_{Z_{h-1}}(d\psi^{(\leq h-1)})\int \tilde P_{Z_{h-1}}(d\psi^{( h)})e^{-\mathcal{\tilde V}^{(h)}\left(\sqrt{Z_h}\psi^{(\leq h)}\right)}
\label{integral_rescaled_splitted}
\end{equation}
which defines, first of all, the measure $P_{Z_{h-1}}(d\psi^{(\leq h-1)})$ as (\ref{tilde_P_z_h-1}) with
\begin{itemize}
\item $Z_{h-1}(\bm k')$ replaced by $Z_{h-1}$ (because of what we explained in the very last point of the latter list),
\item $C_h(\bm k')$ replaced by $C_{h-1}(\bm k')$,
\item $\psi^{(\leq h)}$ replaced by $\psi^{(\leq h-1)}$,
\end{itemize}
and the {\it single scale measure} $\tilde P_{Z_{h-1}}(d\psi^{(h)})$ is given again by (\ref{P_Zh_PBC}) with
\begin{itemize}
\item $Z_{h-1}(\bm k')$ replaced by $Z_{h-1}$,
\item $C_{h}(\bm k')$ replaced by
\begin{equation}
\tilde f_h (\bm k')=Z_{h-1}\left(\frac{C_h^{-1}(\bm k')}{Z_{h-1}(\bm k')}-\frac{C_{h-1}^{-1}(\bm k')}{Z_{h-1}}\right),
\end{equation}
\item $\psi^{(\leq h)}$ replaced by $\psi^{(h)}$.
\end{itemize}
Finally, and this is the {\it definition of the running coupling constants}, we rescale all the fields by $Z_{h-1}$, {\it i.e.} we multiply and divide by the same quantity:
$$\sqrt Z_h\psi^{(\leq h)}=\left(\frac{\sqrt Z_h}{\sqrt Z_{h-1}}\right)\sqrt Z_{h-1}\psi^{(\leq h)}$$
in order to rewrite $\mathcal L\tilde{\mathcal V}^{(h)}$ as 
\begin{equation}
\begin{split}
\mathcal{L}\hat{\mathcal V}^{(h)}\left(\sqrt {Z_{h-1}}\psi^{(\leq h)}\right)=\\
=\gamma^h\nu_hF_{\nu}^{(\leq h)}\left(\sqrt {Z_{h-1}}\psi^{(\leq h)}\right)+\delta_h F_\alpha^{(\leq h)}\left(\sqrt {Z_{h-1}}\psi^{(\leq h)}\right)+\lambda F_{\lambda_h}^{(\leq h)}\left(\sqrt {Z_{h-1}}\psi^{(\leq h)}\right),
\end{split}
\end{equation}
and we rewrite the integral (\ref{integral_rescaled_splitted}) as
\begin{equation}
\frac{1}{\mathcal N_h}\int P_{Z_{h-1}}(d\psi^{(\leq h-1)})\int \tilde P_{Z_{h-1}}(d\psi^{( h)})e^{-\mathcal{\hat V}^{(h)}\left(\sqrt{Z_{h-1}}\psi^{(\leq h)},\right)}
\end{equation}
and, by definition,
\begin{equation}
\begin{split}
\nu_h&=\frac{Z_h}{Z_{h-1}}n_h,\\
\delta_h&=\frac{Z_h}{Z_{h-1}}(a_h-z_h),\\
\lambda_h&=\left(\frac{Z_h}{Z_{h-1}}\right)^2 l_h.
\end{split}
\label{running_coupling_constants_PBC_definition}
\end{equation}
Let us introduce a compact notation: the vector $\vec v$ collects these three constants on scale $h$: $$\vec v_h=(\nu_h,\delta_h,\lambda_h).$$
Now we can perform the integration with the Gaussian Grassman measure $\tilde P_{Z_{h-1}}(d\psi^{(h)})$ associated with the propagator
\begin{equation}
\frac{g^{(h)}(\bm x-\bm y)}{Z_{h-1}}=\sum_{\omega=\pm}e^{-i\omega p_F(x-y)}\frac{g^{(h)}_\omega(\bm x-\bm y)}{Z_{h-1}}
\end{equation}
where 
\begin{equation}
\frac{g^{(h)}_\omega(\bm x-\bm y)}{Z_{h-1}}=\int \tilde P_{Z_{h-1}}(d\psi^{(h)})\psi^{(h)-}_{\omega,\bm x} \psi^{(h)+}_{\omega,\bm y}
\end{equation}
and again
\begin{equation}
g^{(h)}_\omega(\bm x-\bm y)=\frac{1}{L\beta}\sum_{\bm k'\in\mathcal{D}^{\omega}_{L,\beta}}e^{-i\bm k'\cdot (\bm x-\bm y)}\tilde{f}_{h}(\bm k')\left(\hat g^{(h)}_\omega(\bm k')\right)^{-1}
\label{dressed_propagator_scale_h_PBC}
\end{equation}
defining the {\it effective potential} on the next scale:
\begin{equation}
\int \tilde P_{Z_{h-1}}(d\psi^{(h)})e^{-\hat{\mathcal V}^{(h)}(\sqrt Z_{h-1}\psi^{(\leq h)})}=e^{L\beta  e_h-\mathcal V^{(h-1)}(\sqrt Z_{h-1}\psi^{(\leq h-1)})}
\end{equation}
where $ e_h$ is a suitable constant and 
\begin{equation}
\mathcal L\mathcal V^{(h-1)}(\psi^{(\leq h-1)})=\gamma^{h-1}n_{h-1}F_\nu^{(\leq h-1)}+a_{h-1}F_\alpha^{(\leq h-1)}+z_{h-1}F_\zeta^{(\leq h-1)}+l_{h-1}F_\lambda^{(\leq h-1)},
\end{equation}
so that we can iterate the just described procedure.
\begin{rem}
First of all, this iterative procedure gives {\it for free} the way to write the {\it running coupling constants} on scale $h$ as a function of the {\it running coupling constants} om higher scales:
\begin{equation}
\vec v_h=\vec \beta(\vec v_{h+1},\dots,\vec v_0),
\label{beta_function_definition}
\end{equation}
where $\vec \beta(\vec v_{h+1},\dots,\vec v_0)$ is called the {\it beta function}. 
\end{rem}

\subsection{The renormalized tree expansion and renormalized bounds}

It is convenient to directly look at Fig. (\ref{figure_renormalized_trees}): we write $\mathcal{V}^{(0)}$ knowing that there can be endpoints representing contributions from $\mathcal{L} \mathcal{V}^{(1)}$. 
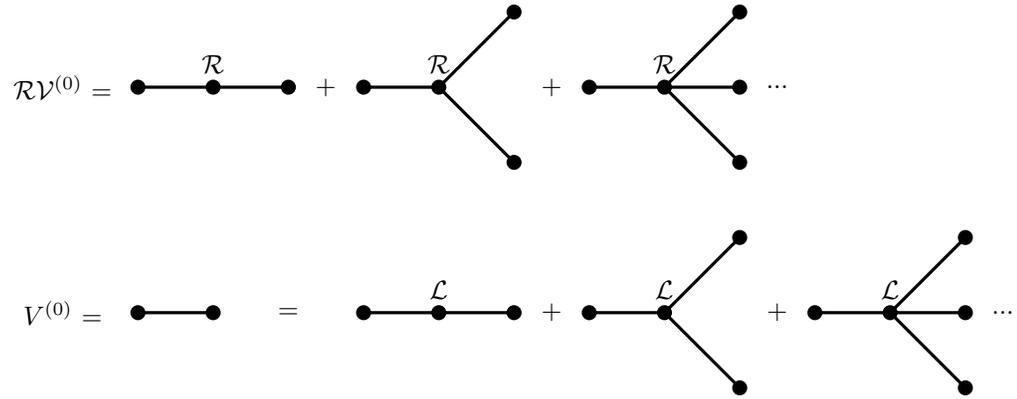
\begin{figure}[htbp]
\centering
\begin{tikzpicture}
[scale=1, transform shape]
\node at (1,3) {$\mathcal R \mathcal V^{(0)}$ =};
\node at (1,0) {$\mathcal \mathcal V^{(0)}$ =};
\node at (4.5, 3) {+};
\node at (7.5,3) {+};
\node at (10.5,3) {...};
\node at (4,0) {=};
\node at (7.5, 0) {+};
\node at (10.5,0) {+};
\node at (13.5,0) {...};
\draw [very thick] (2,3) -- ++ (2,0) ++ (1,0) -- ++ (1,0) -- ++ (1,1) ++ (-1,-1) -- ++ (1,-1) ++ (1,1)  -- ++ (1,0) -- ++ (1,1) ++ (-1,-1) -- ++ (1,-1) ++ (-1,1) -- ++ (1,0);
\fill (2,3) circle (0.1);
\fill (2,3) ++ (1,0) circle (0.1);
\node at (3,3.3) {$\mathcal R$};
\fill (2,3) ++ (1,0) ++ (1,0) circle (0.1);
\fill (2,3) ++ (1,0) ++ (1,0) ++ (1,0) circle (0.1);
\node at (6,3.3) {$\mathcal R$};
\node at (9,3.3) {$\mathcal R$};
\fill (2,3) ++ (1,0) ++ (1,0) ++ (1,0) ++ (1,0) circle (0.1);
\fill (2,3) ++ (1,0) ++ (1,0) ++ (1,0) ++ (1,0) ++ (1,1) circle (0.1);
\fill (2,3) ++ (1,0) ++ (1,0) ++ (1,0) ++ (1,0) ++ (1,-1) circle (0.1);
\fill (2,3) ++ (1,0) ++ (1,0) ++ (1,0) ++ (1,0) ++ (2,0) circle (0.1);
\fill (2,3) ++ (1,0) ++ (1,0) ++ (1,0) ++ (1,0) ++ (3,0) circle (0.1);
\fill (2,3) ++ (1,0) ++ (1,0) ++ (1,0) ++ (1,0) ++ (3,0) ++ (1,1) circle (0.1);
\fill (2,3) ++ (1,0) ++ (1,0) ++ (1,0) ++ (1,0) ++ (3,0)  ++ (1,-1) circle (0.1);
\fill (2,3) ++ (1,0) ++ (1,0) ++ (1,0) ++ (1,0) ++ (3,0) ++ (1,0) circle (0.1);
\draw [very thick] (2,0) -- ++ (1,0);
\fill  (2,0) circle (0.1);
\fill (2,0) ++ (1,0) circle (0.1);
\node at (6,0.3) {$\mathcal L$};
\node at (9,0.3) {$\mathcal L$};
\node at (12,0.3) {$\mathcal L$};
\draw [very thick] (5,0) -- ++ (2,0) ++ (1,0) -- ++ (1,0) -- ++ (1,1) ++ (-1,-1) -- ++ (1,-1) ++ (1,1)  -- ++ (1,0) -- ++ (1,1) ++ (-1,-1) -- ++ (1,-1) ++ (-1,1) -- ++ (1,0);
\fill (5,0) circle (0.1);
\fill (5,0) ++ (1,0) circle (0.1);
\fill (5,0) ++ (1,0) ++ (1,0) circle (0.1);
\fill (5,0) ++ (1,0) ++ (1,0) ++ (1,0) circle (0.1);
\fill (5,0) ++ (1,0) ++ (1,0) ++ (1,0) ++ (1,0) circle (0.1);
\fill (5,0) ++ (1,0) ++ (1,0) ++ (1,0) ++ (1,0) ++ (1,1) circle (0.1);
\fill (5,0) ++ (1,0) ++ (1,0) ++ (1,0) ++ (1,0) ++ (1,-1) circle (0.1);
\fill (5,0) ++ (1,0) ++ (1,0) ++ (1,0) ++ (1,0) ++ (2,0) circle (0.1);
\fill (5,0) ++ (1,0) ++ (1,0) ++ (1,0) ++ (1,0) ++ (3,0) circle (0.1);
\fill (5,0) ++ (1,0) ++ (1,0) ++ (1,0) ++ (1,0) ++ (3,0) ++ (1,1) circle (0.1);
\fill (5,0) ++ (1,0) ++ (1,0) ++ (1,0) ++ (1,0) ++ (3,0)  ++ (1,-1) circle (0.1);
\fill (5,0) ++ (1,0) ++ (1,0) ++ (1,0) ++ (1,0) ++ (3,0) ++ (1,0) circle (0.1);
\end{tikzpicture}
\caption{Effective potential on scale $h=0$, $\mathcal{V}^{(0)}$ splitted into the localized and renormalized contribute.}
\label{figure_renormalized_trees}
\end{figure}
Finally, plugging this splitting of $\mathcal V^{(0)}$ into the graphical representation we have given of $\mathcal{V}^{(-1)}$ in Fig. (\ref{figure_effective_potentiale_scale_0}) we get the expansion of Fig (\ref{figure_local_renormalized_trees_scale_-1}) 
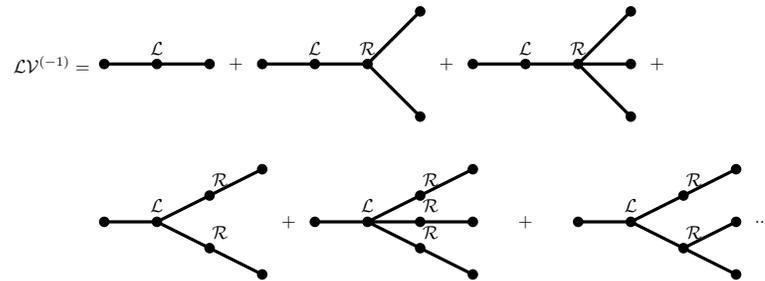
\begin{figure}
\centering
\begin{tikzpicture}
[scale=0.7, transform shape]
\node at (0,11) {$\mathcal L\mathcal V^{(-1)}=$};
\draw [very thick] (1,11) -- ++ (2, 0) ++ (1, 0) -- ++ (2,0)  --  ++ (1,1) ++ (-1,-1)  -- ++ (1, -1) ++ (1,1)  -- ++ (2,0)  -- ++ (1,1) ++ (-1,-1)  -- ++ (1,0) ++ (-1,0) -- ++ (1,-1);
\draw [very thick] (1,8) -- ++ (1,0)  --  ++ (2,1) ++ (-2,-1)  -- ++ (2, -1) ++ (1,1)  -- ++ (1,0)  -- ++ (2,1) ++ (-2,-1)  -- ++ (2,0) ++ (-2,0) -- ++ (2,-1);
\draw [very thick] (1,8) ++ (9,0) -- ++ (1,0)  -- ++ (2,1) ++ (-2,-1) -- ++ (2,-1) ++ (-1,0.5) -- ++ (1,0.5);
\foreach \i in {1,2,3,4,5,6,8,9,10,11} {%
\fill (\i, 11) circle (0.1);
}
\fill (7,12) circle (0.1);
\fill (7,10) circle (0.1);
\fill (11,12) circle (0.1);
\fill (11,10) circle (0.1);
\foreach \i in {1,2,5,6,7,8,10,11}{%
\fill (\i,8) circle (0.1);
}
\fill (4,9) circle (0.1);
\fill (4,7) circle (0.1);
\fill (3,8.5) circle (0.1);
\fill (3,7.5) circle (0.1);
\fill (8,9) circle (0.1);
\fill (8,7) circle (0.1);
\fill (7,8.5) circle (0.1);
\fill (7,7.5) circle (0.1);
\fill (13,9) circle (0.1);
\fill (13,7) circle (0.1);
\fill (12,8.5) circle (0.1);
\fill (12,7.5) circle (0.1);
\fill (13,8) circle (0.1);
\foreach \i in {2,5,9}{
\node at (\i, 11.3) {$\mathcal L$};
}
\foreach \i in {6,10}{
\node at (\i, 11.3) {$\mathcal R$};
}
\foreach \i in {2,6,11}{
\node at (\i, 8.3) {$\mathcal L$};
}
\foreach \i in {2,6,11}{
\node at (\i+1.2, 8.8) {$\mathcal R$};
\node at (\i+1.2, 7.8) {$\mathcal R$};
}
\node at (7.2,8.3) {$\mathcal R$};
\node at (3.5,11) {+};
\node at (7.5,11) {+};
\node at (11.5,11) {+};
\node at (4.5,8) {+};
\node at (9,8) {+};
\node at (13.5,8) {...};
\end{tikzpicture}
\caption{Localized part of the effective potential $\mathcal{V}^{(-1)}$. The renormalized one is exactly the same except for the first vertex following the root, wich is associated to a label $\mathcal R$.}
\label{figure_local_renormalized_trees_scale_-1}
\end{figure}
which can be described as follows:
\begin{itemize}
\item we associate with each vertex $v\in V(\tau)\setminus V_f(\tau)$ a renormalization operator $\mathcal R$ up to the very first vertex $v_0$, which can have associated either an operator $\mathcal R$ or an operator $\mathcal L$ (contributing respectively to the {\it renormalized part} or to the {\it local part} of the effective potential).
\item It is no longer true that each endpoint is at scale $h_{v_f}=1$, indeed there can be endpoints at generic scale $h_v$:
\begin{itemize}
\item $h_v<1$, means that a contribution $\mathcal L\mathcal V^{(h_v)}$ is associated to the vertex $v$,
\item $h_v=1$ means that either a contribution $\mathcal L\mathcal{V}^{(0)}$ or a contribution $\mathcal R\mathcal{V}^{(0)}$ is associated with the vertex $v$.
\end{itemize}
\item If $v$ is an endpoint on scale $h_v\leq -1$, so $h_v=h_{v'}+1$, where $v'$ is the nontrivial vertex immediately preceding $v$.
\item The running coupling constants will be denoted by the variable $\rho_v$: for instance if $h=h_{v'}$, and the contribution to the local part of the effective potential $\mathcal{L}\mathcal{V}^{(h)}$ represented by the endpoint is $F_\nu^{(\leq h)}$, we have $\rho_v=\nu_h$, and so on.
\item The Feynman diagram expansion corresponds, in this case, to a usual expansion in which each cluster value is written as a Taylor expansion $\hat W^{(h)}=\mathcal{L}\hat W^{(h)}+\left(1-\mathcal{L}\right)\hat W^{(h)}$ in such a way that the bound for the remainder $\left(1-\mathcal{L}\right)\hat W^{(h)}$ has the gain we have just discussed $\gamma^{z_v\left(h_v-h_{v'}\right)}$ where
\begin{equation}
z_v=\begin{cases}
1 \mbox{ if } n_v^e=4,\\
2 \mbox{ if } n_v^e=2,\\
0 \mbox{ else},
\end{cases}
\end{equation}
so that $n_v^e/2+m_{2,v}-2+z_v>0$.
\end{itemize}
\paragraph{Renormalized values of the clusters} Obviously the renormalization procedure we described reflects on the bounds of the kernels (values of the clusters). In particular, in the definition (\ref{effective_potential_scale_h_recursive}) we have to replace $\psi^{(\leq h)}\to \sqrt{Z_h}\psi^{(\leq h)}$ and the kernels $\hat W^{(h)}_{2n, \bm \omega}(\bm k_1', \dots,\bm k'_{2n})$ have to be computed taking into account the {\it renormalization procedure} on previous (higher) scales: we call them the {\it renormalized values of the clusters}. In particular, we can rewrite the effective potential as
\begin{equation}
\begin{split}
\mathcal V^{(h)}(\sqrt{Z_h}\psi^{(\leq h)})&=\sum_{n=1}^\infty\sum_{\tau\in\mathcal T_{h,n}}\mathcal V^{(h)}(\tau, \sqrt{Z_h}\psi^{(\leq h)}),\\
\mathcal V^{(h)}(\tau, \sqrt{Z_h}\psi^{(\leq h)})&=\int d\bm x(I_{v_0})\sum_{P_{v_0}\in I_{v_0}}\sqrt{Z_h}^{|P_{v_0}|}\tilde \psi^{(\leq h)}(P_{v_0})\mathcal W^{(h)}(\tau, P_{v_0}, \bm x(I_{v_0})),
\end{split}
\end{equation}
where the kernels
\begin{equation}
\mathcal W^{(h)}(\tau, P_{v_0},\bm x(P_{v_0}))=\int d\bm x(I_{v_0\setminus P_{v_0}})\mathcal W^{(h)}(\tau, P_{v_0}, \bm x(I_{v_0})),
\end{equation}
are the Fourier transforms of the {\it renormalized values} $\hat W^{(h)}_{2n}(\bm k_1,\dots,\bm k_{2n})$ iteratively defines as follows:
\begin{equation}
\begin{split}
\mathcal R\mathcal V^{(h)}\left(\tau,\sqrt{Z_h}\psi^{(\leq h)}\right)=\\=\int d\bm x(I_{v_0})\sum_{P_{v_0}\subset I_{v_0}}\sum_{T\in\bm T}\sum_{\alpha\in A_T}\sqrt{Z_h}^{|P_{v_0}|}\cdot \\ cdot\left[\prod_{f\in P_{v_0}}\partial^{b(f)}_{j(f)} \psi^{(\leq h)\epsilon(f)}_{\bm x(f)}(P_{v_0})\right] \mathcal RW^{(h)}_T(\tau, P_{v_0},\bm x(I_{v_0})),
\end{split}
\end{equation}
where $b(f)\in\{0,1,2\}$, $j(f)\in\{0,1\}$ and $\bm T$ is the set of tree graphs on $\bm x_{v_0}$, obtained by putting together an anchored tree graph $T_v$ for each non trivial vertex $v$. $A_T$ is the set of indices which allows to distinguish the different terms produced by the non trivial $\bm R$ operations and the iterative decomposition of the zeroes. Finally the kernels $W^{(h)}(\tau, P_{v_0},\bm x(I_{v_0}))$ have to be read as the {\it renormalized values of the clusters}:
\begin{equation}
\begin{split}
\mathcal R W^{(h)}_T(\tau, P_{v_0},\bm x(I_{v_0}))=\\=\left[\prod_{v\notin V_f(\tau)}\left(\frac{Z_{h_v}}{Z_{h_{v}-1}}\right)^{\frac{|P_v|}{2}}\right]\left[\prod_{i=1}^{n}(\bm x^i-\bm y^i)^{b(v^*_i)}_{j(v^*_i)}K^{(h_i)}_{{v^*_{i}}}(\bm x_{v^*_i}))\right]\cdot\\
\cdot \left\{\prod_{v\notin V_f(\tau)}\frac{1}{s_v!}\int dP_{T_v}(\bm t_v) \left( \det G_\alpha^{h_v, T_v}(\bm t_v)\right)\cdot\right.\\
\left.\cdot \left[\prod_{\ell \in T_v}(\bm x_\ell- \bm y_\ell)^{b(\ell)}_{j(\ell)}\partial^{q(f_\ell^1)}_{j(f_\ell^1)}\partial^{q(f_\ell^2)}_{j(f_\ell^2)} g^{(h_\ell)}_{\ell}\right]\right\}
\end{split}
\label{renormalized_kernels_explicit_expression_first_version}
\end{equation}  
where $n$ is the number of endpoints being $\tau\in\mathcal T_{h,v}$, $v_1^*,\dots, v_n^*$ are the endpoints of $\tau$, $K^{(h_i)}_{{v^*_{i}}}$ is one of the terms of the local effective potential $\mathcal L \mathcal V^{(h_i)}$, $f_\ell^1$ and $f_\ell^2$ are the labels of the two fields forming the line $\ell$, $b_\alpha(\ell), b_\alpha(v_i^*), q_\alpha(\ell), q_\alpha(v_i^*)\in\{1,2\}$, and the fact that there are as many derivatives as {\it "zeroes"} is technically expressed by the constraint $\sum_{\ell, i}\left(b_\alpha(\ell)+ b_\alpha(v_i^*)- q_\alpha(f_\ell^{(1)})- q_\alpha(f_\ell^{(2)})\right)=0$, while $(\bm x_\ell- \bm y_\ell)^{b_\alpha(\ell)}_{j_\alpha(\ell)}$ are the zeroes we introduced in the renormalization procedure definition, where $j_\alpha\in \{0,1\}$ denotes the component of the vector, and $G^{h_v,T_v}$ has to be read by interpreting
\begin{equation}
\begin{split}
G^{h_v,T_v}_{\alpha;ij,i'j'}= t_{v,i,i'}\partial_{j(f_{ij}^1)}^{q(f_{ij}^1)}\partial_{j(f_{ij}^2)}^{q(f_{ij}^2)}g_{\omega_\ell}^{h_v}(\bm x_{ij}-\bm y_{i'j'}).
\end{split}
\end{equation}
It has to be stressed that this latter expression does not break the Gram structure of the matrix, see \cite{benfatto1993beta} for details. The latter formula is a heavy but schematic representation of how the renormalization acts on the clusters.\\
We can now state the {\it main theorem}, {\it i.e.} the bounds on the renormalized expansion we introduced. We will assume some {\it a priori} bounds on the {\it running coupling constants} we use to prove the estimates on the renormalized kernels. After that, we will check that the bounds we assumed hold, and we will fix the counterterm in the initial Hamiltonian.

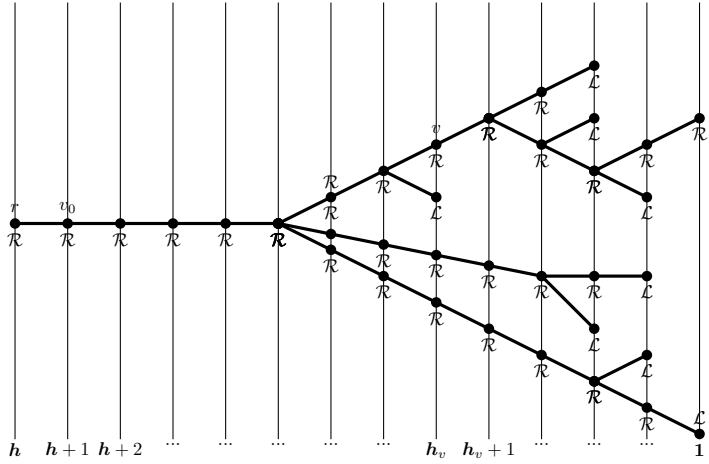
\begin{figure}
\centering
 \begin{tikzpicture} 
[scale=0.7, transform shape]
\foreach \i in {1,2,3,4,5,6,7,8,9,10,11,12,13,14} {%
\draw  (\i,2.9) -- (\i, 11.2); }
\foreach \j in {1,2,3,4,5} {%
\draw [very thick] (\j,7) -- ++ (1,0);
\fill (\j,7) circle (0.1);
\node at (\j, 6.7) {$\mathcal R$};
\fill (6,7) circle (0.1);
\node at (6, 6.7) {$\mathcal R$};
}
\foreach \j in {0,1,2,3,4,5} {%
\draw [very thick] (6+\j, 7 -\j *0.5) -- +(1,-0.5);
\fill (6+\j,7-\j*0.5) circle (0.1);
\node at (6+\j, 6.7-\j*0.5) {$\mathcal R$};}
\fill (6+6, 7-3) circle (0.1);
\node at (12, 4.7) {$\mathcal L$};
\foreach \j in {0,1,2,3} {%
\draw [very thick] (6+\j, 7 +\j *0.5) -- +(1,+0.5);
\fill (6+\j,7+\j*0.5) circle (0.1);
\node at (6+\j, 6.7+\j*0.5) {$\mathcal R$};}
\fill (6+4, 7+2) circle (0.1);
\node at (10, 8.7) {$\mathcal R$};
\foreach \j in {0,1} {%
\draw [very thick] (10+\j, 9 +\j *0.5) -- +(1,+0.5);
\fill (10+\j,9+\j*0.5) circle (0.1);
\node at (10+\j, 8.7+\j*0.5) {$\mathcal R$};}
\fill (12, 10) circle (0.1);
\node at (12, 9.7) {$\mathcal L$};
\foreach \j in {0,1,2} {%
\draw [very thick] (10+\j, 9 -\j *0.5) -- +(1,-0.5);
\fill (10+\j,9-\j*0.5) circle (0.1);
\node at (10+\j, 8.7-\j*0.5) {$\mathcal R$};}
\fill (13,7.5) circle (0.1);
\node at (13, 7.2) {$\mathcal L$};
\foreach \j in {0,1} {%
\draw [very thick] (12+\j, 8 +\j *0.5) -- +(1,+0.5);
\fill (12+\j,8+\j*0.5) circle (0.1);
\node at (12+\j, 7.7+\j*0.5) {$\mathcal R$};
}
\fill(14,9) circle (0.1);
\node at (14, 8.7) {$\mathcal R$};
\foreach \j in {0} {%
\draw [very thick] (12+\j, 4 +\j *0.5) -- +(1,+0.5);
\fill (12+\j,4+\j*0.5) circle (0.1);
\node at (12+\j, 3.7+\j*0.5) {$\mathcal R$};
}
\foreach \j in {0,1} {%
\draw [very thick] (12+\j, 4 -\j *0.5) -- +(1,-0.5);
\fill (12+\j,4-\j*0.5) circle (0.1);
\node at (12+\j, 3.7-\j*0.5) {$\mathcal R$};}
\fill (14,3) circle (0.1);
\node at (14, 3.3) {$\mathcal L$};
\fill (13,4.5) circle (0.1);
\node at (13, 4.2) {$\mathcal L$};
\draw [very thick] (8,8) -- (9, 7.5);
\fill (9,7.5) circle (0.1);
\node at (9, 7.2) {$\mathcal L$};
\draw [very thick] (11,8.5) -- (12, 9);
\fill (12,9) circle (0.1);
\node at (12, 8.7) {$\mathcal L$};
\draw [very thick] (6,7) -- (11,6);
\fill (11,6) circle (0.1);
\node at (11, 5.7) {$\mathcal R$};
\draw [very thick] (11,6) -- (12, 5);
\fill (12,5) circle (0.1);
\draw [very thick]  (11, 6) -- ++ (2,0);
\fill (13,6) circle (0.1);
\node at (13, 5.7) {$\mathcal L$};
\node at (1,2.7) {$\bm h$};
\node at (2,2.7) {$\bm h+1$};
\node at (3,2.7) {$\bm h+2$};
\foreach \i in {4,5,6,7,8} {%
\node at (\i,2.8) {...};}
\node at (9,2.7) {$\bm h_v$};
\node at (10,2.7) {$\bm h_v+1$};
\foreach \i in {11,12,13} {%
\node at (\i,2.8) {...};}
\node at (14,2.7) {$\bm 1$};
\node at (9,8.8) {$ v$};
\node at (1,7.3) {$ r$};
\node at (2,7.3) {$ v_0$};
\fill (7,6.8) circle (0.1);
\node at (7, 7.8) {$\mathcal R$};
\fill (8,6.6) circle (0.1);
\node at (8, 6.3) {$\mathcal R$};
\fill (9,6.4) circle (0.1);
\node at (9, 6.1) {$\mathcal R$};
\fill (10,6.2) circle (0.1);
\node at (10, 5.9) {$\mathcal R$};
\fill (12, 6) circle (0.1);
\node at (12, 5.7) {$\mathcal R$};
\end{tikzpicture}
\caption{Example of a renormalized tree, with $n=9$ endpoints at scales $\leq 1$.}
\label{figure_renormalized_tree}
\end{figure}

\begin{thm}[Renormalized bounds]
\label{theorem_renormalized_bounds}
For renormalized clusters, the {\it renormalized bounds}
\begin{equation}
\begin{split}
\int d\bm x(P_{v_0})\left|\mathcal{W}^{(h)}(\tau, P_{v_0},\bm x (P_{v_0}))\right|\leq C^n \gamma^{-h\left[D(P_{v_0})+z_{v_0}(P_{v_0})\right]}\\
\left(\prod_{v\notin V_f(\tau)}  \gamma^{-\left[D(P_{v})+z_{v}(P_v)\right](h_v-h_{v'})} \right)\left(\prod_{v\in V_f(\tau)\setminus V^*_f(\tau)}|\rho_v|\right)
\end{split}
\end{equation}
where $V^*_f(\tau)$ is the set of endpoints such that no running coupling constants is associated to them, $|\rho_v|\in\{|\nu_{h_v}|,|\delta_{h_v}|, |\lambda_{h_v}|\}$ while $m_{2,v}$ has already been defined as $1$ for  $\nu-$type endpoints and $0$ otherwise.
\end{thm}

\begin{corollary}
Let, $h>h_\beta$. If, for some constant $c_1>0$ these bounds are verified:
\begin{equation}
\sup_{h'>h}|\vec v_{h'}|\equiv \epsilon_h,\hspace{3mm} \sup_{h'>h}\left|\frac{Z_{h'}}{Z_{h'-1}}\right|\leq e^{c_1\epsilon_h^2},
\end{equation}
and if there exists a constant $\bar \epsilon$, depending on $c_1$, such that $\epsilon_h\leq \bar \epsilon$, then, for another suitable constant $c_0$ uniform in $c_1, L$ and $\beta$ the following bounds are true:
\begin{eqnarray}
\sum_{\tau\in\mathcal T_{h,n}}\left[|n_h(\tau)|+|z_h(\tau)|+|a_h(\tau)|+|l_h(\tau)|\right]\leq \left(c_0\epsilon_h\right)^n,
\label{bound_coupling_constants_trees_theorem_PBC}
\\
\sum_{\tau\in\mathcal T_{h,n}}\left| \tilde e_{h+1}(\tau) \right|\leq \gamma^{2h}\left(c_0\epsilon_h\right)^n,
\label{bound_tilde_E_PBC}
\\
\frac{1}{L\beta}\sum_{\tau\in\mathcal T_{h,n}} \int d\bm x(P_{v_0})\left| \mathcal R\mathcal W^{(h)}(\tau, P_{v_0},\bm x(P_{v_0})) \right|\leq \gamma^{-\left(D(P_{v_0})+z_{v_0}\right)h}\left(c_0\epsilon_h\right)^n
\label{renormalized_values_PBC}
\end{eqnarray}
\label{theorem_fundamental_PBC}
\end{corollary}

Since we already discussed the {\it non-renormalized bounds}, we comment only the differences with respect to them.
\begin{proof}
Exploiting the dimensional gains coming from the operator $\mathcal R$ acting as described in equation (\ref{renormalized_kernels_explicit_expression_first_version}), we can repeat the proof of Theorem (\ref{theorem_bound_of_kernels}) by replacing
\begin{equation}
\prod_{v\notin V_f(\tau)}\gamma^{-D(v)(h_v-h_{v'})}\rightarrow \prod_{v\notin V_f(\tau)}\left(\frac{Z_{h_v}}{Z_{h_v-1}}\right)^{|P_v|/2}\gamma^{-[D(v)+z_v](h_v-h_{v'})}
\end{equation}
By the assumption $\sup_{h'>h}Z_{h'}/Z_{h'-1}\leq e^{c_1\epsilon_h^2}\leq$, taking $c_z\epsilon_h^2\leq 1/16$, one gets that
\begin{equation}
\prod_{v\notin V_f(\tau)}(Z_{h_v}/Z_{h_v-1})^{|P_v|/2}\gamma^{-[-2+|P_v|/2+z_v]}\leq \left(\prod_{\bar v }\gamma^{-\frac{1}{40}(h_{\bar v}-h_{\bar v'})}\right)\left(\prod_{v\notin V_f(\tau)}\gamma^{-|P_v|/40}\right)
\label{bound_product_z_h/z_h-1_gamma}
\end{equation}
where $\bar v$ are the non-trivial vertices, and $\bar v'$ is the non tricial vertex immediately preceding $\bar v$. Thanks to the product into the first bracket we bound the sum over the scale labels by $(const.)^n$. The second factor can be used to bound the sums, using
\begin{equation}
\sum_{\tau\in\mathcal T_{h,n}}\sum_{P_v}\sum_T\prod_{v\notin V_f(\tau)}\frac{1}{s_v!}\gamma^{-|P_v|/40}\leq C^n,
\end{equation}
we refer to \cite{benfatto2001renormalization} for details.
\end{proof}
\begin{rem}
As expected, we have just a relevant running coupling constant: $\nu_h$, coming from the fact that each endpoint $v$ with $m_{2,v}=1$ carries a factor $\gamma^{-h_{v'}}$. To have a renormalizable power counting, we {\it hope} to kill it putting a factor $\gamma^{h_{v'}}$ in front of the corresponding running coupling constant, and the strategy is to prove that $n_h$ remains bounded if we fix in a proper way the counterterm $\nu$ in the Hamiltonian.\\
\end{rem}

\begin{rem}
The trees involved in (\ref{renormalized_values_PBC}) are the trees such that a renormalization $\mathcal R$ operator is associated to the first vertex, while the trees involved in (\ref{bound_coupling_constants_trees_theorem_PBC}) correspond to trees such that a $\mathcal L$ operation is associated to the first vertex. The bound (\ref{bound_tilde_E_PBC}) represents the bound of the constant, {\it i.e.} field independent, contribution to the effective potential.
\end{rem}

\paragraph{Short memory property}
The {\it renormalized bounds} have an important consequence: for any $0<\kappa<1$ fixed a priori, the sum over all the trees with root scale $h$ having at least a vertex such that $h_v=k>h_v$ is $O|\lambda|\gamma^{\kappa(h-k)}$: in fact what we need to prove the convergence of the expansion is $-2+|P_v|/2+z_v>0$, and we can {\it rewrite} from $\gamma^{-[-2+|P_v|/2+z_v]}=\gamma^{-\kappa[-2+|P_v|/2+z_v]}\gamma^{-(1-\kappa)[-2+|P_v|/2+z_v]}$, where $\kappa$ has to be chosen in such a way that the bounds over the sums we just described are still valid.\\
As we will see in the next Subsection, this in particular tells us that $\lambda_h$ and $\delta_h$ stay constant because their beta functions vanish.

\subsection{Flow of running coupling constants}
\label{subsection_flow_of_running_coupling_constants_PBC}

From the iterative procedure we set up in this chapter, the flow equations for the running coupling constants $\vec v_h$ ({\it i.e.} the equations linking $\vec v_h$ to $\vec v_k, k\geq h+1$) are
\begin{equation}
\begin{split}
\nu_{h-1}=\gamma \nu_h+\beta_\nu^h (\vec v_h,\dots, \vec v_0),\\
\lambda_{h-1}=\lambda_h+\beta_\lambda^h (\vec v_h,\dots, \vec v_0),\\
\delta_{h-1}=\delta_h+\beta_\nu^h (\vec v_h,\dots, \vec v_0),\\
\frac{Z_{h-1}}{Z_h}=1+\beta_z^h (\vec v_h,\dots, \vec v_0).
\end{split}
\label{running_coupling_constants_flow_PBC}
\end{equation}
The {\it a priori} bounds on the running coupling constants we assumed in Theorem (\ref{theorem_fundamental_PBC}) implies first of all that the {\it absolute summability and analyticity} of the tree expansion kernels, and also that the beta function itself (\ref{beta_function_definition}) is analytic: being the beta function defined in terms of the {\it local parts of the quadratic and quartic kernels of the effective potential $\mathcal V^{(h)}$}.\\
The analyticity of the beta function would suggest, as a natural way to study the flow of the running coupling constant, to truncate the Taylor expansion for the beta function at the lowest non trivial order, try to check whether the {\it approximate flow} verifies the hypothesis of Theorem (\ref{theorem_fundamental_PBC}) and, if so, to prove that the solution is stable under the addition of higher order Taylor approximations. For a qualitative understanding, let us consider the flow equation of $\lambda_h$, assuming that the second order Taylor approximation is non trivial:
$$\lambda_{h-1}=\lambda_h+a_h\lambda_h^2+\dots$$
Of course, the main role is played by $a_h$: if $a_h>a>0$ uniformly in $h$, the truncated flow is divergent as $h\to -\infty$ and the same would be true for the non-truncated flow (in this case, one should introduce a critical scale, below which it is no more possible to apply perturbation theory in $\lambda_h$). If $a_h\leq -a\leq 0$ uniformely in $h$, the truncated flow would be convergent $\lambda_h\to 0$ as $h\to -\infty$, and also the non-truncated flow would be convergent: in this case, we would talk of {\it asymptotic freedom in RG sense}.\\
The fact that the ststem we are studying (\ref{hamiltonian_PBC}) belongs to the Luttinger universality class means that this system realized an intermediate {\it scenario}: one can check that, asymptotically for $h\to -\infty$, $a_h\to 0$, meaning that the truncated flow equation remains analitically close to the initial datum $\lambda_0$ uniformely in $h$. The problem, in this case, is the instability of the truncated flow, so one must show that similar cancellations take place at all orders in perturbation theory. Of course it is a {\it non trivial} and actually {\it very hard} problem, so it is necessary to use some {\it deep argument}, being direct computations not enough.\\
The strategy relies on the fact that the model described by the Hamiltonian $H$ (\ref{hamiltonian_PBC}) is, in a RG sense, {\it close to} the {\it Luttinger model}, which verifies a bunch of symmetries which are not verified by the {\it not-solvable} model we are dealing with (as discussed in the introduction). \\
The idea is to keep as a {\it reference model} the Luttinger model, being able to quantify in a rigorous way this {\it closeness} getting rigorous estimates on the size of the corrections. The first technical step is to recognize, that it is possible to rewrite the propagators $g_\omega^{(i.r.)}$ (and all the single scale propagators $g_\omega^{(h)}$) as the propagator of the {\it infrared Luttinger model} and a remainder.

\begin{lem}
\label{lemma_propagator_luttinger_+_remainder_PBC}
The propagator $g^{(h)}_\omega(\bm x-\bm y)$ in (\ref{dressed_propagator_scale_h_PBC}) can be rewritten as
\begin{equation}
g^{(h)}_\omega(\bm x-\bm y)=g^{(h)}_{0;\omega}(\bm x-\bm y)+C^{(h)}_\omega(\bm x-\bm y),
\label{propagator_as_luttinger_plus_remainder}
\end{equation}
where $C^{(h)}_\omega$ is the remainder of the {\it linear approximation}
\begin{equation}
g^{(h)}_{0;\omega}(\bm x-\bm y)=\frac{1}{L\beta}\sum_{\bm k'\in\mathcal{D}_{L,\beta}^\omega}e^{i\bm k'(\bm x-\bm y)}\frac{\tilde{f}_h(\bm k')}{-ik_0+\omega v_0 k'},
\end{equation}
such that, for any integer $N>1$ we have
\begin{equation}
\left|g_{0;\omega}^{(h)}(\bm x-\bm y)\right|\leq \frac{\gamma^hC_N}{1+(\gamma^h\left|\bm x-\bm y\right|)^N},
\end{equation}
and, with the further assumption $|x-y|\leq L/2$ ans $|x_0-y_0|\leq \beta/2$, we can bound the remainder as
\begin{equation}
|C_\omega^{(h)}(\bm x-\bm y)|\leq \frac{\gamma^{2h}C_N}{1+\left(\gamma^h\left|\bm x-\bm y\right|\right)^N}.
\end{equation}
\end{lem}

An immediate consequence of the latter Lemma is that, on scale $h$, any observable can be naturally decomposed as the sum of a dominant part, expressed in terms of Gallavotti-Nicolò trees whose values is computed considering all the single-scale propagators as $g_{0,\omega}^{(h)}$, and a remainder, that can be written as a sum of trees "containing" at least a propagator $C_\omega^{(h)}$. In particular, we group the running coupling constants of the {\it infrared Luttinger model} into the two components vector 
\begin{equation}
\mu_h=(\lambda_h,\delta_h),
\label{running_coupling_constants_infrared_luttinger_model}
\end{equation}
in order to split the $\beta$-functions $\beta_i^{(h)}$ into a {\it Luttinger model part} plus a {\it remainder}: Lemma (\ref{lemma_propagator_luttinger_+_remainder_PBC}) allows us to rewrite the $\beta$-functions as the Luttinger's ones {\it plus} a remainder as follows: to get the {\it Luttinger model beta function}, first of all we split
\begin{equation}
\beta_i^{(h)}(\mu_h,\nu_h;\dots;\mu_0,\nu_0)=\bar\beta_i^{(h)}(\mu_h;\dots;\mu_1)+\hat  \beta_i^{(h)}(\mu_h,\nu_h;\dots;\mu_1,\nu_1),
\end{equation}  
where $i=\mu,\nu$ where the first term in the right hand side is obtained by putting $\nu_k=0$, $k\geq h$, and then we estract from the first term of the right hand side the Luttinger model $\beta$-function:
\begin{equation}
\bar \beta_i^{h}(\mu_h;\dots;\mu_0)=\hat \beta_i^{h,l}(\mu_h;\dots;\mu_0)+\hat \beta_i^{h,nl}(\mu_h;\dots;\mu_0),
\end{equation}
where the further labels $l$ and $nl$ we introduced mean trivially Luttinger and non-Luttinger, and the first one is obtained simply considering each propagator as the Luttinger one $g_{0,\omega}^{(h)}(\bm x-\bm y)$ so that the $\beta$-function coincides exactly with the $\beta$-function of the infrared Luttinger model.\\
The universal part $\hat \beta_i^{h,l}$ has been studied in deep detail in several papers, so we do not give the complicated details, we refer to \cite{benfatto2005ward}, but we recall here the main result of this paper: the so called {\it asymptotic vanishing of the beta function} (\cite{benfatto2005ward}, Theorem 2 and formula (57)).
\begin{prop}
Let $\mu_h:=(\lambda_h,\delta_h)$ and $|\mu_h|$ small enough. Then
\begin{equation}
|\hat \beta_i^{h,l}(\mu_h,\dots,\mu_h)|\leq C_\alpha |\lambda_h|^2\gamma^{\eta h},
\end{equation}
for $0<\alpha<1$ and a suitable $C_\alpha>0$.
\end{prop}
Finally, we can state the following
\begin{thm}\label{theorem_lambda_nu_solutions}
If $|\lambda |\leq \lambda_0$ with $\lambda_0$ small enough, we can fix once for all a counterterm $\nu^*(\lambda)=:\nu_1$, analytic in $\lambda$, such that the running coupling constants $\{\lambda_h,\nu_h\}_{\leq 1}$, verify $|\nu_h|\leq c |\lambda|\gamma^{(\theta/2)h}$ and $|\lambda_h|\leq c|\lambda|$. Moreover, 
$$z_h\leq 1/2\hspace{3mm} \mbox{ and } \hspace{3mm}
e^{-c|\lambda|^2}\leq\left| \frac{Z_h}{Z_{h-1}}\right|\leq e^{c|\lambda|^2}.
$$
\end{thm}

Before proving the theorem, it is worth commenting this result: this tells us that the running coupling constants $\lambda_h$ and $\delta_h$ stay asymptotically constant, provided we fix an initial datum for $\{\nu_h\}_{h\leq 1}$ such that $\nu_h\to 0$ as $h\to -\infty$ exponentially fast and $\lambda_h, \delta_h$ do not exceed $\epsilon$: in fact we use the freedom of changing the chemical potential {\it correction} $\nu$ to make sure that this happens. Finally, the vanishing of the beta function tells us that the sequence of running coupling constants $\vec v_h=(\nu_h,\delta_h\lambda_h)$ exists and converges exponentially fast to $\vec v_{-\infty}=(0,\delta_{-\infty},\lambda_{-\infty})$, where in particular $\delta_{-\infty},\lambda_{-\infty}$ are analytic in $\lambda$ if $\lambda$ is small enough.

\begin{proof}[Proof of theorem \ref{theorem_lambda_nu_solutions}]
Let us consider the Banach space $\mathcal B_\theta$ of {\it real sequences} $\underline \nu=\{\nu_h\}_{h\leq 1}$ with the norm $||\cdot ||_\theta$ defined by
\begin{equation}
||\underline \nu ||_\theta :=\sup_{k\leq 1}|\nu_k|\gamma^{-k \theta/2}.
\end{equation}
 Actually, we are interested in a {\it closed ball}, so let us consider the ball
\begin{equation}
\mathcal{M}_\theta:=\{\underline \nu=\{\nu_h\}_{h\leq 1}: |\nu_h|\leq c|\lambda|\gamma^{\theta/2}\}.
\end{equation}
The strategy of the proof is the following:
\begin{enumerate}
\item we show that for any $\underline \nu\in\mathcal M_\theta$, both the flow equation for $\nu_h$ and the property $|\lambda_h(\nu)|\leq c|\lambda|$ for some $c>0$ are verified uniformely in $\underline \nu$,
\item we fix the counterterm $\underline \nu\in\mathcal M_\theta$ via an exponentially convergent iterative procedure in such a way that the flow equation for $\nu_h$ is verified.
\item finally, we solve the flow of $Z_h$.
\end{enumerate}
So let us start:
\begin{enumerate}
\item given $\underline \nu\in\mathcal M_\theta$, let us iteratively suppose
\begin{equation}
||\lambda_{k-1}(\underline \nu)-\lambda_k(\underline \nu)||\leq c_0 |\lambda|^2\gamma^{(\theta/2)k}, \mbox{ for } c_0>0, k > h+1.
\end{equation}
First of all, it is true for $h=1$ and, besides, if it is true for any $k>h$, it implies $|\lambda_k|\leq c|\lambda|$.\\
Looking at the flow equation for $\lambda_h$ and the comments about the beta function written as Luttinger's one {\it plus a remainder}, we can further write,
\begin{equation}
\begin{split}
\beta^h_\lambda(\lambda_h,\nu_h; \dots ; \lambda_1, \nu_1)=\\
=\beta^{h,l}_\lambda(\lambda_h,\dots,\lambda_h)+\sum_{k=h+1}^1 D_\lambda^{h,k}+\beta^{h,nl}_\lambda(\lambda_h,\dots,\lambda_1)+\sum_{k\geq h}\nu_k\tilde \beta_\lambda^{h,k}(\lambda_k,\nu_k;\dots;\lambda_1,\nu_1),
\end{split}
\end{equation}
where
\begin{equation}
\begin{split}
|\beta^{h,l}_\lambda|\leq c|\lambda|^2\gamma^{\theta h},\hspace{3mm} |D_{\lambda}^{h,k}|\leq c|\lambda|\gamma^{\theta(h-k)}|\lambda_k-\lambda_h|,\\
|\beta^{h,nl}_\lambda|\leq c|\lambda|^2\gamma^{(\theta/2)h}, \hspace{3mm}  |\tilde \beta_\lambda^{h,k}|\leq c|\lambda|\gamma^{\theta(h-k)}
\end{split}
\end{equation}
It worths remarking that we have the first of these inequalities by the assumption of the vanishing of the Luttinger beta funtion. So
\begin{equation}
\begin{split}
|\lambda_h(\underline \nu)- \lambda_{h+1}(\underline \nu)|\leq c|\lambda^2|\gamma^{\theta(h+1)}+\sum_{k\geq h+2}c|\lambda\gamma^{\theta(h+1-k)}|\sum_{k'=h+2}^k c_0 |\lambda^2|\gamma^{(\theta/2)k'}+\\
c|\lambda|^2\gamma^{(\theta/2)(h+1)}+\sum_{k\geq h+1}c^2|\lambda|^2\gamma^{(\theta/2)k}\gamma^{(\theta(h+1-k))}\leq c_0|\lambda|^2\gamma^{(\theta/2)h}.
\end{split}
\end{equation}
for some $c_0$ large enough. Thanks to the iterative assumption, we get also
\begin{equation}
|\lambda_h(\underline \nu)-\lambda_1(\underline \nu)|\leq c_0|\lambda|^2
\end{equation}
Now, we are left with proving that $\lambda(\underline \nu)$ is a continuous function of $\underline \nu\in\mathcal M_\theta$:
\begin{equation}
\begin{split}
\lambda_h(\underline\nu)-\lambda_k(\underline\nu')=\lambda_1(\underline\nu)-\lambda_1(\underline\nu')+\\+\sum_{h+1\leq k\leq 1}\left[\beta_\lambda^k(\lambda_k(\underline \nu),\nu_k;\dots;\lambda_1\underline \nu),\nu_1)-\beta_\lambda^k(\lambda_k(\underline \nu'),\nu'_k;\dots;\lambda_1\underline \nu'),\nu'_1)\right]
\end{split}
\end{equation}
First of all, we have $|\lambda_1(\underline\nu)-\lambda_1(\underline\nu')|\leq c_0|\lambda||\nu_1-\nu'_1|$. Furthermore, defining $||\underline \nu||_0=\sup_{h\leq 1}|\nu_h|$, if we assume that inductively $|\lambda_k(\underline \nu-)\lambda_k(\underline \nu')|\leq 2c_0 |\lambda| ||\underline \nu-\underline \nu'||_0$, we find (usinge the same decomposition strategy as before) that
\begin{equation}
\begin{split}
|\lambda_h(\underline \nu)-\lambda_h(\underline \nu')|\leq c|\lambda| |\nu_1-\nu'_1|+\\+c|\lambda|\sum_{k\geq h+1}\gamma^{(\theta/2)k}\sum_{k'\geq k}\gamma^{\theta(k-k')}\left(2c_0 |\lambda| ||\underline \nu-\underline \nu'||_0+|\nu_k-\nu'k|\right).
\end{split}
\end{equation}
So, we can choose $c_0$ in such a way that
\begin{equation}
|\lambda_h(\underline \nu)-\lambda_h(\underline \nu')|\leq c|\lambda|||\underline \nu-\underline \nu'||_0.
\end{equation}
\item In order to fix the counterterm, we use a {\it fixed point argument}: indeed we will look
at the recursive relation for $\nu_h$ as the result of the action of an operator acting on a Banach space, and we will prove that the operator {\it generatinf the flow} is a {\it contraction} in this space, so there exists a unique fixed point which will be properly the counterterm.\\
Let us consider the Banach space $\mathcal B_\theta$ of {\it real sequences} $\underline \nu=\{\nu_h\}_{h\leq 1}$ with the norm $||\cdot ||_\theta$ defined by
\begin{equation}
||\underline \nu ||_\theta :=\sup_{k\leq 1}|\nu_k|\gamma^{-k \theta/2}.
\end{equation}
In the context of Banach spaces, we can apply the {\it fixed point theorem for contractions}. Actually, we are interested in a {\it closed ball}, but by the closeness the fixed point argument is valid within it (of course by definition of closeness). So let us consider the ball
\begin{equation}
\mathcal{M}_\theta:=\{\underline \nu=\{\nu_h\}_{h\leq 1}: |\nu_h|\leq c|\lambda|\gamma^{\theta/2}\},
\end{equation}
we will  fix $\underline \nu$ via an exponentially convergent iterative procedure in such a way that the flow equation for $\nu_h$ is satisfied.\\
Let us start from the recursive relation
$$\nu_{h-1}=\gamma \nu_h+\beta_\nu^h(\vec v_h;\dots;\vec v_0)$$
which can be iterated until $h=1$, getting 
$$\nu_{h-1}=\gamma^{2-h}\nu_1+\sum_{k=h}^{1}\gamma^{k-h}\beta_\nu^k(\vec v_k;\dots;\vec v_0),$$ 
meaning that
$$\nu_1=\gamma^{h-2}\nu_{h-1}+\sum_{k=0}^{1-h}\gamma^{k-2}\beta_\nu^k(\vec v_k;\dots;\vec v_0).$$
The latter equation, since we are trying to fix $\underline \nu$ in such a way that $\nu_{-\infty}=0$, can be read as:
\begin{equation}
\nu_1=-\sum_{k=-\infty}^1\gamma^{k-2}\beta_\nu^k(\vec v_k;\dots;\vec v_1),
\end{equation}
from which we should get
\begin{equation}
\nu_h=-\sum_{k\leq h}\gamma^{k-h-1} \beta_\nu^k(\vec v_k;\dots; \vec v_1).
\end{equation}
In order to look at this equation from a {\it fixed point theorem} point of view, let us introduce the operator $\bm T:\mathcal M_\theta\to \mathcal M_\theta$ defined as 
\begin{equation}
\left(\bm T \underline \nu\right)_h=\nu_h=-\sum_{k\leq h}\gamma^{k-h-1} \beta_\nu^k(\vec v_k(\underline \nu);\dots; \vec v_1(\underline \nu)),
\end{equation}
where $\vec v_k(\underline \nu)$ is the vector solution of the equations (\ref{running_coupling_constants_flow_PBC}) as functions of the {\it parameter} $\underline{\nu}$. In this way, we have translated our problem into a fixed point problem for this operator.\\
First of all, we check that the operator is well defined, meaning that it really sends $\mathcal M_\theta$ into itself: thanks to parity cancellations, it is true that 
\begin{equation}
\beta_\nu^h(\vec v_h;\dots;\vec v_1)=\beta_{\nu,1}^h(\mu_h;\dots;\mu_1)+\sum_k\nu_k \tilde{\beta}_{\nu}^{h,k}(\mu_h,\nu_h;\dots;\mu_1\nu_1)
\end{equation}
with 
\begin{equation}
|\beta_{\nu,1}^h|\leq c_1|\lambda|\gamma^{\theta h},\hspace{3mm}|\tilde \beta_{\nu}^{h,k}|\leq c_2 |\lambda|\gamma^{\theta(h-k)}
\end{equation}
where $c_1, c_2$ are suitable constants greater then zero. If we fix $c=2c_1$, we get
\begin{equation}
|(\bm T\underline\nu)_h|\leq \sum_{k\leq h}2c_1|\lambda|\gamma^{k(\theta/2+1)-h}\leq c|\lambda|\gamma^{h\theta/2}.
\end{equation}
Finally, we check that it is actually a contraction: $||(\bm T\underline \nu)-(\bm T \underline \nu')||_\theta\leq c''|\lambda|||\underline \nu-\underline \nu'||_\theta$, indeed
\begin{equation}
\begin{split}
|(\bm T\underline \nu)_h-(\bm T\underline \nu')_h|\leq \sum_{k\leq h}\gamma^{k-h-1}|\beta_\nu^{k}(\vec v_k;\dots;\vec v_1)-\beta_\nu'^{k}(\vec v'_k;\dots;\vec v'_1)|\leq\\
\leq c \sum_{k\leq h}\gamma^{k-h-1}\left[ \gamma^{\theta k}|\lambda_k'(\underline \nu)-\lambda_k'(\underline{\nu}')|+\sum_{k'=k}^1\gamma^{\theta(k-k')}|\lambda||\nu_{k'}-\nu'_{k'}|\right]\leq\\
\leq  c \sum_{k\leq h}\gamma^{k-h-1}\left[ |k|\gamma^{\theta k}|\lambda| ||\underline \nu-\underline{\nu}'||_0+\sum_{k'=k}^1\gamma^{\theta(k-k')}|\lambda|\gamma^{k'\theta/2}||\underline \nu-\underline \nu'|\right]\leq\\
\leq c''|\lambda|\gamma^{h\theta/2}||\underline \nu-\underline \nu '||_\theta.
\end{split}
\end{equation}
So $\bm T$ is a contraction, and there exists a unique ficed point $\underline{\nu}^*$ for $\bm T$ in the closed ball $\mathcal M_\theta$.
\item Now we can use the previous results to claim that there exist two $O(\lambda)$ functions $\eta_z, F_\zeta^h$ such that
\begin{equation}
Z_h=\gamma^{\eta_z(h-1)+F_\zeta^h}.
\end{equation}
Indeed, knowing that $|z_h|\leq c|\lambda|^2$ uniformly in $h$, we can define 
$$\gamma^{-\eta_z}:=\lim_{h\to -\infty}1+z_h,$$
so that 
\begin{equation}
\log_\gamma Z_h=\sum_{k\geq h+1}\log_\gamma\left(1+z_k\right)=\eta_z(h-1)+\sum_{k\geq h+1}r_\zeta^h, \hspace{3mm} r_\zeta^k:=\log_\gamma\left(1+\frac{z_k-z_{-\infty}}{1+z_{-\infty}}\right).
\end{equation}
Now, knowing that $z_{k-1}-z_k$ is either proportional to $\lambda_{k-1}-\lambda_k$ or to $\nu_{k-1}-\nu_k$, we can bound
\begin{equation}
|r_\zeta^k|\leq c \sum_{k'\leq k}|z_{k'-1}-z_{k'}|\leq c|\lambda|^2\gamma^{(\theta/2)k}.
\end{equation}
Finally, if we define $F_\zeta^h:=\sum_{k\geq h+1}r_\zeta^k$ and $F_\zeta^1=0$, then $F_\zeta^h=O(\lambda)$ and $Z_h=\gamma^{\eta_z(h-1)+F_\zeta^h}$.
\end{enumerate}
\end{proof}

\begin{rem}
In light of that, in Corollary (\ref{theorem_bounds_kernels}), we can replace the assumption $Z_h/Z_{h-1}\leq e^{c_1\epsilon^2}$ by $Z_h\simeq A\gamma^{-h\eta}$ for some suitable $A>0$, where $\eta=a\lambda^2+\mathcal O(\lambda^3)$ and the symbol $\simeq$ means that the equivalence is asymptotically true for $h\to -\infty$, and improve the bounds we got in this chapter.
\end{rem}

\chapter{Interacting Fermions on the half line}
\label{chapter_Interacting_fermions_on_the_half_line}
\section{The model}
\label{the_model_DBB}
\subsection{Definition and main result}
We are interested in constructing the ground state of interacting spinless fermions living in a discrete one-dimensional box of mesh size $a=1$ and volume $L\gg 1$ with {\it open boundary conditions}, meaning that the system is defined on a segment instead of on a torus. \\ 
Let $\mathcal F=\oplus_{n=0}^\infty H^{\wedge n}$ be the standard {\it antisymmetric Fock space}, where $\wedge$ denotes the antisymmetric tensor product, and let  $\psi^\pm_x$ be the {\it fermionic creation and annihilation} operators defined on $\mathcal F$, where $x$ is the space coordinate and $\Lambda:=\left\{x\in\mathbb Z: 1\leq x\leq L\right\}$,  $L\in \mathbb N$. Let us define the Hamiltonian
\begin{equation}
H=H_0+\lambda V+ \varpi \mathcal N,
\end{equation}
where
\begin{equation}
\begin{split}
H_0&=T_0-\mu_0 N_0,\\
T_0&=\sum_{x\in\Lambda}\psi^+_x\left(-\Delta^d  \psi^-_x\right)=\sum_{x\in\Lambda}\frac{1}{2}\left(-\psi^+_{x+1}\psi^-_x-\psi^+_{x-1}\psi^-_x+2 \psi^+_x\psi^-_x\right),\\
N_0&=\sum_{x\in \Lambda}\psi^+_x\psi^-_x,
\end{split}
\end{equation}
where, in the formula of $T_0$, we have to interpret $\psi^{\pm}_0=\psi^{\pm}_{L+1}=0$, $\mu_0$ is the chemical potential choosen in such a way that, if we call $\sigma(T_0):=[e_-, e_+]$ the spectral band of the kinetic operator, $\mu\in [e_-+\kappa, e_+-\kappa]$ for some $\kappa>0$ fixed once for all, the interaction of {\it strenght} $\lambda$ is 
\begin{equation}
V=\sum_{x,y \in\Lambda}\psi^+_x\psi^-_x v(x,y)\psi^+_y\psi^-_y,\
\end{equation}
where $v(x,y)=v(y,x)$ is a real, compactly supported function, and satisfies what we call {\it Dirichlet property}, {\it i.e.} it can be written as
\begin{equation}
v(x,y)=\frac{2}{L+1}\sum_{k \in \mathcal{D}^d_{\Lambda}}\sin(kx)\sin(ky)\hat v(k),
\label{potential_v_DBC}
\end{equation}
where $\mathcal D_\Lambda^d=:\left\{k=\frac{n\pi}{L+1}, n=1,\dots,L\right\}$. We stress that the {\it Dirichlet property} of $v( x, y)$ (\ref{potential_v_DBC}) is not crucial at all but it simplifies some technical aspects of the proof.\\
 Finally, the {\it boundary defect} of size $\varpi = \mathcal O(\lambda)$ is
\begin{equation}
\mathcal N =\sum_{x,y\in \Lambda}\psi^+_x\psi^-_y \pi(x,y),
\end{equation}
where $\pi(x,y)$ is a Hermitian matrix such that $\sup_{x\in\Lambda}\sum_{y\in\Lambda}|\pi(x,y)|=1$.\\
We recall here the main result we prove in this section: let $\beta\geq 0$ be the {\it inverse temperature} and let
\begin{equation}
f_{\Lambda,\beta}=-\frac{1}{|\Lambda|\beta}\log \left(Tr \left(e^{-\beta H}\right)\right)
\end{equation}
be the {\it finite volume specific free energy}. Let also
\begin{equation}
f_{\Lambda}=-\frac{1}{|\Lambda|}\lim_{\beta\nearrow \infty}\frac{1}{\beta}\log \left(Tr \left(e^{-\beta H}\right)\right),\hspace{3mm} f_{\infty}=-\lim_{|\Lambda|\nearrow \infty}\frac{1}{|\Lambda|}\lim_{\beta\nearrow \infty} \frac{1}{\beta}\log \left(Tr \left( e^{-\beta H}\right)\right);
\end{equation}
we prove the following result.
\begin{thm}
\label{theorem_main_DBC}
In this framework, there exists a radius $\lambda_0>0$ such that, for any $|\lambda|\leq \lambda_0$ it is possible to fix the {\it boundary defect} $\pi(x,y)$ and its strenght $\varpi=\varpi(\lambda)$ in such a way that, for any $\theta\in (0,1)$, there exists a constant $C_\theta$ such that 
\begin{equation}
\sum_{y\in\Lambda} \left|\pi(x,y)\right| \leq C_\theta \left(\frac{1}{\left(1+|x|\right)^\theta}+\frac{1}{\left(1+|L-x|\right)^\theta}\right),
\end{equation}
and in such a way that $f_\Lambda$ admits a convergent expansion in $\lambda$ and $\varpi$.\\
Moreover
\begin{equation}
\left| f_\Lambda-f_\infty \right|\leq |\lambda|\frac{C_\theta}{L^\theta}.
\end{equation}
\end{thm}
Even though it is not explicitly investigated in this thesis, we stress that a straightforward extension of the proof of this theorem would allow one to construct the correlation functions of the Hamiltonian and to control their boundary corrections.


\subsection{Free Hamiltonian diagonalization and free propagator}
\label{section_the_non_interacting_system_DBC}

It is well known that the Laplacian problem with DBC is {\it diagonalized} by a {\it sine Fourier transform}. Indeed, if we introduce the transformation:
\begin{equation}
\hat \psi^\pm_k=\sum_{x\in\Lambda}\sin (kx)\psi^\pm_x, \hspace{5mm} \psi^\pm_x=\frac{2}{L+1}\sum_{k\in \mathcal D^d_\Lambda}\sin (kx)\hat\psi^\pm_k,
\label{sine_fourier_transform_dbc}
\end{equation}
where $\hat \psi_k^\pm$ creates and annihilates a spinless electron with momentum $k$, the Hamiltonian $H_0$ can be written as a diagonal matrix in the {\it dual space}:
\begin{equation}
H_0=\frac{2}{L+1}\sum_{k\in\mathcal D^d_\Lambda}\hat \psi^+_ke(k)\hat \psi_k^-,
\end{equation}

where $e(k)$ is the dispersion relation:
\begin{eqnarray}
e(k)=1-\cos k -\mu_0,\\
\label{dispersion_relation_DBC}
\mathcal D^d_\Lambda=\left\{k=\frac{n\pi}{L+1}, n\in \mathbb{Z}, n=1,\dots, L\right\}.
\label{momentum_space_DBC}
\end{eqnarray}
In particular, we choose $\mu_0$ in such a way that there exists $p_F\in\mathcal D^d_{\Lambda,\beta}$ such that $e(p_F)=\mathcal O(1/L)$.
\begin{rem}
\label{remark_quasi_particles_issue_DBC}
When we perform the limit $L\to \infty$, $\mathcal{D}^d_\Lambda\to [0,\pi]$ and of course $e(k)$, which is a cosine up to a constant, becomes a function defined in a semi-period of the cosine: $e(\cdot):[0,\pi]\to [-\mu_0, 2-\mu_0]$. This means  that there is a unique point of the domain, and we call it $p_F\in [0,\pi]$, such that $e(p_F)=0$. In light of the previous chapter, it is clear that the zeros of the dispersion relation are fundamental because they correspond to the singularities at zero temperature of the propagator and, since the interesting physics happens near the Fermi points, we are interested in the excitations around these Fermi points.\\
We stress that, while in the translation invariant system there are two symmetric Fermi points $\pm p_F$ and we introduced two different quasi-particles $\{\hat \psi^\pm_\omega\}_{ \omega=\pm}$, in this case the theory naturally suggests the definition of a unique quasi-particle around the unique Fermi point $p_F$.
\end{rem}

\paragraph{Schwinger functions and Free Propagator}
Let $x_0\in[0,\beta)$ be the {\it imaginary time}, let $\bm x=(x,x_0)\in\Lambda\times [0,\beta)$ and let us consider the {\it time-evolved operator} $\psi^{\pm}_{\bm x}=e^{Hx_0}\psi^\pm_x e^{-Hx_0}$. So we can define the $m$-point Schwinger function at finite temperature $T=\beta^{-1}$ as
\begin{equation}
S_{\Lambda,\beta}(\bm x_1, \epsilon_1;\dots;\bm x_m, \epsilon_m):=\left< \psi^{\epsilon_1}(\bm x_1)\dots \psi^{\epsilon_m}(\bm x_m)\right>_{\Lambda,\beta}:=\frac{Tr\left( e^{-\beta H}\bm T \left( \psi^{\epsilon_1}(\bm x_1)\dots \psi^{\epsilon_m}(\bm x_m)\right)\right)}{Tr\left( e^{-\beta H}\right)},
\label{schwinger_function_n_points_DBC}
\end{equation}
where $\epsilon_i\in \{\pm\}$ for $i=1,\dots, m$ and $\bm T$ is the {\it Fermionic time ordering operator}, and where we have introduced a collection $\left\{t_1,\dots,t_m\right\}$ of {\it time variables} such that $t_i\in \left[0,\beta\right)$ $\forall i=1,\dots,m$. The strategy we follow is the same as the previous chapter: we want to derive {\it convergent expansions} for  $f_{\Lambda,\beta}$, uniformly in the volume $|\Lambda|$ and in the inverse temperature $\beta$, and then to take the infinite volume and zero temperature limit $|\Lambda|,\beta\to \infty$ (thermodinamic limit in in the statistical mechanics point of view). In particular, we want to keep track of the {\it finite volume boundary corrections}.
\paragraph{Free Propagator}

The non interacting model described by the Hamiltonian $H_0$ is exactly solvable, meaning that all the Schwinger functions can be exactly computed by simply using the anticommutative (fermionic) {\it Wick rule}, starting from the {\it two point Schwinger function}, {\it i.e}. the {\it propagator}:
\begin{equation}
\begin{split}
\left<\bm T\left(\psi^{-}_{\bm x_1}\dots \psi^{+}_{\bm x_m}\right)\right>_{0,\Lambda,\beta} =\det G ,\\
G_{ij}=\left<\bm T\left(\psi^-_{\bm x_i}\psi^+_{\bm x_j}\right)\right>_{0,\Lambda,\beta}=S^0_{L,\beta}(\bm x, -;\bm y, +),
\end{split}
\end{equation}
where the subscript $0$ means that the expectation value is calculated with respect to the {\it free measure}, and in the first line there are as many creation as annhilation operators. We stress that every $n$-point Schwinger function with $\sum_{i=1}^n\epsilon_i\neq 0$ is identically zero. \\
We do not repeat the discussion of the {\it two point Schwinger function}, which is exaclty the same as the previous chapter (\ref{subsection_free_propagator}), with the only difference that
\begin{equation}
\psi^{\pm}_x=\frac{2}{L+1}\sum_{k\in\mathcal D^d_{\Lambda}}\sin(kx) \hat \psi^{\pm}_{k}.
\end{equation}
So if we use the notation $\bm x=(x,x_0), \bm y=(y,y_0)\in\Lambda\times [0,\beta)$ and $\bm k=(k,k_0)\in \mathcal D_L^d\times \mathcal D_{\beta,M}=:\mathcal D^d_{\Lambda,\beta,M}$, where $\mathcal{D}_{\beta,M}$ has already been defined in (\ref{momenta_space_time}) and $\mathcal D_L^d$ in (\ref{momentum_space_DBC}),
\begin{equation}
\begin{split}
S^0_{L,\beta}(\bm x,-;\bm y,+):= g(\bm x,\bm y)=\\=\frac{2}{\beta(L+1)}\lim_{M\to \infty}\sum_{\bm k\in\mathcal D^d_{\Lambda,\beta,M}}e^{i\delta_Mk_0}e^{-ik_0(x_0-y_0)}\sin(kx)\sin(ky)\hat g(\bm k)
\end{split}
\label{free_propagator_DBC}
\end{equation}
where 
\begin{equation}
\hat g(\bm k):=\frac{1}{-ik_0+e(k)}, \hspace{5mm}\bm k\in\mathcal{D}_{\Lambda,\beta}^d
\label{propagator_momentum_DBC}
\end{equation}
is the same function as the translation invariant case, but the domain changes as already commented: $\hat g$ is singular only in $\bm p_F=(p_F,0)$.\\
As in the previous chapter, the constant $\delta_M=\beta/\sqrt{M}$ is introduced in order to take correctly into account the discontinuity of the propagator $g(\bm x,\bm y)$ at $\bm x=\bm y$, where it has to be defined as $\lim_{x_0-y_0\to 0^-}g(x,x;x_0-y_0)$, in fact the latter definition guarantees that $\lim_{M\to \infty}g_M(\bm x,\bm y):=g(\bm x,\bm y)$ for $\bm x\neq\bm y$, while $\lim_{M\to \infty}g_M(\bm x,	\bm x):=g(x,x;0^-)$ at equal points.\\
As we already commented in Remark (\ref{remark_quasi_particles_issue_DBC}), $\hat g(\bm k)$ is singular when $\bm k=(p_F, 0)$. Since the introduction of an interaction between  the fermions could move this singularity, it is convenient to rewrite
$$\mu_0=\mu+\nu,$$
where $\nu$ is a counterterm which will be eventually suitably chosen in order to fix the position of the singularity
at some interaction-indipendent point.
\paragraph{Symmetries and Fermi points}
\begin{lem}[Reflection rule]
$\forall$ $\bm x,\bm y \in \Lambda\times\left[0,\beta\right)$
\begin{equation}
g(\bm x,\bm y)=g_{2(L+1)}(x-y, x_0-y_0)-g_{2(L+1)}(x+y, x_0-y_0),
\end{equation}
where $g_{2(L+1)}$ is the free propagator of a system described by a hopping Hamiltonian $H_0$ defined on a box of size $2(L+1)$ with {\it periodic boundary conditions, i.e. }
\begin{equation} 
g_{2(L+1)}(x, x_0):= \frac{1}{\beta 2(L+1)}\lim_{M\to \infty}\sum_{k_0\in\mathcal{D}_{\beta,M}}\sum_{k\in\mathcal{D}_{2(L+1)}}e^{-i\bm k\cdot \bm x}\hat g(\bm k),
\end{equation}
where $\mathcal{D}_{2(L+1)}:=\left\{k=\frac{n\pi}{L+1}, n=-(L+1),\dots,L\right\}$ and $\mathcal D_{\beta,M}$ has already been defined in (\ref{momenta_space_time}).
\label{lemma_reflection_trick}
\end{lem}

\begin{proof}
Let us note first of all that 
\begin{equation}
\hat{g}(-k,k_0)=\hat g(k,k_0).
\label{hat_g(k)_parity} 
\end{equation}
So (\ref{free_propagator_DBC}) is
\begin{equation}
\begin{split}
g(\bm x,\bm y)=\frac{2}{\beta (L+1)}\lim_{M\to \infty}\sum_{k \in\mathcal{D}^d_{\Lambda,\beta,M}}e^{-ik_0(x_0-y_0)}\sin(kx)\sin(k y)\hat g(\bm k)=\\
=\frac{2}{\beta (L+1)}\lim_{M\to \infty}\sum_{k_0\in\mathcal{ D}_{\beta,M}}e^{-ik_0(x_0-y_0)}\cdot \\ \cdot \sum_{k \in\mathcal{D}^d_{L}}\frac{e^{ik(x-y)}+e^{-ik(x-y)}-e^{ik(x+y)}-e^{-ik(x+y)}}{4}\hat g(\bm k).
\end{split}
\end{equation}
Using formula (\ref{hat_g(k)_parity}) and that the argument of the second sum vanishes if $k\in (L+1)\mathbb Z$, we can rewrite
\begin{equation}
\begin{split}
g(\bm x,\bm y)= \frac{1}{\beta 2(L+1)}\lim_{M\to \infty}\sum_{k_0\in\mathcal{D}_{\beta,M}}e^{-ik_0(x_0-y_0)}\sum_{k\in\mathcal{D}_{2(L+1)}}e^{-i k(x-y)}\hat g(\bm k)+\\
- \frac{1}{\beta 2(L+1)}\lim_{M\to \infty}\sum_{k_0\in\mathcal{D}_{\beta,M}}e^{-ik_0(x_0-y_0)}\sum_{k\in\mathcal{D}_{2(L+1)}}e^{-i k(x+y)}\hat g(\bm k)=\\
=: g_{2(L+1)}(x-y, x_0-y_0)-g_{2(L+1)}(x+y, x_0-y_0).
\end{split}
\end{equation}
\label{proof_lemma_replicas}
\end{proof}

Let us call
\begin{eqnarray}
g_{2(L+1)}(x-y, x_0-y_0):=g_P(\bm x, \bm y),
\label{propagators_P_definition_noscales}\\
-g_{2(L+1)}(x+y, x_0-y_0):=g_R(\bm x, \bm y),
\label{propagators_PR_definition_noscales}
\end{eqnarray}
where $P$ stays for {\it periodic}, referring to the $2(L+1)$ periodicity in the real-space direction, while $R$ stays for {\it remainder} (we will clarify why it is a remainder after the multiscale decomposition), in such a way that 
\begin{equation}
g(\bm x,\bm y)=\sum_{\sigma\in\{P,R\}}g_\sigma(\bm x,\bm y).
\label{propagator_DBC_as_sum_PR}
\end{equation}
\begin{rem}
Since the parameter $L$ enters only in the real-space component of the problem, from now on, whenever we will mention the $2(L+1)-$periodicity, it will stay for "$2(L+1)-$periodicity in the real-space direction", even when not-explicitly specified.
\end{rem}
\begin{rem}
Following the same ideas used in proof of Lemma (\ref{lemma_reflection_trick}) we can rewrite, $\forall \bm x,\bm y\in \Lambda\times \left[ 0,\beta \right)$,
\begin{equation}
g(\bm x,\bm y)=\sum_{n\in\mathbb {Z}}(-1)^n g_{\infty}(\bm x-r_n\bm y),
\end{equation}
where $g_{\infty}$ is the propagator of a system described by a hopping Hamiltonian defined on $\mathbb Z$ (so translation invariant) and the operator $r_n: \Lambda \times \left[0,\beta\right)\to \mathbb Z \times \left[0,\beta\right)$ is defined as follows

\begin{equation}
\begin{split}
r_ny&=\begin{cases}
y+n(L+1) \mbox{ if } n \mbox{ is even},\\
-y+(1+n)(L+1) \mbox{ if } n \mbox{ is odd}, 
\end{cases}\\
r_ny_0&=y_0 \hspace{3mm} \forall n \in \mathbb{Z}.
\end{split}
\end{equation}
\end{rem}

\begin{rem}
\label{antisymmertic_reflection_remark}
It is worth noting that we could have obtained the same result acting directly on the Grassmann variables $\hat \psi^{\pm}_{\bm k}, \bm k \in \mathcal D_{\Lambda}\times\mathcal D_{\beta}$. Indeed, let us imagine to extend the grassmann variables on $\mathcal D_{2(L+1)}$, defining $\hat \psi^{\pm}_{2(L+1)}(\bm k)$ in such a way that
\begin{equation}
\begin{cases}
\hat \psi^{\pm}_{2(L+1)}(k,k_0)=\hat\psi^{\pm}(k,k_0), \mbox{ if } k \in\{ k \in  \mathcal D_{2(L+1)}\cap k\geq 0\}\equiv \mathcal{D}_{\Lambda},\\
\hat \psi^{\pm}_{2(L+1)}(k,k_0)=-\hat \psi^{\pm}_{2(L+1)}(-k,k_0) \mbox{ if } k \in\{ k \in  \mathcal D_{2(L+1)}\cap k<0\}\equiv \mathcal{D}_{\Lambda}.
\end{cases}
\label{antisymmetric_reflection_particles}
\end{equation}
Because of this symmetry property from now on, with a little abuse of notation, we will call all the momenta space Grassmann variables $\hat \psi^{\pm}_{\bm k}$, and we will take care to specify the domain of $\bm k$ in order to distinguish the original variables from the extended ones.
\end{rem}

\section{Interacting case}
\label{section_the_interacting_case_DBC}
\subsection{Trotter's formula and Grassmann integration}
\paragraph{Formal perturbation theory}
After switching on the interaction, the first step is to derive a  {\it formal perturbation theory} for the specific free energy: we want to compute the generic perturbative order in $\lambda$ of 
$$f_{\Lambda,\beta}:=-\frac{1}{|\Lambda|\beta}\lg \left(Tr \left(e^{-\beta H}\right)\right).$$
Recalling that $H=H_0+\lambda V+\nu N+\mathcal N=: H_0+ U$, where after the substitution $\mu_0=\mu+\nu$ we re-define $H_0=T_0-\mu N$, we use Trotter's product formula
\begin{equation}
\label{trotter's_formula}
e^{-\beta H}=\lim_{n\to \infty} \left[e^{-\beta H_0/n}\left(1-\frac{\beta}{n} U \right)\right]^n
\end{equation}
so that, if we define $$U(t):=e^{t H_0}U e^{-t H_0},$$
we can rewrite
\begin{equation}
\begin{split}
\frac{Tr\left(e^{-\beta H}\right)}{Tr\left(e^{-\beta H_0}\right)}
=1+\sum_{N\geq 1}\left(-1\right)^N \int_0^\beta dt_1\int_0^{t_1} dt_2 \dots \int_0^{t_{N-1}}dt_N \frac{Tr \left(e^{t H_0}U e^{-t H_0}\right)}{Tr\left(e^{-\beta H_0}\right)}.
\end{split}
\end{equation}
The {\it fermionic time-ordering operator} allows us to further rewrite
\begin{equation}
\frac{Tr\left(e^{-\beta H}\right)}{Tr\left(e^{-\beta H_0}\right)}=1+\sum_{N\geq 1}\frac{\left(-1\right)^n}{N!}\left<\bm T \left(\left( U_\beta(\psi)\right)^N\right)\right>_{0,\Lambda,\beta},
\label{expansion_trotter_formula}
\end{equation}
where again $\left<\cdot \right>_{0,\Lambda,\beta}=Tr\left(e^{-\beta H_0}\cdot\right)/Tr\left(e^{-\beta H_0}\right)$, and 
\begin{equation}
\begin{split}
U_\beta(\psi)=\lambda \int_{[0,\beta)}d x_0 \sum_{x\in\Lambda}\int_{[0,\beta)}d y_0 \sum_{y\in\Lambda} \psi^+_{\bm x}\psi^-_{\bm x}v(x,y)\delta_{x_0,y_0}\psi^{+}_{\bm y}\psi^-_{\bm y}+\\
+\varpi \int_{[0,\beta)}d x_0 \sum_{x\in\Lambda}\int_{[0,\beta)}d y_0 \sum_{y\in\Lambda} \psi^+_{\bm x} \pi(x,y)\delta_{x_0,y_0}\psi^{-}_{\bm y}+ \nu\int_{[0\beta)}dx_0\sum_{x\in \Lambda}\psi^+_{\bm x}\psi^-_{\bm x}.
\end{split}
\end{equation}
Now, the $N$-th term of (\ref{expansion_trotter_formula}) can be computed by the {\it fermionic Wick rule} knowing explicitly the free propagator, and following the Feynman rules.

\subparagraph{Feynman rules} 
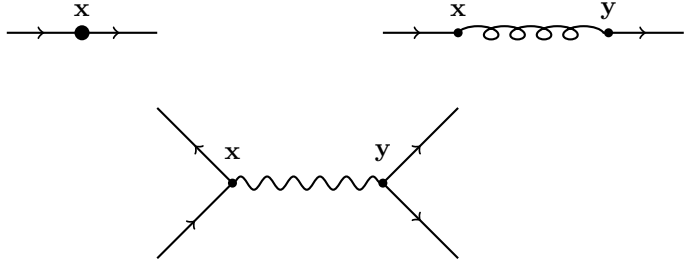
\begin{figure}
\begin{center}
\begin{tikzpicture} 
 [thick,decoration={
    markings,
    mark=at position 0.5 with {\arrow{>}}}] 
\node  at (-2,1.4) {{\bf x}};
\node at (0,1.4) {{\bf y}};
\fill (-2,1) circle (0.06);
\fill (0,1) circle (0.06);
\draw [postaction={decorate}] (-3,0) -- ++(1,1);
\draw [postaction={decorate}](-3,0) ++ (1,1)-- ++ (-1,1);
\draw [postaction={decorate}] (-3,0) ++ (1,1)++ ( 2,0) ++(1,1) ++ (-1,-1) --++ (1,-1);
\draw [postaction={decorate}] (-3,0) ++ (1,1)++ ( 2,0) ++(1,1) ++ (-1,-1) ++ (1,-1) ++ (-1,1) -- ++(1,1);
\draw [-,decorate,decoration=snake] (-3,0) ++ (1,1) ++ (-1,1)++ (1,-1) -- ++(2,0);
\node at (-4,3.3) {\bf x};
\fill (-4,3) circle (0.1);
\draw [postaction={decorate}] (-5,3) -- ++ (1,0);
\draw [postaction={decorate}] (-4,3) -- ++ (1,0);
\node at (1,3.3) {\bf x};
\fill (1,3) circle (0.06);
\node at (3,3.3) {\bf y};
\fill (3,3) circle (0.06);
\draw [postaction={decorate}] (0,3) -- ++ (1,0);
\draw [-,decorate,decoration={coil, aspect=2}] (0,3) ++ (1,0) --++ (2,0);
\draw [postaction={decorate}](0,3) ++ (1,0) ++ (2,0) --++(1,0);
\end{tikzpicture}
\end{center}
\caption{Graph elements, note that there is one element more than (\ref{figure_graph_elements_PBC}), which represents the {\it boundary defect} $\varpi \int_{[0,\beta)} dx_0dy_0\sum_{x,y\in\Lambda}\psi_{\bm x}^+\pi(x,y)\delta_{x_0,y_0}\psi_{\bm y}^-$.}
\label{figure_graph_elements_DBC}
\end{figure}

In order to compute $\left<\bm T\left(U_\beta(\psi))^N\right)\right>_0$, it is easy to check that one can follow these steps:
\begin{itemize}
\item $\forall k,h,l$ such that $0 \leq k,h,l\leq N $ and $k+h+l=N$, draw $k$ graph elements consisting of {\it four legged vertices}, $l$ graph elements consisting of {\it two legged local vertices} and $h$ graph elements consisting of {\it two legged non-local vertices} with the vertices associated to labels $\bm x_i$, $i=1,\dots,N$, in such a way that the {\it four legged vertices} are composed by two entering and to exiting fields, while the {\it two legged vertices} are associated with one exiting and one entering leg, but in the case of the {\it local vertices} the lines touches the same point (so one line enters the same point the same point the other exits), while in the case of the non local vertices the two legs are linked by a further graph element $(\bm x,\bm y)$ (Figure (\ref{figure_graph_elements_DBC}));
\item pair the fields in all possible ways, in such a way that every pair is obtained by contracting an entering and an exiting leg;
\item associate to every pairing the {\it right sign}, which is the sign of the permutation needed to bring every pair of contracted fields next to each other;
\item associate to every linked pair of fields $\left(\psi^-(\bm x_i),\psi^+(\bm x_j)\right)$ an {\it oriented} line connecting the $i-$th with the $j-$th vertex, oriented from $j$ to $i$ ({\it i.e.} from $+$ to $-$ field);
\item associate to every oriented line from $j$ to $i$ value $g(\bm x_i,\bm x_j)$ given by (\ref{free_propagator_DBC});
\item associate to every configuration of pairings, which is called {\it Feynman graph} a value, equal to the product of the sign of the pairing, times $\lambda^k\varpi^h\nu^l$ times the product of the values of all the oriented lines (see, for instance, Figure (\ref{figure_second_order_feynman_graph}));
\item integrate over $\bm x_i$, then perform the sum over all the possible pairings, over $k, h, l$ and over N;
\end{itemize}

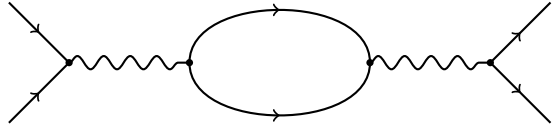
\begin{figure}[h!]
\begin{center}
\begin{tikzpicture}
[scale=0.8, transform shape, thick,decoration={
    markings,
    mark=at position 0.5 with {\arrow{>}}}] 
\fill (-1,0) circle (0.06);
\fill (2,0) circle (0.06);
\fill (-3,0) circle (0.06);
\fill (4,0) circle (0.06);
\draw [postaction={decorate}] (-4,1) -- ++(1,-1);
\draw [postaction={decorate}] (-4,-1)--++ (1,1);
\draw [-,decorate, decoration={snake}] (-3,0) -- ++(2,0);
\draw [postaction={decorate}] (-1,0) to [out=90, in=90, looseness=1] (2,0);
\draw [postaction={decorate}] (-1,0) to [out=-90, in=-90, looseness=1] (2,0);
\draw [-,decorate, decoration={snake}] (2,0) -- ++(2,0);
\draw [postaction={decorate}] (4,0) -- ++(1,-1);
\draw [postaction={decorate}] (4,0)--++ (1,1);
\end{tikzpicture}
\end{center}
\caption{Example of a second order Feynman graph, obtained by contracting two $\lambda$-type endpoints.}
\label{figure_second_order_feynman_graph}
\end{figure}

\paragraph{Grassmann integration}

As explained  in the previous chapter, it is convenient to re-write the {\it free energy} and the {\it Schwinger functions} in terms of Grassmann integrals: first of all we introduce a finite set of {\it Grassmann variables} $\{\hat \psi^{\pm}_{\bm k}\}_{\bm k \in \mathcal D^d_{\Lambda,\beta, M}}$, hence we define the {\it Grassmann integration}

\begin{equation}
P_M(d\psi)=\left(\prod_{\bm k\in\mathcal D^d_{\Lambda,\beta,M}} \left(\frac{\beta(L+1)}{2}\hat g(\bm k)\right)\hat \psi^+_{\bm k}\hat \psi^-_{\bm k} \right) e^{-\sum_{\bm k\in\mathcal D^d_{\Lambda,\beta,M}}\left(\frac{\beta(L+1)}{2}\hat g(\bm k)\right)^{-1}\hat \psi^+_{\bm k}\hat \psi^-_{\bm k}},
\end{equation}

and by, introducing the sine Fourier transform

\begin{equation}
\psi^+_{\bm x}=\frac{2}{\beta(L+1)}\sum_{\bm k\in \mathcal D^d_{\Lambda,\beta,M}}\hat \psi^\pm_{\bm k} e^{-ik_0x_0}\sin (kx),
\end{equation}

we can define the {\it integral}

\begin{equation}
\int P_M(d\psi) \psi^-_{\bm x}\psi^+_{\bm y}=\frac{2}{\beta(L+1)}\sum_{\bm k\in \mathcal D_{\Lambda,\beta,M}}\hat g(\bm k) e^{-ik_0(x_0-y_0)}\sin(kx)\sin(ky),
\end{equation}

while the average of any monomial in the Grassmann variables with respect to the Grassmann integration $P(d\psi)$ is given by the fermionic Wick rule with propagator $g(\bm x,\bm y)$. Using the definitions of Grassmann integration and the Feynman rules just described, it comes out that 
\begin{equation}
\frac{Tr\left(e^{-\beta H}\right)}{Tr\left(e^{-\beta H_0}\right)}=\int P(d\psi)e^{-\mathcal V(\psi)},
\end{equation}
where
\begin{equation}
\begin{split}
\mathcal V(\psi)=\lambda\int_{[0,\beta]}dx_0\sum_{x\in\Lambda}\int_{[0,\beta]}dy_0\sum_{y\in\Lambda}\psi^+_{\bm x}\psi^-_{\bm x}v(x,y)\delta_{x_0,y_0}\psi^+_{\bm y}\psi^-_{\bm y}+\\
+\varpi\int_{[0,\beta]}dx_0\sum_{x\in\Lambda}\int_{[0,\beta]}dy_0\sum_{y\in\Lambda}\psi^+_{\bm x}\pi(x,y)\delta_{x_0,y_0}\psi^-_{\bm y}+\nu\int_{[0,\beta)}\sum_{x\in\Lambda}\psi^+_{\bm x}\psi^-_{\bm x}.
\end{split}
\end{equation} 

\begin{rem}
Starting from that, we can repeat formally the same construction as the previous chapter, getting first of all the formal equation for the free energy (\ref{free_energy_as_sum_of_trunc_expec}) that we recall
\begin{equation*}
f_{\Lambda,\beta}=-\frac{1}{|\Lambda|\beta}\sum_{N\geq 1}\frac{(-1)^N}{N!}\mathcal E^T(\mathcal V;N)=:\sum_{N\geq 1}f^{(N)}_{\Lambda,\beta}.
\end{equation*}
Of course, we can repeat the same argument we used to prove Lemma (\ref{lemma_bounds_no_multiscale_no_determinants}) to get the same rough bound. As in the previous chapter, we can solve the combinatorial problem by using the detereminant expansion (the fact that the free propagator (\ref{free_propagator_DBC}) can be expressed as a proper scalar product is proved in Appendix (\ref{appendix_gram_representation}), looking only at the functions $A^{d}$ and $B^{d}$):
\begin{equation*}
f_{\Lambda,\beta}^{(N)}=-\frac{1}{\beta|\Lambda|}\frac{(-1)^N}{N!}\epsilon^N \sum_{T\in \mathcal T_N}\alpha_T\int d\bm x_1\dots d\bm x_N \prod_{\ell\in T} g_\ell\int dP_T(\bm t)\det G^T(\bm t).
\end{equation*}
In order to solve the divergence problem arising when we take the infinite volume limit, we introduce a multiscale expansion.
\end{rem}

As we already commented in the previous chapter, with a little abuse of notation we called $f_{\Lambda,\beta}$ the difference between the specific free energy of the interacting system and the one of the free system.

\subsection{Multiscale decomposition}

\paragraph{Multiscale decomposition and quasi-particles}
\label{ingoing_outgoing_quasiparticles_subsection}
Let us recall the result of Lemma \ref{lemma_reflection_trick}: 

\begin{equation*}
g(\bm x,\bm y)=g_P(\bm x,\bm y)+g_R(\bm x,\bm y),
\end{equation*}
where $g_P$ and $g_R$ have been defined in (\ref{propagators_P_definition_noscales}) and (\ref{propagators_PR_definition_noscales}), both starting from the periodic propagator on the extended box:
\begin{equation}
g_{2(L+1)}(\bm x)=\frac{1}{\beta2(L+1)}\lim_{M\to \infty}\sum_{\bm k\in\mathcal{D}_{2(L+1)}\times\mathcal{D}_{\beta,M}} e^{-i\bm k\cdot \bm x}\hat g(\bm k).
\label{free_propagator_2(L+1)_PBC)}
\end{equation}
Given that, we can separately perform, on both the propagators $g_P$ and $g_R$, first of all the multiscale decomposition and then the quasi-particle decomposition.
\subparagraph{Infrared and ultraviolet regime}
First of all, let us introduce a smooth $C^{\infty}$ function $\chi:\mathcal{D}_{2(L+1)}\times \mathcal{D}_{\beta, M}\to C^{\infty}([0,1])$ defined in such a way that
\begin{equation}
\chi(\bm k')=
\begin{cases}
1, \mbox{ if } |\bm k'|\leq \gamma^{-1} p_F/2 ,\\
0, \mbox{ if } |\bm k'|\geq p_F/2,
\end{cases}
\label{cut_off_chi_definition_DBC}
\end{equation}
where $\gamma >1$, and $|\bm k|=\sqrt{k_0^2+k^2}$. We know $\cos p_F-\cos k=0$ if  $k=\pm p_F\in \mathcal{D}_{2(L+1)}$, so we rewrite the propagator as:
\begin{equation}
\hat g(\bm k)= \frac{1-\chi(k_0,k+p_F)-\chi(k_0,k-p_F)}{-ik_0+\cos p_F-\cos k}+\frac{\chi(k_0,k+p_F)+\chi(k_0,k-p_F)}{-ik_0+\cos p_F-\cos k}
\end{equation}
defining the {\it ultraviolet} and the {\it infrared} propagator, respectively $\hat g^{(u.v)}$ and $\hat g^{(i.r.)}$:
\begin{eqnarray}
\hat g^{(u.v.)}(\bm k)=\frac{1-\chi(k_0,k+p_F)-\chi(k_0,k-p_F)}{-ik_0+\cos p_F-\cos k},
\label{ultraviolet_momentum_propagator}
\\
\hat g^{(i.r.)}(\bm k)=\frac{\chi(k_0,k+p_F)+\chi(k_0,k-p_F)}{-ik_0+\cos p_F-\cos k}.
\label{infrared_momentum_propagator}
\end{eqnarray}
This decomposition induces a natural decomposition in the {\it real-space}
\begin{equation}
\begin{split}
g^{(u.v.)}_P(\bm x, \bm y)=\frac{1}{\beta 2(L+1)} \lim_{M\to \infty}\sum_{\substack{\bm k \in \mathcal D_{2(L+1)}^{\beta,M}}} e^{-i k_0(x_0-y_0)} e^{-i k(x-y)}\hat g^{(u.v.)}(\bm k),\\
g^{(i.r.)}_P(\bm x, \bm y)=\frac{1}{\beta 2(L+1)} \lim_{M\to \infty}\sum_{\substack{\bm k \in \mathcal D_{2(L+1)}^{\beta,M}}} e^{-i k_0(x_0-y_0)} e^{-i k(x-y)}\hat g^{(i.r.)}(\bm k),\\
g^{(u.v.)}_R(\bm x, \bm y)=\frac{1}{\beta 2(L+1)} \lim_{M\to \infty}\sum_{\substack{\bm k \in \mathcal D_{2(L+1)}^{\beta,M}}} e^{-i k_0(x_0-y_0)} e^{-i k(x+y)}\hat g^{(u.v.)}(\bm k),\\
g^{(i.r.)}_R(\bm x, \bm y)=\frac{1}{\beta 2(L+1)} \lim_{M\to \infty}\sum_{\substack{\bm k \in \mathcal D_{2(L+1)}^{\beta,M}}} e^{-i k_0(x_0-y_0)} e^{-i k(x+y)}\hat g^{(i.r.)}(\bm k).
\end{split}
\end{equation}
where we defined $\mathcal D_{2(L+1)}^{\beta,M}:=\mathcal D_{2(L+1)}\times \mathcal D_{\beta,M},$
and in particular we can introduce the label $\sigma$ in such a way that:
\begin{equation}
\begin{split}
g^{(u.v.)}(\bm x,\bm y)=\sum_{\sigma\in\{P,R\}}g_\sigma^{(u.v.)}(\bm x,\bm y),\\
g^{(i.r.)}(\bm x,\bm y)=\sum_{\sigma\in\{P,R\}}g_\sigma^{(i.r.)}(\bm x,\bm y).
\end{split}
\end{equation}
Using the addition principle (\ref{addition_principle}) we introduce two different sets of Grassmann fields $\{\psi_{\bm x}^{(u.v.)\pm}\}$ and $\{\psi_{\bm x}^{(i.r.)\pm}\}$, with $\bm x\in \Lambda\times [0,\beta)$ and the Gaussian integrations $P_M(\psi^{(u.v.)})$ and $P_M(\psi^{(i.r.)})$ in such a way that
\begin{equation}
\begin{split}
\int P(d \psi^{(u.v.)})\psi^{(u.v.)-}_{\bm x}\psi^{(u.v.)+}_{\bm y}= g^{(u.v.)}(\bm x,\bm y),\\
\int P(d \psi^{(i.r.)})\psi^{(i.r.)-}_{\bm x}\psi^{(i.r.)+}_{\bm y}= g^{(i.r.)}(\bm x,\bm y),
\end{split}
\end{equation}
implying that 
\begin{equation}
\int P(d\psi) e^{-\mathcal V(\psi)}=\int P(d\psi^{(i.r.)})\int P(d\psi^{(u.v)}) e^{-\mathcal V\left(\psi^{(i.r.)}+\psi^{(u.v.)}\right)},
\end{equation}
so that
\begin{equation}
\begin{split}
e^{-\beta |\Lambda| f^{(M)}_{\Lambda,\beta}}=\int P(d\psi^{(i.r.)}) \exp \left( \sum_{n\geq 1}\frac{1}{n!}\mathcal E_{u.v.}^T\left(-\mathcal V\left(\psi^{(i.r.)}+\cdot\right);n\right)\right):=\\ := e^{-\beta |\Lambda| e_{M,0}}\int P(d\psi^{(i.r.)})e^{-\mathcal V_0(\psi^{(i.r.)})}.
\end{split}
\end{equation}
 
where the effective potential $\mathcal V_0(\psi)$ can be written as
\begin{equation}
\mathcal V_0(\psi)=\sum_{n=1}^{\infty}\sum_{\substack{ \bm x_1,\dots,\bm x_{2n}\\ \in\\ \Lambda\times [0,\beta)}} \left(\prod_{j=1}^{n} \psi^{(i.r.)+}_{\bm x_{2j-1}}\psi^{(i.r.)-}_{\bm x_{2j}} \right) W_{M,2n} (\bm x_1,\dots,\bm x_{2n}).
\label{effective_potential_scale_0_DBC}
\end{equation}

\begin{lem}[Ultraviolet integration]
The kernels $W_{M,2n}(\bm x_1,\dots, \bm x_{2n})$ in the previous expansion are given by power series in $\lambda$ convergent in the complex disc $|\lambda|\leq \lambda_0$ for $\lambda_0$ small enough and independent of $M,\Lambda, \beta$, and satisfy the following bound
\begin{equation}
\frac{1}{\beta |\Lambda|}\int d\bm x_1 \dots d\bm x_{2n} \left| W_{M,2n}(\bm x_1,\dots,\bm x_{2n}) \right|\leq C^n |\lambda|^{\max\{1,n-1\}}.
\end{equation}
Moreover, the limits $e_0=\lim_{M\to \infty} e_{M,0}$ and $W_{2n}=\lim_{M\to \infty}(\bm x_1,\dots,\bm x_{2n})$ exist and are reached uniformly in M.
\end{lem}
\begin{rem}
The fact that the limits are reached {\it uniformly} in $M$ tells us that the infrared problem is essentially independent of M. Since in the infrared region $M$ does not play any role, from now on we drop the label $M$.
\end{rem}
As in the previous chapter, we will not prove this Lemma, and we refer, for instance, to \cite{benfatto1993beta}.
\paragraph{Multiscale expansion of the infrared scales and quasi-particles}
\subparagraph{Infrared regime and quasi-particles}

After having integrated the ultraviolet degrees of freedom, we are left with the {\it infrared propagator}

\begin{equation}
\hat g^{(i.r.)}(\bm k)=\frac{\chi(k+p_F,k_0)+\chi(k-p_F,k_0)}{-ik_0+\cos p_F-\cos k}= :\sum_{\omega\in\{\pm1\}}\hat g^{(i.r.)}_{\omega}(\bm k),
\end{equation}
where 
\begin{equation}
\hat{g}^{(i.r.)}_{\omega}(\bm k):=\frac{\chi(k-\omega p_F,k_0)}{-ik_0+\cos p_F-\cos k}.
\end{equation}

We can now define the {\it infrared propagator in real space-time}

\begin{equation}
\begin{split}
g^{(i.r.)}_{2(L+1)}(\bm x)=\frac{1}{\beta2(L+1)}\sum_{\omega\in\{\pm 1\}}\sum_{\bm k\in \mathcal{D}_{2(L+1)}\times\mathcal{D}_{\beta}}e^{-i\bm k\bm x}\frac{\chi(k-\omega p_F,k_0)}{-ik_0+\cos p_F-\cos k}=\\
=\frac{1}{\beta2(L+1)}\sum_{\omega\in\{\pm 1\}}\sum_{\bm k'\in \mathcal{D}^{\omega}_{2(L+1)}\times\mathcal{D}_{\beta}}e^{-i\omega p_F x}e^{-i\bm k'\bm x}\hat g^{(i.r.)}_\omega(\bm k'),
\end{split}
\end{equation} 
 
where we have performed the change of variables $k-\omega p_F=k'$, so $\mathcal{D}^{\omega}_{2(L+1)}=\mathcal{D}_{2(L+1)}-\omega p_F$ and
\begin{equation}
g^{(i.r.)}_\omega(\bm k')=g^{(i.r.)}(k-\omega p_F, k_0).
\end{equation}
Finally, we can define
\begin{equation}
g^{(i.r.)}_{2(L+1)}(\bm x)=\sum_{\omega\in\{\pm 1\}}e^{-i\omega p_F x}g^{(i.r.)}_{2(L+1),\omega}(\bm x),
\end{equation}
where
\begin{equation}
g^{(i.r.)}_{2(L+1),\omega}(\bm x)=\frac{1}{\beta2(L+1)}\sum_{\bm k'\in \mathcal{D}^{\omega}_{2(L+1)}\times\mathcal{D}_{\beta}}e^{-i\bm k'\bm x}\hat g^{(i.r.)}_\omega(\bm k').
\end{equation}

\begin{rem}
From the latter expression we understand the behaviour of the propagator nearby the singularity. Indeed, when the momentum is close to the Fermi momentum $p_F$, ({\it i.e.} $k'\sim 0$), $g^{(i.r.)}_{\omega}(\bm k')\sim \frac{\chi(k_0,k')}{-ik_0+\omega k' \sin p_F}$, so we have a quasi-linear dispersion near the singularity.
\end{rem}
So, we can finally decompose the propagator $g^{(i.r.)}(\bm x,\bm y)$ as
\begin{equation}
\begin{split}
g^{(i.r.)}(\bm x,\bm y)=g_P^{(i.r.)}(\bm x,\bm y)+g_R^{(i.r.)}(\bm x,\bm y)=\\
=\sum_{\omega=\pm}e^{-i\omega p_F(x-y)}g_{P,\omega}^{(i.r.)}(\bm x,\bm y)+\sum_{\omega=\pm}e^{-i\omega p_F(x+y)}g_{R,\omega}^{(i.r.)}(\bm x,\bm y),
\end{split}
\end{equation}
where 
\begin{eqnarray}
g^{(i.r)}_{P,\omega}(\bm x, \bm y):=g^{(i.r.)}_{2(L+1),\omega}(x-y, x_0-y_0),
\label{propagators_P_i.r._defn}\\
g^{(i.r)}_{R,\omega}(\bm x, \bm y):=g^{(i.r.)}_{2(L+1),\omega}(x+y, x_0-y_0),
\label{propagators_PR_i.r._defn}
\end{eqnarray}
that suggests to rewrite the propagator as
\begin{equation}
g^{(i.r.)}(\bm x, \bm y)=\sum_{\sigma\in\{P,R\}}\sum_{\omega=\pm}g^{(i.r.)}_{\sigma,\omega}(\bm x,\bm y)\left(e^{-i\omega p_F(x-y)}\delta_{\sigma,P}+e^{-i\omega p_F(x+y)}\delta_{\sigma,R}\right).
\label{propagator_DBC_as_sum_of_propagators_sigma_omega_infrared}
\end{equation}
Using the {\it addition property} of Gaussian Grassmann measures (\ref{addition_principle}), we can split
\begin{equation}
\psi^{(i.r.)\pm}_{\bm x}=\sum_{\sigma\in\{P,R\}}\sum_{\omega=\pm}e^{\mp ip_F \omega x}\psi^{(i.r.)\pm}_{\sigma,\omega,\bm x},
\label{quasi_particles_decomposition_infrared_DBC}
\end{equation}
associated with the Feynman contraction rule
\begin{equation}
\left<\psi^{(i.r.)-}_{\sigma,\omega}(\bm x)\psi^{(i.r.)+}_{\sigma',\omega'}(\bm y)\right>=g^{(i.r.)}_{\sigma,\omega}(\bm x,\bm y)\delta_{\sigma,\sigma'}\left(\delta_{\omega,\omega'}\delta_{\sigma,P}+\delta_{\omega,-\omega'}\delta_{\sigma,R} \right).
\label{feynman_rules_propagators_sigma_omega_infrared}
\end{equation}

\subparagraph{Decomposition on scales $h\leq 0$}

Once we have defined the infrared scale, we can take a step beyond and rewrite the propagators $\{g^{(i.r.)}_{\sigma,\omega}(\bm x,\bm y)\}_{\sigma\in \{P,R\}}^{\omega=\pm}$ as an infinite sum of {\it single scale propagators} we are going to define. The only trick we use is rewriting the cutoff function $\chi$ as the telescopic series:
\begin{equation}
\chi\left(\bm k'\right)=\sum_{h\leq 0}\left[\chi\left(\gamma^{-h}\bm k'\right)-\chi\left(\gamma^{-h+1}\bm k'\right)\right]=:\sum_{h\leq 0}f_h(\bm k').
\end{equation}

Using it in the very definition of infrared quasi particle propagator we get
\begin{equation}
\hat g_{\omega}^{(i.r.)}( \bm k')=\frac{\chi( k',k_0)}{-ik_0+\cos p_F-\cos (k'-\omega p_F)}=\sum_{h\leq 0}f_{h}(k',k_0)\hat g_{\omega}(\bm k') =:\sum_{h\leq 0}\hat g^{(h)}_{\omega}(\bm k').
\label{propagator_decomposition_scale_h_momentum_space}
\end{equation}
A direct consequence of the latter decomposition is the possibility to define
\begin{equation}
\hat g_{\omega}^{(\leq h)}(\bm k')=\sum_{j\leq h}\hat g_{\omega}^{(j)}(\bm k').
\label{propagator_leq_h_definition_momentum_space}
\end{equation}
It is appropriate now to introduce the real space-time representation of the single scale propagators by Fourier transforming $\hat g^{(h)}_\omega(\bm k)$, 
\begin{equation}
g_{2(L+1),\omega}^{(h)}(\bm x)=\frac{1}{\beta 2 (L+1)}\sum_{\bm k'\in\mathcal{D}^{\omega}_{2(L+1)}\times\mathcal{D}_{\beta}}e^{-i\bm k'\cdot \bm x}\hat g^{(h)}_{\omega}(\bm k'),
\label{propagator_single_scale_real_space}
\end{equation}
which implies, of course, 
\begin{equation}
\label{g^d(h)_definition}
g^{(h)}(\bm x,\bm y)=g_P^{(h)}(\bm x,\bm y)+g_R^{(h)}(\bm x,\bm y)
\end{equation}
and, in a way formally analogous to what we did in the infrared region:
\begin{equation}
g^{(h)}(\bm x,\bm y)=\sum_{\sigma\in\{P,R\}}\sum_{\omega\in \pm}g^{(h)}_{\sigma,\omega}(\bm x, \bm y)\left(e^{-i p_F\omega (x-y)}\delta_{\sigma,P}+e^{-i p_F\omega (x+y)}\delta_{\sigma,R}\right),
\label{gd_sum_of_quasi_particle_propagators}
\end{equation}
where $g^{(h)}_{\sigma,\omega}$ is the analogous of the elements appearing in (\ref{propagators_PR_i.r._defn}) if we replace $\chi_{i.r.}(\bm k)\to f_h(\bm k)$. Of course, we can introduce the quasi-particles fields analogously to what we did in the infrared case (\ref{quasi_particles_decomposition_infrared_DBC}):
\begin{equation}
\psi_{\bm x}^{(h)\pm}=\sum_{\omega=\pm}\sum_{\sigma\in\{P,R\}}e^{\mp i\omega p_F x}\psi^{(h)\pm}_{\sigma,\omega,\bm x},
\label{quasi_particles_decomposition_scale_h_DBC}
\end{equation}
contracting with the Feynman contraction rule
\begin{equation}
\left<\psi^{(h)-}_{\sigma,\omega,\bm x}\psi^{(h)+}_{\sigma',\omega',\bm y}\right>=g^{(h)}_{\sigma,\omega}(\bm x, \bm y)\delta_{\sigma,\sigma'}\left(\delta_{\omega,\omega'}\delta_{\sigma,P}+\delta_{\omega,-\omega'}\delta_{\sigma,R} \right).
\label{feynman_rules_scale_h_DBC}
\end{equation}
Finally, it is useful to decompose:
\begin{eqnarray}
g^{(\leq h)}(\bm x,\bm y)=\sum_{j\leq h}g^{(j)}(\bm x,\bm y),\\
g^{(\leq h)}_{\sigma,\omega}(\bm x,\bm y)=\sum_{j\leq h}g^{(j)}_{\sigma,\omega}(\bm x,\bm y).
\label{propagator_scale_decomposition_DBC_sigma_omega}
\end{eqnarray}
\begin{rem}
When we switch from the original representation in terms of the Grassmann variables $\{\psi^{\pm}\}$ to the {\it quasi particles representation} in terms of $\{\psi^{\pm}_{\sigma,\omega}\}$, we break the Dirichlet boundary conditions: indeed the only  information we get from (\ref{remark_quasi_particles_issue_DBC}), knowing that $\psi^{\pm}(0, x_0)=\psi^{\pm}(L+1,x_0)=0$ $\forall x_0\in \left[0,\beta\right)$ is 
\begin{equation}
\begin{split}
\sum_{\omega}\sum_{\sigma}\psi^{\pm}_{\sigma, \omega}(0,x_0)=0,\\
\sum_{\omega}\sum_{\sigma}e^{\mp ip_F\omega (L+1)}\psi^{\pm}_{\sigma, \omega}(L+1,x_0)=0.\\
\end{split}
\end{equation}
\end{rem}

\subsection{Properties of single-scale free propagators}
\label{subsection_estimates_on_single-scale_free_propagator}

\paragraph{Estimates}
The multiscale analysis involves the norms of the single scale propagators $g^{(h)}$, so it is useful to note that, in the case of $g_{2(L+1),\omega}^{(h)}$ we can directly apply the result of Lemma (\ref{lemma_propagator_faster_any_power}) that we recall here:
\begin{equation}
\left| g^{(h)}_{2(L+1),\omega}(\bm x)\right|\leq \gamma^{h} \frac{C_N}{1+\left(\gamma^h|\bm x|\right)^N}, \hspace{3mm} \forall N\in \mathbb{N}.
\end{equation}
 In subsection (\ref{ingoing_outgoing_quasiparticles_subsection}) we showed how, at each scale $h$, we can rewrite 
\begin{equation}
g^{(h)}(\bm x,\bm y)=\sum_{\omega\in\{\pm 1\}}\left[e^{-i\omega(x-y)}g_{P,\omega}^{(h)}(\bm x, \bm y)+e^{-i\omega(x+y)}g_{R,\omega}^{(h)}(\bm x,\bm y)\right]
\label{propagator_decomposed_in_quasiparticles_DBC}
\end{equation}
where $g^{(h)}_{P,\omega}$ and $g^{(h)}_{R,\omega}$ are defined in terms of $g_{2(L+1),\omega}^{(h)}$ so, again $\forall N\in\mathbb{N}$, we can estimate,
\begin{equation}
\begin{cases}
| g_{P,\omega}^{(h)}(x-y,x_0-y_0)|\leq \gamma^h\frac{C_N}{1+\left(\gamma^h \left|(x-y,x_0-y_0)\right|\right)^N},\\
| g_{R,\omega}^{(h)}(x+y,x_0-y_0)|\leq \gamma^h\frac{C_N}{1+\left(\gamma^h \left|(x+y,x_0-y_0)\right|\right)^N}.
\end{cases}
\label{bounds_propagator_faster_than_any_power_DBC}
\end{equation}

\begin{corollary}
\label{corollary_norms_propagators_DBC}
Thanks to (\ref{bounds_propagator_faster_than_any_power_DBC}), we can bound the norms $||\cdot||_\infty$ and $||\cdot||_1$ of the quasi particles propagators $\{g_{\sigma,\omega}\}_{\sigma\in\{P,R\}}^{\omega=\pm}$ as
\begin{eqnarray}
||g^{(h)}_{P,\omega}||_\infty:=\sup_{\bm x,\bm y}|g^{(h)}_{P,\omega}(\bm x,\bm y)|\leq C\gamma^h, \label{infty_norm_diagonal_DBC}\\
||g^{(h)}_{R,\omega}||_\infty:=\sup_{\bm x,\bm y}|g^{(h)}_{R,\omega}(\bm x,\bm y)|\leq C\gamma^h,\label{infty_norm_off_diagonal_DBC}\\
||g^{(h)}_{P,\omega}||_1:=\left|\frac{1}{2(L+1)\beta}\int d\bm x d\bm y g^{(h)}_{P,\omega}(\bm x,\bm y)\right|\leq C\gamma^{-h},\label{1_norm_diagonal_DBC}\\
||g^{(h)}_{R,\omega}||_1:=\left|\frac{1}{2(L+1)\beta}\int d\bm x d\bm y g^{(h)}_{P,\omega}(\bm x,\bm y)\right|\leq C\gamma^{h_L-h}\gamma^{-h},\label{1_norm_off_diagonal_DBC}
\end{eqnarray}
where ${h_L}:=\lfloor \log_\gamma \frac{1}{L+1}\rfloor$, {\it i.e.} $\gamma^{h_L}\sim \frac{1}{L+1}$ and $C>0$. 
\end{corollary}

In light of that, we can define, for each $N=1,2,\dots$, a function that will be useful in the following
\begin{equation}
\label{definition_rho_h^N}
\rho^{(N)}_h(x)=\sup_{\substack{y\in\Lambda, \\ x_0,y_0\in [0,\beta)}}\frac{C_N}{\left(1+\gamma^h|(x+y,x_0-y_0)|\right)^N}\leq \frac{C_N}{(1+\gamma^hd_{L}(x))^N},
\end{equation}
where $d_L(x)=\min_{x\in \Lambda}\{|x|,|x-L||\}$, so that it holds
\begin{equation}
\left| g^{(h)}_R(\bm x,\bm y)\right|\leq ||g^{(h)}||_\infty \rho^{(N)}_h(x).
\label{bound_g_R_g_infty_rho}
\end{equation}
From now on, with a slight abuse of notation, we will denote 
$$\frac{1}{(1+\gamma^hd_{L}(x))} =: \frac{1}{(1+\gamma^h|x|)}.$$
\begin{rem}
\label{remark_anchorage_property_norm_1_infty}
We will say that the dimensional gain $\gamma^{h_L-h}$ (scale jump) we have in (\ref{1_norm_off_diagonal_DBC}) with respect to (\ref{1_norm_diagonal_DBC}) and (\ref{norm_1_propagator}) is due to the {\it "anchorage property"}: the propagator $g^{(h)}_{R,\omega}$ does not depend on the distance of the arguments, so we can perform both the integrals over the positions using the decay properties of the propagator, without getting a volume factor. 
\end{rem}

\paragraph{Gram representation} The propagator $g^{(h)}_P{(\bm x-\bm y)}$ has already been studied in the previous chapter, but we have to check that also $g^{(h)}_R(\bm x,\bm y)$ admits a Gram representation.

\begin{lem}{Gram estimate}
Let $M$ be a square matrix whose entries are $M_{ij}=\left<A_i, B_j\right>$ where $A_i$ and $B_j$ are vectors in a Hilbert space with scalar product $\left<\cdot,\cdot \right>$, then
\begin{equation}
|\det M|\leq \prod_{i}||A_i|| ||B_j||
\end{equation} 
where $||\cdot ||$ is the norm induced by the scalar product.
\end{lem}

We do not prove this lemma, and we refer to \cite{gentile2001renormalization}, Theorem A.1.

\section{Non-renormalized expansion and properties of kernels}
\label{section_Non-renormalized expansion and properties of kernels}
By combining the multiscale expansion of the free propagator and the properties of the Grassmann integration, our goal is to compute
\begin{equation}
\begin{split}
e^{-\beta |\Lambda| f_{\Lambda,\beta}}=\int P(d \psi^{(\leq h)})\int P(d \psi^{(h+1)})\int P(d \psi^{(h+2)})\dots \int P(d\psi^{(0)})e^{-\mathcal V_0(\psi^{(0)}+\psi^{(\leq 1)})}=\\
=e^{-\beta |\Lambda| e_h}\int P(d\psi^{(\leq h)})e^{\mathcal V^{(h)}(\psi^{(\leq h)})}
\end{split}
\end{equation}
where the Grassmann integration is 
\begin{equation}
\begin{split}
P(d\psi^{(\leq h)})=\prod_{\sigma\in\{P,R\}}\prod_{\omega=\pm}\prod_{\bm x\in\mathcal{D}^{d}_{\Lambda,\beta}}\left(\left[ \det g^{(\leq h)}\right]^{-1} \psi^{(\leq h)+}_{\sigma,\omega}(\bm x) \psi^{(\leq h)-}_{\sigma,\omega} (\bm x)\right)\\ \exp\left[-\sum_{\sigma\in\{P,R\}}\sum_{\omega=\pm}\sum_{\bm x,\bm y\in \Lambda\times [0,\beta)}\psi^{(\leq h)+}_{\sigma,\omega}(\bm x)\left[g^{(\leq h)}\right]^{-1}_{\sigma,\omega}(\bm x,\bm y)\psi_{\sigma,\omega}^{(\leq h)-}(\bm y)\right],
\label{measure_quasi_particles_real_space_DBC}
\end{split}
\end{equation}
and the effective potentials can be written as:
\begin{equation}
\label{non_renormalized_effective_potential_PBD}
\mathcal V^{(h)}(\psi^{\leq h})=\sum_{n=1}^{\infty}\int d\bm x_1\dots d\bm x_{2n}\left(\prod_{j=1}^{\infty}\psi^{(\leq h)+}_{\bm x_{2j-1}}\psi^{(\leq h)-}_{\bm x_{2j}}\right)  W_2^{(h)}(\bm x_1,\dots, \bm x_{2n})
\end{equation}
where the integrals has to be interpreted as
$$\int d\bm x=\int_{[0,\beta)}dx_0\sum_{x\in\Lambda}.$$
From now on, each integral has to be interpreted in this way.
\begin{rem}
\label{remark_comparison_norms_diagonal_off_diagonal_DBC}
From corollary (\ref{corollary_norms_propagators_DBC}) we recognize that, while in norm-$\infty$ there is no difference between translation invariant propagators $g^{(h)}_{P,\omega}$ and the non translation invariant ones $g^{(h)}_{R,\omega}$, in particular, there is no difference between the {\it dimensional bound} $||g^{(h)}_{P,\omega}||_\infty$ (which scales as the translation invariant one $||g^{(h)}_{\omega}||_\infty$) and $||g^{(h)}_{R,\omega} ||_\infty$, by comparing the $||\cdot||_1$ norms we recognize that $||g^{(h)}_{R,\omega}||_1$ has a scale gain $\gamma^{h_L-h}$ with respect to $||g^{(h)}_{P,\omega}||_1$, which scales as the translation invariant one $||g^{(h)}_\omega||$.\\
Looking at the measure (\ref{measure_quasi_particles_real_space_DBC}) associated with a propagator labeled by two indices and using the {\it addition principle }(\ref{addition_principle}) we recognize that, in constructing the trees, we have to assign to each half-line a further label $\sigma\in\{P,R\}$.
\end{rem}

\paragraph{Non-renormalized expansion}

Taking into account what we have pointed out in remarks (\ref{remark_comparison_norms_diagonal_off_diagonal_DBC}), it is trivial to claim that the theorem (\ref{theorem_bound_of_kernels}) holds also in the framework described by the Hamiltonian $H$, where $\epsilon=\max\{|\nu|,|\lambda|,|\varpi|\}$. Indeed, it is enough to notice that:
\begin{itemize}
\item we can roughly localize the non-local endpoints: 
$$ \int dy \left| \pi(x,y)\right|\leq  \frac{C_\theta}{1+|x|^\theta}\leq  C_\theta, \hspace{5mm} \forall \hspace{3mm} 0<\theta\leq 1, $$
and consider the $\varpi$-type endpoints as $\nu$-type endpoints, where just the associated constant changes ($\varpi$ instead of $\nu$). 
\item after remark (\ref{remark_comparison_norms_diagonal_off_diagonal_DBC}), it is clear that the dominant part of the propagators (as expected) is the {\it translation-invariant one} scaling, in the sense of the norms $||\cdot||_\infty$ and $||\cdot||_1$, as the translation invariant propagator the dimensional analysis of section (\ref{section_multiscale_analysis}) is based on. So, in order to get the dimensional estimate for {\it non-renormalized kernels} as in theorem (\ref{theorem_bound_of_kernels}), we can roughly consider each propagator as a $g^{(h)}_{P,\omega}$ one. This choice corresponds to estimate only the dominant part (the part which survives when $\beta, L\to \infty$). Obviously, being interested in the finite size correction to the specific free energy and the Schwinger functions, it will be necessary to consider in a more sophisticated way the off-diagonal propagators, as it will be clear in the {\it renormalization procedure}, in which we will be forced to take into account the $\bm x$-dependence of propagators,
\end{itemize}
After these two simplifications, we can retrace the proof of Theorem (\ref{theorem_bound_of_kernels}), after having redefined $\epsilon=\max\{|\nu|,|\lambda|,|\varpi|\}$, and get exactly the same bounds. Consequently, we have exactly the same {\it a priori} classification of the clusters into {\it marginal, relevant} and {\it irrelevant ones}. So again, in order to express the {\it specific free energy} and the {\it Schwinger functions} as convergent series we have to expand in a different way the kernels with two and four external legs, indentifying properly the source of troubles. We stress that we cannot use as a black box the machinery we set up in the previous chapter (\ref{chapter_fermions_PBC}) because of the localization procedure is crucially based on the translation invariance of the system.

\subsection{Properties of the kernels}

First of all, let us stress that we can rewrite the non-renormalized effective potential (\ref{non_renormalized_effective_potential_PBD}) using (\ref{quasi_particles_decomposition_scale_h_DBC}):

\begin{equation}
\begin{split}
\mathcal V^{(h)}(\psi^{(\leq h)})=\sum_{n=1}^\infty\sum_{\bm \omega}\sum_{\bm \sigma}\int \bm x_1\dots\bm x_{2n}\cdot \\ \cdot\left(\prod_{j=1}^n\psi^{(\leq h)+}_{\omega_{2j-1},\sigma_{2j-1},\bm x_{2j-1}}\psi^{(\leq h)-}_{\omega_{2j},\sigma_{2j},\bm x_{2j}}e^{-ip_F(\omega_{2j-1}x_{2j-1}-\omega_{2j}x_{2j})}\right) W^{(h)}_{2n}(\bm x_1,\dots,\bm x_{2n}),
\end{split}
\label{effective_potential_scale_h_DBC}
\end{equation}
where $\bm \omega=\left\{\omega_1,\dots,\omega_{2n}\right\}$, $\bm \sigma=\left\{\sigma_1,\dots,\sigma_{2n}\right\}$.
\begin{rem}\label{remark_kernels_independent_of_quasi_particles} By construction the kernels are independent of the {\it external configuration of $\bm \omega$ and $\bm \sigma$}. 
The only thing that could break this structure is the renormalization and localization procedure, so we will take care of proving that in fact it is not the case.
\end{rem}
\paragraph{Boundary conditions not-invariant under RG integration}

\begin{prop}
\label{proposition_kernels_not_diagonal_in_k_DBC}
The Dirichlet boundary conditions are not invariant under a single {\it renormalization group step}, meaning that $W^{(h)}_2(\bm x,\bm y)$ in (\ref{effective_potential_scale_h_DBC}) {\it  is not diagonal in the sine Fourier base}, {\it i.e.}:
\begin{equation}
W^{(h)}_2(\bm x,\bm y)=\left[\frac{2}{\beta(L+1)}\right]^2\sum_{\bm k_1,\bm k_2\in\mathcal D^d_{\Lambda,\beta}}e^{i(k_{1_0}x_0-k_{2_0}y_0)}\sin (k_1x) \sin (k_2y)\hat W^{(h)}(\bm k_1, \bm k_2)\delta_{k_{1_0},k_{2_0}}
\end{equation}
where $\hat W_2(\bm k_1,\bm k_2)\neq 0$  for some suitable $\bm k_1,\bm k_2$, and we used the notation $\bm k_1=(k_1,k_{1_0})$ and $\bm k_2=(k_2,k_{2_0})$.
\end{prop}

Of course it is enough to show a counter-example: in fact the first order {\it non-local tadpole} (see Figure (\ref{figure_tadpoles}), the graph on the right) already breaks the diagonal form. We construct in full detail the counter-example in Appendix (\ref{appendix_non_local_tadpole}).

\begin{figure}
\centering
\begin{tikzpicture}
 [thick,decoration={
    markings,
    mark=at position 0.5 with {\arrow{>}}}] 
\fill (1,1) circle (0.06);
\fill (1,2) circle (0.06);
\fill (6,1) circle (0.06);
\fill (8,1) circle (0.06);
\node at (1,0.8) {\bf x};
\node at (6,0.8) {\bf x};
\node at (8,0.8) {\bf y};
\draw [postaction={decorate}] (0,1) -- ++(1,0);
\draw [postaction={decorate}] (1,1) -- ++(1,0);
\draw [-, decorate, decoration={snake}] (1,1) -- ++(0,1);
\draw [->] (1,2) to [out=0, in=0, looseness=1] (1,3);
\draw  (1,3) to [out=180, in=180, looseness=1] (1,2);
\draw [postaction={decorate}] (5,1) -- ++ (1,0);
\draw [-,decorate,decoration={snake}] (6,1) -- ++ (2,0);
\draw [postaction={decorate}] (8,1) -- ++ (1,0);
\draw  [postaction={decorate}] (6,1) to [out=60, in=120, looseness=1.5] (8,1);
\end{tikzpicture}
\caption{Two first order Feynman diagrams: the local tadpole on the left, and the non-local tadpole on the right.}
\label{figure_tadpoles}
\end{figure}
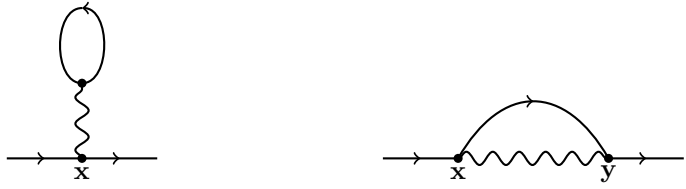

\begin{rem}
\label{remark_scale_decomposition_main_difference_DBC}
This deep difference with respect to the translation invariant theory we studied in Chapter (\ref{chapter_fermions_PBC}) underlines a technical problem. We know that there exists a Fourier base such that the free propagator can be written in a diagonal form (\ref{free_propagator_DBC}). This allows us to define the single scale propagators in a very natural way by rewriting the two dimensional space $\mathcal D^d_{\Lambda,\beta}$ as union of annuli of radii $\gamma^{h}$ and $\gamma^{h-1}$, exactly in the same conceptual way we followed in the translation invariant case (\ref{cut_off_chi_definition}).\\
The novelty comes when we try to dress the theory: the previous Proposition (\ref{proposition_kernels_not_diagonal_in_k_DBC}) tells us that we cannot trivially extend this localization procedure to the case we are dealing with in this chapter:
\begin{itemize}
\item in the translation invariant case, the $\delta(\cdot)$ function in the definition (\ref{effective_potential_scale_h_recursive}) guarantees that the entering and exiting fields $\{\hat \psi^{(\leq h)\pm}_{\bm k}\}$ carry the same momentum so, in particular, live at the same scale $h$. The renormalization process preserves this structure, so in the dressing procedure (see Subsection (\ref{subsection_anomalous_integration_PBC})) we move the first-order localized term from the interaction to the Grassman integration in a trivial way, because it has the same property of the covariance and there is no ambiguity about the momentum and the scale splitting;
\item in the case we are treating in this chapter, on the one hand the free propagator, so the covariance of the Gaussian Grassman integration, is diagonal in $\bm k$, and consequently we can rewrite it via a trivial {\it scale decomposition}, as in fact we have done; on the other hand the quadratic term we would like (by analogy with the already known case) to move from the interaction to the measure, is not diagonal in $\bm k$ space so there is not a well defined notion of scale comparable with the notion we used for the propagator.
\end{itemize}
\end{rem}

\section{Renormalization Group}
\label{section_renormalization_group_DBC}
The idea we are led by is the same as the previous Chapter (\ref{chapter_fermions_PBC}): so far we have expanded the quantities we are interested in using a trivial cluster expansion, which does not conclude the analysis because of the divergences that arise when we perform the sum, over all the possible trees, of the {\it marginal} and {\it relevant} contributions, corresponding respectively to the {\it quartic} and {\it quadratic} diagrams in the expansion. The trick will be to change a bit our point of view on this expansion.\\
We will use the results we obtained as a black box for the {\it dominant theory}, and some new ideas will be introduced in order to deal with the technical problems arising from the presence of the boundary. The strategy we will follow is the following:
\begin{itemize}
\item in order to use the dimensional bounds we showed in Corollary (\ref{corollary_norms_propagators_DBC}), we rewrite each propagator as the combination of four quasi-particles propagators as in (\ref{gd_sum_of_quasi_particle_propagators}). This means that, in Gallavotti-Nicolo's trees, the {\it field labels} $f$ are associated with $\left\{\bm x(f), \epsilon(f),\omega(f), \sigma(f)\right\}$, where $\bm x, \epsilon$ and $\omega$ are the same labels we introduced in the previous chapter, while $\sigma(f)\in\{P,R\}$,
\item we will rewrite the quartic terms as the sum of the {\it bulk quartic terms} ({\it i.e.} the translation invariant ones we introduced in the previous chapter (\ref{chapter_fermions_PBC})), consisting of the clusters containing only $P-$type propagators integrated over a properly extended domain in order to get the right quantities, and a {\it remainder}, consisting of clusters containing at least either one $R$-labeled propagator or a $\varpi-$type endpoint. \\
We will show that the  presence of the $R$-labeled propagators makes the quartic clusters irrelevant, so we do not need to renormalize them. So the study of the running coupling constant associated to the quartic term $\lambda_h$ will be exactly the same of the translation invariant part;
\item we apply the same idea to the quadratic term, but we discover a more complicated situation: 
\begin{itemize}
\item by construction, the {\it bulk quadratic term} will be treated as in the previous chapter: we will split it in an {\it order zero localized term} that will be compensated by a properly chosen constant counterterm $\nu$, and an {\it order one localized term} we will dress the propagator with,
\item by a more sophisticated dimensional analysis, we will show that the remainder part of the quadratic terms is {\it marginal} (not {\it relevant}) so, differently from the quartic term case, it is still necessary to perform an {\it order zero localization}. At this point the boundary shows its effects: since the kernels are not translation invariant (Propostion (\ref{proposition_kernels_not_diagonal_in_k_DBC})), the localization process gives rise to a {\it running coupling function} instead of a running coupling constant. The role played by the counterterm $ \mathcal N$ is to compensate, step by step, these {\it marginal non-local} terms.
\end{itemize}

\end{itemize}

We want to look at the {\it localization operator} as the {\it composition of two localization operators}: the one extracting the bulk contribution, the other extracting the dominant contributions to the Taylor expansion, as we did in the previous chapter. In particular:
\begin{itemize}
\item we want to keep as reference a theory with Dirichlet boundary conditions, so we want to extract from the quadratic terms the contribution that can be written as
\begin{equation*}
W^{d(h)}(\bm x,\bm y)=\frac{2}{\beta(L+1)}\sum_{\bm k\in\mathcal D^d_{\Lambda,\beta}} e^{-ik_0(x_0-y_0)} \sin (kx) \sin(ky)\hat W^{d(h)}(\bm k),
\end{equation*}
for some $\hat W^{d(h)}$ we will define in the next section.
\item since, as it will be clear a posteriori, the boundary corrections to the quartic term are irrelevant, we want to extract from the quartic kernel
$$\bar W^{(h)}_4(\bm x_1,\bm x_2,\bm x_3,\bm x_4)=\bar W^{(h)}_4(\bm x_1-\bm x_4, \bm x_2-\bm x_4,\bm x_3- \bm x_4),$$
which is $2(L+1)$ periodic in the space-direction of all the variables.
\end{itemize}

\subsection{"Preliminar" localization: $\mathcal L_{\mathcal B}$}

Let us first define the localization operator extracting the bulk terms $\mathcal L_{\mathcal B}$ and the renormalization operator $\mathcal R_{\mathcal B}$, where $\mathcal B$ stays for {\it bulk}, and then explain how they operate on the trees.

\begin{itemize}
\item
if $2n=2$
\begin{eqnarray}
\mathcal L_{\mathcal B} W^{(h)}_2(\bm x,\bm y):= &W_2^{d(h)}(\bm x,\bm y), \label{localization_L_D_definition}\\
\mathcal R_{\mathcal B}W^{(h)}_2(\bm x,\bm y) :=&\mathcal W_2^{(h)}(\bm x,\bm y)\label{remainder_R_D_definition},
\end{eqnarray}
where $W_2^{d(h)}(\bm x, \bm y)$ can be written as
\begin{equation}
W_2^{d(h)}(\bm x,\bm y)=\frac{2}{\beta(L+1)}\sum_{\bm k\in\mathcal D_{\Lambda,\beta}^d}e^{-ik_0(x_0-y_0)}\sin(kx) \sin(ky) \hat W^{d(h)}_2(\bm k),
\end{equation}
for a suitable $\hat W^{d(h)}(\cdot)$ that we will explicitly build up in the next paragraph, while $$\mathcal R_{\mathcal B}W_2^{(h)}=\mathcal W^{(h)}_2(\bm x,\bm y)=W^{(h)}_2(\bm x,\bm y)-W^{d(h)}_2(\bm x,\bm y)$$ contains, in its cluster representation, at least one {\it non-translation invariant} graph element, {\it i.e.} either a remainder propagator $g^{(k)}_R(\bm x,\bm y)$, $k\geq h$, or a $\varpi$-type endpoint.
\item if $2n=4$,
\begin{eqnarray}
\mathcal L_{\mathcal{B}}W^{(h)}_4(\bm x_1,\bm x_2,\bm x_3, \bm x_4):=&\bar W^{(h)}_4(\bm x_1-\bm x_4,\bm x_2-\bm x_4,\bm x_3-\bm x_4),\label{localization_L_T_definition}\\
\mathcal R_{\mathcal{B}}W^{(h)}_4(\bm x_1,\bm x_2,\bm x_3, \bm x_4):=&\mathcal W_4^{(h)}(\bm x_1, \bm x_2,\bm x_3, \bm x_4),\label{remainder_R_T_definition}
\end{eqnarray}
while $\mathcal R_{\mathcal B}W_4^{(h)}=\mathcal W^{(h)}_4(\bm x_1,\bm x_2,\bm x_3,\bm x_4)=W^{(h)}_4(\bm x_1,\bm x_2,\bm x_3,\bm x_4)-\bar W^{(h)}_4(\bm x_1,\bm x_2,\bm x_3,\bm x_4)$ contains, in its cluster representation, at least one {\it non-translation invariant} graph element, {\it i.e.} either a remainder propagator $g^{(k)}_R(\bm x,\bm y)$, $k\geq h$, or a $\varpi$-type endpoint
\end{itemize}

\paragraph{Rigorous definition of $\mathcal L_{\mathcal B}$}

Let us recall that 
$$g^{(h)}(\bm x,\bm y)=g^{(h)}_P(\bm x-\bm y)+g^{(h)}_R(\bm x, \bm y),$$
where $g_P(\bm x-\bm y)$ is $2(L+1)$-periodic in the {\it real-space direction}.  Besides, we stress that
\begin{equation}
\begin{split}
g_P^{(h)}((x,x_0),(y,y_0))=-g_R^{(h)}((-x,x_0),(y,y_0))=\\
=-g_R^{(h)}((x,x_0),(-y,y_0))=g_P^{(h)}((-x,x_0),(-y,y_0)).
\end{split}
\label{g_P_g_R_symmetry}
\end{equation}
\subparagraph{Quadratic terms}

Let us first write down a complete decomposition of $W_2^{(h)}$, that we are going to comment term by term in the following:

\begin{equation}
\label{expansion_of_quadratic_terms}
\begin{split}
W_2^{(h)}(\bm x,\bm y)= W_2^{diff(h)}(\bm x,\bm y)+\left(W_2^{(h)}(\bm x,\bm y)-W_2^{diff(h)}(\bm x,\bm y)\right)=\\
=\bar W_2^{(h)}(\bm x-\bm y)+\left(W_2^{(h)}(\bm x,\bm y)-W_2^{diff(h)}(\bm x,\bm y)\right)+\left(W_2^{diff(h)}(\bm x,\bm y)-\bar W_2^{(h)}(\bm x-\bm y)\right)=\\
=W_2^{d(h)}(\bm x,\bm y)+\left(W_2^{(h)}(\bm x,\bm y)-W_2^{diff(h)}(\bm x,\bm y)\right)+\\+ 
\left(W_2^{diff(h)}(\bm x,\bm y)-\bar W_2^{(h)}(\bm x-\bm y)\right)+ \bar W_2^{(h)}(x+y, x_0-y_0),
\end{split}
\end{equation}
where
\begin{itemize}
\item $W^{diff(h)}_2(\bm x,\bm y)$ is defined as the kernel associated to {\it the sum of all those trees such that there are no $\varpi$-type endpoints and $\sigma(f)=P$ for each $f\in\bigcup_{v\in V_f(\tau)} I_v$}. This implies, by construction, that $$\left(W_2^{(h)}(\bm x,\bm y)-W_2^{diff(h)}(\bm x,\bm y)\right)$$ is associated with the trees such that there are {\it at least either two field labels associated with $\sigma(f)=R$ or a $\varpi$-type endpoint and, using Lemma (\ref{lemma_gram_hadamard_scalar_product_off_diagonal_propagators}), it can be rewritten using the determinant expansion we introduced in order to overcome the combinatorial problem}. We anticipate, and it will be discussed later, that there might be the case in which there is no remainder propagator associated with any $\ell\in T$, but the remainder propagators belong only to the matrix $G^{h_v,T_v}$.\\
Recall that $g^{(h)}_P$ is $2(L+1)$-periodic in the space direction, while the integrals over the {\it inner points} of the cluster are performed on $\Lambda$. So, in order to get a translation invariant $2(L+1)$ periodic function, we have to extend in a proper way the integration domain for any inner point $\bm z_i$
\begin{equation*}
\sum_{z_i\in \Lambda}\to \sum_{z_i=-L-1}^L.
\end{equation*}
\item $\bar W^{(h)}_2(\bm x-\bm y)$ is obtained starting from $W^{diff(h)}$ in the following way: 
\begin{enumerate}
\item let $\{\bm x, \bm z_1,\dots, \bm z_n, \bm y\}$ be the {\it space-time} variables associated with the endpoints of the tree related to $W^{diff(h)}$,
\item keeping fixed the endpoints associated with $\bm x$ and $ \bm y$, for each possible unordered ({\it i.e.} the order of the elements does not play any role) $k$-tuple, $1\leq k \leq n$, of endpoints we perform the following operation on the position labels:
$$\left\{(z_{i_1},z_{{i_1}_0});\dots; (z_{i_k},z_{{i_k}_0})\right\}\to \left\{(-z_{i_1},z_{{i_1}_0});\dots; (-z_{i_k},z_{{i_k}_0})\right\}, \hspace{3mm} \forall \hspace{3mm} 1\leq k \leq n,$$
\item if we add all the $(\bm x,\bm y)$-depending kernels associated with the trees obtained in the previous point to $W^{(diff(h))}$, we obtain  $\bar W^{(h)}_2(\bm x,\bm y)=\bar W^{(h)}_2(\bm x-\bm y)$ that, by construction, is a $2(L+1)$ periodic function, in the space direction, depending on the difference of its arguments. 
\end{enumerate}
By construction, and thanks to the symmetry (\ref{g_P_g_R_symmetry}), $\bar W^{(h)}_2(\bm x-\bm y)-W^{diff(h)}_2(\bm x,\bm y)$ contains at least a remainder propagator $g_R(\bm z_i,\bm z_j)$.
\item $\bar W_2^{(h)}(x+y, x_0-y_0)$ is simply obtained changing $(y, y_0)\to (-y, y_0)$ in all the trees involved in the previous step, so also $\bar W_2^{(h)}(x+y, x_0-y_0)$ contains at least a remainder propagator. 
\end{itemize}
So, by a straightforward calculation analogous to what we showed in the proof of Lemma (\ref{lemma_reflection_trick}), we obtain that
\begin{equation}
\begin{split}
W^{d(h)}_2(\bm x,\bm y):=&\bar W^{(h)}_2(\bm x- \bm y)- \bar W^{(h)}_2(x+y; x_0-y_0)=\\=&\frac{2}{\beta(L+1)}\sum_{\bm k\in\mathcal D^d_{\Lambda,\beta}} e^{-ik_0(x_0-y_0)} \sin(kx) \sin(k y) \hat W^{d(h)}(\bm k),
\end{split}
\end{equation}
thanks to the same argument we used in Lemma (\ref{lemma_reflection_trick}).
\paragraph{Quartic term}

Since the manipulations of the tree are basically the same as in the quadratic case, let us underline only the differences:
\begin{itemize}
\item now there are four external points to be kept fixed: $(\bm x_1,\bm x_2, \bm x_3, \bm x_4)$,
\item we can analogously define $W^{diff(h)}_4$ as the sum of all the trees such that there are no $\varpi-$type endpoints and such that $\sigma(f)=P$ for all the field labels $f\in\cup_{v\in V_f(\tau)} I_v$.
\item we can define $\bar W^{(h)}_4$ following  the same procedure over the endpoints as before, so that
\begin{equation}
\label{explicit_expansion_quartic_terms}
\begin{split}
W^{(h)}_4(\bm x_1,\bm x_2, \bm x_3, \bm x_4)= \bar W^{(h)}_4(\bm x_1-\bm x_4, \bm x_2-\bm x_4, \bm x_3-\bm x_4)+\\+ \left( W^{diff(h)}(\bm x_1,\bm x_2, \bm x_3, \bm x_4)-\bar W^{(h)}_4(\bm x_1-\bm x_4, \bm x_2-\bm x_4, \bm x_3-\bm x_4)\right)+\\+ \left(W^{(h)}_4(\bm x_1,\bm x_2, \bm x_3, \bm x_4)- W^{diff(h)}(\bm x_1,\bm x_2, \bm x_3, \bm x_4)\right).
\end{split}
\end{equation}
\end{itemize}

\begin{rem}
\label{remark_R_B_decomposition_R1_R2}
We introduce a notation that will make the proof of the main theorem more readable:
\begin{equation}
\mathcal R_\mathcal B=\mathcal R^{(1)}_\mathcal B+\mathcal R^{(2)}_\mathcal B,
\end{equation}
acting as follows:
\begin{eqnarray}
\begin{aligned}
\mathcal{R}^{(1)}_\mathcal B W_2^{(h)}(\bm x,\bm y)&=W_2^{(h)}(\bm x,\bm y)-W_2^{diff(h)}(\bm x,\bm y),\\
\mathcal{R}^{(2)}_\mathcal B W_2^{(h)}(\bm x,\bm y)&= W_2^{diff(h)}(\bm x,\bm y)-\bar W_2^{(h)}(\bm x-\bm y)+\bar W^{(h)}_2(x+y,x_0-y_0),\\
\mathcal{R}^{(1)}_\mathcal B W_4^{(h)}(\bm x_1,\bm x_1,\bm x_3, \bm x_4)&=W_4^{(h)}(\bm x_1,\bm x_1,\bm x_3,\bm x_4)-W_4^{diff(h)}(\bm x_1,\bm x_2,\bm x_3,\bm x_4),\\
\mathcal{R}^{(2)}_\mathcal B W_4^{(h)}(\bm x_1,\bm x_2,\bm x_3, \bm x_4)&= W_4^{diff(h)}(\bm x_1,\bm x_2,\bm x_3,\bm x_4)-\bar W_4^{(h)}(\bm x_1,\bm x_2,\bm x_3,\bm x_4).
\end{aligned}
\end{eqnarray}
Furthermore, we underline that 
\begin{enumerate}
\item the operator $\mathcal R^{(1)}_\mathcal B$ simply {\bf selects} the trees such that at least
\begin{itemize}
\item either two $f\in I_{v_0}$ are associated with $\sigma (f)=R$,
\item or an endpoint is a $\varpi-$type endpoint,
\end{itemize}
but it does not modify anything of the tree,
\item the operator $\mathcal R^{(2)}_\mathcal B$ operates on the trees that do not contain neither field labels $f\in I_{v_0}$ associated with $\sigma(f)=R$ nor $\varpi$-type end points, and {\bf modifies} the coordinates-labels of the tree.
\end{enumerate}

\end{rem}

\subsection{Definition of localization}

\paragraph{$\mathcal L_\mathcal T$ and $\tilde{\mathcal L}_\mathcal T$ localization operators}

Let us recall that
\begin{equation*}
\begin{split}
\mathcal L_{\mathcal B}W^{(h)}_2(\bm x,\bm y)&=W^{d(h)}_2(\bm x,\bm y),\\
\mathcal R_{\mathcal B}W^{(h)}_2(\bm x,\bm y)&=\mathcal W^{(h)}_2(\bm x,\bm y),\\
\mathcal L_{\mathcal B}W^{(h)}_4(\bm x_1,\bm x_2,\bm x_3,\bm x_4)&=\bar W^{(h)}_4(\bm x_1-\bm x_4,\bm x_2-\bm x_4,\bm x_3-\bm x_4),\\
\mathcal R_{\mathcal B}W^{(h)}_4(\bm x_1,\bm x_2,\bm x_3,\bm x_4)&=\mathcal W^{(h)}_4(\bm x_1,\bm x_2,\bm x_3,\bm x_4),
\end{split}
\end{equation*}
where we recall that $W^{d(h)}_2$ is {\it diagonal in the Dirichlet basis} and $\bar W^{(h)}_4$ is translation invariant and $2(L+1)\times \beta$-periodic. Plugging these decompositions in the expression of the effective potential, 
\begin{equation}
\begin{split}
\mathcal V^{(h)}(\psi^{(\leq h)})=\int d\bm x d\bm y\psi^{(\leq h)+}_{\bm x}\psi^{(\leq h)-}_{\bm y}W^{d(h)}_2(\bm x,\bm y)+
\int d\bm x d\bm y\psi^{(\leq h)+}_{\bm x}\psi^{(\leq h)-}_{\bm y}\mathcal W^{(h)}_2(\bm x,\bm y)+\\
+\int d\bm x_1\dots d\bm x_4\psi^{(\leq h)+}_{\bm x_1}\psi^{(\leq h)+}_{\bm x_2}\psi^{(\leq h)-}_{\bm x_3}\psi^{(\leq h)-}_{\bm x_4} \bar W_4^{(h)}(\bm x_1-\bm x_4,\bm x_2-\bm x_4,\bm x_3-\bm x_4)+\\
+\int d\bm x_1\dots \bm x_4\psi^{(\leq h)+}_{\bm x_1}\psi^{(\leq h)+}_{\bm x_2}\psi^{(\leq h)-}_{\bm x_3}\psi^{(\leq h)-}_{\bm x_4}\mathcal W_4^{(h)}(\bm x_1,\bm x_2,\bm x_3,\bm x_4)+\\
+\sum_{n\geq 3}\int d\bm x_1\dots d\bm x_{2n}\left[\prod_{j=1}^n\psi^{(\leq h)+}_{\bm x_{2j-1}}\psi^{(\leq h)-}_{\bm x_{2j}}\right] W^{(h)}_{2n}(\bm x_1,\dots,\bm x_{2n}).
\end{split}
\label{effective_potential_decomposed_2-4_kerlels_DBC}
\end{equation}
Let us recall that,
\begin{equation}
\mathcal L_{\mathcal B}W^{(h)}_{2}(\bm x,\bm y)=W^{d(h)}_2(\bm x,\bm y)=\frac{2}{\beta(L+1)}\sum_{\bm k\in \mathcal{D}^d_{\Lambda,\beta}}e^{ik_0(x_0-y_0)}\sin (kx)\sin (ky) \hat W^{d(h)}_2(\bm k),
\end{equation}
that implies that the quadratic term can be rewritten in momentum representation as
\begin{equation}
\int d\bm x d\bm y\psi^{(\leq h)+}_{\bm x}\psi^{(\leq h)-}_{\bm y}W^{d(h)}_2(\bm x,\bm y)=\frac{2}{\beta(L+1)}\sum_{\bm k\in\mathcal{D}^d_{\Lambda,\beta}}\hat \psi^{(\leq h)+}_{\bm k}\hat \psi^{(\leq h)-}_{\bm k}\hat W^{d(h)}_{2}(\bm k).
\label{2el_W^d_in_diagonal_form}
\end{equation}
It is worth noting that, by writing this quadratic term in this diagonal form, we overcome the {\it quasi-particles definition} issue we pointed out in Remark (\ref{remark_quasi_particles_issue_DBC}): indeed we know that in the {\it original dual space} $\mathcal D^d_{\Lambda,\beta}$ there is only one Fermi point $\bm p_F=(p_F,0)$, so we can perform the change of variable $\bm k=\bm k'+\bm p_F$ and rewrite the latter expression as
\begin{equation}
\begin{split}
\frac{2}{\beta(L+1)}\sum_{\bm k'\in \mathcal D'^{d}_{\Lambda,\beta}}\hat \psi^{(\leq h)+}_{\bm k'+\bm p_F}\hat \psi^{(\leq h)-}_{\bm k'+\bm p_F}\hat W^{d(h)}_{2}(\bm k'+\bm p_F).
\end{split}
\label{2el_W^d_in_diagonal_quasiparticles_form}
\end{equation}
Of course this expression suggests us a {\it natural way to localize} directly in the {\it momentum-space} by Taylor expanding the kernel around the Fermi point $\bm p_F$, analogously to what we have done in (\ref{localization_2el_infinite_volume_limit}).\\
As we have done in the previous Chapter (\ref{chapter_fermions_PBC}), in sake of simplicity we give here the definition of Localization at infinite volume, and we refer to Appendix (\ref{appendix_finite_volume_loc_DBC}) to the rigorous definitions that take into account the finite volume corrections.
\begin{itemize}
\item {\bf Case $2n=2$, kernel $W^{d(h)}_2=\mathcal L_\mathcal BW^{(h)}_2$} As we already commented, formulae (\ref{2el_W^d_in_diagonal_form}) and (\ref{2el_W^d_in_diagonal_quasiparticles_form}) allow us to localize proceeding by analogy with the translation invariant case, defining a localization procedure directly in the dual space:
\begin{equation}
\begin{split}
\mathcal L_\mathcal T \left[\mathcal L_\mathcal B\left(\int d\bm x d\bm y \psi^{(\leq h)+}_{\bm x}\psi^{(\leq h)-}_{\bm y}W^{(h)}_2(\bm x,\bm y)\right)\right]=\\
\\\mathcal L_\mathcal T  \left(\frac{2}{\beta(L+1)}\sum_{\bm k'\in \mathcal D'^{d}_{\Lambda,\beta}} \hat \psi^{(\leq h)+}_{\bm k'+p_F}\hat \psi^{(\leq h)-}_{\bm k'+p_F}\hat W^{d(h)}_2(\bm k'+\bm p_F)\right)=\\
=\frac{2}{\beta(L+1)}\sum_{\bm k'\in \mathcal D'^{d}_{\Lambda,\beta}} \hat \psi^{(\leq h)+}_{\bm k'+p_F}\hat \psi^{(\leq h)-}_{\bm k'+p_F}\left(\hat W^{d(h)}_2(\bm p_F)+k_0\partial_{k_0} \hat W^{d(h)}_2(\bm p_F)+k'\partial_k \hat W^{d(h)}_2(\bm p_F)\right)
\end{split}
\end{equation}
where we stress once more that this localization definition is independent of the quasi-particles.
\item {\bf Case $2n=2$, kernel $\mathcal W^{(h)}_2=\mathcal R_\mathcal B W^{(h)}_2$} Because of the non-diagonality of the kernel $\mathcal W^{(h)}_2$, there is no advantage in defining a localization procedur in $\bm k$ space, so we work directly in the real spacetime:
\begin{equation}
\begin{split}
\tilde{\mathcal L}_\mathcal T \mathcal L_\mathcal B\left(\int d\bm x d\bm y \psi^{(\leq h)+}_{\bm x}\psi^{(\leq h)-}_{\bm y}W^{(h)}_2(\bm x,\bm y)\right)=\\
=\tilde{\mathcal L}_\mathcal T\int d\bm x d\bm y \psi^{(\leq h)+}_{\bm x}\psi^{(\leq h)-}_{\bm y}\mathcal W^{(h)}_2(\bm x,\bm y)=\\
=\left.\int d\bm x d\bm y \psi^{(\leq h)+}_{\bm x}\psi^{(\leq h)-}_{\bm y}\right|_{x_0=y_0} \mathcal W^{(h)}_2(\bm x,\bm y).
\end{split}
\end{equation}
We used the symbol $\tilde{•}$ to stress the fact that $\tilde{\mathcal L}_\mathcal T$ operates only on the time variables.
\item {\bf Case $2n=4$, kernel $\bar W^{(h)}_4=\mathcal L_\mathcal B W^{(h)}_4$}
In this case, as in the previous chapter, we define
\begin{equation}
\begin{split}
\mathcal L_\mathcal T\left( \mathcal L_\mathcal B \int d\bm x_1\dots d\bm x_4 W_4^{(h)}(\bm x_1,\bm x_2,\bm x_3, \bm x_4) e^{-ip_F(\omega_1 x_1 +\omega_2 x_2 - \omega_3 x_3 - \omega_4 x_4)}\right. \cdot \\\left. \cdot 
\psi^{(\leq h)+}_{\sigma_1,\omega_1,\bm x_1}\psi^{(\leq h)+}_{\sigma_2,\omega_2,\bm x_2}\psi^{(\leq h)-}_{\sigma_3,\omega_3,\bm x_3}\psi^{(\leq h)-}_{\sigma_4,\omega_4,\bm x_4}\right)=\\
=\int d\bm x_1\dots d\bm x_4 \mathcal L_\mathcal B W_4^{(h)}(\bm x_1,\bm x_2,\bm x_3, \bm x_4) e^{-ip_F(\omega_1 x_1 +\omega_2 x_2 - \omega_3 x_3 - \omega_4 x_4)} \cdot \\ \cdot 
\psi^{(\leq h)+}_{\sigma_1,\omega_1,\bm x_4}\psi^{(\leq h)+}_{\sigma_2,\omega_2,\bm x_4}\psi^{(\leq h)-}_{\sigma_3,\omega_3,\bm x_4}\psi^{(\leq h)-}_{\sigma_4,\omega_4,\bm x_4}.
\end{split}
\end{equation}
\item {\bf Case $2n=4$, kernel $\mathcal W^{(h)}_4=\mathcal R_\mathcal B W^{(h)}_4$}
In this case, we do not need to renormalize
\begin{equation}
\begin{split}
\mathcal L_\mathcal T\left( \mathcal R_\mathcal B \int d\bm x_1\dots d\bm x_4 W_4^{(h)}(\bm x_1,\bm x_2,\bm x_3, \bm x_4) e^{-ip_F(\omega_1 x_1 +\omega_2 x_2 - \omega_3 x_3 - \omega_4 x_4)}\right. \cdot \\ \cdot 
\psi^{(\leq h)+}_{\sigma_1,\omega_1,\bm x_1}\psi^{(\leq h)+}_{\sigma_2,\omega_2,\bm x_2}\psi^{(\leq h)-}_{\sigma_3,\omega_3,\bm x_3}\psi^{(\leq h)-}_{\sigma_4,\omega_4,\bm x_4}=0,
\end{split}
\end{equation}
so $\mathcal R_\mathcal T\mathcal R_\mathcal B=\mathcal R_\mathcal B$ if it operates on a quartic term (the same holds for $\tilde{\mathcal L}_\mathcal T$, so for $\tilde{\mathcal R}_\mathcal T\mathcal R_\mathcal B$).
\item Finally, if $2n\geq 6$
\begin{equation}
\mathcal L_\mathcal T\int d\bm x_1\dots d \bm x_{2n}\left(\prod_{j=1}^n\psi^{(\leq h)+}_{\bm x_{2j-1}}\psi^{(\leq h)-}_{\bm x_{2j}}\right)W_{2n}^{(h)}(\bm x_1,\dots, \bm x_{2n})=0,
\end{equation}
and the same holds for $\tilde {\mathcal L}_\mathcal T$.
\end{itemize}
\paragraph{Composition of $\mathcal L_{\mathcal B}$, $\mathcal L_\mathcal T$  and $\tilde{\mathcal L}_\mathcal T$} 
By composing the operators we introduced, we define a linear operator $\mathcal L$, and consequently $\mathcal R=\left(1-\mathcal L\right)$, acting as follows:
\begin{itemize}

\item if $2n=2$:
\begin{eqnarray}
\mathcal L=\mathcal L_\mathcal T\mathcal L_{\mathcal B}+\tilde {\mathcal L}_\mathcal T\mathcal R_{\mathcal B},\\
\mathcal R= \mathcal R_\mathcal T\mathcal L_{\mathcal B}+\tilde {\mathcal R}_\mathcal T\mathcal R_{\mathcal B},
\end{eqnarray}
\item if $2n=4$
\begin{eqnarray}
\mathcal L=\mathcal L_\mathcal T\mathcal L_{\mathcal B},\\
\mathcal R= \mathcal R_{\mathcal T}\mathcal L_{\mathcal B}+\mathcal R_{\mathcal B},
\end{eqnarray}
\item if $2n=6$ $\mathcal R$ acts as the identity.
\end{itemize}

\begin{rem}\label{remark_commutation_renormalization operators}
As we commented in Remark (\ref{remark_R_B_decomposition_R1_R2}), the operator $\mathcal R^{(1)}_\mathcal B$ does not modify the trees, so it holds
$$\tilde{\mathcal R}_\mathcal T\mathcal R^{(1)}_\mathcal B=\mathcal R^{(1)}_\mathcal B\tilde{\mathcal  R}_\mathcal T.$$
On the other hand, also the operator $\tilde {\mathcal R}_\mathcal T$ commutes with $\mathcal R^{(2)}_\mathcal B$ since the first one acts on the time-component of the coordinates, while the second one on the space-component, so
$$\tilde{\mathcal R}_\mathcal T\mathcal R^{(2)}_\mathcal B=\mathcal R^{(2)}_\mathcal B\tilde{\mathcal  R}_\mathcal T.$$
\end{rem}

The local part of the effective potential at scale $h$ is
\begin{equation}
\mathcal L\mathcal V^{(h)}(\psi^{(\leq h)})=\gamma^h n_h F_\nu^{(\leq h)}+\gamma^h \left(\pi_h\ast F_{\varpi}^{(\leq h)}\right)+z_h F_{\zeta}^{(\leq h)}+a_h F_{\alpha}^{(\leq h)}+ l_h F_{\lambda}^{(\leq h)},
\label{linearized_effective_potential_DBC}
\end{equation}
where 
\begin{eqnarray}
F_\nu^{(\leq h)}=\frac{2}{\beta(L+1)}\sum_{\bm k'\in\mathcal D'^d_{\Lambda,\beta}}\hat \psi^{(\leq h)+}_{\bm k'+\bm p_F}\hat \psi^{(\leq h)-}_{\bm k'+\bm p_F}\\
\label{linearized_effective_potentia_nu_DBC}
 \pi_h\ast F_{\varpi}^{(\leq h)}=\int d\bm x\int dy\psi^{(\leq h)+}_{\bm x}\psi^{(\leq h)-}_{\bm y}\left.\right|_{y_0=x_0} \pi_h(x,y)\\
\label{linearized_effective_potentia_tildenu_DBC}
F_{\zeta}^{(\leq h)}=\frac{2}{\beta(L+1)}\sum_{\bm k'\in \mathcal D'^d_{\Lambda,\beta}}(-ik_0)\hat \psi^{(\leq h)+}_{\bm k'+\bm p_F}\hat \psi^{(\leq h)-}_{\bm k'+\bm p_F}\\
\label{linearized_effective_potentia_zeta_DBC}
F_{\alpha}^{(\leq h)}=\frac{2}{\beta(L+1)}\sum_{\bm k'\in\mathcal D'^d_{\Lambda,\beta}} v_0k'\hat \psi^{(\leq h)+}_{\bm k'+\bm p_F}\hat \psi^{(\leq h)-}_{\bm k'+\bm p_F},\\
\label{linearized_effective_potentia_alpha_DBC}
F_{\lambda}^{(\leq h)}=\sum_{\bm \omega}\sum_{\bm \sigma}\int_{[0,\beta)} dx_0\sum_{ x\in\Lambda}\psi^{(\leq h)+}_{\sigma_1,\omega_1,\bm x}\psi^{(\leq h)+}_{\sigma_2,\omega_2,\bm x}\psi^{(\leq h)-}_{\sigma_3,\omega_3,\bm x}\psi^{(\leq h)-}_{\sigma_4,\omega_4,\bm x}.
\label{linearized_effective_potentia_lambda_DBC}
\end{eqnarray}

where $v_0=\sin p_F$.\\
While the constants $(n_h, z_h, a_h, l_h)$ behave as the analogous quantities appearing in (\ref{local_effective_potential_scale_h_PBC}), so we will call them the {\it (non-rescaled) bulk running coupling constants}, the novelty is the function $\pi_h(x, y)$, due to the boundary effects, defined as
\begin{equation}
\gamma^h \pi_h(x,y):=\int_{[0,\beta)} dy_0\mathcal W^{(h)}_2(x,y;x_0-y_0).
\label{definition_tilde_n}
\end{equation}

\begin{rem}
\label{remark_arbitrariness_localization_point_DBC}
We point out once more that we define the linear operator $\mathcal L$ to rewrite $\mathcal V^{(h)}(\psi^{(h)})=\mathcal L\mathcal V^{(h)}(\psi^{(h)})+\left(1-\mathcal L\right)\mathcal V^{(h)}(\psi^{(h)})$ in such a way that the term $\left(1-\mathcal L\right)\mathcal V^{(h)}(\psi^{(h)})$ is irrelevant: in defining $\pi_h(x,y)$, we have choosen to localize only in the time variable, while we could have choosen a localization consisting in a combination of zeroth order localization both in space and in time: since, in order to prove the theorem, this reduces to an aesthetic choice, we prefer the definition we defined in order to drop the dependence on quasi-particle labels, and to have, scale by scale, the term $\pi_h\ast F_{\varpi}^{d(\leq h)}$ having formally the same shape as $ \mathcal N$.
\end{rem}

\subsection{Scale h integration and dressed theory}
\label{subsection_anomalous_integration_and_dressed_theory_DBC}

Let us comment the terms which the local part consists of:
\begin{itemize}
\item $l_h F_\lambda^{(\leq h)}$ reproduces, on scale $h$, the initial two points interaction with a different interaction potential, which effectively behaves as  {\it bulk potential}, being the same we found in the translation invariant setting;
\item $n_hF_\nu^{(\leq h)}$ reproduces the counterterm operator of the initial $\mathcal V(\psi)$, where the constant value $\nu$ is replaced by $n_h$;
\item the sum of $a_h F_\alpha^{(\leq h)}$ and $z_hF_\zeta^{(\leq h)}$  has the same shape, up to $\mathcal O(k'^2)$ terms, as 
\begin{equation*}
\begin{split}
\left(\hat g^{(h)}(\bm k'+\bm p_F)\right)^{-1}=\left(-ik_0+(1-\cos k')\cos p_F+ v_0\sin k'\right) f^{-1}_h(\bm k')=\\
=\left(-ik_0+v_0k'+(1-\cos k')\cos p_F+ v_0(\sin k'-k')\right) f^{-1}_h(\bm k')=:\\
=:\left(-ik_0+ v_0 k'+t_{h}(k')\right) f^{-1}_h(\bm k'),
\end{split}
\end{equation*} 
with constants $a_h$ and $z_h$ replacing $1$, and where we called $t_{h}(k')$ the $\mathcal O(k'^2)$ term (even though the subscript $h$ does not play any rule at the moment, it is anyway a convenient notation for what we will do in the next steps). 
\item the term $\pi_h \ast F_{\varpi}^{(\leq h)}$ deserves some deeper comment. Indeed, its role is the same role played by the counterterm 
$$\varpi \int_{[0,\beta)}dx_0\sum_{ x\in\Lambda}\int_{[0,\beta)}dy_0\sum_{ y\in\Lambda}\psi^{(\leq 0)}_{\bm x}\psi^{(\leq 0)}_{\bm y} \pi(x,y)\delta_{x_0,y_0},$$
in the scale zero effective potential $\mathcal V^{(0)}$. The reader may be surprised by the fact that we localized only in the {\it time variable} keeping track of a {\it non-local} counterterm, since a localization in real-space variables (analogous to what we defined in $F^{(\leq h)}_\nu$) not only would give us the right scale gain ({\it i.e.} an irrelevant remainder), but it would produce a {\it local} counterterm. However, one should recall that the {\it real space localization} has to be performed, in order to get the right scale gains, on the {\it quasi-particle} fields: so a priori, after a real-space localization, being the momentum no longer preserved, we would get {\it four different} running coupling functions $\{\pi_{\omega,\omega'}\}_{\omega,\omega'\in\{\pm\}}$.  So, analogously to what we did in Remark (\ref{remark_uniqueness_of_counterterm_PBC}), we should exploit some symmetry of the kernels to find out that actually these four functions originate from the same function: in fact it is not the case, as it can be checked at the lowest order (the tadpole Feynman graph), and the fact that these functions are different reflects obviously the fact that these kernels are non-local so, in order to compensate them, it is necessary to introduce a non-local counterterm. We refer to the next chapter for some heuristic discussion about how to deal with these terms.
\end{itemize}

\paragraph{Dressed theory}

So, we can iteratively perform the integral
$$\int P(d\psi^{(\leq h)})e^{-\mathcal V^{(h)}(\psi^{\leq h})},$$
by including, step by step, the terms (\ref{linearized_effective_potentia_zeta_DBC}) and (\ref{linearized_effective_potentia_alpha_DBC}) into the Gaussian Grassman measure, as follows.\\
Let us introduce a sequence of constants $\{Z_h\}_{h\leq 0}$, and $Z_0=1$, and let us define a function
\begin{equation}
C_h(\bm k')^{-1}=\sum_{j=h_L}^hf_j(\bm k')
\end{equation}
where, as already mentioned, $h_L$ is defined by $-\lfloor \log_\gamma L \rfloor =h_L$. Let us imagine that the Grassman variables $\psi^{(0)},\dots, \psi^{(h+1)}$ have been already integrated, so that we are left with
\begin{equation}
\int P_{Z_h}\left(d\psi^{(\leq h)}\right)e^{-\mathcal V^{(h)}\left(\sqrt{ Z_h} \psi^{(\leq h)}\right)},
\end{equation}
where, up to a constant, the Gaussian Grassman measure is 
\begin{equation}
\begin{split}
P_{Z_h}\left(d\psi^{(\leq h)}\right)=\left(\prod_{\bm k'\in \mathcal D'^{d}_{\Lambda,\beta}}d\psi^{(\leq h)+}_{\bm k',+}d\psi^{(\leq h)-}_{\bm k',+}\right)\\
\exp\left\{-\frac{2}{\beta(L+1)} \sum_{\bm k'\in \mathcal D'^{d}_{\Lambda,\beta}}C_h(\bm k')Z_h \left[-ik_0+v_0 k'+t_{h}(k')\right]\psi^{(\leq h)+}_{\bm k',+}\psi^{(\leq h)-}_{\bm k',+}\right\}.
\end{split}
\end{equation}
Now we recall that, in the previous subsection, we rewrote 
$$\mathcal V^{(h)}=\mathcal L\mathcal V^{(h)}+\mathcal R \mathcal V^{(h)},$$
with $\mathcal R:=\left( 1-\mathcal L\right)$, and up to the linear coefficients, the linear combination $a_hF^{(\leq h)}_\alpha+z_h F^{(\leq h)}_{\zeta}$ has the same structure as the linear part of the covariance the integration is associated with. So we move, after proper manipulation that are included in the following definitions, these to terms into the measure, leaving in the interaction  the terms $n_hF^{(\leq h)}_\nu$ and $\pi_h\ast F^{(\leq h)}_\varpi$, besides the (harmless) renormalized part $\mathcal R \mathcal V^{(h)}$.\\
So, if $\mathcal N_h$ is a suitable renormalization constant, we rewrite
\begin{equation}
\int P_{Z_h}(d\psi^{(\leq h)})e^{-\mathcal V^{(h)}\left(\sqrt{Z_h}\psi^{(\leq h)}\right)}= \frac{1}{\mathcal N_h} \int \tilde P_{Z_{h-1}}(d\psi^{(\leq h)})e^{-\tilde{\mathcal V}^{(h)}\left(\sqrt{Z_h}\psi^{(\leq h)}\right)},
\label{integral_dressed_theory_leqh_DBC}
\end{equation}
with
\begin{equation}
\begin{split}
\tilde P_{Z_{h-1}}\left(d\psi^{(\leq h)}\right)=\left(\prod_{\bm k'\in \mathcal D'^{d}_{\Lambda,\beta}}d\psi^{(\leq h)+}_{\bm k'+p_F}d\psi^{(\leq h)-}_{\bm k'+p_F}\right)\\
\exp\left\{-\frac{2}{\beta(L+1)} \sum_{\bm k'\in \mathcal D'^{d}_{\Lambda,\beta}}C_h(\bm k')Z_{h-1}(\bm k') \left[-ik_0+v_0 k'+\vartheta_h(\bm k')\right]\psi^{(\leq h)+}_{\bm k'+p_F}\psi^{(\leq h)-}_{\bm k'+p_F}\right\},
\end{split}
\label{measure_tildePZ_h-1_DBC}
\end{equation}
where:
\begin{equation}
\label{dressing_measure_defn_tildeV_DBC}
\begin{split}
Z_{h-1}(\bm k')&=Z_h\left(1+C_h^{-1}(\bm k')z_h\right),\\
Z_{h-1}&=Z_h(1+z_h),\\
\tilde{\mathcal V}^{(h)}&=\mathcal L\tilde{\mathcal V}^{(h)}+\left(1-\mathcal L\right)\mathcal V^{(h)},\\
\mathcal L\tilde{\mathcal V}^{(h)}&=\gamma^hn_h F^{(\leq h)}_\nu+\gamma^h \pi_h\ast F^{(\leq h)}_\varpi+(a_h-z_h)F^{(\leq h)}_\alpha+l_hF^{(\leq h)}_\lambda.
\end{split}
\end{equation}
and $\vartheta_h(\bm k')$ is iteratively defined as follows
\begin{equation}
\vartheta_h(\bm k')=
\begin{cases}
t_0( k') &\mbox{ if } h=0,\\
\frac{Z_{h+1}}{Z_{h}(\bm k')}\vartheta_{h+1}(\bm k') &\mbox{ if } h<0.
\end{cases}
\end{equation}
Of course, in this way we {\it dressed} the whole infrared theory ({\it i.e.} the linear part of the {\it whole propagator} $g^{(\leq h)}$), so in order to perform the {\it next single scale integration} in the iterative procedure we have to define a single scale measure, meaning that, using the addition principle (\ref{addition_principle}) and the change of the Gaussian Grassman integration measure (\ref{change_of_integration_measure_property})  we have to rewrite the integral (\ref{integral_dressed_theory_leqh_DBC}) as
\begin{equation}
\frac{1}{\mathcal N_h}\int P_{Z_{h-1}}(d\psi^{(\leq h-1)})\int \tilde P_{Z_{h-1}}(d\psi^{(h)})e^{-\tilde{\mathcal V}^{(h)}\left(\sqrt{Z_h}\psi^{(\leq h)}\right)}
\end{equation}
where, on the one hand, $P_{Z_{h-1}}(d\psi^{(\leq h-1)})$ is given by formula (\ref{measure_tildePZ_h-1_DBC}) with
\begin{itemize}
\item $Z_{h-1}(\bm k')$ replaced by $Z_{h-1}$,
\item $C_h(\bm k')$ replaced by $C_{h-1}(\bm k')$,
\item $\psi^{(\leq h)}$ replaced by $\psi^{(\leq h-1)}$,
\end{itemize}
while on the other hand the {\it single scale dressed measure} $\tilde P_{Z_{h-1}}(d\psi^{(h)})$ is also given by (\ref{measure_quasi_particles_real_space_DBC}) with
\begin{itemize}
\item $Z_{h-1}(\bm k')$ replaced by $Z_{h-1}$,
\item $C_h(\bm k')$ replaced by $\tilde f^{-1}_h(\bm k')$, where
\begin{equation}
\tilde f^{-1}_h(\bm k')=Z_{h-1}\left(\frac{C^{-1}_h(\bm k')}{Z_{h-1}(\bm k')}-\frac{C^{-1}_{h-1}(\bm k')}{Z_{h-1}}\right),
\end{equation}
\item $\psi^{(\leq h)}$ replaced by $\psi^{(h)}$.
\end{itemize}
It is worth remarking that the {\it scaling properties} of $\tilde f_h(\bm k')$ are the same as $f_h(\bm k')$, {\it i.e.} $\tilde f_h(\bm k')$ is a compact support function, with support of width $O\left(\gamma^h\right)$ and at a distance $O\left(\gamma^h\right)$ from $\bm p_F$ .\\
Finally, we rescale the Grassman fields in such a way that (\ref{integral_dressed_theory_leqh_DBC}) can be rewritten as
\begin{equation}
\frac{1}{\mathcal N_h}\int P_{Z_{h-1}}\left(d\psi^{\leq (h-1)}\right)\int \tilde P_{Z_{h-1}}\left(d\psi^{(h)}\right)e^{\hat{\mathcal V}^{(h)}\left(\sqrt{Z_{h-1}}\psi^{(\leq h)}\right)}
\end{equation}
where $\hat{ \mathcal V}^{(h)}$ is such that its local part is given by
\begin{equation}
\begin{split}
\mathcal L\hat{\mathcal V}^{(h)}\left(\sqrt{Z_{h-1}}\psi^{(\leq h)}\right)=\gamma^h\nu_h F_\nu^{(\leq h)}\left(\sqrt{Z_{h-1}}\psi^{(\leq h)}\right)+\\+\gamma^h\left(\varpi_h\ast F_\varpi^{(\leq h)}\right)\left(\sqrt{Z_{h-1}}\psi^{(\leq h)}\right)+\delta_h F_\alpha^{(\leq h)}+\lambda_h F_\lambda^{(\leq h)}\left(\sqrt{Z_{h-1}}\psi^{(\leq h)}\right)
\end{split}
\label{localized_effective_potential_DBC}
\end{equation}
defining the {\it running coupling constants} and the new {\it running coupling functions}:
\begin{equation}
\label{running_coupling_functions_DBC}
\begin{split}
\nu_h&=\frac{Z_h}{Z_{h-1}}n_h,\\
\delta_h&=\frac{Z_h}{Z_{h-1}}(a_h-z_h),\\
\lambda_h&=\left(\frac{Z_h}{Z_{h-1}}\right)^2l_h,\\
\varpi_h(x,y)&=\frac{Z_h}{Z_{h-1}}\pi_{h}(x,y), \hspace{3mm}.
\end{split}
\end{equation}
which we group together defining
\begin{equation}
\label{vec_v_h(x)}
\vec v_h(x,y)=\left(\nu_h, \delta_h, \lambda_h, \varpi_{h}(x,y)\right).
\end{equation}
At this point, we can perform the {\it single scale integration}
\begin{equation}
\int \tilde P_{Z_{h-1}}(d\psi^{(h)})e^{-\hat{\mathcal V}^{(h)}\left(\sqrt{Z_{h-1}}\psi^{(\leq h)}\right)}=e^{-\mathcal{V}^{(h-1)}\left(\sqrt{Z_{h-1}}\psi^{(\leq h-1)}\right)+L\beta {e}_h}
\end{equation}
where of course we have reconstructed the formal situation of (\ref{linearized_effective_potential_DBC}), with $h\to h-1$:
\begin{equation}
\begin{split}
\mathcal L\mathcal V^{(h-1)}(\psi^{(\leq h-1)})=\gamma^{h-1} n_{h-1} F_\nu^{(\leq h-1)}(\psi^{(\leq h-1)})+\gamma^{h-1} \left(\pi_{h-1}\ast F_{\varpi}^{(\leq h-1)}(\psi^{(\leq h-1)})\right)+\\+z_{h-1} F_{\zeta}^{(\leq h-1)}(\psi^{(\leq h-1)})+a_{h-1} F_{\alpha}^{(\leq h-1)}(\psi^{(\leq h-1)})+ l_{h-1} F_{\lambda}^{(\leq h-1)}(\psi^{(\leq h-1)})
\end{split}
\end{equation}
so that we can apply iteratively the same scheme. We remark that, iterating this procedure, we can write $\vec v_h(x,y)$ in terms of $\vec v_h'(x,y)$, $h'\geq h+1$:
\begin{equation}
\vec v_h=\vec \beta(\vec v_{h+1}(x,y),\dots,\vec v_0(x,y); x,y)
\end{equation}
where $\vec \beta\left(\vec v_{h+1}(x,y),\dots,\vec v_0(x,y); x,y\right)$ is called the {\it beta function}.\\
In order to set up an iterative integration process, we have to be sure that the {\it dressing procedure} we just defined does not break th property of the kernels we pointed out in Remark (\ref{remark_kernels_independent_of_quasi_particles}), {\it i.e.} the fact that the kernels are independent of the {\it quasi-particles}. Indeed:
\begin{prop}
The interaction $\hat{\mathcal V}$ has the form
\begin{equation}
\hat{\mathcal V}\left(\psi^{(\leq h)}\right)=\sum_{n\geq 1}\int d\bm x_1\dots d\bm x_{2n}\left(\prod_{j=1}^n \psi^{(\leq h)+}_{\bm x_{2j-1}}\psi^{(\leq h)-}_{\bm x_{2j}}\right)W^{(h)}_{2n}(\bm x_1,\dots,\bm x_{2nv}),
\end{equation}
{\it i.e.} the kernels $W_{2n}^{(h)}$ are independent of quasi-particles.
\end{prop}
\begin{proof}
We can proceed iteratively. By construction, $\mathcal V^{(\leq 0)}$ is independent of quasi-particles and, in general, if we go on integrating with respect to the {\it bare integration} $P_h(\psi^{(h)})$, all the effective interactions $\mathcal V^{(\leq h)}$ are independent of 	quasi-particles.\\ 
So let us suppose that $\hat {\mathcal V}^{(\leq h)}$ is independent of quasi-particles, and let us check that the dressing procedure we just described does not break this structure.
First of all, let us note that the integration $P(\psi^{(\leq h)})$ (\ref{measure_quasi_particles_real_space_DBC}) is associated with the propagator $g^{(\leq h)}$, independent of quasi-particles. Besides, we dress this propagator with the local part (\ref{localization_W^d_DBC}) that, as we already pointed out, is independent of quasi particles, as is the remainder 
$$\left(1-\mathcal L_\mathcal T\right)\sum_{\bm k\in\mathcal D^d_{\Lambda,\beta}}\hat\psi^{(\leq h)+}_{\bm k}\hat\psi^{(\leq h)-}_{\bm k}\hat W^{d(h)}(\bm k),$$
that we left in the effective interaction $\tilde{\mathcal V}^{(h)}$.\\
Finally:
\begin{itemize}
\item the dressed measure $\tilde P(\psi^{(\leq h)})$ is still associated with a propagator independent of quasi-particles,
\item the effective potential $\tilde{\mathcal V}^{(h)}$ is obtained, starting from $\mathcal{V}^{(h)}$, by replacing $$\sum_{\bm k\in\mathcal D^d_{\Lambda,\beta}}\hat\psi^{(\leq h)+}_{\bm k}\hat\psi^{(\leq h)-}_{\bm k}\hat W^{d(h)}(\bm k)\to \left(1-\mathcal L_\mathcal T\right)\sum_{\bm k\in\mathcal D^d_{\Lambda,\beta}}\hat\psi^{(\leq h)+}_{\bm k}\hat\psi^{(\leq h)-}_{\bm k}\hat W^{d(h)}(\bm k),$$ which is also independent of quasi-particles as we just commented.
\end{itemize}
\end{proof}

\paragraph{Renormalized propagator and renormalized effective potential} The propagator associated with the {\it dressed Gaussian Grassman measure} $\tilde P_{Z_{h-1}}$, in the real-space representation, is
\begin{equation}
\begin{split}
\frac{g^{(h)}(\bm x,\bm y)}{Z_{h-1}}=\sum_{\omega=\pm}\frac{\left(e^{-i\omega (x- y)}g_{P,\omega}^{(h)}(\bm x-\bm y)+e^{-i\omega (x+ y)}g_{R,\omega}^{(h)}(\bm x,\bm y)\right)}{Z_{h-1}},
\end{split}
\end{equation}
where, with a slight abuse of notation, we call $g_{\sigma,\omega}^{(h)}$ the analogous of the already defined propagator (\ref{propagator_decomposed_in_quasiparticles_DBC}) where we replace $f_h(\bm k')$ by $\tilde f_h(\bm k')$:
\begin{equation}
\begin{split}
g_{P,\omega}^{(h)}(\bm x,\bm y)=\frac{1}{\beta 2(L+1)}\sum_{\bm k'\in \mathcal D_{\Lambda,\beta}}e^{-i\bm k'\cdot (\bm x-\bm y)}\frac{\tilde f_h(\bm k')}{-ik_0+e(\bm k'+\omega p_F)}.\\
g_{R,\omega}^{(h)}(\bm x,\bm y)=\frac{1}{\beta 2(L+1)}\sum_{\bm k'\in \mathcal D_{\Lambda,\beta}}e^{-ik'(x+y)}e^{-ik'_0(x_0-y_0)}\frac{\tilde f_h(\bm k')}{-ik_0+e(\bm k'+\omega p_F)}.
\end{split}
\end{equation}
Besides, it is important to remember that, if we did not perform the renormalization procedure just described, the effective potential $\mathcal V^{(h)}$ would have the same shape we have already shown in (\ref{effective_potential_decomposed_2-4_kerlels_DBC}). Actually, the renormalization procedure we just described  produces a new sequence of effective potential, that we call {\it renormalized effective potentials}, being of the same form of (\ref{effective_potential_decomposed_2-4_kerlels_DBC}) where we replace $\psi^{(\leq h)}$ by $\sqrt{Z_h}\psi^{(\leq h)}$, and the kernels $W^{(h)}_{2n}$ by what we call the {\it renormalized values of the clusters}. Properly, the effective potentials can be written as
\begin{equation}
\begin{split}
\mathcal V^{(h)}\left(\sqrt{Z_h}\psi^{(\leq h)}\right)=\sum_{n=1}^{\infty}\sum_{\tau\in\mathcal T_{h,n}}\mathcal V^{(h)}\left(\tau, \sqrt{Z_h}\psi^{(\leq h)}\right),\\
\mathcal V^{(h)}\left(\tau,\sqrt{Z_h}\psi^{(\leq h)}\right)=\int d\bm x(I_{v_0})\sum_{P_{v_0}\subset I_{v_0}}\sqrt{Z_h}^{|P_{v_0}|}\tilde \psi^{(\leq h)}(P_{v_0}) W^{(h)}(\tau, P_{v_0},\bm x(I_{v_0}))
\end{split}
\end{equation}
where, in fact, the kernels $W^{(h)}(\tau, P_{v_0},\bm x(I_{v_0}))$ have to be read as the {\it renormalized values of the clusters}, that we discuss in the following subsection.

\subsection{The renormalized tree expansion}
\label{subsection_the_renormalized_tree_expansion}
\begin{figure}
\centering
 \begin{tikzpicture} 
[scale=0.7, transform shape]
\foreach \i in {1,2,3,4,5,6,7,8,9,10,11,12,13,14} {%
\draw  (\i,2.9) -- (\i, 11.2); }
\foreach \j in {1,2,3,4,5} {%
\draw [very thick] (\j,7) -- ++ (1,0);
\fill (\j,7) circle (0.1);
\node at (\j, 6.7) {$\mathcal R$};
\fill (6,7) circle (0.1);
\node at (6, 6.7) {$\mathcal R$};
}
\foreach \j in {0,1,2,3,4,5} {%
\draw [very thick] (6+\j, 7 -\j *0.5) -- +(1,-0.5);
\fill (6+\j,7-\j*0.5) circle (0.1);
\node at (6+\j, 6.7-\j*0.5) {$\mathcal R$};}
\fill (6+6, 7-3) circle (0.1);
\node at (12, 4.7) {$\mathcal L$};
\foreach \j in {0,1,2,3} {%
\draw [very thick] (6+\j, 7 +\j *0.5) -- +(1,+0.5);
\fill (6+\j,7+\j*0.5) circle (0.1);
\node at (6+\j, 6.7+\j*0.5) {$\mathcal R$};}
\fill (6+4, 7+2) circle (0.1);
\node at (10, 8.7) {$\mathcal R$};
\foreach \j in {0,1} {%
\draw [very thick] (10+\j, 9 +\j *0.5) -- +(1,+0.5);
\fill (10+\j,9+\j*0.5) circle (0.1);
\node at (10+\j, 8.7+\j*0.5) {$\mathcal R$};}
\fill (12, 10) circle (0.1);
\node at (12, 9.7) {$\mathcal L$};
\foreach \j in {0,1,2} {%
\draw [very thick] (10+\j, 9 -\j *0.5) -- +(1,-0.5);
\fill (10+\j,9-\j*0.5) circle (0.1);
\node at (10+\j, 8.7-\j*0.5) {$\mathcal R$};}
\fill (13,7.5) circle (0.1);
\node at (13, 7.2) {$\mathcal L$};
\foreach \j in {0,1} {%
\draw [very thick] (12+\j, 8 +\j *0.5) -- +(1,+0.5);
\fill (12+\j,8+\j*0.5) circle (0.1);
\node at (12+\j, 7.7+\j*0.5) {$\mathcal R$};
}
\fill(14,9) circle (0.1);
\node at (14, 8.7) {$\mathcal L$};
\foreach \j in {0} {%
\draw [very thick] (12+\j, 4 +\j *0.5) -- +(1,+0.5);
\fill (12+\j,4+\j*0.5) circle (0.1);
\node at (12+\j, 3.7+\j*0.5) {$\mathcal R$};
}
\foreach \j in {0,1} {%
\draw [very thick] (12+\j, 4 -\j *0.5) -- +(1,-0.5);
\fill (12+\j,4-\j*0.5) circle (0.1);
\node at (12+\j, 3.7-\j*0.5) {$\mathcal R$};}
\fill (14,3) circle (0.1);
\node at (14, 3.3) {$\mathcal L$};
\fill (13,4.5) circle (0.1);
\node at (13, 4.2) {$\mathcal L$};
\draw [very thick] (8,8) -- (9, 7.5);
\fill (9,7.5) circle (0.1);
\node at (9, 7.2) {$\mathcal L$};
\draw [very thick] (11,8.5) -- (12, 9);
\fill (12,9) circle (0.1);
\node at (12, 8.7) {$\mathcal L$};
\draw [very thick] (6,7) -- (11,6);
\fill (11,6) circle (0.1);
\node at (11, 5.7) {$\mathcal R$};
\draw [very thick] (11,6) -- (12, 5);
\fill (12,5) circle (0.1);
\draw [very thick]  (11, 6) -- ++ (2,0);
\fill (13,6) circle (0.1);
\node at (13, 5.7) {$\mathcal L$};
\node at (1,2.7) {$\bm h$};
\node at (2,2.7) {$\bm h+1$};
\node at (3,2.7) {$\bm h+2$};
\foreach \i in {4,5,6,7,8} {%
\node at (\i,2.8) {...};}
\node at (9,2.7) {$\bm h_v$};
\node at (10,2.7) {$\bm h_v+1$};
\foreach \i in {11,12,13} {%
\node at (\i,2.8) {...};}
\node at (14,2.7) {$\bm 1$};
\node at (9,8.8) {$ v$};
\node at (1,7.3) {$ r$};
\node at (2,7.3) {$ v_0$};
\fill (7,6.8) circle (0.1);
\node at (7, 7.8) {$\mathcal R$};
\fill (8,6.6) circle (0.1);
\node at (8, 6.3) {$\mathcal R$};
\fill (9,6.4) circle (0.1);
\node at (9, 6.1) {$\mathcal R$};
\fill (10,6.2) circle (0.1);
\node at (10, 5.9) {$\mathcal R$};
\fill (12, 6) circle (0.1);
\node at (12, 5.7) {$\mathcal R$};
\end{tikzpicture}
\label{figure_renormalized_tree_DBC}
\caption{Example of a renormalized tree, with $n=9$ endpoints at scales $\leq 1$.}
\end{figure}
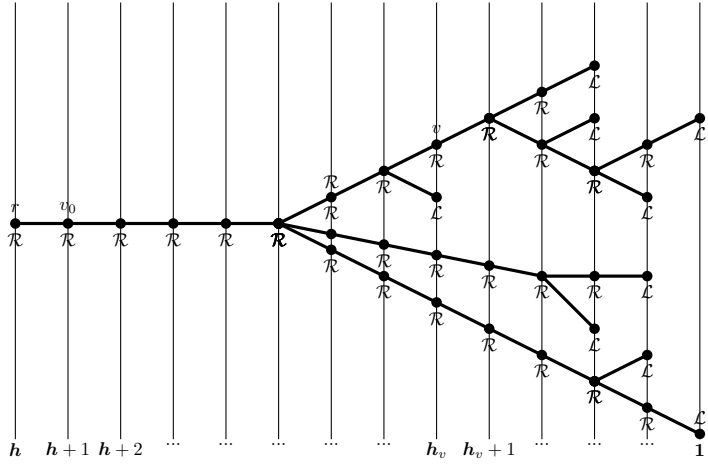
As usual it is convenient to give a graphical representation of the renormalized and localized effective potentials in terms of trees. One starts drawing $\mathcal V^{(-1)}$ as in figure (\ref{figure_effective_potentiale_scale_0}), using the representation for $\mathcal V^{(0)}$ as the sum of renormalized and localized part, and then iterates the procedure on $\mathcal V^{(-1)}$ itself. Finally, one gets the family of renormalized trees given by the same trees we got expanding naively $\mathcal V^{(h)}$ with the following differences:
\begin{itemize}
\item Each vertex $v\neq V_f(\tau)$ is labeled by an operator 
$$\mathcal R\in\left\{\mathcal R_\mathcal T\mathcal L_{\mathcal B}, \tilde{\mathcal R}_\mathcal T\mathcal R_{\mathcal B}, \mathcal R_\mathcal B\right\},$$
 up to the first vertex $v_0$ which can be labeled either by $$\mathcal R\in\left\{\mathcal R_\mathcal T\mathcal L_{\mathcal B}, \tilde{\mathcal R}_\mathcal T\mathcal R_{\mathcal B}, \mathcal R_\mathcal B\right\},\mbox{ or by }\mathcal L\in \left\{ \mathcal L_\mathcal T\mathcal L_\mathcal{B}, \tilde{ \mathcal L}_\mathcal T\mathcal L_\mathcal{B}\right\}.$$
\item There are endpoints $v\in V_f(\tau)$ scale label $h_v\leq 1$ (while, in the non-renormalized expansion, each endpoint lives at scale $h_v=1$). If $v\in V_f(\tau)$ and $h_v<0$, a constribution $\mathcal L \mathcal V^{(h)}$ among (\ref{linearized_effective_potentia_nu_DBC})-(\ref{linearized_effective_potentia_lambda_DBC}) is associated with $v$, while if $h_v=0$ either a contribution $\mathcal L \mathcal V^{(0)}$ or a contribution $\mathcal R\mathcal V^{(0)}$ is associated with $v$. Of course, in this way the endpoints are associated with running coupling constants and functions by the label $v$, meaning that if $r_v=\nu_h$ and $h=h_{v'}$, $v$ is associated with $F^{(\leq h)}_\nu$, if $r_v=\varpi_h$ and $h=h_{v'}$, $v$ is associated with $\left(\varpi \ast F^{(\leq h)}_\varpi\right)$ and so on.
\item the hierarchical structure of the tree, and the  very definition of the operators $\mathcal L_{\mathcal B}$, $\mathcal R_{\mathcal B}$, implies some ordering constraints on the remainder operators labeling the vertices: let us suppose that some vertex $v\in V(\tau)$ is labeled by a renormalization operator $\mathcal R_\mathcal T\mathcal R_{\mathcal B}$ or $\mathcal R_\mathcal B$. So, for each $w\prec v$,
\begin{equation}
\mathcal L_{\mathcal B} \mathcal W^{(h_w)}(\tau_w,P_w,\bm x(P_w))= 0.
\end{equation}
This means that, given $v\in V(\tau)$ labeled by $\mathcal R_\mathcal T\mathcal R_{\mathcal B}$ or $\mathcal R_\mathcal B$, any $w\prec v$ is necessairly labeled by $\mathcal R_0\mathcal R_{\mathcal B}$ or $\mathcal R_\mathcal B$. {\it Vice versa} it is not possible that a vertex labeled by $\mathcal R\in\{\mathcal R_\mathcal T\mathcal L_{\mathcal B}, \mathcal L_{\mathcal B}\}$ is an ancestor of vertices labeled by $\mathcal R\in\{\mathcal R_\mathcal T\mathcal R_B\}$.
\end{itemize}

This inductive definition of the renormalized effective potential is convenient since it is possible to get some estimates on the kernels of the effective potential which we use to show that the multiscale expansion is well behaved. We already pointed out that the multiscale integration induces a natural definition of the so called {\it running coupling constants and functions} $\vec v_k(x,y)$, $h<k\leq 1$ the effective kernels can be thought functions of. Our strategy is to:
\begin{enumerate}
\item first of all, we will consider $\vec v_k$ as an arbitrary sequence we will do some smallness assumptions on, without requiring that they are solution of the beta function (\ref{vec_v_h(x)}),
\item once we know that, under these smallness assumptions, the kernels of the effective potential are well defined, we prove that in particular the {\it running coupling constants and functions} are solutions of the beta function. 
\end{enumerate}

\subsection{Renormalized bounds}

Let us recap what we have done so far and what we want to do. In Section (\ref{section_Non-renormalized expansion and properties of kernels}) we inferred, knowing that the {\it bulk contributions dominate with respect to the remainder contributions}, that the {\it sources of problems} for the convergence of the expansion are the {\it quadratic and quartic kernels}, as in the previous chapter (Theorem (\ref{theorem_bound_of_kernels})). In light of this fact we set up a renormalization procedure, consisting in dressing, scale by scale, the Grassmann integration with a marginal contribution coming from the effective potential. Of course we are left with proving that this procedure improves the bounds of the kernels in such a way that we can perform the sum over all the scales, {\it i.e. we want to prove the following theorem}.

\begin{thm}
\label{theorem_renormalized_bounds_DBC}
Let $\tau\in\mathcal T_{h,n}$ a renormalized tree, $h>h_L$, and $\mathcal W^{(h)}(\tau, P_{v_0}, \bm x(P_{v_0}))$ the respective renormalized kernel. So it holds:
\begin{equation}
\label{bounds_renormalized_kernels_DBC}
\begin{split}
\frac{1}{|\Lambda|\beta}\int d\bm x(P_{v_0})\left|\mathcal W^{(h)}(\tau,P_{v_0},\bm x(P_{v_0}))\right|\leq C^n\gamma^{-h[D(P_{v_0})+z_{v_0}]}\\
\left(\prod_{v\neq V_f(\tau)}\gamma^{-[D(P_{v})+z_{v}](h_v-h_v')}\right)\left(\prod_{v\in V_f(\tau)}r_v\right)
\end{split}
\end{equation}
where $r_v\in \{|\nu_h|,|\lambda_h|, |\delta_h|, \sup_{x\in\Lambda}\int_\Lambda dy|\varpi_h(x,y)|\},$
$z_v=\theta$ if $G_v$ has four external lines, $z_v=1+\theta$ if $G_v$ has two external lines and $C$ is an independent constant.
\end{thm}
It is convenient to start with looking at the {\it explicit formula} of the renormalized effective potentials:
\begin{equation}
\begin{split}
\mathcal V^{(h)}\left(\sqrt{Z_h}\psi^{(\leq h)}\right)=\sum_{n=1}^{\infty}\sum_{\tau\in\mathcal T_{h,n}}\mathcal V^{(h)}\left(\tau, \sqrt{Z_h}\psi^{(\leq h)}\right),
\end{split}
\end{equation}
where, if $v_0$ is the first vertex of $\tau$ and $\tau_1,\dots, \tau_{s_{v_0}}$ are the subtrees of $\tau$ with root $v_0$, and $\mathcal V^{(h)}$ is inductively defined as
\begin{equation}
\begin{split}
\mathcal V^{(h)}(\tau,\sqrt{Z_h}\psi^{(\leq h)})=\\
\frac{(-1)^{s_{v_0}+1}}{s_{v_0}!}\mathcal E_{h+1}^T\left(\mathcal V^{(h+1)}(\tau_1, \sqrt{Z_h}\psi^{(\leq h+1)});\dots ; \mathcal V^{(h+1)}(\tau_1, \sqrt{Z_h}\psi^{(\leq h+1)})\right)
\end{split}
\end{equation}
where
\begin{equation}
\mathcal E_{h_v}^T\left(\tilde \psi(P_{v_1}), \dots, \tilde \psi(P_{v_{s_0}})\right)=\sum_{T_v}\alpha_{T_v}\prod_{\ell \in T_v}g^{(h_\ell)}_{\ell}\int dP_{T_{v}}(\bm t) \det G^{T_v}(\bm t)
\end{equation}
where $T_v$ is a set of lines forming an {\it anchored tree} between the clusters of points $P_{v_1}, \dots, P_{v_{s_v}}$, $\alpha_T$ is a sign, and $g^{(h_\ell)}_{\ell}=g_{\sigma(\ell),\omega(\ell)}^{(h_\ell)}(\bm x(f_\ell),\bm y(f'_{\ell}))$, so:
\begin{equation}
\begin{split}
\mathcal R\mathcal V^{(h)}\left(\tau,\sqrt{Z_h}\psi^{(\leq h)}\right)=\\=\int d\bm x(I_{v_0})\sum_{P_{v_0}\subset I_{v_0}}\sum_{T\in\bm T}\sum_{\alpha\in A_T}\sqrt{Z_h}^{|P_{v_0}|}\left[\prod_{f\in P_{v_0}}\partial^{b_\alpha(f)}_{j_\alpha(f)} \psi^{(\leq h)\epsilon(f)}_{\bm x(f)}(P_{v_0})\right] \mathcal R_{\alpha}W^{(h)}(\tau, P_{v_0},\bm x(I_{v_0})),
\end{split}
\end{equation}
where $b_\alpha(f)\in\{0,1,2\}$, $j_\alpha(f)\in\{0,1\}$, and $A_T$ is a set of indices that formally allow to distinguish the different terms produced by the non trivial $\mathcal R\in\{\mathcal R_\mathcal T\mathcal L_\mathcal B, \tilde{\mathcal R}_\mathcal T\mathcal R_B,\mathcal R_\mathcal B\}$.
\paragraph{Explicit expression of renormalized kernels}

\begin{rem}When $\mathcal R=\mathcal R_\mathcal T\mathcal L_\mathcal B$, thanks to the very structure of the renormalized trees we explained in Subsection (\ref{subsection_the_renormalized_tree_expansion}), we are in the same case of Theorem (\ref{theorem_renormalized_bounds}), there is nothing new to comment.
So we are left with controlling the cases $\mathcal R_\alpha\in\{\tilde{\mathcal R}_\mathcal T\mathcal R^{(1)}_\mathcal B, \tilde{\mathcal R}_\mathcal T\mathcal R^{(2)}_	\mathcal B, \mathcal R^{(1)}_\mathcal B, \mathcal R^{(2)}_\mathcal B\}$.
\end{rem}
Let us comment properly the multiscale structure of $\mathcal R_\alpha W^{(h)}$ in these four cases.\\
In general, using Remarks (\ref{remark_R_B_decomposition_R1_R2}) and (\ref{remark_commutation_renormalization operators}), we can re-write
\begin{equation}
\begin{split}
\mathcal R_\alpha W^{(h)}(\tau, P_{v_0},\bm x(I_{v_0}))=\\=\left[\prod_{v\notin V_f(\tau)}\left(\frac{Z_{h_v}}{Z_{h_{v}-1}}\right)^{\frac{|P_v|}{2}}\right]
\mathcal R^{(\tau)}_{v_0,\alpha,\mathcal B}\left( \left\{\prod_{v\notin V_f(\tau)}\frac{1}{s_v!}\int dP_{T_v}(\bm t_v) \left( \det G_\alpha^{h_v, T_v}(\bm t_v)\right)\cdot\right.\right.\\\left.
\left.\cdot \left[\prod_{\ell \in T_v}(\bm x_\ell- \bm y_\ell)^{b_\alpha(\ell)}_{j_\alpha(\ell)}\partial^{q_\alpha(f_\ell^1)}_{j_\alpha(f_\ell^1)}\partial^{q_\alpha(f_\ell^2)}_{j_\alpha(f_\ell^2)} g^{(h_\ell)}_{\ell}\right]\right\}\left[\prod_{i=1}^{n}(\bm x^i-\bm y^i)^{b_\alpha(v^*_i)}_{j_\alpha(v^*_i)}K^{(h_i)}_{{v^*_{i}}}(\bm x_{v^*_i}))\right]\right)
\end{split}
\label{renormalized_kernels_explicit_expression}
\end{equation}
where 
\begin{itemize}
\item we used the commutation of the renormalization operators (Remark (\ref{remark_commutation_renormalization operators})), acting first with the {\it Taylor renormalization operators}, so that  $b_\alpha(\ell), b_\alpha(v_i^*), q_\alpha(\ell), q_\alpha(v_i^*)\in\{1,2\}$, and the fact that there are as many derivatives as {\it "zeroes"} is technically expressed by the constraint $\sum_{\ell, i}\left(b_\alpha(\ell)+ b_\alpha(v_i^*)- q_\alpha(f_\ell^{(1)})- q_\alpha(f_\ell^{(2)})\right)=0$, while $(\bm x_\ell- \bm y_\ell)^{b_\alpha(\ell)}_{j_\alpha(\ell)}$ are the zeroes we introduced in the renormalization procedure definition, where $j_\alpha\in \{0,1\}$ denotes the component of the vector, $K^{(h_i)}_{{v^*_{i}}}$ is one of the terms of the local effective potential $\mathcal L \mathcal V^{(h_i)}$, and $G_\alpha^{h_v,T_v}$ is the matrix whose entries are
\begin{equation}
\begin{split}
G^{h_v,T_v}_{\alpha, ij,i'j'}= t_{v,i,i'}\partial_{j_\alpha(f_{ij}^1)}^{q_\alpha(f_{ij}^1)}\partial_{j_\alpha(f_{ij}^2)}^{q_\alpha(f_{ij}^2)}g_{\sigma_\ell,\omega_\ell}^{h_v}(\bm x_{ij},\bm y_{i'j'}).
\end{split}
\label{matrix_G_h_v_t_v}
\end{equation}
\item $R_{v_0,\alpha,\mathcal B}^{(\tau)}$ is a formal way to represent the {\it bulk renormalization} operations, so it has to be interpreted as iteratively defined in a way dependent on the structure of the renormalized tree, as discussed in defining the {\it localization and renormalization operators}, Subsection (\ref{subsection_the_renormalized_tree_expansion}).
\end{itemize}
In particular, $R_{v_0,\alpha,\mathcal B}^{(\tau)}$ has to be thought of as a {\it composition} of $\mathcal R^{(i)}_\mathcal B$, $i=1,2$ operators acting on the vertices of $V(\tau)$ in the following way:
\begin{itemize}
\item {\bf Case $\mathcal R^{(1)}_\mathcal B$} As we commented in Remark (\ref{remark_R_B_decomposition_R1_R2}), $\mathcal R^{(1)}_\mathcal B$ does not modify the trees, but it selects the trees having at least a non-translation-invariant element (either $g_R$ propagators or $\varpi$-endpoints). 
\item {\bf Case $\mathcal R^{(2)}_\mathcal B$} As we commented in (\ref{remark_R_B_decomposition_R1_R2}), $\mathcal R^{(2)}_\mathcal B$ modifies the trees that do not contain any non-translation-invariant element. So, in the right hand side of (\ref{renormalized_kernels_explicit_expression_first_version}), the product within the last brackets does not include, by construction, $\varpi$-type endpoints;  $\mathcal R_\mathcal B^{(2)}$, by modifying the coordinates associated with the endpoints, modifies {\it  in general} the coordinates of some propagators and the coordinates involved in the determinant.
\end{itemize}

\paragraph{Definition of the weight functions}

It is worth pointing out that,  for the porpouses of the estimates that we are looking for, we can associate the {\it decay properties} related to non-local counterterms $\gamma^h \varpi_h(x,y)$  and remainder propagators $g_{R,\omega}^{(h)}(\bm x,\bm y)$ with vertices instead of with lines of the spanning tree, and it will be useful during the proof. Indeed:
\begin{itemize}
\item the inductive hypoyhesis on $\varpi_h(\cdot, \cdot)$
\begin{equation}
\int dy \left|\varpi_h(x,y)\right|\leq  |\lambda| \frac{C_\theta}{1+\gamma^{\theta h}|x|^\theta}, \hspace{3mm} 0<\theta \leq 1,
\label{varpi_endpoints_decay}
\end{equation}
symmetrically in $\bm x \leftrightarrow \bm y$. Let us define 
\begin{equation}
||\varpi_h||^{(\theta)}_{\infty,1}=\sup_{x\in\Lambda}(1+\gamma^h|x|)^\theta\int dy \left|\varpi_h(x,y)\right|.
\end{equation}

\begin{figure}
\centering
\begin{tikzpicture}
 [thick,decoration={
    markings,
    mark=at position 0.5 with {\arrow{>}}}] 
\fill (0,1) circle (0.06);
\node at (0,0.6) {\bf x};
\fill (2,1) circle (0.06);
\node at (2,0.6) {\bf y};
\fill (7,1) circle (0.06);
\node at (7,0.6) {\bf x};
\node at (7,1.4) {$\varpi_h(x)$};
\node [regular polygon, regular polygon sides=4,
        minimum size=3mm, fill] at (7,1) {};
\draw [postaction={decorate}] (-1,1) -- ++ (1,0);
\draw [-,decorate, decoration={coil, aspect=2}] (0,1) --++ (2,0);
\draw [postaction={decorate}] (2,1) -- ++ (1,0);
\draw [postaction={decorate}] (6,1) -- ++ (1,0);
\draw [postaction={decorate}] (7,1) -- ++ (1,0);
\draw [->, very thick] (3.5,1) -- ++ (2,0);
\end{tikzpicture}
\centering
\begin{tikzpicture}
 [thick,decoration={
    markings,
    mark=at position 0.5 with {\arrow{>}}}] 
\fill (0,1) circle (0.06);
\node at (0,0.6) {\bf x};
\fill (2,1) circle (0.06);
\node at (2,0.6) {\bf y};
\fill (8,1) circle (0.06);
\node at (6,0.6) {\bf x};
\node at (8,0.6) {\bf y};
\node at (6,1.4) {$\rho_h(x)$};
\node [regular polygon, regular polygon sides=4,
        minimum size=3mm, fill] at (6,1) {};
\draw [very thick, dashed] (0,1) --++ (2,0);
\draw [very thick] (6,1) -- ++ (2,0);
\draw [->, very thick] (3.5,1) -- ++ (2,0);
\end{tikzpicture}
\caption{First line: localization of a $\varpi$-type endpoint: $\varpi_h(x)=\int_{\Lambda}dy|\varpi_h(x,y)|$. In the second line we replace a $R$-labeled propapagot (dashed line) by a $P$-labeled propagator, dressing one of the two vertices with a weight function $	\rho$.}
\label{figure_localization_varpi_endpoints}
\end{figure}
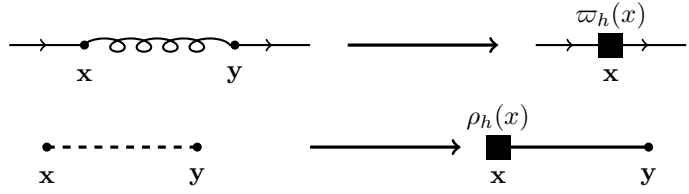

This means that, in studying a {\it renormalized} or {\it linearized} tree $\tau \in \mathcal T_{h,n}$ (meaning that the endpoints can live at any scale $h< k \leq 0$, and they represent the terms appearing in the linearized part of the effective potential $\mathcal L \mathcal V^{(k)}$ (\ref{localized_effective_potential_DBC})), if we call $n_{\varpi}$ the number of $\varpi$-type endpoints (which are the only non-local endpoints), we can reduce the number of integration points by integrating $n_\varpi$ integration points as in (\ref{varpi_endpoints_decay}), see Figure (\ref{figure_localization_varpi_endpoints}): in this way, for the purpose of an upper bound, we are replacing a non local graph element representing $\gamma^h \varpi_h(x,y)$ by a local graph element, {\it i.e.} a vertex, and we associate with this vertex a {\it weight} that, with a slight abuse of notation and by suitably fixing a $\bar \theta$ in formula (\ref{varpi_endpoints_decay}), we call $\gamma^h \varpi_h(x)$, defined as 
$$ \varpi_h(x)= |\lambda|\frac{C_{\bar \theta}}{1+\gamma^{\bar \theta h}|x|^{\bar{\theta}}},$$
\item following the same idea of associating a {\it weight} to vertices, let us recall the result of Corollary (\ref{corollary_norms_propagators_DBC})
\begin{equation}
\frac{1}{\beta}\left| \int d\bm x d\bm y g_{R,\omega}(\bm x,\bm y)\right| \leq \gamma^{- 2h},
\end{equation}
so, for the purpose of a dimensional estimate, we can replace the propagator associated with the line $(\bm x,\bm y)$ by a translation invariant propagator $g_P$ provided we {\it dress} one of the two vertices linked by the propagator with a proper weight: let us recall the definition of $\rho_h^{(N)}$  (\ref{definition_rho_h^N}), for each $N=1,2,\dots$: 
$$\rho^{(N)}_h(x)= \frac{C_N}{1+\left(\gamma^h |x|\right)^N},$$ where $C_N$ is the same as Corollary (\ref{corollary_norms_propagators_DBC}) (again, we can arbitrarily choose $x$ or $y$), indeed:
\begin{equation}
\begin{split}
\frac{1}{\beta}\left| \int d\bm x d\bm y g^{(h)}_{R,\omega}(\bm x,\bm y)\right| \leq C \gamma^{-2h},\\
\frac{1}{\beta}\left| \int d\bm x d\bm y g^{(h)}_{P,\omega}(\bm x-\bm y)\rho_h(x)\right|= \left|\int d\bm y' g_{P,\omega}^{(h)}(\bm y')\int dx \rho_h(x) \right| \leq C \gamma^{- 2h},
\end{split}
\label{reminder_propagators_as_periodic_and_weight}
\end{equation}
if we call $\rho_h(\cdot):=\rho_h^{(\bar N)}$ with $\bar N$ suitably fixed {\it a priori} (see Figure \ref{figure_localization_varpi_endpoints}).
\item we will use, during the proof, the already commented estimate (\ref{bound_g_R_g_infty_rho}), that we recall here:
$$|g_R^{(h)}(\bm x,\bm y)|\leq ||g_R^{(h)}||_\infty \rho^{(N)}_h(x),\hspace{5mm}\forall \hspace{3mm} N=1,2,\dots$$
\end{itemize}
\begin{rem}
From now on we use the symbol $\rho_h(\cdot)$ to denote $\rho^{(\bar N)}_h$ for some suitably fixed $\bar N$, so in particular we can think:
\begin{equation}
\rho_h(x)\leq \frac{C_{\bar N}}{\left(1+\gamma^h |x|\right)^{\bar N}}.
\end{equation}
We stress that for any $\theta\in (0,1)$ there exists a constant $C_\theta>0$ such that
\begin{equation}
\sup_x \left| \rho_h(x) \right|\leq C_\theta,\hspace{3mm} \sup_x\left| \varpi_h(x)\right|\leq C_\theta |\lambda|,
\end{equation}
meaning that, not taking advantage of the decay properties of $\rho_h(\cdot)$ and $\varpi_h(\cdot)$ we would get, for the clusters containing at least one element which breaks translation invariance, the same bound as the dominant (translation invariant) part.
\end{rem}
Let us unify the notation by introducing the weight function

\begin{equation}
w_h(h)\in\{\rho_h(x), \varpi_h(x)\}.
\label{definition_weight_functions}
\end{equation}

\paragraph{"Simplification" of the tree and definition of $V_\mathcal B(\tau)$} We can {\it simplify} the hierarchical renormalization structure of the tree $\tau$ by suitably using Corollary (\ref{corollary_norms_propagators_DBC}) in bounding the norm $||\cdot||_1$ of the renormalized kernels:
\begin{equation}
\frac{1}{|\Lambda|\beta}\int d\bm x(P_{v_0})\left|\mathcal W^{(h)}\left(\tau,P_{v_0}, \bm x(P_{v_0})\right)\right|.
\end{equation}

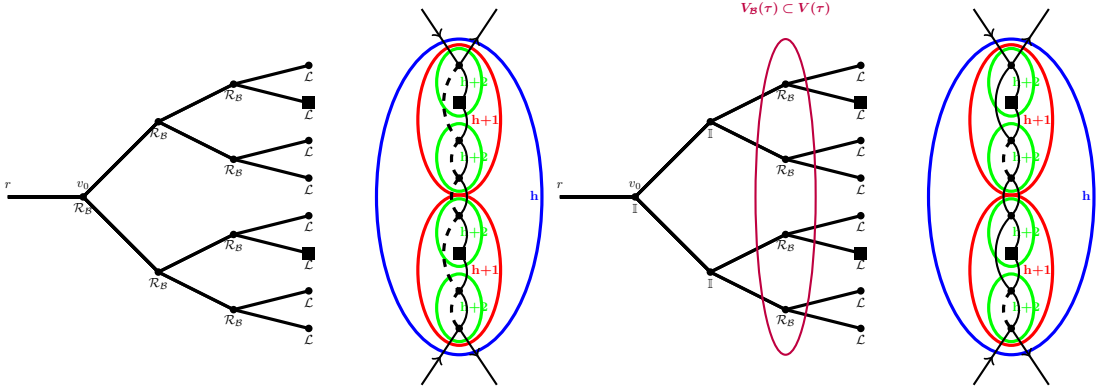
\begin{figure}
\begin{tikzpicture}
[scale=0.5,transform shape,thick,decoration={
    markings,
    mark=at position 0.5 with {\arrow{>}}}] 
\node at (1,7.3) {$r$};
\node at (3,7.3) {$v_0$};
\foreach \i in {-1,1} {%
\foreach \j in {-1,1}{%
\foreach \l in {-1,1}{%
\draw [very thick] (1,7) -- ++ (2,0) -- ++ (2,\i*2) -- ++ (2,\j*1) -- ++ (2,\l*0.5);
\fill (1,7) ++ (2,0) ++ (2,\i*2) ++ (2,\j*1) ++ (2,\l*0.5) circle (0.1);
\node at  (1+2+2+2+2,7+\i*2+\j+\l*0.5-0.3)   {$\mathcal L$};
}
\fill (1,7) ++ (2,0) ++ (2,\i*2) ++ (2,\j*1)  circle (0.1);
\node at (1+2+2+2, 7+\i*2+\j - 0.3){$\mathcal R_\mathcal B$};
}
\fill (1,7) ++ (2,0) ++ (2,\i*2) circle (0.1);
\node at  (1+2+2,7+\i*2-0.3)   {$\mathcal R_\mathcal B$};
}
\fill (1,7) ++ (2,0) circle (0.1);
\node at (1+2,7-0.3)  {$\mathcal R_\mathcal B$};
\node [regular polygon, regular polygon sides=4,
        minimum size=3mm, fill] at (9, 9.5) {};
        \node [regular polygon, regular polygon sides=4,
        minimum size=3mm, fill] at (9, 5.5) {};
  \foreach \i in {-1,1} {%
\foreach \j in {-1,1}{%
\foreach \l in {-1,1}{%
\fill (1,7) ++ (4,0) ++ (2,0) ++ (2,\i*2) ++ (2,\j*1) ++ (2,\l*0.5) circle (0.1);
}}}
\node [regular polygon, regular polygon sides=4,
        minimum size=3mm, fill] at (13, 9.5) {};
        \node [regular polygon, regular polygon sides=4,
        minimum size=3mm, fill] at (13, 5.5) {};
\foreach \i in {-1,1}{%
\foreach \j in {-1,1}{%
\draw [green, very thick] (1,7.05) ++ (4,0) ++ (2,0) ++ (2,0) ++ (2,\i*2) ++ (2,\j*1) ellipse (0.6 and 0.9);
\node at (1+4+2+2+2+2+0.4, 7.05+\i*2+\j) {\textcolor{green}{\bf h+2}};
}
\draw [red, very thick] (1,7.05) ++ (4,0) ++ (4,0) ++ (2,0) ++ (2,\i*2) ellipse (1.1 and 2);
\node at (1+4+2+2+2+2+0.7, 7.05+\i*2) {\textcolor{red}{\bf h+1}};
}
\draw [blue, very thick] (1,7) ++ (10,0) ++ (2,0) ellipse (2.2 and 4.2);
\node at (1+4+2+2+2+2+2, 7.05) {\textcolor{blue}{\bf h}};
\foreach \i in {0,1,2,3,4,5,6} {%
\draw (13,10.5 - \i) to [out=-45, in=45, looseness=1] (13,10.5 -1 -\i);}
\draw [very thick, dashed] (13,10.5)to [out=225, in=-225, looseness=1] (13,8.5);
\draw [very thick, dashed] (13,6.5)to [out=225, in=-225, looseness=1] (13,4.5);
\draw [very thick, dashed](13,7.5)to [out=225, in=-225, looseness=1] (13,6.5);
\draw [very thick, dashed] (13,8.5)to [out=225, in=-225, looseness=1] (13,7.5);
\draw [very thick, dashed] (13,4.5)to [out=225, in=-225, looseness=1] (13,3.5);
\draw [postaction={decorate}] (12,12) -- (13,10.5);
\draw [postaction={decorate}] (13,10.5) -- (14,12);
\draw [postaction={decorate}] (12,2) -- (13,3.5);
\draw [postaction={decorate}] (13,3.5) -- (14,2);
\end{tikzpicture}
\begin{tikzpicture}
[scale=0.5,transform shape, thick,decoration={
    markings,
    mark=at position 0.5 with {\arrow{>}}}] 
\node at (1,7.3) {$r$};
\node at (3,7.3) {$v_0$};
\foreach \i in {-1,1} {%
\foreach \j in {-1,1}{%
\foreach \l in {-1,1}{%
\draw [very thick] (1,7) -- ++ (2,0) -- ++ (2,\i*2) -- ++ (2,\j*1) -- ++ (2,\l*0.5);
\fill (1,7) ++ (2,0) ++ (2,\i*2) ++ (2,\j*1) ++ (2,\l*0.5) circle (0.1);
\node at  (1+2+2+2+2,7+\i*2+\j+\l*0.5-0.3)   {$\mathcal L$};
}
\fill (1,7) ++ (2,0) ++ (2,\i*2) ++ (2,\j*1)  circle (0.1);
\node at (1+2+2+2, 7+\i*2+\j - 0.3){$\mathcal R_\mathcal B$};
}
\fill (1,7) ++ (2,0) ++ (2,\i*2) circle (0.1);
\node at  (1+2+2,7+\i*2-0.3)   {$\mathbb I$};
}
\fill (1,7) ++ (2,0) circle (0.1);
\node at (1+2,7-0.3)  {$\mathbb I$};
\node [regular polygon, regular polygon sides=4,
        minimum size=3mm, fill] at (9, 9.5) {};
        \node [regular polygon, regular polygon sides=4,
        minimum size=3mm, fill] at (9, 5.5) {};
  \foreach \i in {-1,1} {%
\foreach \j in {-1,1}{%
\foreach \l in {-1,1}{%
\fill (1,7) ++ (4,0) ++ (2,0) ++ (2,\i*2) ++ (2,\j*1) ++ (2,\l*0.5) circle (0.1);
}}}
\node [regular polygon, regular polygon sides=4,
        minimum size=3mm, fill] at (13, 9.5) {};
        \node [regular polygon, regular polygon sides=4,
        minimum size=3mm, fill] at (13, 5.5) {};
\foreach \i in {-1,1}{%
\foreach \j in {-1,1}{%
\draw [green, very thick] (1,7.05) ++ (4,0) ++ (2,0) ++ (2,0) ++ (2,\i*2) ++ (2,\j*1) ellipse (0.6 and 0.9);
\node at (1+4+2+2+2+2+0.4, 7.05+\i*2+\j) {\textcolor{green}{\bf h+2}};
}
\draw [red, very thick] (1,7.05) ++ (4,0) ++ (4,0) ++ (2,0) ++ (2,\i*2) ellipse (1.1 and 2);
\node at (1+4+2+2+2+2+0.7, 7.05+\i*2) {\textcolor{red}{\bf h+1}};
}
\draw [blue, very thick] (1,7) ++ (10,0) ++ (2,0) ellipse (2.2 and 4.2);
\node at (1+4+2+2+2+2+2, 7.05) {\textcolor{blue}{\bf h}};
\foreach \i in {0,1,2,3,4,5,6} {%
\draw (13,10.5 - \i) to [out=-45, in=45, looseness=1] (13,10.5 -1 -\i);}
\draw  (13,10.5)to [out=225, in=-225, looseness=1] (13,8.5);
\draw (13,6.5)to [out=225, in=-225, looseness=1] (13,4.5);
\draw (13,7.5)to [out=225, in=-225, looseness=1] (13,6.5);
\draw [very thick, dashed] (13,8.5)to [out=225, in=-225, looseness=1] (13,7.5);
\draw [very thick, dashed] (13,4.5)to [out=225, in=-225, looseness=1] (13,3.5);
\draw [postaction={decorate}] (12,12) -- (13,10.5);
\draw [postaction={decorate}] (13,10.5) -- (14,12);
\draw [postaction={decorate}] (12,2) -- (13,3.5);
\draw [postaction={decorate}] (13,3.5) -- (14,2);
\draw [purple] (7,7) ellipse (0.8 and 4.2);
\node at (7,12) {\textcolor{purple}{$\bm{V_\mathcal B(\tau)\subset V(\tau)}$}};
\end{tikzpicture}
\caption{Simplification process of the hierarchical structure of the Renormalization Operators. The first two figures on the left represent the original renormalized tree we consider and the respective {\it Feynman diagrams structure}; on the right there is the "simplified tree"and the respective {\it Feynman diagrams structure}. In the {\it Feynman diagrams structure}, the black lines are the $P$-type propagators, the dashed lines are the $R-$type propagators, the black dots are $\lambda$-type endpoints and the black squares are $\varpi$-type endpoints. In sake of simplicity, each of the clusters has $4$ external legs and contains one non translation invariant element which is not contained in any sublcusters, so in this particular graph there is no need of Taylor renormalization. After the simplification procedure, only the "{\it innermost}" non-translation-invariant elements survive, and the main goal of this section will be to show that they are enough to renormalizethe whole tree. Finally, the vertices of the tree surrounded by the purple line belong to the set $V_\mathcal B(\mathcal \tau)\subset V(\tau)$.}
\label{figure_simplification_figure_after}
\end{figure}

 First of all using Remark (\ref{remark_commutation_renormalization operators}), as in (\ref{renormalized_kernels_explicit_expression}), we imagine to have already {\it applied} the {\it Taylor renormalization operators} $\mathcal R_\mathcal T, \tilde{\mathcal R}_{\mathcal T}$, so in Figure (\ref{figure_simplification_figure_after}) we drop in sake of simplicity the symbols $\mathcal R_\mathcal T, \tilde{\mathcal R}_\mathcal T$.\\ 
Starting from each leaf of the Gallavotti-Nicolo ({\it "Taylor renoemalized"}) tree, we descend the tree toward the root until we meet for the first time a vertex $v\in V(\tau)$ labeled by $\mathcal R_{\mathcal B}\in\{\mathcal R^{(1)}_\mathcal B, \mathcal R^{(2)}_	\mathcal B\}$ : from this point on all the ancestors $w\prec v$ are labeled by $\mathcal R_{\mathcal B}$, but there are two possibilities: 
\begin{itemize}
\item either there are {\it neither} remainder propagators {\it nor} $\varpi$-type endpoints at scale $h_w$ \footnote{In sake of clarity we repeat that, by {\it having a propagator (resp. an endpoint) at scale $h_w$} we properly mean that the propagator (resp. the endpoint) is an element of the cluster $G_w$, but it is not an element of any of the subclusters $G_{\bar w}\subset G_{w}$, where $\bar w$ is a descendent of $w$.}, and we do nothing, 
\item or there is at least {\it either} a remainder propagator {\it or} $\varpi$-type endpoint at scale $h_w< h_v$, and for each of them we use Corollary (\ref{corollary_norms_propagators_DBC}) $$||g_R^{(h_w)}||_\infty\leq ||g_P^{(h_w)}||_\infty,\hspace{3mm} ||g_R^{(h_w)}||_1\leq ||g_P^{(h_w)}||_1,\hspace{3mm} \sup_{x,y}|\varpi_{h_w}(x)|\leq C_{\theta'},$$
for some $\theta'$ fixed {\it a priori} to, respectively, replace $g_R^{(h_w)}$ by $g^{(h_w)}_P$ and $|\varpi_{h_w}(x)|$ by $C_{\theta'}$ for some $\theta'$ suitably chosen a priori.
\end{itemize}
\begin{rem}
From now on, we will say that $\mathcal R_\mathcal T\mathcal R_\mathcal B$ (or $\tilde{\mathcal R}_\mathcal T\mathcal R_\mathcal B$) {\it acts in a non-trivial way} only on those clusters $G_v$ containing a non-trivial weight function just  at scale $h_v$, and we call $V_{\mathcal B}(\tau)\subset V(\tau)$ the set of the vertices of the tree labeled in a non trivial way by $\mathcal R_\mathcal T\mathcal R_{\mathcal B}$ (or $\tilde{\mathcal R}_\mathcal T\mathcal R_\mathcal B$).
\end{rem}

\paragraph{Proof of Theorem \ref{theorem_renormalized_bounds_DBC}}

Once we introduced all the possible and useful simplifications and defined the notations, we can start proving Theorem (\ref{theorem_renormalized_bounds_DBC}). In particular, the proof will consist of two fundamental Lemmata that we are going to introduce:
\begin{itemize}
\item thanks to Lemma (\ref{lemma_bound_determinants_inside_renormalized_kernels}), we will reduce the problem of {\it bounding the kernels in Theorem (\ref{theorem_renormalized_bounds_DBC})} to the problem of {\it bounding the integral over a spanning tree whose vertices are weighted by the weight functions $w_h(\cdot)$} we just introduced.
\item thanks to Lemma (\ref{lemma_effective_gain}) we exploit the presence of these weight functions to get the dimensional gains that renormalize the tree. In particular, we will prove Lemma (\ref{lemma_effective_gain}) {\it via} two auxiliary lemmata: \begin{itemize}
\item Lemma (\ref{lemma_transfer}) tells us that we can re-arrange the spanning tree as we need, by moving the weight functions $w_h(\cdot)$ inside the cluster at scale $h$ it belongs to,
\item Lemma (\ref{lemma_integral_w_g}) tells us where the dimensional gains actually come from,
\end{itemize}
\item putting together Lemmata (\ref{lemma_bound_determinants_inside_renormalized_kernels}) and (\ref{lemma_effective_gain}), we obtain in a straightforward way the desired bound in Theorem (\ref{theorem_renormalized_bounds_DBC}).
\end{itemize}

\begin{lem}
\label{lemma_bound_determinants_inside_renormalized_kernels}
\begin{equation}
\begin{split}
\left|\int d\bm x(P_{v_0}) \mathcal R_{\alpha} W^{(h)}(\tau, P_{v_0},\bm x(I_{v_0}))\right|
\leq  C \left[\prod_{v\notin V_f(\tau)}\left(\frac{Z_{h_v}}{Z_{h_{v}-1}}\right)^{\frac{|P_v|}{2}}\right]\int d\bm x(P_{v_0}) \cdot
\\ \cdot \left\{\prod_{v\notin V_f(\tau)}\frac{1}{s_v!}\gamma^{h_vq_{\alpha,G^{h_v,T_v}}}||g^{(h_v)}||^{\frac{1}{2}\left(\sum_{j=1}^{s_v}|P_{v_j}|-|P_v|-2(s_v-1)\right)}_{\infty}\cdot\right.\\
\left.\cdot \left[\prod_{\ell \in T_v}\left|(\bm x_\ell- \bm y_\ell)^{b_\alpha(\ell)}_{j_\alpha(\ell)}\partial^{q(f_\ell^1)}_{j(f_\ell^1)}\partial^{q(f_\ell^2)}_{j(f_\ell^2)} g^{(h_\ell)}_{\ell}\right|\right]\right\}\cdot \\ \cdot\left[\prod_{i=1}^{n}\left|(\bm x^i-\bm y^i)^{b(v^*_i)}_{j(v^*_i)}K^{(h_i)}_{{v^*_{i}}}(\bm x_{v^*_i})\right|\right]\left(\prod_{v\in V_\mathcal{B}(\tau)}w_h(x_v)\right),
\end{split}
\end{equation}
where the terms $\gamma^{h_vq_{\alpha,G^{h_v,T_v}}}$ take into account the dimensional gains coming from the derivatives in (\ref{matrix_G_h_v_t_v}), the argument of the square brackets in the last line has to be read as in (\ref{renormalized_kernels_explicit_expression}) where we replaced $\varpi_h(x)$ by a constant, and the argument of the last brackets are the weight functions we defined in (\ref{definition_weight_functions}).
\end{lem}
We will use another Lemma to bound the integral appearing in the r.h.s. of the latter formula.
\begin{lem}
\label{lemma_effective_gain}
\begin{equation}
\begin{split}
\frac{1}{|\Lambda|\beta}\int d\bm x(P_{v_0})\prod_{v\notin V_f(\tau)}\gamma^{h_vq_{\alpha,G^{h_v,T_v}}}\left(\left[\prod_{\ell \in T_v}\left|(\bm x_\ell- \bm y_\ell)^{b(\ell)}_{j(\ell)}\partial^{q(f_\ell^1)}_{j(f_\ell^1)}\partial^{q(f_\ell^2)}_{j(f_\ell^2)} g^{(h_\ell)}_{\ell}\right|\right]\right)\cdot \\ \cdot\left[\prod_{i=1}^{n}\left|(\bm x^i-\bm y^i)^{b(v^*_i)}_{j(v^*_i)}K^{(h_i)}_{{v^*_{i}}}(\bm x_{v^*_i})\right|\right]\left(\prod_{v\in V_\mathcal{B}(\tau)}w_h(x_v)\right) \leq \\
\leq \left( \prod_{v\in V_f(\tau)}\rho_v\right)\left(\prod_{v\notin V_f(\tau)}\gamma^{-h_v(s_v-1)}\right)\left(\prod_{v\notin V_f(\tau)}\gamma^{-z_v(h_v-h_{v'})}\right)
\end{split}
\end{equation}
where 
\begin{equation}
z_v=\begin{cases}
\theta &\mbox{ if } |P_v|=4,\\
1+\theta &\mbox{ if } |P_v|=2,
\end{cases} \hspace{5mm} m_{2,v}=\begin{cases}
1 \mbox{ if $v$ is of type $\nu$ or $\varpi$},\\
0 \mbox{ otherwise}.
\end{cases}
\end{equation}
\end{lem}

\begin{proof}[Proof of Theorem (\ref{theorem_renormalized_bounds_DBC})]
By putting together the results of Lemmata (\ref{lemma_bound_determinants_inside_renormalized_kernels}) and (\ref{lemma_effective_gain}) we can bound the right hand side of (\ref{bounds_renormalized_kernels_DBC}) by
\begin{equation}
\begin{split}
\left(\prod_{v\in V_f(\tau)}\rho_v\right)\left(\prod_{v\notin V_f(\tau)}\gamma^{\frac{h_v}{2}\left(\sum_{j=1}^{s_v}|P_{v_j}|-|P_v|-2(s_v-1)\right)}\gamma^{-h_v(s_v-1)}\right)\cdot \\ \cdot\left(\prod_{v\notin V_f(\tau)}\gamma^{-z_h(h_v-h_{v'})}\right)
\end{split}
\end{equation}
\end{proof}

\subparagraph{Proofs of Lemmata (\ref{lemma_bound_determinants_inside_renormalized_kernels}) and (\ref{lemma_effective_gain})} 

\begin{proof}[Proof of Lemma (\ref{lemma_bound_determinants_inside_renormalized_kernels})]

Let us start by considering the action of the operator $\mathcal R_{\mathcal{B}}^{(1)}$ on a vertex $v\in V_\mathcal B(\tau)$. 
\paragraph{Action of $\mathcal R^{(1)}_{\alpha,\mathcal B}$} We refer again to Remarks (\ref{remark_R_B_decomposition_R1_R2}) and (\ref{remark_commutation_renormalization operators}), and we recall the formal representation
\begin{equation*}
\begin{split}
 \mathcal R^{(1)}_{\alpha,v,\mathcal B}\left\{\int dP_{T_v}(\bm t_v) \left( \det G_\alpha^{h_v, T_v}(\bm t_v)\right)\cdot\right.\\\left.
\cdot \left[\prod_{\ell \in T_v}\left|(\bm x_\ell- \bm y_\ell)^{b(\ell)}_{j(\ell)}\partial^{q(f_\ell^1)}_{j(f_\ell^1)}\partial^{q(f_\ell^2)}_{j(f_\ell^2)} g^{(h_\ell)}_{\ell}\right|\right]\right\}\\ \left[\prod_{v_i^*\in V_f^{(\varpi)}}\left|(\bm x^i-\bm y^i)^{b(v^*_i)}_{j(v^*_i)}K^{(h_i)}_{{v^*_{i}}}(\bm x_{v^*_i}))\right|\right]\left[\prod_{v_i^*\in V_f\setminus V_f^{(\varpi)}}\left|(\bm x^i-\bm y^i)^{b(v^*_i)}_{j(v^*_i)}K^{(h_i)}_{{v^*_{i}}}(\bm x_{v^*_i}))\right|\right]
\end{split}
\end{equation*}
where $\mathcal R^{(1)}_{\alpha,v,\mathcal B}$ formally means that the vertex $v\in V_\mathcal B(\tau)$ is renormalized by $\mathcal R^{(1)}_\mathcal B$, with the constraints given by the structure of the subtree $\tau_v$ encoded in $\alpha$.
\begin{itemize}
\item {\bf Case 0} The argument of the first square brackets of the last line gives us directly the weight functions $\varpi_h(\cdot)$ coming from the $\varpi-$type endpoints, and they are independent of $\det G_\alpha^{h_v,T_v}$. To bound $\det G_\alpha^{h_v,T_v}$, we use as usual the Gram-Hadamard inequality (\ref{lemma_gram_hadamard_for_G}).
\end{itemize}
So we are left with showing that, for any $v\in V_\mathcal B(\tau)$, we can extract a weight function $\rho_{h_v}(\cdot)$ from the propagators. There are two different cases: given $v\in V_\mathcal B(\tau)$, either there is at least one remainder propagator belonging to the spanning tree $T_v$, or each of the propagators belonging to the spanning tree is a $P$-labeled propagator and $G_\alpha^{h_v,T_v}$ has a block of remainder propagators.
\begin{itemize}
\item {\bf Case 1: given $v\in V_\mathcal B(\tau)$, there is at least one $R$-labeled propagators belonging to the spanning tree $T_v$:} first of all, we use the Gram-Hadamard inequality (\ref{lemma_gram_hadamard_for_G}) to bound the determinant. To extract the weight function from the remainder propagator belonging to the spanning tree, we use (\ref{reminder_propagators_as_periodic_and_weight}) to bound
$$\int d\bm x_\ell d\bm y_\ell |g_{\ell,R}|\leq  c \int d\bm x_\ell d\bm y_\ell |g_{\ell,P}| \rho_{h_\ell}(x_\ell), \hspace{3mm} \ell\in T_v.$$
\item {\bf Case 2: given $v\in V_\mathcal B(\tau)$, there are no $R$-labeled propagators belonging to the spanning tree $T_v$, and $G_\alpha^{h_v,T_v}$ has a block of remainder propagators.} Let us call this set of vertices $\bar V_\mathcal B(\tau)\subseteq V_\mathcal B(\tau)$. The basic idea is to expand, for each $v\in\bar V_\mathcal B(\tau)$ , $\det G_\alpha^{h_v,T_v}(\bm t_v)$ using the very definition of {\it determinant}, along a row of remainder propagators:
\begin{equation}
det G_\alpha^{h_v,T_v}=\sum_{i=1}^{s_v}(-1)^{i+j}t_{ij}\partial_{j_\alpha(f_{ij}^1)}^{q_\alpha(f_{ij}^1)}\partial_{j_\alpha(f_{ij}^2)}^{q_\alpha(f_{ij}^2)}g_R^{h_v}(\bm x(i),\bm x(j)) G^{h_v, T_v}_{\alpha, ij},
\label{detG_expanded_along_a_line}
\end{equation}
where we recall that $T_v$ is defined as the set of lines such that $T=\cup_{v\notin V_f(\tau)} T_{v}$, and where $G_{\alpha, ij}^{h_v,T_v}$ is the determinant of the matrix obtained starting from $G_\alpha^{h_v,T_v}$ and erasing the row $i$ and the column $j$. Once we extracted the remainder propagators, by using the bound (\ref{bound_g_R_g_infty_rho}) we get:
\begin{equation}
\begin{split}
\left| \int d\bm x(P_{v_0})\prod_{\ell\in T_v} g_\ell \int P(d\bm t) \det G_\alpha^{h_v, T_v}(\bm t)\right|\leq \\ \leq c_{v,0} \left|\int d\bm x(P_{v_0})\prod_{\ell\in T_v} g_\ell  \cdot \right.\\ \left.\cdot \int P(d\bm t) \sum_{i}  (-1)^{i+j}t_{i_j} \partial_{j_\alpha(f_{ij}^1)}^{q_\alpha(f_{ij}^1)}\partial_{j_\alpha(f_{ij}^2)}^{q_\alpha(f_{ij}^2)}g_R^{(h_v)}(\bm x(i),\bm x(j))G_{\alpha,ij}^{h_v, T_v}(\bm t) \right| \leq \\
\leq  c_{v,1} \int d\bm x(P_{v_0})\prod_{\ell\in T_v} |g_\ell | \cdot \\ .\cdot\int P(d\bm t) \sum_{i}   \rho_{h_v}(x(j)) ||\partial_{j_\alpha(f_{ij}^1)}^{q_\alpha(f_{ij}^1)}\partial_{j_\alpha(f_{ij}^2)}^{q_\alpha(f_{ij}^2)}g_R^{(h_v)}||_\infty ||G_{\alpha, ij}^{h_v, T_v}(\bm t)||_\infty \leq\\
\leq C_{v} ||g^{(h_v)}||^{\frac{1}{2}\left(\sum_{i=1}^{s_v}|P_{v_i}|-|P_v|-2(s_v-1)\right)}_\infty \gamma^{h_vq_{\alpha,G^{h_v,T_v}}}\int d\bm x(P_{v_0})\left(\prod_{\ell\in T_v} \left|g_\ell \right| \right)\rho_{h_v}(x(j)).
\end{split}
\label{bound_spanning_tree_propagator_expanded}
\end{equation}
where $C_v$ depends on the size of (number of propagators belonging to)  the cluster $G_v$. Since there are, in general, more than one vertices of this type, we are left with controlling $\prod_{v\in \bar V_\mathcal B(\tau)}C_v$. 
\end{itemize}
The worst case possible is when $\bar V_\mathcal B(\tau)=V_\mathcal B(\tau)$, so we are forced to expand, for each $v\in V_{ \mathcal B}(\tau)$, the determinant of the $\left(\sum_{i=1}^{s_v}|P_{v_1}|-|P_{v}|\right)/2\times\left(\sum_{i=1}^{s_v}|P_{v_1}|-|P_{v}|\right)/2$ matrix $G^{h_v,T_v}$, so that $\prod_{v\in V_\mathcal B(\tau)}C_v\leq \prod_{v\in V_{{\mathcal B}}(\tau)}\left(\frac{\sum_{i=1}^{s_v}|P_{v_1}|-|P_{v}|}{2}\right)$. We want to prove that:
\begin{equation}
\prod_{v\in V_{{\mathcal B}}(\tau)}\left(\frac{\sum_{i=1}^{s_v}|P_{v_1}|-|P_{v}|}{2}\right)\leq c e^n, \hspace{3mm} \forall \tau\in\mathcal T_{h,n},
\end{equation}
where $n$ is the number of the endpoints. Let us prove the latter bound: thanks to the hierarchical structure of the set of vertices $V_{\mathcal B}(\tau)$ we just explained,
\begin{equation}
\prod_{v\in V_{\mathcal B}(\tau)}\left(\frac{\sum_{i=1}^{s_v}|P_{v_1}|-|P_{v}|}{2}\right)\leq c_1 \prod_{i=1}^k s_{v_i},\mbox{ with the constraints: } \begin{cases} 1\leq k\leq n,\\ \sum_{i=1}^k s_{v_i}=n.\end{cases}
\end{equation}
So
\begin{equation}
\begin{split}
\prod_{i=1}^k s_{v_i}=e^{\sum_{i=1}^k \log s_{v_i}}=e^{k\left(\frac{1}{k}\sum_{i=1}^k \log s_{v_i}\right)}\leq e^{k \log \left(\frac{1}{k}\sum_{i=1}^k s_{v_i}\right)}= \\ =e^{k\log \frac{n}{k}}\leq c e^n, \hspace{2mm} \forall 1\leq k \leq n.
\end{split}
\end{equation}
\paragraph{Action of $\mathcal R^{(2)}_{\mathcal B}$} Remark (\ref{remark_R_B_decomposition_R1_R2}) and the fact that
\begin{itemize}
\item all the propagators of the starting tree are $g_P$,
\item $T=\cup_{v\notin V_f(\tau)}T_{v}$,
\item  the action of $\mathcal R_\mathcal B^{(2)}$ keeps fixed at least one among $\bm x$ and $\bm y$,
\end{itemize}
ensure that, as soon as we change the sign of some space variable, at least a remainder propagator appears on the spanning tree $T_v$ for each $v\in V_\mathcal B(\tau)$. It can be proved iteritively: starting from the innermost subclusters ({\it i.e.} the vertices of the tree immediately preceeding the leaves), we can look at the action of $\mathcal R^{2}_\mathcal B$ as giving three different subcases:
\begin{enumerate}
\item $\mathcal R^{2}_\mathcal B$ does not change the sing of the coordinate of any vertex belonging to the clusters, so nothing changes,
\item $\mathcal R^{2}_\mathcal B$ changes the sing of the coordinate of each of the vertices belonging to the clusters, so nothing changes by symmetry with the previous point,
\item $\mathcal R^{2}_\mathcal B$ changes the sing of the coordinate of a subset of the vertices belonging to the clusters, leaving at least one of the vertices unchanged: so at least one of the propagators belonging to the spanning tree becomes, by the symmetry properties of the propagators, a remainder propagator.
\end{enumerate}
Of course, in the case 1 and 2 the subclusters have to be, if necessary, "Taylor renormalized".\\
Iteratively, we apply these three points on the bigger clusters with an ingredient more: there is at least a line of the spanning tree connecting the vertices of the cluster we are considering with a vertex belonging to some subclusters: so even in cases 1 and 2 there could appear some $R$-labeled propagator on the spanning tree. This mechanism ensures that, if in some cluster the $P$ symmetry is broken, for sure it is broken on the spanning tree. Besides, the fact that at least one among $\bm x, \bm y$ is kept fixed, ensures that at some point of the tree this symmetry is broken.\\
 Analogously to {\bf Case 1}, by using (\ref{reminder_propagators_as_periodic_and_weight}) and the Gram-Hadamard (\ref{lemma_gram_hadamard_for_G}) inequality we get the result.
\end{proof}

Now we prove Lemma (\ref{lemma_effective_gain}).

\begin{proof}[Proof of Lemma (\ref{lemma_effective_gain})]
Let us recall that we want to prove
\begin{equation*}
\begin{split}
\frac{1}{|\Lambda|\beta}\int d\bm x(P_{v_0})\prod_{v\notin V_f(\tau)}\gamma^{h_vq_{\alpha,G^{h_v,T_v}}}\left(\left[\prod_{\ell \in T_v}\left|(\bm x_\ell- \bm y_\ell)^{b(\ell)}_{j(\ell)}\partial^{q(f_\ell^1)}_{j(f_\ell^1)}\partial^{q(f_\ell^2)}_{j(f_\ell^2)} g^{(h_\ell)}_{\ell}\right|\right]\right)\cdot \\ \cdot\left[\prod_{i=1}^{n}\left|(\bm x^i-\bm y^i)^{b(v^*_i)}_{j(v^*_i)}K^{(h_i)}_{{v^*_{i}}}(\bm x_{v^*_i})\right|\right]\left(\prod_{v\in V_\mathcal{B}(\tau)}w_h(x_v)\right) \leq \\
\leq \left( \prod_{v\in V_f(\tau)}\rho_v\right)\left(\prod_{v\notin V_f(\tau)}\gamma^{-h_v(s_v-1)}\right)\left(\prod_{v\notin V_f(\tau)}\gamma^{-z_v(h_v-h_{v'})}\right)
\end{split}
\end{equation*}
{\bf Observation 1:} if we bounded, for each $v\in V_\mathcal B(\tau)$, $|w_{h_v}(x_v)|\leq C_\theta$ for a suitably fixed $\theta\in (0,1)$, we would get an analogous bound, provided we replaced $z_v$ by
\begin{equation*}
\tilde z_v=\begin{cases}
2 \mbox{ if } |P_v|=2 \mbox{ and } v\in V(\tau)\setminus V_\mathcal B(\tau),\\
1 \mbox{ if } |P_v|=4 \mbox{ and }  v\in V(\tau)\setminus V_\mathcal B(\tau),\\
1 \mbox{ if } |P_v|=2 \mbox{ and }  v\in V_\mathcal B(\tau),\\
0 \mbox{ if } |P_v|=4 \mbox{ and }  v\in V_\mathcal B(\tau).
\end{cases}
\end{equation*}
This observation is a consequence of Remark (\ref{remark_commutation_renormalization operators}) and of the dimensional gains coming from the {\it Taylor} renormalization operators we described in the previous chapter. We will get the further gains by exploiting the presence of $\left(\prod_{v\in V_\mathcal{B}(\tau)}w_{h_v}(x_v)\right)$. \\
{\bf Observation 2:} once we reconstructed the bound of the determinants, we are left with computing an integral along the spanning tree $T=\cup_{v}T_v$ formally analogous to the one we bounded in proving Theorem (\ref{theorem_renormalized_bounds}), with the only difference that {\it some of the vertices of the spanning tree are weighted by a weight function $w_{h_v}(x_v)$}. Moreover, let us recall that the spanning tree $T=\cup_{v}T_v$ has a hierarchical sturcture.
In fact, we will exploit this hierarchical structure to obtain the dimensional gains we need, and we will proceed in two steps:
\begin{enumerate}
\item first of all, we show that we can arbitrarily transfer the function $w_h(\cdot)$ from a vertex belonging to a cluster at scale $h$ to any vertex belonging to the same cluster at scale $h$ (Lemma (\ref{lemma_transfer})); 
\item then, we  prove that in fact we can transfer the function $w_h(\cdot)$ to a vertex belonging to some cluster at lower scale $\bar h< h$ which contains as a subcluster the cluster at scale $h$, gaining a dimensional factor $\gamma^{\theta(\bar h- h)}$, (Lemma (\ref{lemma_integral_w_g})). 
\end{enumerate}

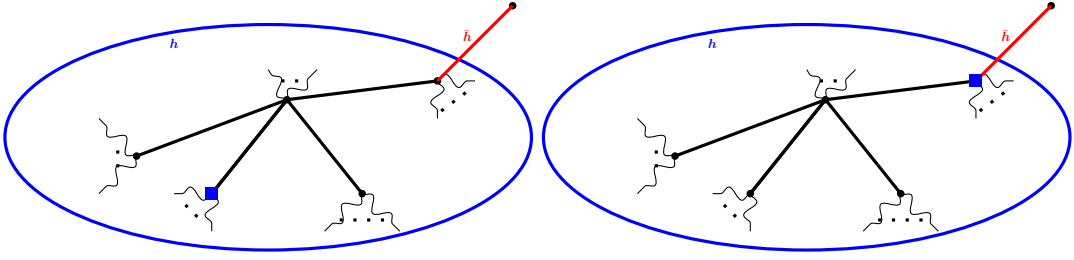
\begin{figure}
\begin{tikzpicture}
[scale=0.5, transform shape]
\draw [blue, very thick] (7.5,7.5) ellipse (7 and 3);
\node at (5,10) {\textcolor{blue}{$\bm h$}};
\node at (12.8, 10.2) {\textcolor{red}{{$\bar {\bm h}$}}};
\fill (6, 6) circle (0.1);
\fill (8, 8.5) circle (0.1);
\fill (12, 9) circle (0.1);
\fill (4,7) circle (0.1);
\fill (10, 6) circle (0.1);
\fill (14,11) circle (0.1);
\draw [very thick] (4,7) -- (8,8.5);
\draw [-,decorate, decoration={snake}] (4,7) -- (3,6);
\draw [-,decorate, decoration={snake}] (4,7) -- (3,8);
\draw [very thick, loosely dotted] (3.5, 6.7) -- (3.5,7.3);
\draw [-,decorate, decoration={snake}] (6,6) -- (5,6);
\draw [-,decorate, decoration={snake}] (6,6) -- (6,5);
\draw [very thick, loosely dotted] (5.3, 5.7) -- (5.7,5.3);
\draw [-,decorate, decoration={snake}] (8,8.5) -- (7.5,9.3);
\draw [-,decorate, decoration={snake}] (8,8.5) -- (8.8,9.3);
\draw [very thick, loosely dotted] (7.85,9) -- (8.45,9);
\draw [-,decorate, decoration={snake}] (10,6) -- (9,5);
\draw [-,decorate, decoration={snake}] (10,6) -- (11,5);
\draw [very thick, loosely dotted] (9.4,5.3) -- (10.6,5.3);
\draw [-,decorate, decoration={snake}] (12,9) -- (13,9);
\draw [-,decorate, decoration={snake}] (12,9) -- (12,8);
\draw [very thick, loosely dotted] (12.1,8.2) -- (12.9,8.8);
\draw [very thick] (8,8.5) -- (6,6);
\draw [very thick] (8,8.5) -- (6,6);
\draw [very thick] (8,8.5) -- (10,6);
\draw [very thick] (8,8.5) -- (12,9);
\draw [very thick, red] (12,9) -- (14,11);
\node [regular polygon, blue, regular polygon sides=4,
        minimum size=3mm, fill] at (6, 6) {};
\end{tikzpicture}
\begin{tikzpicture}
[scale=0.5, transform shape]
\draw [blue, very thick] (7.5,7.5) ellipse (7 and 3);
\node at (5,10) {\textcolor{blue}{$\bm h$}};
\node at (12.8, 10.2) {\textcolor{red}{{$\bar {\bm h}$}}};
\fill (6, 6) circle (0.1);
\fill (8, 8.5) circle (0.1);
\fill (12, 9) circle (0.1);
\fill (4,7) circle (0.1);
\fill (10, 6) circle (0.1);
\fill (14,11) circle (0.1);
\draw [very thick] (4,7) -- (8,8.5);
\draw [-,decorate, decoration={snake}] (4,7) -- (3,6);
\draw [-,decorate, decoration={snake}] (4,7) -- (3,8);
\draw [very thick, loosely dotted] (3.5, 6.7) -- (3.5,7.3);
\draw [-,decorate, decoration={snake}] (6,6) -- (5,6);
\draw [-,decorate, decoration={snake}] (6,6) -- (6,5);
\draw [very thick, loosely dotted] (5.3, 5.7) -- (5.7,5.3);
\draw [-,decorate, decoration={snake}] (8,8.5) -- (7.5,9.3);
\draw [-,decorate, decoration={snake}] (8,8.5) -- (8.8,9.3);
\draw [very thick, loosely dotted] (7.85,9) -- (8.45,9);
\draw [-,decorate, decoration={snake}] (10,6) -- (9,5);
\draw [-,decorate, decoration={snake}] (10,6) -- (11,5);
\draw [very thick, loosely dotted] (9.4,5.3) -- (10.6,5.3);
\draw [-,decorate, decoration={snake}] (12,9) -- (13,9);
\draw [-,decorate, decoration={snake}] (12,9) -- (12,8);
\draw [very thick, loosely dotted] (12.1,8.2) -- (12.9,8.8);
\draw [very thick] (8,8.5) -- (6,6);
\draw [very thick] (8,8.5) -- (6,6);
\draw [very thick] (8,8.5) -- (10,6);
\draw [very thick] (8,8.5) -- (12,9);
\draw [red, very thick] (12,9) -- (14,11);
\node [regular polygon,blue, regular polygon sides=4,
        minimum size=3mm, fill] at (12, 9) {};
\end{tikzpicture}
\caption{Graphical explanation of Lemma (\ref{lemma_transfer}): the blue square represents a weight function at scale $h$, and our goal is to {\it move it} along the spanning tree until the vertex shared with the red propagator, living at scale $\bar h\leq h$.}
\label{figure_lemma_transfer}
\end{figure}
Let $$\rho_h(\bm x,\bm y)=\gamma^h\frac{C_{\bar N}}{1+(\gamma^h|\bm x-\bm y|)^{\bar N}},\hspace{3mm} \rho^{(q_1,q_2;j_1,j_2)}_\ell=\partial^{q(f_\ell^1)}_{j(f_\ell^1)}\partial^{q(f_\ell^2)}_{j(f_\ell^2)}\gamma^{h_\ell}\frac{C_{\bar N}}{1+(\gamma^h|\bm x(f_\ell^1)-\bm x(f_\ell^2)|)^{\bar N}},$$
for some suitably fixed $\bar N$, and let $\tilde K^{(h_i)}_{{v'^*_{i}}}(\bm x_{v'^*_i})$ be the contribution obtained by replacing, in the iterative definition of the endpoints contributions $K^{(h_i)}_{{v'^*_{i}}}(\bm x_{v'^*_i})$, each of the propagators by the suitable $\rho^{(q_1,q_2;j_1,j_2)}_\ell$.
\begin{lem}
\label{lemma_transfer} 
Let  $v\in V_\mathcal B(\tau)$, $\tau_v\subset \tau$ the subtree whose first vertex is $v$, and let $\bm s\in \bm x(P_v)$. So
\begin{equation}
\begin{split}
\int d\bm x(P_{v})\prod_{v'\notin V_f(\tau_v)}\left[\prod_{\ell \in T_{v'}}\left|(\bm x_\ell- \bm y_\ell)^{b(\ell)}_{j(\ell)}\partial^{q(f_\ell^1)}_{j(f_\ell^1)}\partial^{q(f_\ell^2)}_{j(f_\ell^2)} g^{(h_\ell)}_{\ell}\right|\right]\cdot \\ \cdot\left[\prod_{i=1}^{n}\left|(\bm x^i-\bm y^i)^{b(v'^*_i)}_{j(v'^*_i)}K^{(h_i)}_{{v'^*_{i}}}(\bm x_{v'^*_i})\right|\right]w_{h_{v}}(s) \leq \\
\leq C^{n^0_v-1} \int d\bm x(P_{v})\prod_{v'\notin V_f(\tau_v)}\left[\prod_{\ell \in T_{v'}}\left|(\bm x_\ell- \bm y_\ell)^{b(\ell)}_{j(\ell)}\rho^{(q_1,q_2;j_1,j_2)}_\ell\right|\right]\cdot \\ \cdot\left[\prod_{i=1}^{n}\left|(\bm x^i-\bm y^i)^{b(v'^*_i)}_{j(v'^*_i)} \tilde K^{(h_i)}_{{v'^*_{i}}}(\bm x_{v'^*_i})\right|\right]w_{h_v}(x)
\end{split}
\label{integral_lemma_transfer}
\end{equation}
$\forall \hspace{2mm} \bm x\in \bm x(P_v)$, where $n^0_v$ is the number of endpoints following $v$.
\end{lem}
\begin{proof}
Let us prove the Lemma by considering $w_h(s)=\varpi_h(s)$, where $s$ is the space-component of the integration point $\bm s$ (the proof for $w_h(s)=\rho_h(s)$ is conceptually the same).\\
Of course, being $|w_h(s)|\leq C_\theta/(1+\gamma^h|x|)^\theta$ for any $\theta\in (0,1)$ and a suitable $C_\theta$, and
\begin{equation}
\frac{C_\theta}{(1+\gamma^h|s|)^{\theta}}=\frac{C_\theta}{(1+\gamma^{ h}|x|)^{\theta}}\left(\frac{1+\gamma^{ h}|x|}{1+\gamma^{h}|s|}\right)^\theta
\end{equation} 
we can bound the latter integral in (\ref{integral_lemma_transfer}) by
\begin{equation*}
\begin{split}
C \int d\bm x(P_{v})\left[\prod_{\ell \in T_v}\left|(\bm x_\ell-\bm y_\ell)^{b(\ell)}_{j(\ell)}\partial^{q(f_\ell^1)}_{j(f_\ell^1)}\partial^{q(f_\ell^2)}_{j(f_\ell^2)} g^{(h_\ell)}_{\ell}\right|\right]\cdot \\ \cdot\left[\prod_{i=1}^{n}\left|(\bm x^i-\bm y^i)^{b(v^*_i)}_{j(v^*_i)}K^{(h_i)}_{{v^*_{i}}}(\bm x_{v^*_i}))\right|\right]w_{h_v}(x)\left(\frac{1+\gamma^{ h}|x|}{1+\gamma^{h}|s|}\right)^\theta
\end{split}
\end{equation*}
so the main goal is to replace the factor within the last brackets by a constant. By definition of spanning tree $T_v=\cup_{v'\notin V_f(\tau_v)}T_{v'}$, there exists a connected path of lines belonging to the spanning tree that connects $s$ to $x$, so we can expand along the tree $|x|^\theta$:
\begin{equation}
|x|^\theta\leq C_\theta \left(|s|^\theta+\sum_{i=0}^{m-1}|z_{i+1}-z_i|^\theta\right),
\end{equation}
where $z_0:=x$, $z_m:=s$ and $z_1,\dots,z_{m-1}$ are the real space coordinates associated with the vertices of the path. So
\begin{equation}
 \left(\frac{1+\gamma^{\theta h}|x|^{\theta}}{1+\gamma^{\theta h}|s|^{\theta}}\right)\leq  C_\theta\left(\left(\frac{1+\gamma^{\theta h}|s|^{\theta}}{1+\gamma^{\theta h}|s|^{\theta}}\right) + \sum_{i=0}^{m-1}\gamma^{\theta h}|z_{i+1}-z_i|^\theta\right).
\end{equation}
By construction, for each couple of points $(z_{i+1}, z_i)$ there is a propagator $g_{P,\omega}^{(k)}(\bm z_{i+1}- \bm z_i)$ with $k\geq h$ such that 
$$\int d(z_{i+1_0}-z_{i_0}) \int d(z_{i+1}-z_i)|z_{i+1}-z_i| |g^{(k)}_{P,\omega}(\bm z_{i+1}-\bm z_i))|\leq \gamma^{-k} \gamma^{-k},$$
so that, in order to get the bound we are interested in, we can replace, inside the integral, $$\gamma^{\theta h}|z_{i+1}-z_i|^\theta\leq c_1 \gamma^{\theta (h-k)}\leq c_1,$$ since $h-k\leq 0$.
\end{proof}

\begin{figure}
\centering
\begin{tikzpicture}
[scale=0.5, transform shape]
\draw [blue, very thick] (7.5,7.5) ellipse (7 and 3);
\node at (5,10) {\textcolor{blue}{$\bm h$}};
\node at (12.8, 10.2) {\textcolor{red}{{$\bar {\bm h}$}}};
\node at (13.7,9.9) {\textcolor{red}{$\gamma^{\theta(h-\bar h)}$}};
\fill (6, 6) circle (0.1);
\fill (8, 8.5) circle (0.1);
\fill (12, 9) circle (0.1);
\fill (4,7) circle (0.1);
\fill (10, 6) circle (0.1);
\fill (14,11) circle (0.1);
\draw [very thick] (4,7) -- (8,8.5);
\draw [-,decorate, decoration={snake}] (4,7) -- (3,6);
\draw [-,decorate, decoration={snake}] (4,7) -- (3,8);
\draw [very thick, loosely dotted] (3.5, 6.7) -- (3.5,7.3);
\draw [-,decorate, decoration={snake}] (6,6) -- (5,6);
\draw [-,decorate, decoration={snake}] (6,6) -- (6,5);
\draw [very thick, loosely dotted] (5.3, 5.7) -- (5.7,5.3);
\draw [-,decorate, decoration={snake}] (8,8.5) -- (7.5,9.3);
\draw [-,decorate, decoration={snake}] (8,8.5) -- (8.8,9.3);
\draw [very thick, loosely dotted] (7.85,9) -- (8.45,9);
\draw [-,decorate, decoration={snake}] (10,6) -- (9,5);
\draw [-,decorate, decoration={snake}] (10,6) -- (11,5);
\draw [very thick, loosely dotted] (9.4,5.3) -- (10.6,5.3);
\draw [-,decorate, decoration={snake}] (12,9) -- (13,9);
\draw [-,decorate, decoration={snake}] (12,9) -- (12,8);
\draw [very thick, loosely dotted] (12.1,8.2) -- (12.9,8.8);
\draw [very thick] (8,8.5) -- (6,6);
\draw [very thick] (8,8.5) -- (6,6);
\draw [very thick] (8,8.5) -- (10,6);
\draw [very thick] (8,8.5) -- (12,9);
\draw [red, very thick] (12,9) -- (14,11);
\node [regular polygon,red, regular polygon sides=4,
        minimum size=3mm, fill] at (14, 11) {};
\end{tikzpicture}
\caption{This figure has to be thought of as linked to Figure (\ref{figure_lemma_transfer}): by integrating together the blue dot and the red propagator in figure (\ref{figure_lemma_transfer}), we can transfer the blue dot into the red one outside the cluster and get a dimensional gain: this is the result of Lemma (\ref{lemma_effective_gain}).}
\end{figure}
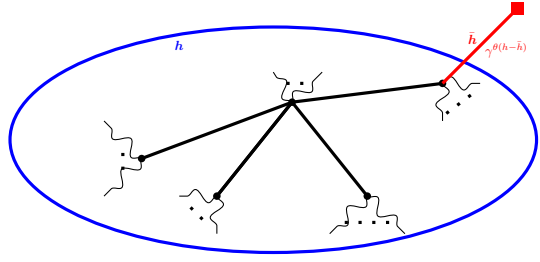

\begin{lem}
\label{lemma_integral_w_g}
Let $w_h(x)\in\{\rho_h(x),\varpi_h(x)\}$ and $\bar h< h$. So
\begin{equation}
\left| \int d\bm y w_h(y)g^{(\bar h)}_{P,\omega}(\bm y-\bm x)\right|\leq C_\theta \gamma^{-\bar h} \gamma^{\alpha(\bar h- h)}w_{\bar h}(x),\mbox{ where } \alpha=\begin{cases}
1 \mbox{ if } w_h(\cdot)=\rho_h(\cdot),\\
\theta \mbox{ if } w_h(\cdot)=\varpi_h(\cdot).
\end{cases}
\end{equation}
\end{lem}

We present a detailed proof of this Lemma in Appendix (\ref{appendix_proof_lemma_effective_gain}) because, even thought simple, it is long and it would make the proof we are involved in less understandable.

\begin{rem}
\label{remark_gain_renormalization_reminder}
Lemma (\ref{lemma_transfer}) tells us that, if there is some weight function at scale $h$, we can associate it with the vertex shared by the kernel at scale $h$ and some propagator at scale $\bar h \leq h$, so it is natural to use Lemma (\ref{lemma_integral_w_g}), which tells us that by integrating a propagator $g^{(\bar h)}_{P,\omega}$ "against" a weight function living at some higher scale $h>\bar h$ (in particular, the one "coming" from the kernel $\mathcal W^{(h)})$, we {\it improve the usual bound} $\gamma^{-\bar h}$ by a factor $\gamma^{\theta(\bar h- h)} w_{\bar h}(\cdot)$, where in particular $\sup_x|w_{\bar h}(x)|\leq C$:
\begin{itemize}
\item we can "associate" the factor $\gamma^{\theta(\bar h-h)}$ to the cluster at scale $h$ the weight function came from: this means that we can extract, from the presence of a weight function, a scale gain in RG language, thanks to which the marginal terms become irrelevant, and the relevant one become marginal, being $0< \theta< 1$,
\item moreover, we transferred the weight function to scale $\bar h$ and of course, if we need it, we further transfer $w_{\bar h}(\cdot)$ to smaller scales getting some scale gain that iteratively improves the power counting.
\end{itemize}
\end{rem}
Now we are left with using these technical Lemmata to prove the bound (\ref{lemma_effective_gain}): the core of the proof consists of a precise integration prescription, that systematically uses Lemmata (\ref{lemma_transfer}) and (\ref{lemma_integral_w_g}) to iteratively renormalize the kernels.
\subparagraph{Integration procedure} 
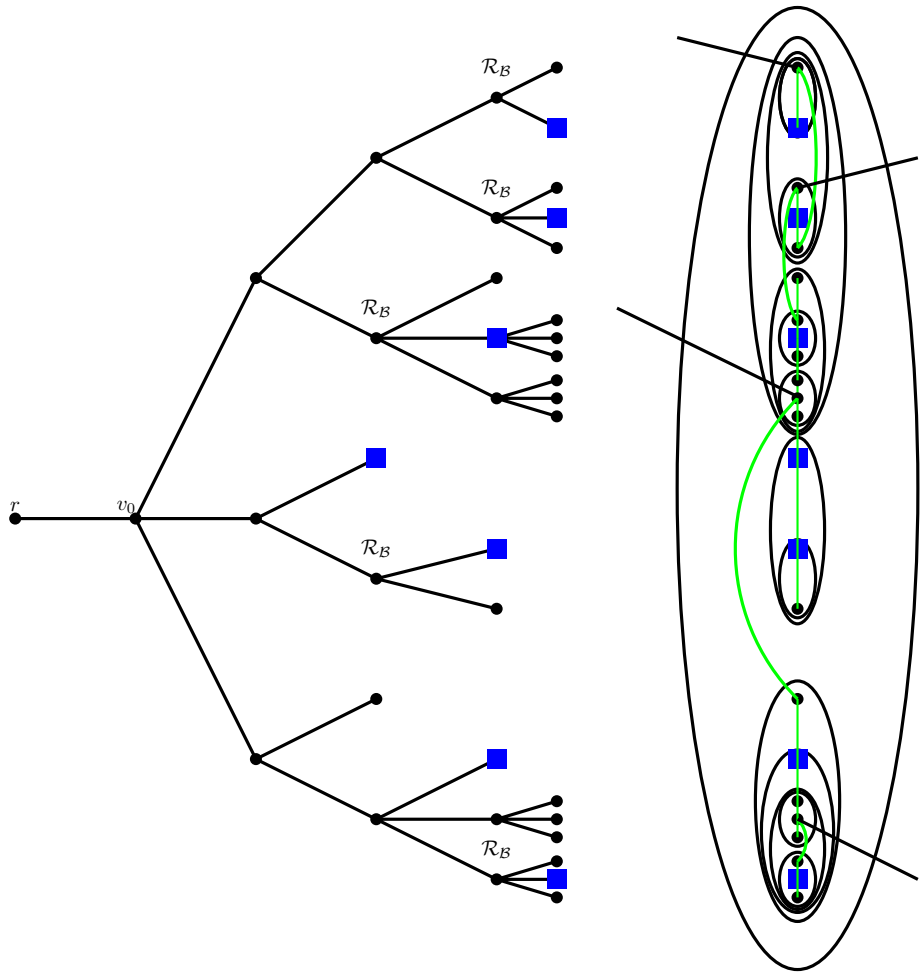
\begin{figure}
\centering
\begin{tikzpicture}
[very thick, scale=0.8, transform shape]
\draw [thick, green] (13,9) -- (13,10);
\draw (0,7) -- ++ (2,0) -- ++ (2,4) ++ (-2,-4) -- ++ (2, 0) ++ (-2,0) -- ++ (2, -4);
\fill (0,7) circle (0.1);
\fill (2,7) circle (0.1);
\fill (4,11) circle (0.1);
\fill (4,7) circle (0.1);
\fill (4,3) circle (0.1);
\draw (4,11) -- ++ (2,2) -- ++ (2,1) -- ++ (1,0.5) ++ (-1,-0.5) -- ++ (1,-0.5);
\draw (4,11) ++ (2,2) -- ++ (2,-1) -- ++ (1,0.5) ++ (-1,-0.5) -- ++ (1,-0.5) ++ (-1, 0.5) -- ++ (1,0);
\draw (4,11) -- ++ (2,-1) -- ++ (2,1) ++ (-2,-1) -- ++ (2,0) -- ++ (1,0.3) ++ (-1,-0.3) -- ++ (1,0) ++ (-1,0) -- ++ (1,-0.3) ++ (-3, 0.3) -- ++ (2,-1) -- ++ (1,0.3) ++ (-1,-0.3) -- ++ (1,0) ++ (-1,0) -- ++ (1,-0.3);
\draw (4,7) -- ++ (2,1) ++ (-2,-1) -- ++ (2,-1) -- ++ (2,0.5) ++ (-2,-0.5) -- ++ (2,-0.5);
\draw (4,3) -- ++ (2,1) ++ (-2,-1) -- ++ (2,-1) -- ++ (2,0) -- ++ (1,0.3) ++ (-1,-0.3) -- ++ (1,0) ++ (-1,0) -- ++ (1,-0.3) ++ (-3, 0.3) -- ++ (2,-1) -- ++ (1,0.3) ++ (-1,-0.3) -- ++ (1,0) ++ (-1,0) -- ++ (1,-0.3);
\draw (6,2) -- ++ (2,1);
\foreach \i in {2,4,6,8,10,13} {
\fill (6,\i) circle (0.1);
}
\foreach \i in {1,2,3,5.5,6.5,9,10,11,12,14}{
\fill (8,\i) circle (0.1);
}
\foreach \i in {0.7,1,1.3,1.7,2,2.3,8.7,9,9.3,9.7,10,10.3,11.5,12,12.5,13.5,14.5}{
\fill (9,\i) circle (0.1);
}
\node [regular polygon,blue, regular polygon sides=4,
        minimum size=3mm, fill] at (8, 3) {};
        \node [regular polygon,blue, regular polygon sides=4,
        minimum size=3mm, fill] at (9, 1) {};
        \node [regular polygon,blue, regular polygon sides=4,
        minimum size=3mm, fill] at (6, 8) {};
        \node [regular polygon,blue, regular polygon sides=4,
        minimum size=3mm, fill] at (8, 6.5) {};
        \node [regular polygon,blue, regular polygon sides=4,
        minimum size=3mm, fill] at (8, 10) {};
        \node [regular polygon,blue, regular polygon sides=4,
        minimum size=3mm, fill] at (9, 12) {};\node [regular polygon,blue, regular polygon sides=4,
        minimum size=3mm, fill] at (9, 13.5) {};
\foreach \i in {0.7,1,1.3,1.7,2,2.3,3,4,5.5,6.5,8,8.7,9,9.3,9.7,10,10.3,11,11.5,12,12.5,13.5,14.5}
{
\fill (13, \i) circle (0.1);
}
\foreach \i in {1,2,9,10}{
\draw (13,\i) ellipse (0.3 and 0.45);
}
\draw (13,6) ellipse (0.3 and 0.65);
\draw (13,12) ellipse (0.3 and 0.65);
\draw (13,14) ellipse (0.3 and 0.65);
\draw (13,14) ellipse (0.3 and 0.65);
\draw (13,13) ellipse (0.5 and 1.75);
\draw (13,9.8) ellipse (0.45 and 1.35);
\draw (13,6.8) ellipse (0.45 and 1.55);
\draw (13,1.5) ellipse (0.45 and 1);
\draw (13,1.8) ellipse (0.6 and 1.35);
\draw (13,2.3) ellipse (0.7 and 2);
\draw (13, 11.7) ellipse (0.8 and 3.3);
\draw (13, 7.5) ellipse (2 and 8);
\node [regular polygon,blue, regular polygon sides=4,
        minimum size=1mm, fill] at (13, 3) {};
        \node [regular polygon,blue, regular polygon sides=4,
        minimum size=1mm, fill] at (13, 1) {};
        \node [regular polygon,blue, regular polygon sides=4,
        minimum size=1mm, fill] at (13, 8) {};
        \node [regular polygon,blue, regular polygon sides=4,
        minimum size=1mm, fill] at (13, 6.5) {};
        \node [regular polygon,blue, regular polygon sides=4,
        minimum size=1mm, fill] at (13, 10) {};
        \node [regular polygon,blue, regular polygon sides=4,
        minimum size=1mm, fill] at (13, 12) {};\node [regular polygon,blue, regular polygon sides=4,
        minimum size=1mm, fill] at (13, 13.5) {};
\draw [thick, green] (13,13.5) -- ++ (0,1);
\draw [thick, green] (13,12.5) -- ++ (0,-0.5) -- ++ (0,-0.5);
\draw [thick, green] (13,11) -- ++ (0,-0.7) ++ (0,-0.3) -- ++ (0,-0.3) -- ++ (0,-0.4) ++ (0,-0.3) -- ++ (0,-0.3) -- ++ (0,-0.7) -- ++ (0,-1.5) -- ++ (0,-1) ++ (0,-1.5) -- ++ (0,-1) -- ++ (0,-0.7) -- ++ (0,-0.3) -- ++ (0,-0.3) ++ (0,-0.4) -- ++ (0,-0.3) -- ++ (0,-0.3);
\draw [green ](13,1.3)  to [out =45, in = -45 ]  (13, 2);
\draw [green ](13,14.5)  to [out =-45, in = 45, looseness=0.5 ]  (13, 11.5);
\draw [green ](13,12.5)  to [out =225, in = -225, looseness=0.5 ]  (13, 10.3);
\draw [green, thick] (13,10.3) -- ++ (0,-0.3);
\draw [green ](13,9)  to [out =225, in = -225, looseness=1 ]  (13, 4);
\draw (13, 14.5) -- (11,15);
\draw (13, 12.5) -- (15,13);
\draw (13, 9.03) -- (10,10.5);
\draw (13, 2) -- (15,1);
\node at (8,14.5) {$\mathcal R_\mathcal B$};
\node at (8,12.5) {$\mathcal R_\mathcal B$};
\node at (6,10.5) {$\mathcal R_\mathcal B$};
\node at (6,6.5) {$\mathcal R_\mathcal B$};
\node at (8,1.5) {$\mathcal R_\mathcal B$};
\node at (0,7.2) {$r$};
\node at (1.85,7.2) {$v_0$};
\end{tikzpicture}
\caption{Example of a renormalized tree (left) and its respective cluster structure (right). Only the $\mathcal R_\mathcal B$ operators are explicitly written, and the blue squares are the weight functions $w_h(\cdot)$. The union of the green lines on the right represents the {\it spanning tree}, while the four black lines are the four external legs. In Figure (\ref{figure_lemma_integration_prescription_step_two}) we describe the first step of the integration.}
\label{figure_integration_prescription_step_one}
\end{figure}

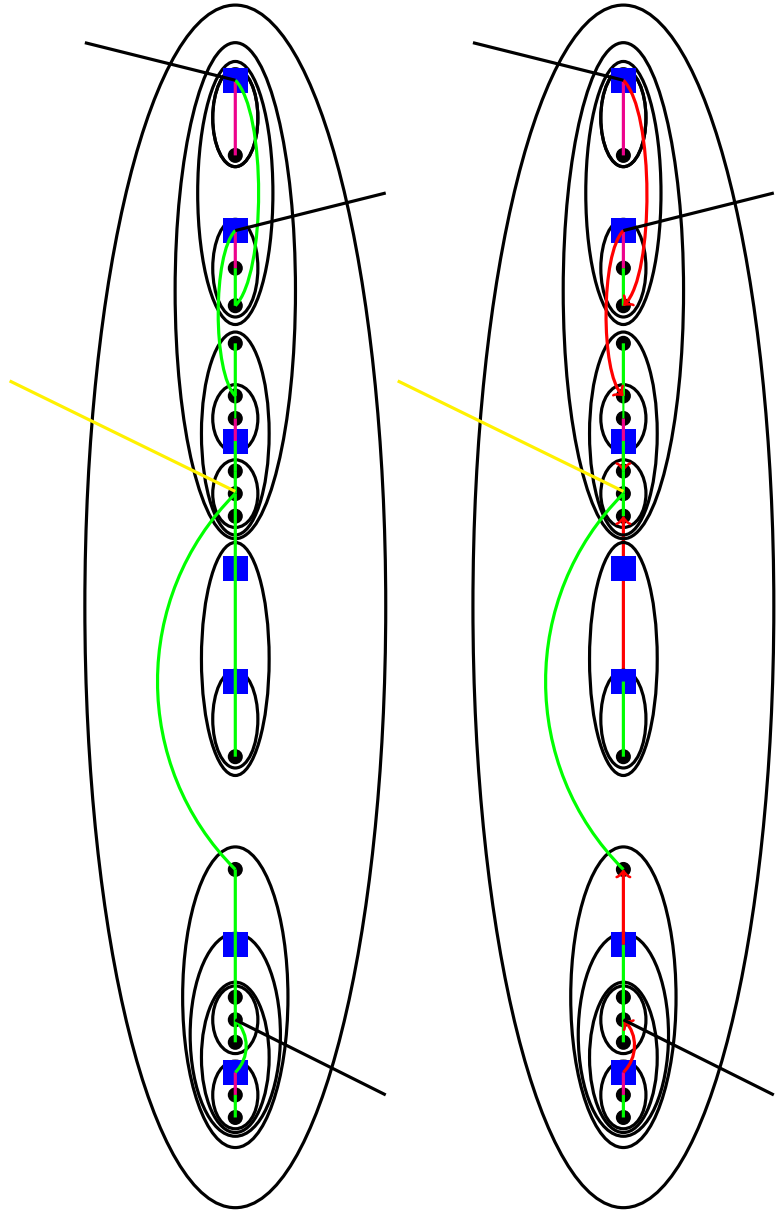
\begin{figure}
\centering
\begin{tikzpicture}
[very thick, scale=1, transform shape]
\foreach \i in {0.7,1,1.3,1.7,2,2.3,3,4,5.5,6.5,8,8.7,9,9.3,9.7,10,10.3,11,11.5,12,12.5,13.5,14.5}
{
\fill (13, \i) circle (0.1);
}
\foreach \i in {1,2,9,10}{
\draw (13,\i) ellipse (0.3 and 0.45);
}
\draw (13,6) ellipse (0.3 and 0.65);
\draw (13,12) ellipse (0.3 and 0.65);
\draw (13,14) ellipse (0.3 and 0.65);
\draw (13,14) ellipse (0.3 and 0.65);
\draw (13,13) ellipse (0.5 and 1.75);
\draw (13,9.8) ellipse (0.45 and 1.35);
\draw (13,6.8) ellipse (0.45 and 1.55);
\draw (13,1.5) ellipse (0.45 and 1);
\draw (13,1.8) ellipse (0.6 and 1.35);
\draw (13,2.3) ellipse (0.7 and 2);
\draw (13, 11.7) ellipse (0.8 and 3.3);
\draw (13, 7.5) ellipse (2 and 8);
\node [regular polygon,blue, regular polygon sides=4,
        minimum size=1mm, fill] at (13, 3) {};
        \node [regular polygon,blue, regular polygon sides=4,
        minimum size=1mm, fill] at (13, 1.3) {};
        \node [regular polygon,blue, regular polygon sides=4,
        minimum size=1mm, fill] at (13, 8) {};
        \node [regular polygon,blue, regular polygon sides=4,
        minimum size=1mm, fill] at (13, 6.5) {};
        \node [regular polygon,blue, regular polygon sides=4,
        minimum size=1mm, fill] at (13, 9.7) {};
        \node [regular polygon,blue, regular polygon sides=4,
        minimum size=1mm, fill] at (13, 12.5) {};
        \node [regular polygon,blue, regular polygon sides=4,
        minimum size=1mm, fill] at (13, 14.5) {};
        \draw [very thick, magenta] (13, 13.5) -- (13,14.5);
        \draw [very thick, green]  (13, 12) -- ++ (0,-0.5);
        \draw [very thick, green]  (13, 9.3) -- ++ (0,-0.3);
\draw [very thick, magenta] (13,12.5) -- ++ (0,-0.5);
\draw [very thick, magenta] (13,10) -- ++ (0,-0.3);
\draw [very thick, magenta] (13,1) -- ++ (0,0.3);
\draw [very thick, green] (13,11) -- ++ (0,-0.7) ++ (0,-0.3) ++ (0,-0.3) -- ++ (0,-0.4) ++ (0,-0.3) -- ++ (0,-0.3) -- ++ (0,-0.7) -- ++ (0,-1.5) -- ++ (0,-1) ++ (0,-1.5) -- ++ (0,-1) -- ++ (0,-0.7) -- ++ (0,-0.3) -- ++ (0,-0.3) ++ (0,-0.4) ++ (0,-0.3) -- ++ (0,-0.3);
\draw [green ](13,1.3)  to [out =45, in = -45 ]  (13, 2);
\draw [green ](13,14.5)  to [out =-45, in = 45, looseness=0.5 ]  (13, 11.5);
\draw [green ](13,12.5)  to [out =225, in = -225, looseness=0.5 ]  (13, 10.3);
\draw [green, thick] (13,10.3) -- ++ (0,-0.3);
\draw [green ](13,9)  to [out =225, in = -225, looseness=1 ]  (13, 4);
\draw (13, 14.5) -- (11,15);
\draw (13, 12.5) -- (15,13);
\draw [yellow] (13, 9.03) -- (10,10.5);
\draw (13, 2) -- (15,1);
\end{tikzpicture}
\centering
\begin{tikzpicture}
[very thick,scale=1, transform shape]
\foreach \i in {0.7,1,1.3,1.7,2,2.3,3,4,5.5,6.5,8,8.7,9,9.3,9.7,10,10.3,11,11.5,12,12.5,13.5,14.5}
{
\fill (13, \i) circle (0.1);
}
\foreach \i in {1,2,9,10}{
\draw (13,\i) ellipse (0.3 and 0.45);
}
\draw [->, red] (13,6.5) -- (13,8);
\draw [->, red] (13,8) -- (13,8.7);
\draw [->, red] (13,9.7) -- (13,9.3);
\draw (13,6) ellipse (0.3 and 0.65);
\draw (13,12) ellipse (0.3 and 0.65);
\draw (13,14) ellipse (0.3 and 0.65);
\draw (13,14) ellipse (0.3 and 0.65);
\draw (13,13) ellipse (0.5 and 1.75);
\draw (13,9.8) ellipse (0.45 and 1.35);
\draw (13,6.8) ellipse (0.45 and 1.55);
\draw (13,1.5) ellipse (0.45 and 1);
\draw (13,1.8) ellipse (0.6 and 1.35);
\draw (13,2.3) ellipse (0.7 and 2);
\draw (13, 11.7) ellipse (0.8 and 3.3);
\draw (13, 7.5) ellipse (2 and 8);
\node [regular polygon,blue, regular polygon sides=4,
        minimum size=1mm, fill] at (13, 3) {};
        \node [regular polygon,blue, regular polygon sides=4,
        minimum size=1mm, fill] at (13, 1.3) {};
        \node [regular polygon,blue, regular polygon sides=4,
        minimum size=1mm, fill] at (13, 8) {};
        \node [regular polygon,blue, regular polygon sides=4,
        minimum size=1mm, fill] at (13, 6.5) {};
        \node [regular polygon,blue, regular polygon sides=4,
        minimum size=1mm, fill] at (13, 9.7) {};
        \node [regular polygon,blue, regular polygon sides=4,
        minimum size=1mm, fill] at (13, 12.5) {};
        \node [regular polygon,blue, regular polygon sides=4,
        minimum size=1mm, fill] at (13, 14.5) {};
        \draw [very thick, magenta] (13, 13.5) -- (13,14.5);
        \draw [very thick, green]  (13, 12) -- ++ (0,-0.5);
        \draw [very thick, green]  (13, 9.3) -- ++ (0,-0.3);
\draw [very thick, magenta] (13,12.5) -- ++ (0,-0.5);
\draw [very thick, magenta] (13,10) -- ++ (0,-0.3);
\draw [very thick, magenta] (13,1) -- ++ (0,0.3);
\draw [very thick, green] (13,11) -- ++ (0,-0.7) ++ (0,-0.3) ++ (0,-0.3) -- ++ (0,-0.4) ++ (0,-0.3) -- ++ (0,-0.3)  ++ (0,-0.7)  ++ (0,-1.5) -- ++ (0,-1) ++ (0,-1.5)  ++ (0,-1) -- ++ (0,-0.7) -- ++ (0,-0.3) -- ++ (0,-0.3) ++ (0,-0.4) ++ (0,-0.3) -- ++ (0,-0.3);
\draw [red, -> ](13,3)  --  (13, 4);
\draw [red, -> ](13,1.3)  to [out =45, in = -45 ]  (13, 2);
\draw [red, -> ](13,14.5)  to [out =-45, in = 45, looseness=0.5 ]  (13, 11.5);
\draw [red, -> ](13,12.5)  to [out =225, in = -225, looseness=0.5 ]  (13, 10.3);
\draw [green, thick] (13,10.3) -- ++ (0,-0.3);
\draw [green ](13,9)  to [out =225, in = -225, looseness=1 ]  (13, 4);
\draw (13, 14.5) -- (11,15);
\draw (13, 12.5) -- (15,13);
\draw [yellow] (13, 9.03) -- (10,10.5);
\draw (13, 2) -- (15,1);
\end{tikzpicture}
\caption{Step one of the integration procedure. The starting point is the cluster structure of figure (\ref{figure_integration_prescription_step_one}). First of all we fixed one of the external legs, the yellow one, as the root. Using Lemma (\ref{lemma_transfer}), we moved the weight functions, {\it i.e.} the blue squares, along the spanning tree (along the lines that were green in Figure (\ref{figure_integration_prescription_step_one}) and that are {\it magenta} in this Figure), in order to use Lemma (\ref{lemma_effective_gain}) getting the gain and going toward the root of the spanning tree. The first step of the integration is done along the red arrows (on the right), "{\it living}" at a higher scale with respect to the blue dots they are attached to.}
\label{figure_lemma_integration_prescription_step_two}
\end{figure}

At this point we can describe the integration procedure to bound:
\begin{equation*}
\begin{split}
\frac{1}{|\Lambda|\beta}\left|\int d\bm x(P_{v_0})\prod_{v\notin V_f(\tau)}\gamma^{h_vq_{\alpha,G^{h_v,T_v}}}\left(\left[\prod_{\ell \in T_v}(\bm x_\ell-\bm y_\ell)^{b(\ell)}_{j(\ell)}\partial^{q(f_\ell^1)}_{j(f_\ell^1)}\partial^{q(f_\ell^2)}_{j(f_\ell^2)} g^{(h_\ell)}_{\ell}\right]\right)\cdot \right.\\\left. \cdot\left[\prod_{i=1}^{n}(\bm x^i-\bm y^i)^{b(v^*_i)}_{j(v^*_i)}K^{(h_i)}_{{v^*_{i}}}(\bm x_{v^*_i}))\right]\left(\prod_{v\in V_\mathcal{B}(\tau)}w_{h_v}(x_v)\right)\right|
\end{split}
\end{equation*}
{\bf Step 1: re-arranging the weight functions} Using  Lemma (\ref{lemma_transfer}), we can move for each $v\in V_\mathcal B(\tau)$, the weight functions $w_{h_v}(\cdot)$ along the spanning tree, re-arranging it in such a way that
\begin{itemize}
\item all the weight functions decorate vertices which are shared by propagators at different scales (in order to apply the result of Lemma (\ref{lemma_integral_w_g})),
\item if we imagine to erase the propagator at higher scale, the weighted vertex is still connected to the root of the rooted spanning tree by a subtree including the propagator at lower scale (of course in this procedure is still arbitrary, because in general there is not only a vertex of this kind, and in particular there could be vertices having the properties we are requiring, and being attached to propagators at different scales: we comment it in the remark at the end of this proof).
\end{itemize}
{\bf Step 2: integration order}
\begin{itemize}
\item we can arbitrairly choose an external line of the cluster $G_{v_0}$, and consider it as the {\it root} of the tree giving a natural order relation to the tree;
\item moreover, we use the {\it quantifier "for each $\theta\in (0,1)$"} appearing in the hypothesys of the weight functions $\varpi_{h}(\cdot)$: considering the {\it rooted spanning tree} rooted in the selected external leg we just defined, we pick the {\it closest vertex}, in the sense of the tree distance, to the root (if there are more then one at the same distance, we can arbitrarly choose one of them).
\begin{itemize}
\item if this vertex is associated with $w_{h^*}(\cdot)=\rho_{h^*}(\cdot)$, for some $h^*>h$, we do nothing,
\item if this vertex is associated with $w_{h^*}(\cdot)=\varpi_{h^*}(\cdot)$, for some $h^*>h$, we {\it fix once for all}, for all the other weight functions associated with the vertices of the spanning tree, $\theta\equiv \theta'$ in the decay hypothesis, except for the selected vertex, for which we keep the quantifier {\it "for each $\theta\in (0,1)$"}.
\end{itemize}
\end{itemize}
After these manipulations, the integral we are bounding has to be read as follows:
\begin{itemize}
\item all the propagators have to be read, with the purpos of an upper bound, as $P$-type propagators,
\item each of the  weight functions $w_{h_v}(\cdot)$, $v\in V_\mathcal B(\tau)$, is associated with some special vertex of the spanning tree: {\it i.e.} it is associated with the exiting point of some propagator $g^{(h_\ell)}_\ell$ at scale $h_\ell<h_v$, so that in the integration procedure we can think of all the weight functions as associated to some propagator: 
\end{itemize}
\begin{equation}
|w_{h'} \ast g_\ell|
\end{equation}
where $h'>h_\ell$. In particular, $h'$ is the scale of the cluster we want to {\it renormalize}, while $h_\ell$ is the scale of the line $\ell$ exiting this cluster.\\

{\bf Iterative integration} At this point, starting from the leaves of the rooted spanning tree we start the integration procedure, observing the following prescriptions:
\begin{itemize}
\item if none of the vertices linked by $\ell\in T_v$ is {\it weighted}, so we get the same bound as {\it usual} from the integration of the single line of the spanning tree,
\item if one of the two vertices linked by $\ell\in T_v$ is weighted, we use Lemma (\ref{lemma_integral_w_g}) to integrate $\int d\bm x_\ell |w_{h'}\ast g_\ell|$, both transfering the weight function to the lower scale $h_\ell$ and getting the {\it scale jump} $\gamma^{\theta(h_\ell-h')}$. 
\begin{rem}
Let us stress that this procedure (see Figures (\ref{figure_integration_prescription_step_one}) and (\ref{figure_lemma_integration_prescription_step_two})) is constructed in such a way that, in using Lemma (\ref{lemma_transfer}), we never transfer the weight function along a line of the spanning tree that we already used to transfer weight functions. This fact ensures that the constants $C^{n_v^0-1}$ appearing inthe bound in Lemma (\ref{lemma_transfer}) do not accumulate at all.
\end{rem}
{\bf Observation:}
\begin{itemize}
\item the scale jump $\gamma^{\theta(h_\ell-h')}$ is exactly the gain factor associated to the operators $\mathcal R_{\mathcal B}$ acting on the cluster at scale $h'$,
\item we transferred the weight function to scale $h_\ell$, and so we can use again Lemma (\ref{lemma_transfer}) to move the weight function following the rules described in Step 1. We stress that it could of course be the case that the propagator at scale $h_\ell$ we use to {\it transfer the weight function} links two vertices at scales $(h',h^*)$, so that the weight function $w_{h_\ell}(\cdot)$ would be attached to some vertex belonging to some subcluster at scale $h^*>h_{\ell}$, but it is not an issue: of course we consider this vertex as a vertex of the cluster at scale $h_\ell$, and we can still use Lemma (\ref{lemma_transfer}), as it is clear in the proof of Lemma.
\end{itemize}
\item in this way we are left with computing an integral over a {\it pruned spanning tree} having the same formal structure of the one we starded from, so we can iterate the procedure with this further prescription: since it could happen that at some scale $\bar h$ we get several weight functions $w_h(\cdot)$, so
\begin{itemize}
\item if one of these weight function is the one associated with the {\it special vertex}, we bound all the others by $C_{\theta'}$, keeping the weight function associated with the quantifier "{\it for each $\theta\in (0,1)$}.
\item otherwise, we arbitrarily bound all of them but one by $C_{\theta'}$, and we use Lemma (\ref{lemma_transfer}) to move the weight function following the rule of Step 1.
\end{itemize}
\item the very last non trivial step consists of integrating the cluster $G_v$ containing the weight function $\varpi_{h_v}$ we used in order to preserve the {\it quantifier "for each $\theta\in (0,1)$"}: we perform this integral as follows, getting a {\it gain factor} $\gamma^{h_L-h_{v}}$.\\
We use the coordinate the weight finction $w_h(\cdot)$ is associated with as the root of a change of variables by which we associate an effective variable to each $\ell\in T_v$. We use all the propagators $g_\ell$, $\ell\in T_v$ to integrate these effective variables, {\it i.e.} we bound each line of the spanning tree by $||g_\ell^{h_\ell}||_\infty$, and we use the weight function to integrate the overleft variable, getting:
\begin{equation}
\frac{1}{L} \int_0^L dx \left| w_{h}(x)\right| \leq c_{|\lambda|,w}  \gamma^{\theta\left(h_L-h\right)}
\end{equation}
where $h_L$ has already been defined after (\ref{1_norm_off_diagonal_DBC}), and
\begin{equation}
c_{|\lambda|,w}= \begin{cases}
|\lambda| C_\theta, &\mbox{ if } w_h(\cdot)=\varpi_h(\cdot),\\
c, &\mbox{ if } w_h(\cdot)=\rho_h(\cdot).
\end{cases}
\end{equation}
\begin{proof}
From the hypothesis it follows that, when $w_h(\cdot)=\varpi_h(\cdot)$:
\begin{equation}
\begin{split}
\frac{1}{L}\int_{\Lambda} dx \left|\varpi_{h}(x)\right| \leq \frac{1}{L} \int_\Lambda dx \frac{c |\lambda|}{1+\gamma^{\theta h}|x|^{\theta}}\leq C_\theta |\lambda| \gamma^{-\theta h}\gamma^{h_L\theta}= C_\theta |\lambda| \gamma^{\theta(h_L-h)},
\end{split}
\end{equation}
while, if $w_h(\cdot)=\rho_h(\cdot)$, for each $N=1,2,\dots$:
\begin{equation}
\begin{split}
\frac{1}{L}\int_{\Lambda} dx \rho_{h}(x)= \frac{1}{L} \int_\Lambda dx \frac{C_N}{1+\gamma^{N h}|x|^{N}}\leq C'_N \frac{1}{L}\gamma^{-h}\leq C \gamma^{h_L-h}.
\end{split}
\end{equation}
\end{proof}

 Of course, we can rewrite this gain factor as
\begin{equation}
\label{scale_jump}
\gamma^{\theta(h_L-h_v)}=\gamma^{\theta(h_L-h)}\left[\gamma^{\theta(h-h_1)}\gamma^{\theta(h_1-h_2)}\cdots \gamma^{\theta(h_m-h_v)} \right]
\end{equation}
where $h<h_1<\dots <h_m< h_v$, and the factor included in square brackets renormalizes all the inclapsulated subclusters of $G_1\supset G_2\supset\dots\supset G_m\supset G_v$, while the overleft scale jump $\gamma^{\theta(h_L-h)}$, that at the moment is not crucial in order to prove this theorem, will become crucial in proving the {\it main theorem}.
\end{itemize}
\begin{rem}
First of all, it is worth noting that, following this procedure, having a {\it non translation invariant element} at scale $h_v$ corresponding to the vertex $v\in V(\tau)$ is enough to get a dimensional gain for all the ancestors of $v$: this fact justifies the "simplification" procedure we introduced over the hierarchical organization of the renormalization operators.\\
We already pointed out that there could be several vertices we can possibly choose to glue the weight function to, attached to different propagators living at different scales, so one can think that the more convenient thing to do is to choose the propagator, among those, living at the lowest scale in order to maximize the scale jump in the gain of Lemma (\ref{lemma_integral_w_g}). In fact all the possible choises are completely equivalent for our pourposes:
\begin{itemize}
\item on the one hand, if we choose the propagator living at the lowest scale, we can rewrite the gain we get
$$\gamma^{\theta (h^{\min}_\ell-h')}=\gamma^{\theta (h^{\min}_\ell-h_1)}\gamma^{\theta (h_1-h_2)}\dots \gamma^{\theta (h_n-h')}, \mbox{ where } h_\ell^{\min}<h_1<h_2<\dots<h',$$
are the scales of all the subclusters contained in the cluster at scale $h_\ell^{\min}$, in order to get the {\it right} gain related to the renormalization operators $\mathcal R_{\mathcal B}$,
\item on the other hand, thanks to the {\it transfer of the weight function} to lower scales (Lemma (\ref{lemma_integral_w_g})), we can get the same factor step by step (or splitting the scale jump as we want).
\end{itemize}
\end{rem}

\end{proof}

\begin{corollary}
\label{theorem_bounds_kernels}
Let $\tau\in\mathcal T_{h,n}$ a renormalized tree, $h>h_L$, and $\mathcal W^{(h)}(\tau, P_{v_0}, \bm x(P_{v_0}))$ the respective kernel. If, for some constant $c_1>0$ and if for any $\theta\in (0,1)$ there exists a constant $C_\theta$ such that these bounds are verified
\begin{equation}
\begin{split}
\int_{y\in\Lambda}dy \left| \varpi_{h'}(x,y)\right|\leq |\lambda| \frac{C_\theta}{(1+\gamma^{h'}|x|)^\theta}, \\
\sup_{h'>h}\left(\max\{|\nu_{h'}|,|\delta_{h'}|, |\lambda_{h'}|, |z_{h'}|\}\right)\equiv \epsilon_h, \hspace{3mm} \sup_{h'>h}\left| \frac{Z_{h'}}{Z_{h'-1}} \right|\leq e^{c_1\epsilon_h^2}
\end{split}
\end{equation}
and if there exists a constant $\bar \epsilon$, depending on $c_1$, such that $\epsilon_h\leq \bar \epsilon$, then, for another suitable constant $c_0$ uniform in $c_1$, $L$ and $\beta$ the following bounds are true
\begin{equation}
\sum_{\tau\in\mathcal T_{h,n}}\left[|n_h(\tau)|+|z_h(\tau)|+|a_h(\tau)|+|l_h(\tau)|+||\varpi_h(\tau)||_{\infty,1}\right]\leq \left(c_0\epsilon_h\right)^n,
\end{equation}
\begin{equation}
\sum_{\tau\in\mathcal T_{h,n}}|e_{h+1}|\leq \gamma^{2h}(c_0\epsilon_h)^n
\end{equation}
and
\begin{equation}
\begin{split}
\frac{1}{|\Lambda|\beta}\sum_{\tau\in\mathcal T_{hn}}\int d\bm x(P_{v_0})\left|\mathcal R\mathcal W^{(h)}(\tau,P_{v_0},\bm x(P_{v_0}))\right|\leq \gamma^{-h\left(D(P_{v_0})+z_{v_0}\right)}(c_0\epsilon_h)^n
\end{split}
\end{equation}
where for each $\theta \in (0,1)$,
\begin{equation*}
z_{v_0}=\begin{cases}
1+\theta, &\mbox{ if $G_V$ has two external lines},\\
\theta, &\mbox{ if $G_V$ has four external lines}.
\end{cases}
\end{equation*}
\end{corollary}

\begin{proof}
Exploiting the dimensional gains coming from the operator $\mathcal R$ acting as described in equation (\ref{renormalized_kernels_explicit_expression}), we can repeat the proof of Theorem (\ref{theorem_bound_of_kernels}) by replacing
\begin{equation}
\prod_{v\notin V_f(\tau)}\gamma^{-D(v)(h_v-h_{v'})}\rightarrow \prod_{v\notin V_f(\tau)}\left(\frac{Z_{h_v}}{Z_{h_v-1}}\right)^{|P_v|/2}\gamma^{-[D(v)+z_v](h_v-h_{v'})}
\end{equation}
By the assumption $\sup_{h'>h}Z_{h'}/Z_{h'-1}\leq e^{c_1\epsilon_h^2}\leq$, taking $c_z\epsilon_h^2\leq 1/16$, one gets that
\begin{equation}
\prod_{v\notin V_f(\tau)}(Z_{h_v}/Z_{h_v-1})^{|P_v|/2}\gamma^{-[-2+|P_v|/2+z_v]}\leq \left(\prod_{\bar v }\gamma^{-\frac{1}{40}(h_{\bar v}-h_{\bar v'})}\right)\left(\prod_{v\notin V_f(\tau)}\gamma^{-|P_v|/40}\right)
\label{bound_product_z_h/z_h-1_gamma_DBC}
\end{equation}
where $\bar v$ are the non-trivial vercies, and $\bar v'$ is the non tricial vertex immediately preceding $\bar v$. Thanks to the product into the first bracket we bound the sum over the scale labels by $(const.)^n$. The second factor can be used to bound the sums, using
\begin{equation}
\sum_{\tau\in\mathcal T_{h,n}}\sum_{P_v}\sum_T\prod_{v\notin V_f(\tau)}\frac{1}{s_v!}\gamma^{-|P_v|/40}\leq C^n,
\end{equation}
we refer to \cite{benfatto2001renormalization} for details.
\end{proof}

\subsection{Proof of the main theorem}
Let us recall the main result:
\begin{thm}
\label{theorem_main_theorem_introduction}
There exists a radius $\lambda_0>0$ such that, for any $|\lambda|\leq \lambda_0$ it is possible to fix the {\it boundary defect} $\pi(x,y)$ and its strenght $\varpi=\varpi(\lambda)$ in such a way that, for any $\theta\in (0,1)$, there exists a constant $C_\theta$ such that 
\begin{equation}
\sum_{y\in\Lambda} \left|\pi(x,y)\right| \leq C_\theta \left(\frac{1}{\left(1+|x|\right)^\theta}+\frac{1}{\left(1+|L-x|\right)^\theta}\right),
\end{equation}
in such a way that $f_\Lambda$ admits a convergent expansion in $\lambda$ and $\varpi$.\\
Moreover
\begin{equation}
\left| f_\Lambda-f_\infty \right|\leq |\lambda|\frac{C_\theta}{L^\theta}.
\end{equation}
\end{thm}
First of all, let us recall that the {\it diagrams} that contribute to the {\it specific free energy} are the so called {\it vacuum diagrams}, {\it i.e.} the diagrams such that $\left| P_{v_0}\right|=0$.\\
As we have done in the case of the kernels $W^{(h)}_2$ and $W_4^{(h)}$, we can split the {\it free fermi energy} into the {\it bulk term} and a {\it remainder}:
\begin{equation}
f_{\Lambda,\beta}=f^{(P)}_{\Lambda,\beta}+f^{(R)}_{\Lambda,\beta},
\end{equation}
where, by construction, all the diagrams contributing to $f^{(R)}_{\Lambda,\beta}$ contain at least either a remaider propagator or a non local endpoint.\\
In order to explicitly control the boundary corrections, we define
\begin{equation}
f_{\Lambda}=\lim_{\beta\nearrow \infty}f_{\Lambda,\beta},\hspace{3mm} f_{\infty}=\lim_{|\Lambda|\nearrow \infty} f_{\Lambda}.
\end{equation}
and we study the difference:
\begin{equation}
|f_\infty-f_{\Lambda}|,
\end{equation}
knowing that $|f_\infty-f^{(P)}_{\Lambda}|\leq \frac{C}{L^2}$, which can be proved by proceeding as in (\cite{giuliani2013universal}).\\
Using {\it exactly the same technique} as in proving the Theorem (\ref{theorem_bounds_kernels}) with the constraints $\left| P_{v_0}\right|=0$, and keeping track of the {\it scale jump} $\gamma^{\theta(h_L-h)}$ (\ref{scale_jump}) we already commented  in the proof of Theorem (\ref{theorem_renormalized_bounds_DBC}) we get, for each $\theta\in (0,1)$,
\begin{equation}
|f_\infty-f_{\Lambda}|\leq |\lambda|c_\theta\sum_{h\leq 1} \gamma^{2h} \gamma^{\theta(h_L-h)}\leq |\lambda|\frac{C_\theta}{L^{\theta}}.
\end{equation}
The boundary defect $\pi(x,y)$ and its strenght $\varpi$ will be fixed in the next subsection (\ref{subsection_flow_rcc_DBC}).

\subsection{Flow of running coupling constants and functions}
\label{subsection_flow_rcc_DBC}
From the iterative procedure we set up,  we can write the flow equations for $\vec v_h(x)$ for the quantities we defined in (\ref{running_coupling_functions_DBC}):
\begin{equation}
\label{flows_running_coupling_DBC}
\begin{split}
\nu_{h-1}&=\gamma \nu_h+\beta_\nu^h(\vec v_h(x,y),\dots,\vec v_0(x,y);x,y),\\
\lambda_{h-1}&=\lambda_h+\beta_\lambda^h(\vec v_h(x,y),\dots,\vec v_0(x,y);x,y),\\
\delta_{h-1}&=\delta_h+\beta_\delta^h(\vec v(x,y),\dots,\vec v_0(x,y);x,y),\\
\frac{Z_{h-1}}{Z_h}&=1+\beta^h_z(\vec v_h(x,y),\dots, \vec v_0(x,y);x,y),\\
\varpi_{h-1}(x,y)&=\gamma \varpi_{h}(x,y)+\beta^h_\varpi(\vec v_h(x,y),\dots, \vec v_0(x,y);x,y).
\end{split}
\end{equation}
The convergence of the multiscale expansion has been proved under the hypotesys of that the running coupling constants and functions are small enough. Now, we have to show that, choosing $\lambda$ small enough and fixing once for all the counterterm $\nu$ (which is an analytic function of $\lambda$ as we have already seen in subsection (\ref{subsection_flow_of_running_coupling_constants_PBC})) and $\varpi  N(x,y)$ as functions of $\lambda$, such hypotesis are verified.\\
The strategy is to write down the Taylor expansion for the beta function (convergent as long as the hypotesis are fulfilled), truncate this Taylor expansion at lowest non-trivial order, check whether the {\it approximate flow} still verifies the hypotesis, and finally prove that the solution of this approximate flow is stable under the addition of higher order Taylor approximation.\\
The idea is that the beta function of this model is asympyoyically close to the beta function of the Luttinger model with an ultraviolet cut-off, so it belongs to the Luttinger liquid universality class (that we introduced in the Introduction). The main difference, as we already mentioned in the introduction, is that the reference model shows more symmetries than the models of the universality class, that can be used to show that the beta function $\beta_\lambda^{(h)}$, in the reference model, is asymptotically zero. Thanks to the {\it asymptotic closeness} of the models, the same holds for the model we are studying.\\
It is worth stressing that, by the very definition of running coupling constants and functions, we can rewrite (\ref{flows_running_coupling_DBC}) as
\begin{equation}
\begin{split}
\nu_{h-1}&=\gamma \nu_h+\beta_\nu^h(\vec v_h,\dots,\vec v_0),\\
\lambda_{h-1}&=\lambda_h+\beta_\lambda^h(\vec v_h,\dots,\vec v_0),\\
\delta_{h-1}&=\delta_h+\beta_\delta^h(\vec v,\dots,\vec v_0),\\
\frac{Z_{h-1}}{Z_h}&=1+\beta^h_z(\vec v_h,\dots, \vec v_0),\\
\varpi_{h-1}(x,y)&=\gamma \varpi_{h}(x,y)+\beta^h_\varpi(\vec v_h(x,y),\dots, \vec v_0(x,y);x,y),
\end{split}
\end{equation}
that basically means that, while the bulk constants enter in the flow equation of $\varpi_h(x,y)$, $\varpi_h(x,y)$ does not enter in the flow equations of the bulk running coupling constants. As a consequence, we can assume to have already studied the bulk flow equations of the running coupling constants $(\nu_h,\lambda_h,\delta_h,Z_h)$, and study the flow equation of the running coupling function.

\paragraph{Fixing the non local counterterm}

Let us study the flow of $\varpi_{h}(x,y)$, that has already been defined as
\begin{equation}
\gamma^h\varpi_h(x,y):=\frac{Z_{h-1}}{Z_h}\int_0^\beta dy_0 \mathcal W^{(h)}(\bm x,\bm y)
\end{equation}
and let us recall the flow
\begin{equation}
\varpi_{h-1}(x,y)=\gamma \varpi_{h}(x,y)+\beta^h\varpi(\vec v_h(x,y),\dots, \vec v_0(x,y);x,y).
\label{flow_varpi(x,y)}
\end{equation}
So, from the very last of (\ref{flows_running_coupling_DBC}), we get
\begin{equation}
\varpi_1(x,y)= -\sum_{k=h}^1\gamma^{k-2}\beta^{(k)}_\varpi(\vec v_h(x,y),\dots, \vec v_0(x,y);x,y),
\end{equation}
so
\begin{equation}
\varpi_h(x,y)=-\sum_{k\leq h}\gamma^{k-h-1}\beta^{(k)}_\varpi(\vec v_h(x,y),\dots, \vec v_0(x,y);x,y).
\label{varpi_as_sum_over_betas}
\end{equation}
Let us recall the definition of the norm:
\begin{equation}
||\varpi_h||^{(\theta)}_{\infty,1}=\sup_{x\in \Lambda}\left(1+\gamma^{ h}|x|\right)^\theta\sum_{y\in\Lambda}|\varpi_h(x,y)|,
\end{equation}
allowing us to define the Banach space $\mathcal B$  as follows.
\begin{defn}
Let $\mathcal B$ be the set of the real sequences $\underline \varpi (x,y):=\left\{\varpi_h(x,y)\right\}_{h\leq 1}$ with norm $$|| \underline \varpi||^{(\theta)}:=\sup_{h\leq 1}||\varpi_h||^{(\theta)}_{\infty,1}.$$
Besides, let us define the closed ball  $\mathcal M_{\bar \theta}=:\mathcal M\subset \mathcal B$: let us fix $\bar \theta$, and let us define
\begin{equation}
\label{closed_ball_banach_space}
\mathcal M:=\left\{\underline \varpi(x,y): \forall\hspace{1mm} \theta\leq \bar \theta, \hspace{1mm} ||\underline \varpi^{(h)} ||^{(\theta)}\leq |\lambda| C\right\}.
\end{equation}
\end{defn}
\begin{rem}
We define such a Banach space because, to fix the initial value of 
\begin{equation}
 \pi(x,y)=:\varpi_1(x,y),
\end{equation}
we will look for a fixed point of the flow equation (\ref{flow_varpi(x,y)}) using the Banach fixed point theorem; in particular, we are interested in the elements  belonging to the closed ball $\mathcal M$, and the closeness guarantees that strarting from an initial datum inside the ball $\mathcal M$, the fixed point belongs to $\mathcal M$.
\end{rem}
Let us start with defining an operator $\bm T$ acting on $\mathcal M$ as
\begin{equation}
\left(\bm T \underline \varpi(x,y)\right)_h=-\sum_{k\leq h}\gamma^{k-h-1}\beta^{(k)}_\varpi(\vec v_h(\underline \varpi;x,y), \dots, \vec v_0(\underline \varpi;x,y);x,y).
\label{flow_varpi_as_operator}
\end{equation}
{\bf Claim} If we find a fixed point $\underline \varpi^*(x,y)$ of (\ref{flow_varpi_as_operator}), the solution will be such that $\varpi_h(x,y)$ is {\it small} as desired.\\ \noindent
In order to find the fixed point for the operator $\bm T$:
\begin{enumerate}
\item we have to check that it leaves $\mathcal M$ invariant, {\it i.e.} that \begin{equation}
T:\mathcal M\to \mathcal M,
\end{equation}
\item we have to check that $\bm T$ is a contraction in $\mathcal M$, {\it i.e.}
\begin{equation}
||\bm T \underline\varpi-\bm T\underline\varpi'||^{(\theta)}\leq ||\underline \varpi - \underline \varpi'||^{(\theta)}.
\end{equation}
\end{enumerate}
Let us prove the Claim.
\begin{proof} [Proof of Claim]
\begin{enumerate}
\item  Let us prove that $\bm T:\mathcal M\to \mathcal M$. Let us recall that, by definition, $$\beta_{\varpi}^{(h)}(\vec v_h(x,y),\dots,\vec v_0(x,y);x;y)$$ is the sum of all possible Feynman graphs whose internal lines live at scale $\geq h$, such that there must be at least one line living right at scale $h$ and at least one element breaking the translation invariance, {\it i.e.} either $\gamma^k\varpi_k(x,y)\delta_{x_0,y_0}$ or $g^{(k)}_{R,\omega}(\bm x,\bm y)$, $k\geq h$. \\
So, first of all, let us check that, for each $\theta \leq \bar theta$
\begin{equation}
\sup_{x\in\Lambda}\left(1+\gamma^{\theta h}|x|^\theta\right)\sum_{y\in\Lambda}\left|\beta_{\varpi}^{(h)}(\vec v_h(x,y),\dots,\vec v_0(x,y);x,y)\right|\leq C |\lambda|,
\label{norm_infty1_beta_functions}
\end{equation}
As we already explained, we can assume without loss of generality that:
\begin{itemize}
\item all the remainder propagators belong to the spanning tree,
\item we transfered the {\it anchorage property} to the vertices, localizing the non-local counterterms and rewriting $g_{R,\omega}^{(h)}\to g_{P,\omega}^{h}\ast w_h$, and we call $n_w\geq 1$ the number of {\it weighted vertices}.
\end{itemize}
We keep track of only one {\it weight function}, bounding the contribution of $n_w-1$ non-local endpoints using the hypothesis: for each $\theta'\leq \bar \theta$ (resp. $N=1,2,\dots$) there exist a constant $C_{\theta'}$ (resp. $C_N$) such that 
\begin{equation}
 |w_h(x)|\leq \begin{cases} |\lambda|\frac{C_{\theta'}}{\left(1+\gamma^h|x|\right)^{\theta'}}\leq |\lambda| C_{\theta'}, \hspace{3mm} &\mbox{ if } w_h(\cdot) =\varpi_h(\cdot),\\
 \frac{C_N}{(1+\gamma^h|x|)^N}\leq C_N, & \mbox{ if }w_h(\cdot)=\rho_h(\cdot).
 \end{cases}
\end{equation}
We stress that the constant $|\lambda|$, if there are not non-local endpoints but only remainder propagators, arises from the fact that there must be at least a {\it four-external legs} endpoint to contribute to $\beta_\varpi^{(h)}$. In light of this fact, from now on we assume without loss of generality us assume that $w_k(\cdot)=\varpi_h(\cdot)$, 
So by construction there is a vertex of the spanning tree associated with the {\it weight function} $\varpi_h(\cdot)$. As in proof of Theorem (\ref{theorem_renormalized_bounds}), the determinant expansion is the same of the {\it translation invariant case}, while the novelty is the weight function $\varpi_h(\cdot)$ appearing in the integration over the spanning tree: so we are interested in bounding, for each $\theta\leq \bar \theta$:
\begin{equation*}
\begin{split}
\sup_{x\in \Lambda}(1+\gamma^h|x|)^\theta \left| \int d(\underline{ \bm z}\setminus \bm x)\prod_{v\notin V_f(\tau)}\gamma^{h_vq_{\alpha,G^{h_v,T_v}}}\left(\left[\prod_{\ell \in T_v}(\bm x_\ell-\bm y_\ell)^{b(\ell)}_{j(\ell)}\partial^{q(f_\ell^1)}_{j(f_\ell^1)}\partial^{q(f_\ell^2)}_{j(f_\ell^2)} g^{(h_\ell)}_{\ell}\right]\right)\cdot \right.\\\left. \cdot\left[\prod_{i=1}^{n}(\bm x^i-\bm y^i)^{b(v^*_i)}_{j(v^*_i)}K^{(h_i)}_{{v^*_{i}}}(\bm x_{v^*_i}))\right]\varpi_k(s)\right|\leq\\
\leq \sup_{x\in \Lambda}(1+\gamma^h|x|)^\theta \left| \int d(\underline{ \bm z}\setminus \bm x)\prod_{v\notin V_f(\tau)}\gamma^{h_vq_{\alpha,G^{h_v,T_v}}}\left(\left[\prod_{\ell \in T_v}(\bm x_\ell-\bm y_\ell)^{b(\ell)}_{j(\ell)}\partial^{q(f_\ell^1)}_{j(f_\ell^1)}\partial^{q(f_\ell^2)}_{j(f_\ell^2)} g^{(h_\ell)}_{\ell}\right]\right)\cdot \right.\\\left. \cdot\left[\prod_{i=1}^{n}(\bm x^i-\bm y^i)^{b(v^*_i)}_{j(v^*_i)}K^{(h_i)}_{{v^*_{i}}}(\bm x_{v^*_i}))\right] \frac{|\lambda|C_{\theta'}}{(1+\gamma^k|s|)^{\theta'}}\right|\leq \\
\leq  \sup_{x\in \Lambda} \left| \int d(\underline{ \bm z}\setminus \bm x)\prod_{v\notin V_f(\tau)}\gamma^{h_vq_{\alpha,G^{h_v,T_v}}}\left(\left[\prod_{\ell \in T_v}(\bm x_\ell-\bm y_\ell)^{b(\ell)}_{j(\ell)}\partial^{q(f_\ell^1)}_{j(f_\ell^1)}\partial^{q(f_\ell^2)}_{j(f_\ell^2)} g^{(h_\ell)}_{\ell}\right]\right)\cdot \right.\\\left. \cdot\left[\prod_{i=1}^{n}(\bm x^i-\bm y^i)^{b(v^*_i)}_{j(v^*_i)}K^{(h_i)}_{{v^*_{i}}}(\bm x_{v^*_i}))\right] \frac{|\lambda|C_{\theta'}(1+\gamma^h|x|)^\theta}{(1+\gamma^k|s|)^{\theta'}}\right|\leq
\end{split}
\end{equation*}
The strategy is the same we used in proving the Lemma (\ref{lemma_transfer}): by definition of spanning tree $T$, there must be a connected path of lines $\ell \in T$ connecting $x$ (the non-integrated point, so the anchored point) to one of the points $s$: $\{(x,z_1);(z_1,z_2),\dots,(z_{i-1},s)\}\subseteq T$. So, setting $z_0:=x$ and $z_i=s$,
\begin{equation}
|x|^\theta\leq c_{0,\theta}\left(\sum_{j=1}^i |z_{i-1}-z_i|^\theta+|s|^\theta\right),
\end{equation}
so that 
\begin{equation}
\left(1+\gamma^h|x|\right)^\theta\leq c_{1,\theta}\left[ \left(1+\gamma^h|s|\right)^\theta+\sum_{j=0}^i\gamma^{\theta h} |z_{i-1}-z_i|^\theta\right]
\end{equation}
Then each term of the sum $|z_i-z_j|^\eta$ is associated to a line of the spanning tree, so to a propagator $g^{(k)}_{P,\omega}(\bm z_i-\bm z_j)$ so that, analogously to what we did in proving Lemma (\ref{lemma_transfer}), in order to get the bounds we want we can replace, inside the integrals,  $|z_i-z_j|^\theta\leq \gamma^{-\theta k}$ so $$\gamma^{\theta h}|z_i-z_j| \leq \gamma^{\theta (h-k)}\leq 1,$$ being, by construction, $h\leq k$.\\
 So we are left with proving that
 $$\frac{\left(1+\gamma^h|s|\right)^\theta}{(1+\gamma^k|s|)^{\theta'}}\leq c, \hspace{3mm} \mbox{ for any } \theta,\theta'\leq \bar \theta,$$
 and indeed it is:
 \begin{equation}
 \frac{\left(1+\gamma^h|s|\right)^\theta}{(1+\gamma^k|s|)^{\theta'}}\leq \begin{cases}
 \left(1+\gamma^{h-k}\right)^\theta\leq 2,& \mbox{ if } |s|\leq \gamma^{-k},\\
\frac{\left(1+\gamma^h\gamma^{-h}\right)^\theta}{(1+\gamma^k{\gamma^{-k}})^{\theta'}}\leq 1, &\mbox{ if } \gamma^{-k}\leq |s|\leq \gamma^{-h},\\
c \frac{\gamma^{\theta h}|s|^\theta}{\gamma^{\theta' k}|s|^{\theta'}}\leq c \gamma^{\theta'(h-k)}\leq c, &\mbox{ if } \gamma^{-h}\leq |s|.
 \end{cases}
 \end{equation}

Now, we use this bound to verify that $||\left(\bm T \underline \varpi(x,y)\right)_h||_{\infty,1}\leq |\lambda| c$. Indeed
\begin{equation}
\begin{split}
\sup_{x\in\Lambda}\left(1+\gamma^{\theta h}|x|^{\theta}\right)\sum_{k\leq h}\gamma^{k-h-1}\left(\sum_{y\in\Lambda}\beta_{\varpi}^{(k)}(\vec v_h(x,y),\dots,\vec v_0(x,y);x,y)\right)\leq \\
\leq \sum_{k\leq h}\sup_{x\in\Lambda}\left(\gamma^{k-h-1}+\gamma^{k-h-1}\gamma^{\theta h}|x|^{\theta}\right)\left(\sum_{y\in\Lambda}\beta_{\varpi}^{(k)}(\vec v_h(x,y),\dots,\vec v_0(x,y);x,y)\right)\leq \\
\leq  \sum_{k\leq h}\Biggl (C_{\theta'} c_{2,\theta}|\lambda|\gamma^{k-h-1}+\\+\left.\gamma^{(k-h)(1-\theta)}\sup_{x\in\Lambda}\left(\sum_{y\in\Lambda}\gamma^{\theta k}|x|^{\theta}\beta_{\varpi}^{(k)}(\vec v_h(x,y),\dots,\vec v_0(x,y);x,y)\right)\right)\leq c_{3,\theta}C_{\theta'}|\lambda|,
\end{split}
\end{equation}
so $\left(\bm T \underline \varpi(x,y)\right)\in \mathcal M$ if we choose $C= c_{3,\theta}C_{\theta'}$.
\item Let us check that $\bm T$ is a contraction, {\it i.e.} $||\bm T\underline \varpi- \bm T\underline \varpi'||\leq ||\underline \varpi-\underline \varpi'||$. First of all, let us remark that, by the very definition of the running coupling functions (\ref{vec_v_h(x)}) $$\vec v_h(x,y)=\left(\nu_h,\delta_h,\lambda_h,\varpi_h(x,y)\right),$$
the running coupling constants $(\nu_h,\delta_h,\lambda_h)$ do not depend on the running coupling functions $\varpi_{h}(x)$. So we can split
\begin{equation}
\beta_{\varpi}^{(k)}(\vec v_h,\dots,\vec v_0;x,y)=\beta_{\varpi=0}^{(k)}(\vec v_h,\dots,\vec v_0;x,y)+\bar\beta_{\varpi}^{(k)}(\vec v_h(x,y),\dots,\vec v_0(x,y);x,y)
\end{equation}
where $\bar \beta_{\varpi}^{(k)}$ corresponds, in the Feynman graphs picture, to the contribution of all the diagrams containing at least a $\varpi_{h}$ term, and $\beta_{\varpi=0}^{(k)}$ is the remainder. So, if we use the notation
$$\vec v_h(x,y)=\left(\nu_h,\delta_h,\lambda_h,\varpi_h(x,y)\right), \hspace{5mm} \vec v'_h(x,y)=\left(\nu_h,\delta_h,\lambda_h,\varpi'_h(x,y)\right),$$
the difference between the beta functions depending on different {\it running coupling functions} depends only on the diagrams containing at least a $\varpi_{h}$ term:
\begin{equation}
\begin{split}
\beta_{\varpi}^{(k)}(\vec v_h(x,y),\dots,\vec v_0(x,y);x,y)-\beta_{\varpi}^{(k)}(\vec v'_h(x,y),\dots,\vec v'_0(x,y);x,y)=\\=\bar\beta_{\varpi}^{(k)}(\vec v_h(x,y),\dots,\vec v_0(x,y);x,y)-\bar\beta_{\varpi}^{(k)}(\vec v'_h(x,y),\dots,\vec v'_0(x,y);x,y)
\end{split}
\end{equation}
So, using the notation $\vec v_k(x,y)=:\vec v_k$
\begin{equation}
\begin{split}
||\left(\bm T\underline \varpi\right)_h- \left(\bm T\underline \varpi'\right)_h||^{(\theta)}\leq \\
\leq \sum_{k\leq h}\sup_{x\in\lambda}(1+\gamma^{\theta h}|x|^\theta)\gamma^{k-h-1}\sum_{y\in\lambda}\left|\beta_{\varpi}^{(k)}(\vec v_h,\dots,\vec v_0;x,y)-\beta_{\varpi}^{(k)}(\vec v'_h,\dots,\vec v'_0;x,y)\right|=\\
=\sum_{k\leq h}\gamma^{k-h-1} \sum_{y\in\lambda}\left| \bar \beta_{\varpi}^{(k)}(\vec v_h,\dots,\vec v_0;x,y)-\bar \beta_{\varpi}^{(k)}(\vec v'_h,\dots,\vec v'_0;x,y)\right|+\\+\sup_{x\in\Lambda}\gamma^{h\theta} |x|^\theta \sum_{k\leq h}\gamma^{k-h-1} \sum_{y\in\lambda}\left|\bar \beta_{\varpi}^{(k)}(\vec v_h,\dots,\vec v_0;x,y)-\bar \beta_{\varpi}^{(k)}(\vec v'_h,\dots,\vec v'_0;x,y)\right|\leq \\ 
 \leq c' \sum_{k\leq h}\gamma^{k-h-1} \sum_{y\in\lambda}\left|\bar \beta_{\varpi}^{(k)}(\vec v_h,\dots,\vec v_0;x,y)-\bar \beta_{\varpi}^{(k)}(\vec v'_h,\dots,\vec v'_0;x,y)\right|+\\ + c'\sum_{k\leq h}\gamma^{(k-h)(1-\theta)-1} \sup_{x\in\Lambda}\gamma^{k \theta} |x|^\theta \sum_{y\in\lambda}\left|\bar \beta_{\varpi}^{(k)}(\vec v_h,\dots,\vec v_0;x,y)-\bar \beta_{\varpi}^{(k)}(\vec v'_h,\dots,\vec v'_0;x,y)\right|\leq \\
\leq c'' \sum_{k\leq h}\gamma^{k-h-1} |\lambda| \left| \varpi_k(x)-\varpi'_k(x)\right| \sum_{k'\geq k} \gamma^{\theta(k-k')}+
\\
+\sum_{k\leq h}\gamma^{(k-h)(1-\theta)-1} \sup_{x\in\Lambda}\gamma^{k \theta} |x|^\theta |\lambda| \left| \varpi_k(x)-\varpi'_k(x)\right| \sum_{k'\geq k} \gamma^{\theta(k-k')} \leq C \lambda ||\underline \varpi-\underline \varpi'||^{(\theta)}.
\end{split}
\end{equation}
\end{enumerate}
Finally, we fix $$\varpi_1(x,y)=\varpi \pi(x,y),$$ and $\varpi$ is fixed by imposing that
\begin{equation}
\varpi =\sup_{x\in \Lambda}\left|\int_\Lambda dy \varpi_1(x,y)\right|.
\end{equation}
\end{proof}

\paragraph{Bulk running coupling constants}

Since the bulk running coupling constants are {\it space independent}, the strategy to study them is conceptually the same as the previous chapter (\ref{chapter_fermions_PBC}), even though at finite volume they are not exactly the same constants. Let us denote by $\{\bar \nu_h,\bar \delta_h,\bar \lambda_h\}$ the running coupling constants at finite volume of the translation invariant setting described by the Hamiltonian (\ref{hamiltonian_PBC}) in a volume $|\bar \Lambda|=2(L+1)$.
So, let us recall where the bulk running coupling constants come from: $\lambda_h$ comes from $\mathcal L_\mathcal T \mathcal L_{\mathcal B}\mathcal V^{(h)}(\psi^{(\leq h)})$.  Thanks to Theorem (\ref{theorem_renormalized_bounds_DBC}), we infer that, even at finite volume, $\lambda_h=\bar \lambda_h$, since the difference between the quartic terms in the two settings is an irrelevant term. The running coupling constants coming from the quadratic terms localization deserve a deeper comment. Indeed, by construction $\mathcal L_{\mathcal B} W^{(h)}_2=\bar W^{(h)}_2(x-y, x_0-y_0)-\bar W^{(h)}_2(x+y,x_0-y_0)$ consists of the first term, which is exactly the {\it quadratic kernel} of the translation invariant theory defined in the box $\bar \Lambda$, and the second one which is a {\it remainder term} in the sense of the norm $||\cdot||_1$. This suggests to treat $\delta_h$ and $z_h$, coming from $\mathcal L_\mathcal T\mathcal L_{\mathcal B} W^{(h)}_2$, morally as we treated $\lambda_h$, {\it i.e.} by inferring that $|\delta_h-\bar \delta_h|$ and $|z_h-\bar z_h|$ are irrelevant quantities, so we can actually reduce our study to the translational invariant case one.\\
So we are left with fixing the {\it constant counter-term} $\nu_h$: at a formal technical level ({\it i.e.} the fixed point argument), there is no difference with respect to what we have done in the very last section of the previous chapter (\ref{subsection_flow_of_running_coupling_constants_PBC}); anyway, since $\nu_h$ is a {\it relevant running coupling constants}, we cannot proceed as we did for the other constants, because $|\nu_h-\bar \nu_h|$ is a {\it marginal} quantity. At this point it should be clear that, following the {\it definition} of $\nu_h$ that we have chosen, the actually {\it relevant} contribution to $\nu_h$ comes from the linearization of the operators associated with the integral kernel $\bar W_2^{(h)}(\bm x-\bm y)$ and, by applying the same estimates of Theorem (\ref{theorem_renormalized_bounds_DBC}) we get, at finite volume, $|\bar \nu_h-\nu_h|\leq \gamma^{h_L-h}$. To conclude, only when the thermodinamic limit is reached the bulk counterterm on the halfline is the same as the one of the system defined on the whole line.\\
Since all these considerations have a meaning only at finite volume, we underline that the difference between the finite volume and infinite volume running coupling constants has already been rigorously studied in \cite{giuliani2013universal}, during the study of the flow of running coupling constants.

\chapter{Conclusion}

\label{chapter_conlcusion}

\section{Summary}
\label{section_summary}
With the purpose of extending the Constructive Renormalization Group formalism to systems defined in general domains, we attacked a multiscale problem that breaks the {\it translation invariance symmetry} in the simplest possible way: we considered a system of spinless fermions hopping on a 1D semi-infinite lattice with Dirichlet boundary conditions, in the presence of a {\it weak density-density interaction} (of size $\lambda$).\\
We showed, by rigorous Renormalization Group methods, that, if the perturbation is {\it weak enough}, it is possible to fix a {\it quadratic boundary counterterm}, localized at the boundary as ${|\lambda|C_\theta}/{\left((1+|x|\right)^\theta}$ for each $\theta\in(0,1)$, in such a way that the specific free energy is expressed as an analytic function of the perturbation size. In particular, we derived constructive bounds on the difference between the {\it finite volume} specific free energy $f_\Lambda$ and its thermodynamic limit $f=\lim_{|\Lambda|\nearrow \infty } f_\Lambda$:
$$|f-f_\Lambda|\leq |\lambda| \frac{C_\theta}{L^\theta}, \hspace{5mm} \forall \hspace{5mm} 0<\theta<1.$$\\
Our proof involves a systematic treatment of what we call {\it boundary terms}. In particular we developed a method thanks to which, in a {\it multiscale} language, given a family of incapsulated clusters, the presence of a {\it non-translation invariant} element in the innermost cluster $G_v$ is enough to get a {\it  dimensional gain}, that improves the renormalization analysis, for each of the clusters $G_w \supseteq G_v$ containing $G_v$.\\
If, on the one hand, this improvement is enough to renormalize the quartic boundary contributions, on the other hand it allows us to conclude that the quadratic boundary contributions are {\it marginal}. The fact that the boundary conditions are not invariant under RG integrations makes it technically difficoult to absorb this quadratic boundary contributions into the Grassmann integration, so we decided to introduce a quadratic {\it boundary correction},  localized at the boundary as $|\lambda|C_\theta/(1+|x|)^\theta$ for each $\theta\in (0,1)$, to control them.\\
It is worth pointing out that we did not take care of  keeping track of the $\theta$-dependent constants coming from the bounding procedure, thanks to which one could find {\it more explicit} bounds for the corrections to the free energy. \\
In particular one expects, in this way, to be able to express $C_\theta$ as
\begin{equation}
C_\theta=C\frac{1}{(1-\theta)^\alpha},
\end{equation}
for some suitable $\alpha>0$.\\
This would allow us to choose {\it the optimal $\theta\in (0,1)$} by fixing $\theta$ in such a way that
\begin{equation}
\frac{d}{d\theta}\left(\frac{C_\theta}{L^\theta}\right)=C\frac{d}{d\theta}\left(\frac{1}{(1-\theta)^{\alpha}L^\theta}\right)=0
\end{equation}
so the {\it optimal}
\begin{equation}
|f-f_\Lambda|\leq C |\lambda| \frac{\left(\log L\right)^\alpha}{L},
\end{equation}
for some $C>0$.\\
In order not to make heavier the analysis, we decided neither to give these detalis nor to discuss the construction of the Schwinger functions not even in the translation invariant case (Chapter (\ref{chapter_fermions_PBC})). Indeed, even in the translation invariant setting, the construction of Schwinger functions requires an {\it adapted multiscale analysis} slightly different from the one we set up to construct the specific free energy (see \cite{gentile2001renormalization} Section $12$ for an introduction, \cite{benfatto1993beta} for the details). A modification of this {\it multiscale argument}, in the spirit of the modification we introduced in Chapter (\ref{chapter_Interacting_fermions_on_the_half_line}) with respect to Chapter (\ref{chapter_fermions_PBC}), would extend the control of the boundary correction to the case of Schwinger functions: as already mentioned, one expects different behaviours depending on the {\it comparison} between the mutual distance and the {\it distance from the boundary.}

\section{Outlook}
\label{section_outlook}

The next natural step is the program we started this thesis with: {\it i.e.} to {\it invert the counterterm}, meaning properly to {\it build the ground state} of a system described by the Hamiltonian:
\begin{equation*}
H=H_0+\lambda V,
\end{equation*}
with Dirichlet boundary conditions.\\
This corresponds to finding a way to re-sum the quadratic boundary contributions into the Grassmann integration, absorbing the boundary effects into the dressed propagator. The expected result would be encoded in some {\it space dependent} critical exponent $\eta(x)$ in the long distance behaviour of the dressed propagaor.  In the context of Luttinger liquids some {\it non rigorous results} have been obtained, via {\it non-rigorous} bosonization techniques, in  \cite{fabrizio1995interacting, meden2000luttinger, grap2009renormalization,mattsson1997properties}, where the presence of {\it space dependent} critical exponents is investigated by comparing the two-point correlation functions in different asymptotic regimes: both the variables {\it well inside the bulk}, one {\it well inside the bulk} and the other {\it close to the boundary}, both of them {\it close to the boundary}.\\
As we pointed out, the main technical problem in doing this is the fact that {\it the boundary conditions are not invariant under RG iterations}, so absorbing the boundary corrections into the propagator would break the multiscale structure, {\it mixing up} different scales in a {\it non-hierarchical way}. This complication can basically be summarized by saying that the momentum,  being non-preserved, is not the {\it right quantum number} to look at, so in order to solve this problem one should be able to {\it diagonalize}, scale by scale, the single scale Laplacian with covariance $\left( g^{(h)}\right)^{-1}$ perturbed by a {\it weak and localized potential}:
$$\left(\left(g^{(h)}\right)^{-1}-W_2^{(h)}\right)^{-1}(\bm x, \bm y)= \tilde g^{(h)}(\bm x,\bm y)=\sum_{j\in \mathcal D}\hat{\tilde g}^{(h)}_j \varphi_j^*(\bm x)\varphi_j(\bm y),$$
for some suitable {\it dual space} we denote by $\mathcal D$ and orthonormal basis $\left\{\varphi_j(\bm x)\right\}_{j\in\mathcal D}^{\bm x\in\Lambda\times [0,\beta)}$, where $W_2^{(h)}(\bm x, \bm y)$ is such that 
$$\frac{1}{\beta}\int_{[0,\beta)} dx_0 \int_{[0,\beta)} dy_0 \sum_{y\in\Lambda} \left| W_2^{(h)}(\bm x,\bm y)\right|\leq |\lambda| \gamma^h e^{-\gamma^h|x|}.$$
In particular, one should imagine the dual space $\mathcal{D}$ as the {\it energy-space}, being the energy preserved. In other words, it seems that the main difficoulty of the problem is to solve, scale by scale, a {\it scattering problem}, or a {\it perturbed Schr\"odinger equation}, getting explicit expressions both for the eigenfunctions and for the spectrum of the system. Being a quite common problem, there is a huge literature about it mostly interested in the spectral property of the {\it perturbed system} (see {\it e.g.} the review about rank-one perturbations of the Laplacian \cite{simon1995spectral}, or \cite{kato2013perturbation} for a {\it scattering theory} point of view). However at this point it should be clear that, in order to {\it construct} the observables we are interested in via a RG method, we need an {\it "explicit enough"} representation of the covariance allowing us to exploit the {\it selfsimilar structure} of the theory at {\it different scales} in order to iterate this procedure. This quite natural, but challenging, method, seems in fact to match with a {\it multiscale implementation of} the ideas used first by Symanzik, then by Diehl {\it et al.} to study $\phi^4_{4-\epsilon}$ theories in non trivial domains \cite{symanzik1981schrodinger, diehl1981field2,diehl1983universality} via non-rigorous RG methods, in order to investigate the Casimir effect. Even though, on the one hand, the single scale problem seems to be reasonable, we expect its multiscale implementation to be non-trivial. \\
Let us stress, in light of the fact that the boundary correction to the {\it quadratic part of the effective action are marginal}, the same novelties we met in dealing with {\it 1D spinless fermions} would come out even in the case of systems with {\it irrelevant interactions} ({\it e.g.} Ising model, see \cite{mastropietro2008non} Chapter 9 for a RG language treatment of this topic). In fact we expect that our analysis is adaptable to 2D statistical models as Ising, dimers {\it etc.} provided one is able to find a {\it manageable} fermionic representation of the starting model (see again \cite{mastropietro2008non} Chapter 9 for the Ising model, see \cite{giuliani2015height} for dimers, both in translation invariant settings).\\

A further challenging topic related to a full understanding of the problem we started to study is the investigation of the Kondo model \cite{kondo1964resistance} (both the original and the {\it multichannel one}) around the {\it strong coulping regime fixed point}. Indeed, a natural way to study the {\it Kondo effect} seems to be a {\it conformal field theory approach} \cite{affleck1995conformal}: studying the strong coupling regime basically means to assume that the interaction with the impurity (that we assume to be sitting at the origin) is much stronger than the kinetic part. This assumption would imply that the fermion sitting at the origin is bounded to the impurity forming a singlet state with it. So an arbitary electron configuration occurs on all other sites, but other electrons are {\it forbidden} to enter the origin, since that would destroy the singlet state costing a {\it big energy}: that is the reason why the impurity at the origin has roughly the same effect as a Dirichlet boundary condition. Anyway, a rigorous understanding of the Kondo model is still far, even though a first step, based on rigorous hierarchical RG methods, has been completed in \cite{benfatto2015kondo}.
\appendix

\chapter{Estimates of single scale propagators}
\label{appendix_propagator_decay_property}
\begin{proof}[ Proof of Lemma \ref{lemma_propagator_faster_any_power}]
Let us introduce the discrete derivatives $\partial_{\bm k}$, and the {\it directional discrete derivative as} $\tilde \partial_i=\tilde e_i\cdot \partial_{\bm k}$, where $i=0,1$, $\tilde e_0=(1,0)$ and $\tilde e_1=(0,1)$, as follows: fiven a compactly supported function $\hat F(\bm k):\mathcal D_{\Lambda,\beta}\to \mathbb C$,
\begin{equation}
\tilde \partial_i \hat F(\bm k)= \tilde e_i\cdot \partial_{\bm k} \hat F(\bm k)=\frac{\hat F(\bm k+\tilde e_i\Delta k_i)-\hat F(\bm k)}{\delta k_i},
\end{equation}
where by the definition of $\mathcal D_|\Lambda|$ and $\mathcal D_{\beta}$, $\Delta k_0=2\pi/\beta$ and $\delta k_1=2\pi/L$. So, looking at the Fourier transform
\begin{equation}
F(\bm x)=\frac{1}{|\Lambda|\beta}\sum_{\bm k\in\mathcal D_{L,\beta}} e^{-i\bm k\cdot \bm x}\hat F(\bm k),
\end{equation}
it is immediate to check that
\begin{equation}
\sum_{\bm k\in\mathcal D_{\Lambda,\beta}}e^{-i\bm k\cdot {\bm x}} \partial_i \hat F(\bm k)=\left(\frac{e^{i\Delta k_i \tilde x_i}-1}{\Delta k_i}\right)\sum_{\bm k\in\mathcal D_{\Lambda,\beta}}e^{-i\bm k\cdot \bm x}\hat F(\bm k).
\end{equation}
Now let us notice that
\begin{equation}
\frac{e^{i\Delta k_i \tilde {x}_i}-1}{\Delta k_i}=e^{i\Delta k_i \tilde x_i}i\frac{2}{\Delta k_i}\sin \left(\tilde x_i\Delta k_i /2\right)\Rightarrow \left|\frac{e^{i\Delta k_i \tilde {x}_i}-1}{\Delta k_i}\right|=\left|\frac{2}{\Delta k_i}\sin \left(\tilde x_i\Delta k_i /2\right)\right|,
\end{equation}
and that
\begin{equation}
|x_0|\leq \frac{\pi}{2}\frac{\sin ( x_0\pi/\beta)}{\pi/\beta}=:\frac{\pi}{2}d_\beta(x_0), \hspace{3mm} |x|\leq \frac{\pi}{2}\frac{sin(x \pi/L )}{\pi/L}=:\frac{\pi}{2}d_\Lambda(x).
\end{equation}
we get
\begin{equation}
\begin{split}
|x_0|^2|F(\bm x)|\leq \frac{\pi^2}{4}d_\beta^2(x_0)|F(\bm x)|=\frac{\pi^2}{4} \left|\frac{e^{i\Delta k_0 x_0}-1}{\Delta k_0}\right|^2 \left|\frac{1}{|\Lambda \beta|}\sum_{\bm k\in\mathcal D_{\Lambda,\beta}}e^{-i\bm k\cdot \bm x}\hat F(\bm k)\right|=\\
=\frac{\pi^2}{4}  \left|\frac{1}{|\Lambda \beta|}\sum_{\bm k\in\mathcal D_{\Lambda,\beta}}e^{-i\bm k\cdot \bm x}\partial_{k_0}^2\hat F(\bm k)\right|\leq \frac{\pi^2}{4}\frac{1}{|\Lambda|\beta}\sum_{\bm k\in\mathcal D_{|\Lambda|,\beta}}\left|\partial^2_{k_0}\hat F(\bm k)\right|.
\end{split}
\end{equation}
and
\begin{equation}
\begin{split}
|x|^2|F(\bm x)|\leq \frac{\pi^2}{4}d_\Lambda^2(x)|F(\bm x)|=\frac{\pi^2}{4} \left|\frac{e^{i\Delta k x}-1}{\Delta k}\right|^2 \left|\frac{1}{|\Lambda \beta|}\sum_{\bm k\in\mathcal D_{\Lambda,\beta}}e^{-i\bm k\cdot \bm x}\hat F(\bm k)\right|=\\
=\frac{\pi^2}{4}  \left|\frac{1}{|\Lambda \beta|}\sum_{\bm k\in\mathcal D_{\Lambda,\beta}}e^{-i\bm k\cdot \bm x}\partial_{k}^2\hat F(\bm k)\right|\leq \frac{\pi^2}{4}\frac{1}{|\Lambda|\beta}\sum_{\bm k\in\mathcal D_{|\Lambda|,\beta}}\left|\partial^2_{k}\hat F(\bm k)\right|.
\end{split}
\end{equation}
Let us now consider $F(\bm x)=g^h_{\omega}(\bm x)$
\begin{equation}
g_{\omega}^{(h)}(\bm x-\bm y)=\frac{1}{L\beta}\sum_{\bm k\in\mathcal{D}_{L,\beta}}e^{-i\bm k\cdot(\bm x-\bm y)}f_h(\bm k-\omega p_F)\hat g(\bm k).
\end{equation}
and let us observe that we can procede as before to get, for all $N\geq 0$ and a suitable constant $C$,
\begin{equation}
\begin{split}
|x_0|^{2N}|g_\omega^{(h)}(\bm x)|\leq \frac{C^N}{\beta |\Lambda|}\sum_{\bm k\in\mathcal D_{\Lambda,\beta}}\left|\partial_{k_0}^{2N}\left[f_h(\bm k-\omega p_f)\hat g(\bm k)\right]\right|,\\
|x|^{2N}|g_\omega^{(h)}(\bm x)|\leq \frac{C^N}{\beta |\Lambda|}\sum_{\bm k\in\mathcal D_{\Lambda,\beta}}\left|\partial_{k}^{2N}\left[f_h(\bm k-\omega p_f)\hat g(\bm k)\right]\right|,
\end{split}
\end{equation}
that can be unified in
\begin{equation}
|\bm x|^{2N}|g_\omega^{(h)}(\bm x)|\leq \frac{C^N}{\beta |\Lambda|}\sum_{\bm k\in\mathcal D_{\Lambda,\beta}}\left|\partial_{\bm k}^{2N}\left[f_h(\bm k-\omega p_f)\hat g(\bm k)\right]\right|.
\end{equation}
Observing that, on the support of $f_h(k-\omega p_F,k_0)$, $|k_0|\sim |k-\omega p_F|\sim \gamma^h $ and $\hat g(k)\sim \gamma^{-k}$ the derivative can be dimensionally estimated, when acting either on $f_h$ of on $\hat g$, by a factor $\gamma^{-h}$, while the support of $f_h$ is an annulus whose volume is proportional to $\gamma^{2h}$, meaning that $$\frac{1}{|\Lambda|\beta}\sum_{\bm k\in\mathcal D_{\Lambda,\beta}}f_h(\bm k-\omega p_F)\sim \gamma^{2h}.$$ 
Putting all together these estimates we get, for some $C_N$, the bound
\begin{equation}
|\bm x|^{2N}|g^{(h)}_\omega(\bm x)|\leq C_N\gamma^{2h}\gamma^{-h}\gamma^{-2Nh}.
\end{equation}

\end{proof}

Thanks to these estimates, we can also prove the bounds of $||g||_1$ and $||g||_{\infty}$ we used in Lemma (\ref{lemma_bounds_no_multiscale_no_determinants})
\begin{corollary}
Let $$g(\bm x-\bm y)=\frac{1}{\Lambda\beta}\sum_{\bm k\in\mathcal D_{\Lambda,\beta, M}}e^{-i\bm k\cdot (\bm x-\bm y)}\hat g(\bm k).$$
 So
\begin{eqnarray}
||g||_1=\frac{1}{L\beta}\int_{[0,\beta)}dx_0\sum_{x\in\Lambda}\int_{[0,\beta)}dy_0\sum_{y\in\Lambda} g(\bm x-\bm y)\leq C\beta,\\
||g||_\infty\leq CM.
\end{eqnarray}
\end{corollary}
\begin{proof}
It is enough to extend the multiscale decomposition to the ultraviolet regime $2 \geq h \leq \lfloor h_M$, where $\log_\gamma M \rfloor$
\begin{eqnarray}
||g||_1\leq c\left(\sum_{h=h_\beta}^1 \gamma^{-h}+\sum_{h=2}^{h_M} 
\gamma^{-h}\right)\leq C \beta,\\
||g||_\infty \leq c\left(\sum_{h=h_\beta}^1 \gamma^{h}+\sum_{h=2}^{h_M} \gamma^{h}\right)\leq C M
\end{eqnarray}
\end{proof}

\chapter{Finite volume localization (PBC)}
\label{appendix_finite_volume_localization}
\section*{Momentum space-time localization}
{\it In this case in which periodic boundary conditions are imposed, we have a certain freedom to choose the point to localize at, because of the symmetry of $\mathcal{D}_{L,\beta}^{\omega}$ with respect to $0$. So let us call
\begin{equation}
\bar{\bm k}_{\eta, \eta'}=\left(\eta \frac{\pi}{L},\eta'\frac{\pi}{\beta}\right), \hspace{5mm} \eta,\eta'=\pm 1,
\label{bar_k_taylor_expansion_points}
\end{equation}
whose usefulness will be clear in a while.}\\
Let $\mathcal{L}$ be the {\it localization operator} acting on the {\it effective potentials} in the following way:
\begin{itemize}
\item the terms with more then $6$ external legs cause no problems, so we have nothing to extract:
$$\mathcal{L} \left(\frac{1}{\left(\beta L\right)^{2n}}\sum_{n\geq 3}\sum_{\bm k_1,\dots,\bm k_n\in\mathcal{D}_{L,\beta}}\prod_{j=1}^n\left(\hat \psi_{\bm k_{2j-1}}^{(\leq 0)+}\hat \psi_{\bm k_{2j}}^{(\leq 0)-}\right)\hat W_{2n}(\bm k_1,\dots,\bm k_{2n})\delta\left(\sum_{j=1}^{2n}(j+1)\bm k_j\right)\right)=0$$
\item on the terms with $4$ external legs,
\begin{equation}
\begin{aligned}
\mathcal L \left(\frac{1}{\left(\beta L\right)^4}\sum_{\bm k'_1,\bm k'_2,\bm k'_3,\bm k'_4 \in \mathcal D^{\bm \omega}_{L,\beta}}\psi^{(\leq 0)+}_{\omega_1,\bm k'_1}\psi^{(\leq 0)+}_{\omega_2, \bm k'_2}\psi^{(\leq 0)-}_{\omega_3, \bm k'_3}\psi^{(\leq 0)-}_{\omega_4, \bm k'_4}\right.\\ \left .\hat W_{4,\bm \omega}(\bm k'_1,\bm k'_2,\bm k'_3,\bm k'_4)\delta_{\bm k'_1+\bm k'_2,\bm k'_3+\bm k'_4}\delta_{\omega_1+\omega_2,\omega_3+\omega_4}\right)=\\
\frac{1}{\left(\beta L\right)^4}\sum_{\bm k'_1,\bm k'_2,\bm k'_3,\bm k'_4 \in \mathcal D^{\bm \omega}_{L,\beta}}\psi^{(\leq 0)+}_{\omega_1,\bm k'_1}\psi^{(\leq 0)+}_{\omega_2, \bm k'_2}\psi^{(\leq 0)-}_{\omega_3, \bm k'_3}\psi^{(\leq 0)-}_{\omega_4, \bm k'_4}\\ \hat W_{4,\bm \omega}(\bar{\bm k}_{++},\bar{\bm k}_{++},\bar{\bm k}_{++},\bar{\bm k}_{++})\delta_{\bm k'_1+\bm k'_2,\bm k'_3+\bm k'_4}\delta_{\omega_1+\omega_2,\omega_3+\omega_4},
\end{aligned}
\end{equation}
\item  on the terms with $2$ external legs
\begin{equation}
\begin{split}
\mathcal L\left( \frac{1}{L\beta}\sum_{\bm k\in\mathcal{D}^{\omega}_{L,\beta}} \psi_{\omega,\bm k'}^{(\leq 0)+}\hat \psi_{\omega,\bm k'}^{(\leq 0)-}W_{2,\omega}(\bm k') \right)
=\frac{1}{L\beta}\sum_{\bm k\in\mathcal{D}^{\omega}_{L,\beta}} \psi_{\omega,\bm k'}^{(\leq 0)+}\hat \psi_{\omega,\bm k'}^{(\leq 0)-}\cdot \\
\cdot\frac{1}{4}\left(\sum_{\eta,\eta'=\pm 1}\hat W^{(h)}_{2,\bm \omega}(\bar{\bm k}_{\eta\eta'})\right)\left[1+\left(\eta \frac{L}{\pi}\left(b_L+a_Le(k'+\omega p_F)\right)+\eta'\frac{\beta}{\pi}k_0\right)\right]
\end{split}
\end{equation}
where
\begin{equation}
a_L=\frac{\pi/L}{\sin\left(\pi/L\right)},\hspace{3mm} 
b_L= \cos(p_F) \frac{\left(\cos(\pi/L)-1\right)\pi/L}{\sin(\pi/L)}
\end{equation}
If, on the one hand, this is the rigorous definition taking into account the finite size effect that, of course, disappear once we take the thermodinamic and zero temperature limit, the reader should keep in mind that
\begin{equation}
\begin{split}
\lim_{L\to\infty}a_L=1, \hspace{3mm}\lim_{L\to\infty}b_L=0,
\end{split}
\end{equation}
implying that
\begin{equation}
\begin{split}
\lim_{L,\beta\to \infty}\mathcal L\left( \frac{1}{L\beta}\sum_{\bm k\in\mathcal{D}^{\omega}_{L,\beta}} \psi_{\omega,\bm k'}^{(\leq 0)+}\hat \psi_{\omega,\bm k'}^{(\leq 0)-}W_{2,\omega}(\bm k') \right)
=\int d\bm k' \psi_{\omega,\bm k'}^{(\leq 0)+}\hat \psi_{\omega,\bm k'}^{(\leq 0)-}\cdot \\
\cdot\left[\hat W^{(h)}_{2,\bm \omega}(0)+\left(e( k'+\omega p_F)\frac{\partial \hat{W}^{(h)}_{2,\omega}}{\partial k'}(0)+k_0 \frac{\partial \hat W^{(h)}_{2,\omega}}{\partial k_0}(0)\right)\right].
\end{split}
\end{equation}
\end{itemize}
Finally, we simply define the {\it renormalization operator}
\begin{equation}
\mathcal R=1-\mathcal L
\end{equation}
where $1$ has to be read as the {\it identity operator}.
\begin{rem}
First of all, it is clear that the only condition to recover the  {\it infinite volume limit} (\ref{localization_2el_infinite_volume_limit}) is 
$$a_L=1+O\left(\frac{1}{L^2}\right),\hspace{3mm}b_L=O\left(\frac{1}{L^2}\right).$$
The choice (\ref{localization_2el_infinite_volume_limit}) is better then the others because it reproduces the property which is true in the infinite volume limit: $\mathcal{L}$ and $\mathcal{R}$ are projectors onto to orthogonal spaces:
\begin{equation}
\mathcal L^2=\mathcal L,\hspace{3mm}\mathcal R^2=\mathcal R,\hspace{3mm}\mathcal{L}\mathcal{R}=\mathcal{R}\mathcal{L}=0. 
\end{equation} 
Besides, it is worth noting that that in momenta space representation the operators $\mathcal L$ and $\mathcal R$ act directly on the kernels $\hat W_{2n,\omega}$, so it is not necessary to use the heavy notation we used to express their action on the whole operator represenation. Anyway, we prefer to be heavier but complete: indeed this representation corresponds, in a trivial way, to a real-space representation (via a Fourier transorm) that we will explain in a while, and that will be the crucial representation in the next chapter when, due to the lost of momenta conservation, it will be useless to represent the system in Fourier space.
\end{rem}
\section*{Real-space localization}
\label{appendix_real_space-time_localization}

In subsection (\ref{subsection_renormalization_group_PBC}) we defined the {\it localization} operator directly in the Fourier space formalism: it corresponds, trivially, to compute the kernels that have to be renormalized (two and four external legs kernels) at Fermi momentum, based on the fact that we know that the propagator (\ref{free_propagator_PBC}) is singular at Fermi point. Technically, it corresponds to keep in what we call the {\it local part} only the order zero term (in the case of quartic kernels) or the order zero and order one terms (in the case of quadratic terms) of the Taylor expansion around the Fermi points. Of course, there is a counterpart in the real space-time representation, that we are going to explain. Of course, there is nothing more to do than just Fourier transform the formulae we have given in the momentum-space.
\begin{itemize}
\item If $2n=4$, by Fourier transforming the (\ref{localization_4el_finite_volume_limit}) we get
\begin{equation}
\begin{split}
\mathcal L \sum_{\bm x_1,\dots,\bm x_4} W^{(h)}_{\bm \omega, 4}(\bm x_1, \bm x_2, \bm x_3, \bm x_4) \psi^{(\leq h)+}_{\omega_1,\bm x_1} \psi^{(\leq h)+}_{\omega_2,\bm x_2}\psi^{(\leq h)-}_{\omega_3,\bm x_3}\psi^{(\leq h)-}_{\omega_4,\bm x_4}=\\
=  \sum_{\bm x_1,\dots,\bm x_4} W^{(h)}_{\bm \omega, 4}(\bm x_1, \bm x_2, \bm x_3, \bm x_4)e^{i\bar{\bm k}_{++}\cdot(\bm x_1-\bm x_4)} \psi^{(\leq h)+}_{\omega_1,\bm x_4} \\e^{i\bar{\bm k}_{++}\cdot(\bm x_2-\bm x_4)}\psi^{(\leq h)+}_{\omega_2,\bm x_4}
e^{-i\bar{\bm k}_{++}\cdot(\bm x_3-\bm x_4)}\psi^{(\leq h)-}_{\omega_3,\bm x_4}\psi^{(\leq h)-}_{\omega_4,\bm x_4}=\\
=\sum_{\bm x_1,\dots,\bm x_4} W^{(h)}_{\bm \omega, 4}(\bm x_1, \bm x_2, \bm x_3, \bm x_4)e^{i\bar{\bm k}_{++}\cdot(\bm x_1+\bm x_2-\bm x_3-\bm x_4)} \psi^{(\leq h)+}_{\omega_1,\bm x_4} \psi^{(\leq h)+}_{\omega_2,\bm x_4}\psi^{(\leq h)-}_{\omega_3,\bm x_4}\psi^{(\leq h)-}_{\omega_4,\bm x_4}
\end{split}
\label{localization_4el_real_spacetime}
\end{equation}
where, by recalling the definition (\ref{bar_k_taylor_expansion_points})
$$\bar{\bm k}_{++}=\left(\frac{\pi}{L},\frac{\pi}{\beta}\right)$$
we notice that the {\it localization operator} in real space-time representation {\it acts} on the annihilation and creation operators by computing all of them in the same space-time point ({\it localization}) and multiplying them by an oscillating factor that keeps track of the finiteness of the volume:
$$e^{i\bar{\bm k}_{++}(\bm x_i-\bm x_4)}=e^{i\pi\left(\frac{x_i-x_4}{L}+\frac{x_{0_i}-x_{0_4}}{\beta}\right)}.$$
In writing the formula (\ref{localization_4el_real_spacetime}) we have choosen arbitrarly the point $\bm x_4$ to localize at, but by noting that the function $W^{(h)}_{\bm \omega, 4}(\bm x_1, \bm x_2, \bm x_3, \bm x_4)e^{i\bar{\bm k}_{++}\cdot(\bm x_1+\bm x_2-\bm x_3-\bm x_4)}$ is translation invariant and periodic in the space-time components (respectively with period $L$ and $\beta$) it follows that we can replace each $e^{i\epsilon_i \bar{\bm k}_{++}(\bm x_i-\bm x_4)}\psi^{(\leq h)\epsilon_i}_{\bm x_4}$ with $e^{i\epsilon_i \bar{\bm k}_{++}(\bm x_i-\bm x_k)}\psi^{(\leq h)\epsilon_i}_{\bm x_k}$, $k=1,2,3,4$: {\it i.e.} we can chose the {\it localization point}, so we have a {\it freedom in the choice of the localization point}.\\
By definition $$\mathcal{R}=1-\mathcal L,$$ so if we define the localization as in (\ref{localization_4el_real_spacetime}), we get
\begin{equation}
\begin{split}
\mathcal R\sum_{\bm x_1,\dots,\bm x_4} W^{(h)}_{\bm \omega, 4}(\bm x_1, \bm x_2, \bm x_3, \bm x_4) \psi^{(\leq h)+}_{\omega_1,\bm x_1} \psi^{(\leq h)+}_{\omega_2,\bm x_2}\psi^{(\leq h)-}_{\omega_3,\bm x_3}\psi^{(\leq h)-}_{\omega_4,\bm x_4}=\\
=\sum_{\bm x_1,\dots,\bm x_4} W^{(h)}_{\bm \omega, 4}(\bm x_1, \bm x_2, \bm x_3, \bm x_4) \left(\prod_{i=1}^4\psi^{(\leq h)\epsilon_i}_{\omega_i,\bm x_i}-\prod_{i=1}^4e^{i \epsilon_i \bar{\bm k}_{++}(\bm x_i-\bm x_4)}\psi^{(\leq h)\epsilon_i}_{\omega_i,\bm x_i}\right)
\end{split}
\end{equation}
where the term in brackets can be read as
\begin{equation}
\begin{split}
\psi^{(\leq h)+}_{\omega_1,\bm x_1} \psi^{(\leq h)+}_{\omega_2,\bm x_2}D^{1,1(\leq h)-}_{\bm x_3,\bm x_4,\omega 3}\psi^{(\leq h)-}_{\omega_4,\bm x_4}+\\
+e^{-i\bar{\bm k}_{++}(\bm x_3-\bm x_4)} \psi^{(\leq h)+}_{\omega_1,\bm x_1} D^{1,1(\leq h)+}_{\bm x_2,\bm x_4, \omega_2}\psi^{(\leq h)-}_{\omega_3,\bm x_4}\psi^{(\leq h)-}_{\omega_4,\bm x_4}+\\
+e^{i\bar{\bm k}_{++}(\bm x_2-\bm x_3)}D^{1,1(\leq h)+}_{\bm x_1,\bm x_4,\omega_1} \psi^{(\leq h)+}_{\omega_2,\bm x_4}\psi^{(\leq h)-}_{\omega_3,\bm x_4}\psi^{(\leq h)-}_{\omega_4,\bm x_4}
\end{split}
\end{equation}
where in sake clarity we used the synthetic notation:
\begin{equation}
D_{\bm y,\bm x,\omega}^{1,1(\leq h)\epsilon}=\psi_{\bm y,\omega}^{(\leq h),\epsilon}-e^{i \epsilon\bar{\bm k}_{++}(\bm y-\bm x)}\psi^{(\leq h)+}_{\bm x,\omega}
\end{equation}
and we have chosen the name $D$ because, when the new {\it field} $D$ is contracted, it corresponds to a {\it derivative-propagator}.\\
An equivalent representation of the remainder is
\begin{equation}
\begin{split}
\mathcal R\sum_{\bm x_1,\dots,\bm x_4}W^{(h)}_{\bm \omega, 4}(\bm x_1, \bm x_2, \bm x_3, \bm x_4) \prod_{i=1}^4\psi^{(\leq h)\epsilon_i}_{\omega_i,\bm x_i}=\\
=\sum_{\bm x_1,\dots,\bm x_4} \prod_{i=1}^4\psi^{(\leq h)\epsilon_i}_{\omega_i,\bm x_i} \Bigl[W^{(h)}_{\bm \omega, 4}(\bm x_1, \bm x_2, \bm x_3, \bm x_4)-\\ -\delta_{\bm x_3,\bm x_4}\sum_{\bm y_3} W^{(h)}_{\bm \omega, 4}(\bm x_1, \bm x_2, \bm y_3, \bm x_4)e^{-i\bar{\bm k}_{++}(\bm y_3-\bm x_4)} \Bigr]+\\
+\sum_{\bm x_1,\dots,\bm x_4} \prod_{i=1}^4\psi^{(\leq h)\epsilon_i}_{\omega_i,\bm x_i} \delta_{\bm x_3,\bm x_4}\sum_{\bm y_3} \Bigl(W^{(h)}_{\bm \omega, 4}(\bm x_1, \bm x_2, \bm y_3, \bm x_4)e^{-i\bar{\bm k}_{++}(\bm y_3-\bm x_4)} -\\
- \delta_{\bm x_2,\bm x_4}\sum_{\bm y_2}W^{(h)}_{\bm \omega, 4}(\bm x_1, \bm x_2, \bm y_3, \bm x_4)e^{i\bar{\bm k}_{++}(\bm y_2-\bm x_3)}\Bigr)+\\
+\sum_{\bm x_1,\dots,\bm x_4} \prod_{i=1}^4\psi^{(\leq h)\epsilon_i}_{\omega_i,\bm x_i} \delta_{\bm x_2,\bm x_4}\delta_{\bm x_3,\bm x_4} \sum_{\bm y_2,\bm y_3}\Bigl(W^{(h)}_{\bm \omega, 4}(\bm x_1, \bm y_2, \bm y_3, \bm x_4)e^{i\bar{\bm k}_{++}(\bm y_2-\bm x_3)}-\\ 
-\delta_{\bm x_1,\bm x_4}\sum_{\bm y_1}W^{(h)}_{\bm \omega, 4}(\bm y_1, \bm y_2, \bm y_3, \bm x_4)e^{i\bar{\bm k}_{++}(\bm y_1+\bm y_2-\bm y_3-\bm x_4)}\Bigr)
\end{split}
\end{equation}
where we have used
\begin{equation}
\delta_{\bm x,\bm y}=\frac{1}{L\beta}\sum_{\bm k'\in\mathcal{D}_{L,\beta}^\omega}e^{i\bm k'(\bm x-\bm y)}.
\end{equation}

\item If $2n=4$, by Fourier transforming (\ref{localization_2el_infinite_volume_limit}) we get 
\begin{equation}
\begin{split}
\mathcal{L}\sum_{\bm x,\bm y}W^{(h)}_{\bm \omega,2}(\bm x-\bm y)\psi^{(\leq h)+}_{\omega,\bm x}\psi^{(\leq h)-}_{\omega,\bm y}=\\ 
=\sum_{\bm x,\bm y}W^{(h)}_{\bm \omega,2}(\bm x-\bm y)\psi^{(\leq h)+}_{\omega,\bm x}T^{1(\leq h)-}_{\bm y,\bm x,\omega}=\\
\sum_{\bm x,\bm y}W^{(h)}_{\bm \omega,2}(\bm x-\bm y)T^{1(\leq h)+}_{\bm x,\bm y,\omega}\psi^{(\leq h)-}_{\omega,\bm y}
\end{split}
\label{localization_2el_real_spacetime}
\end{equation}
where
\begin{equation}
\begin{split}
T^{1(\leq h)\epsilon}_{\bm y,\bm x,\omega}=\psi^{(\leq h)\epsilon}_{\omega,\bm x}c_\beta(y_0-x_0)[c_L(y-x)+b_Ld_L(y-x)]+\\
+[\partial_1 \psi^{(\leq h)\epsilon}_{\omega,\bm x}+\frac{i\cos p_F}{2}\partial_1^2\psi^{(\leq h)\epsilon}_{\omega,\bm x}c_\beta(y_0-x_0)a_Ld_L(y-x)]+\\
+\partial_0\psi^{(\leq h)\epsilon}_{\omega,\bm x} d_\beta (y_0-x_0)c_L(y-x)
\end{split}
\end{equation}
where
\begin{equation}
\begin{split}
d_L(x)=\frac{L}{\pi}\sin\left(\frac{\pi x}{L}\right), \hspace{5mm} d_\beta(x_0)=\frac{\beta}{\pi}\sin\left(\frac{\pi x}{\beta}\right),\\
c_L=\cos\left(\frac{\pi x}{L}\right),\hspace{5mm}c_\beta=\cos\left(\frac{\pi x_0}{\beta}\right).
\end{split}
\end{equation}
Again, by definition,
\begin{equation}
\begin{split}
\mathcal{R}\sum_{\bm x,\bm y}W^{(h)}_{\bm \omega,2}(\bm x-\bm y)\psi^{(\leq h)+}_{\omega,\bm x}\psi^{(\leq h)-}_{\omega,\bm y}=\\
=\sum_{\bm x,\bm y}W^{(h)}_{\bm \omega,2}(\bm x-\bm y)\psi^{(\leq h)+}_{\omega,\bm x}D^{2(\leq h)-}_{\bm y,\bm x\omega}=\\
=\sum_{\bm x,\bm y}W^{(h)}_{\bm \omega,2}(\bm x-\bm y)D^{2(\leq h)+}_{\bm x,\bm y,\omega}\psi^{(\leq h)-}_{\omega,\bm y}
\end{split}
\end{equation}
with of course
\begin{equation}
D^{2(\leq h)\epsilon}_{\bm y,\bm x,\omega}=\psi^{(\leq h)\epsilon}_{\omega,\bm y}-T^{1(\leq h)\epsilon}_{\bm y,\bm x,\omega}.
\end{equation}
The remainder can be represented as
\begin{equation}
\begin{split}
\mathcal{R}\sum_{\bm x,\bm y}W^{(h)}_{\bm \omega,2}(\bm x-\bm y)\psi^{(\leq h)+}_{\omega,\bm x}\psi^{(\leq h)-}_{\omega,\bm y}=\sum_{\bm x,\bm y}\psi^{(\leq h)+}_{\omega,\bm x}\psi^{(\leq h)-}_{\omega,\bm y}\Bigl(W^{(h)}_{\omega, 2}(\bm x-\bm y)-\\
-\delta_{\bm x,\bm y}\sum_{\bm z} W^{(h)}_{\omega, 2}(\bm x-\bm z)c_\beta(z_0-x_0)[c_L(z-x)+b_Ld_L(z-x)]-\\
-[-\partial_1\delta_{\bm y,\bm x}+\frac{i\cos p_F}{2}\partial_1^2\delta_{\bm x,\bm y}]\sum_{\bm z}W^{(h)}_{\omega, 2}(\bm x-\bm z)c_\beta(z_0-x_0)a_Ld_L(z-x)-\\
-\partial_0\delta_{\bm x,\bm y}\sum_{\bm x}W^{(h)}_{\omega, 2}(\bm x-\bm z)d_\beta(z_0-x_0)c_L(z-x).
\end{split}
\end{equation}
\end{itemize}

\chapter{Gram representation}
\label{appendix_gram_representation}
In the proof of the following Lemma, concerning the Gram-Hadamard representation of the remainder propagators, is included the prove of the fact that also the free propagator (\ref{free_propagator_PBC}) can be expressed as a scalar product of vectors of a suitable chosen Hilbert space.\\
The proof of the assumption we made in Lemma (\ref{lemma_gram_hadamard_for_G}) is included in the proof of the following Lemma (\ref{lemma_gram_hadamard_scalar_product_off_diagonal_propagators}): in particular it is enough to look at formula (\ref{g_2(L+1)_scalar_product}) and the explicit expressions of $A^{(h)}_{2(L+1)}$, $B^{(h)}_{2(L+1)}$ in formula (\ref{explicit_expressions_functions_A_B_gram_hadamard}).
\begin{lem}
\label{lemma_gram_hadamard_scalar_product_off_diagonal_propagators}
The {\it remainder propagator} $g^{(h)}_{R}(x+y,x_0-y_0)$ defined in Lemma (\ref{lemma_reflection_trick}) and the definition of $g^{(h)},$ can be written as a scalar product. As a consequence, we are allowed to build up the matrix $G^T(\bm t)$ to use the formula (\ref{free_energy_determinant_expansion_gram_hadamard}) and the Gram-Hadamard estimate (Lemma (\ref{lemma_gram_hadamard_inequality})).
\end{lem}

\begin{proof}
Let us recall that the formula (\ref{g^d(h)_definition}) says that
\begin{equation}
g_R^{(h)}(\bm x,\bm y):=g^{(h)}_{2(L+1)}((x,x_0),(-y,y_0))=g^{(h)}_{2(L+1)}(\bm x-\bm y)-g^{(h)}(\bm x,\bm y),
\end{equation}
where
\begin{equation}
\begin{split}
g^{(h)}_{2(L+1)}(\bm x-\bm y)=\frac{1}{\beta 2(L+1)}\sum_{\bm k\in\mathcal D_{2(L+1),\beta}} e^{-i\bm k(\bm x-\bm y)}\frac{f_h(\bm k)}{-ik_0+e(k)},\\
g^{(h)}(\bm x,\bm y)=\frac{2}{\beta (L+1)}\sum_{\bm k\in\mathcal D^d_{\Lambda,\beta}} e^{-i k_0( x_0 - y_0)}\sin(kx)\sin(ky)\frac{f_h(\bm k)}{-ik_0+e(k)},
\end{split}
\end{equation}

The idea of the proof is to show that we can write both the propagators in the right hand side of the latter formula as scalar products, and then to use the linearity of the scalar product. Let us introduce four functions $A^{d(h)}(\bm x,\cdot), B^{d(h)}(\bm x,\cdot), A^{(h)}_{2(L+1)}(\bm x,\cdot),B^{(h)}_{2(L+1)}(\bm x,\cdot)$ such that, if we use the notation $\left<A(\bm x,\cdot), B(\bm y, \cdot)\right>=\int d\bm z\bar A(\bm x,\bm z)B(\bm z,\bm y)$, we can write
\begin{eqnarray}
g^{(h)}_{2(L+1)}(\bm x-\bm y)=\left<A^{(h)}_{2(L+1)}(\bm x,\cdot), B^{(h)}_{2(L+1)}(\bm y, \cdot)\right>, \label{g_2(L+1)_scalar_product}\\
g^{(h)}(\bm x,\bm y)=\left<A^{d(h)}(\bm x,\cdot), B^{d(h)}(\bm y, \cdot)\right>. \label{g^d_scalar_product}
\end{eqnarray}
and it is easy to check that a good choice is 
\begin{equation}
\begin{split}
A^{(h)}_{2(L+1)}(\bm x, \bm y)=\frac{1}{2\beta(L+1)}\sum_{\bm k\in\mathcal D_{2(L+1),\beta}}e^{-i\bm k (\bm x-\bm y)}\sqrt{f_h(\bm k)}\frac{1}{k_0^2+e^2(k)},\\
B^{(h)}_{2(L+1)}(\bm x, \bm y)=-\frac{1}{2\beta(L+1)}\sum_{\bm k\in\mathcal D_{2(L+1),\beta}}e^{-i\bm k (\bm x-\bm y)}\sqrt{f_h(\bm k)}(ik_0+e(k)),\\
A^{d(h)}(\bm x, \bm y)=\frac{2\chi(x>0,y>0)}{\beta(L+1)}\sum_{\bm k\in\mathcal D^d_{\beta}}e^{-i k_0 (\bm x_0-\bm y_0)}\sin(kx)\sin(ky)\sqrt{f_h(\bm k)}\frac{1}{k_0^2+e^2(k)},\\
B^{d(h)}(\bm x, \bm y)=-\frac{2\chi(x>0,y>0)}{\beta(L+1)}\sum_{\bm k\in\mathcal D^d_{\beta}}e^{-i k_0 (\bm x_0-\bm y_0)}\sin(kx)\sin(ky)\sqrt{f_h(\bm k)}(ik_0+e(k)).
\end{split}
\label{explicit_expressions_functions_A_B_gram_hadamard}
\end{equation}
Now we can check that there exist $\tilde A_R(\bm x,\cdot)$ and $\tilde B_R(\bm x,\cdot)$ (where the slightly different notation stays for the fact that they live in a bigger space with respect to the previous ones, as we explain below), such that
\begin{equation}
g^{(h)}_R(\bm x,\bm y)= \left<\tilde A^{(h)}_R(\bm x,\cdot), \tilde B^{(h)}_R(\bm y,\cdot)\right>.
\end{equation}
Let $\odot$ denote some kind of product between the space of the functions $A,B$ and of the new vectors we introduce: $u_{A^d}, u_{A_{2(L+1)}},u_{B^d},u_{B_{2(L+1)}}$, such that, if $\cdot$ is the usual scalar product between these vectors,
\begin{eqnarray}
u_{A_{2(L+1)}}\cdot u_{B_{2(L+1)}}=1,\\
u_{A_{2(L+1)}}\cdot u_{B^d}=0,\\
u_{A^d}\cdot u_{B_{2(L+1)}}=0,\\
u_{A^d}\cdot u_{B^d}=1.
\end{eqnarray}
Besides, let us interpret
\begin{equation}
\left<u_i\odot A(\bm x,\cdot),u_j\odot (\bm y,\cdot)\right>=u_i\cdot u_j\int d\bm z \bar A(\bm x,\bm z)B(\bm y,\bm z)=u_i\cdot u_j \left<A(\bm x,\cdot), B(\bm y,\cdot)\right>.
\end{equation}
Thanks to these definition, we can check that by defining
\begin{eqnarray}
\tilde A^{(h)}_R=A^{d(h)}\odot u_{A^d}+i A^{(h)}_{2(L+1)}\odot u_{A_{2(L+1)}},\\
\tilde B^{(h)}_R=B^{d(h)}\odot u_{B^d}-i B^{(h)}_{2(L+1)}\odot u_{B_{2(L+1)}},\\
\end{eqnarray}
we get
\begin{equation}
\begin{split}
\left<\tilde A^{(h)}_R(\bm x,\cdot),\tilde B^{(h)}_R(\bm y, \cdot)\right>=\left<\tilde A^{d(h)}(\bm x,\cdot),\tilde B^{d(h)}(\bm y, \cdot)\right>-\left<\tilde A^{(h)}_{2(L+1)}(\bm x,\cdot),\tilde B^{(h)}_{2(L+1)}(\bm y, \cdot)\right>=\\ =g^{(h)}(\bm x,\bm y)-g^{(h)}_{2(L+1)}(\bm x-\bm y)
\end{split}
\end{equation}
so finally
\begin{equation}
g^{(h)}_R(\bm x,\bm y)=\left<\tilde A^{(h)}_R(\bm x,\cdot),\tilde B^{(h)}_R(\bm y, \cdot)\right>.
\end{equation}
Furthermore, in order to use the Gram-Hadamard estimate, let us note that
\begin{eqnarray}
||\tilde{A}^{(h)}_R||^2=||A^{(h)}_{2(L+1)}||^2+||A^{d(h)}||^2,\\
||\tilde{B}^{(h)}_R||^2=||B^{(h)}_{2(L+1)}||^2+||B^{d(h)}||^2.
\end{eqnarray}
so that, thanks to formulae (\ref{g_2(L+1)_scalar_product}) and (\ref{g^d_scalar_product}),
\begin{equation}
||\tilde{A}^{(h)}_R||||\tilde{B}^{(h)}_R||\leq C\gamma^h,
\end{equation}
for some $C>0$.
\end{proof}

\chapter{Breaking of DBC: a counter-example}
\label{appendix_non_local_tadpole}
\begin{proof}
In constructing this counter-example, we use the assumption that $v(x,y)$ is diagonal in the {\it sine Fourier base}. It is enough to show a counterexample: the {\it non-local tadpole} (see Figure (\ref{figure_tadpoles}), the element on the right):
\begin{equation}
\begin{split}
\int_{[0,\beta)} dx_0\sum_{x\in\Lambda}\int_{[0,\beta)} dy_0\sum_{y\in\Lambda} \psi^{(h)+}_{\bm x}\psi^{(h)-}_{\bm y}v(x,y)\delta_{x_0,y_0}g^{(h)}(\bm x,\bm y)=\\
=\frac{1}{\beta^3}\left(\frac{2}{(L+1)}\right)^4  \sum_{\substack{ k_{1_0}, k_{2_0}, k_{4_0}\\ \in \\ \mathcal D_{\beta}}}\sum_{\substack{ k_1,  k_2,  k_3,  k_4\\ \in \\ \mathcal D_{\Lambda}^d}}\int_{[0,\beta)} dx_0\sum_{x\in\Lambda}\sum_{y\in\Lambda}\hat \psi^{(h)+}_{\bm k_1}\hat \psi^{(h)-}_{\bm k_2}\hat v( k_3) \hat g^{(h)}(\bm k_4) \cdot \\ \cdot e^{-ix_0(k_{1_0}-k_{2_0})}\left(\sin (k_1x)\sin (k_3x)\sin (k_4x)\right)\left(\sin (k_2y)\sin (k_3y)\sin (k_4y)\right)=\\
=\left(\frac{1}{\beta}\right)^2\left(\frac{2}{(L+1)}\right)^4\sum_{k_0, k_{4_0}\in\mathcal D_\beta} \sum_{\substack{  k_1,  k_2, k_3, k_4\\ \in \\ \mathcal D_{\Lambda}^d}}\sum_{x\in\Lambda}\sum_{y\in\Lambda}\hat \psi^{(h)+}_{(k_1, k_0)}\hat \psi^{(h)-}_{(k_2, k_0)} \hat v( k_3) \hat g^{(h)}(\bm k_4) \cdot \\ \cdot \left(\sin (k_1x)\sin (k_3x)\sin (k_4x)\right)\left(\sin (k_2y)\sin (k_3y)\sin (k_4y)\right).
\end{split}
\end{equation}
Knowing that
\begin{equation}
\begin{split}
\sin (k_1x)\sin (k_2 x) \sin(k_3 x)=\frac{1}{(2i)^3}\sum_{\substack{\sigma_1, \sigma_2, \sigma_3 \\ \in \\ \{\pm 1\}} }\sigma_1\sigma_2\sigma_3 e^{i(\sigma_1k_1+\sigma_2k_2+\sigma_3k_3)x}=\\=\frac{1}{(2i)^2}\sum_{\substack{\sigma_2, \sigma_3 \\ \in \\ \{\pm 1\}} }\sigma_2\sigma_3 \sin((k_1+\sigma_2 k_2+\sigma_3k_3)x)
\end{split}
\end{equation}
we can rewrite
\begin{equation}
\begin{split}
\sum_{x\in\Lambda}\sum_{y\in\Lambda}\left(\sin (k_1x)\sin (k_3x)\sin (k_4x)\right)\left(\sin (k_2y)\sin (k_3y)\sin (k_4y)\right)=\\
= \sum_{\substack{\sigma_3,\sigma_4,\omega_3,\omega_4\\ \in \\ \{\pm 1\}}}\sigma_3\sigma_4\omega_3\omega_4 \sum_{x\in\Lambda}\sum_{y\in\Lambda} \sin ((k_1+\sigma_3k_3+\sigma_4k_4)x)\sin ((k_2+\omega_3k_3+\omega4k_4)y)
\end{split}
\label{sin^6_in_counterexample}
\end{equation}
Finally
\begin{equation}
\begin{split}
 \sum_{x\in\Lambda}\sum_{y\in\Lambda} \sin ((k_1+\sigma_3k_3+\sigma_4k_4)x)\sin ((k_2+\omega_3k_3+\omega_4k_4)y)=\\
 =\left(L\delta (k_1+\sigma_3k_3+\sigma_4k_4)-2\sum_{x\in \Lambda} e^{-i(k_1+\sigma_3k_3+\sigma_4k_4)x}\right)\cdot\\
\cdot \left(L\delta (k_2+\omega_3k_3+\omega_4k_4)-2\sum_{y\in \Lambda} e^{-i(k_2+\omega_3k_3+\omega_4k_4)y}\right)
\end{split}
\end{equation}
In particular,
\begin{equation}
\begin{split}
\left(L\delta (k_1+\sigma_3k_3+\sigma_4k_4)-2\sum_{x\in \Lambda} e^{-i(k_1+\sigma_3k_3+\sigma_4k_4)x}\right)\cdot\\
\cdot \left(L\delta (k_2+\omega_3k_3+\omega_4k_4)-2\sum_{y\in \Lambda} e^{-i(k_2+\omega_3k_3+\omega_4k_4)y}\right)=L^2\delta(k_1-k_2)\delta (k_1+\sigma_3k_3+\sigma_4k_4)
\end{split}
\end{equation}
if and only if $\omega_3=\sigma_3$ and $\omega_4=\sigma_4$. Since in equation (\ref{sin^6_in_counterexample}) there is a sum over all the possible values of  $\omega_3, \sigma_3, \omega_4, \sigma_4$, we have proved that the non-local tadpole is in fact a counter-example.
\end{proof}

\chapter{Finite volume localization definition, DBC}
\label{appendix_finite_volume_loc_DBC}

\begin{itemize}
\item {\bf Case $2n=2$, kernel $W^{d(h)}_2=\mathcal L_{\mathcal B}W_2^{(h)}$} As we pointed out in formulae (\ref{2el_W^d_in_diagonal_form}) and (\ref{2el_W^d_in_diagonal_quasiparticles_form}), $$\int d\bm x d\bm y\psi^{(\leq h)+}_{\bm x}\psi^{(\leq h)-}_{\bm y}W^{d(h)}_2(\bm x,\bm y)=\frac{2}{\beta(L+1)}\sum_{\bm k'\in \mathcal D'^{d}_{\Lambda,\beta}}\hat \psi^{(\leq h)+}_{\bm k'+\bm p_F}\hat \psi^{(\leq h)-}_{\bm k'+\bm p_F}\hat W^{d(h)}_{2}(\bm k'+\bm p_F),$$
so we can proceed in localizing by analogy with the translation invariant case and define a localization procedure directly in the dual space. In order to take into account the finiteness of the volume we cannot directly localize at $\bm p_F$, but we are forced to localize at the nearest possible points to $\bm p_F$ belonging to the domain. While in the translation invariant case, so in the domain $\mathcal D_{\Lambda,\beta}$, $\bm p_F$ has four equidistant nearest neighbors $\bar{\bm k}_{\eta,\eta'}$, in the domain $\mathcal D^d_{\Lambda,\beta}$ there are only two of them:
\begin{equation}
\underline{\bm k}_\eta=\left(\frac{\pi}{L+1},\eta \frac{\pi}{\beta}\right),\hspace{5mm} \eta=\pm 1.
\label{underline_k_eta}
\end{equation}

\begin{equation}
\begin{split}
\mathcal L_\mathcal T\left[\mathcal L_{\mathcal B}\left(\int _{[0,\beta)}dx_0 dy_0\sum_{ x, y \in \Lambda}\psi^{(\leq h)+}_{\bm x}\psi^{(\leq h)-}_{\bm y}W^{(h)}_2(\bm x,\bm y)\right)\right]=\\
=\mathcal L_\mathcal T \left(\frac{2}{\beta(L+1)}\sum_{\bm k'\in \mathcal D'^{d}_{\Lambda,\beta}}\hat \psi^{(\leq h)+}_{\bm k'+\bm p_F}\hat \psi^{(\leq h)-}_{\bm k'+\bm p_F}\hat W^{d(h)}_{2}(\bm k'+\bm p_F) \right)= \\=\frac{2}{\beta(L+1)}\sum_{\bm k'\in \mathcal D'^d_{\Lambda,\beta}}\hat \psi^{(\leq h)+}_{\bm k'+\bm p_F}\hat \psi^{(\leq h)-}_{\bm k'+\bm p_F}\cdot \\
\cdot \frac{1}{2}\left(\sum_{\eta=\pm 1}\hat W^{d(h)}_{2}(\underline{\bm k}_\eta-\bm p_F)\right)\left[1+\left(\frac{L}{\pi}\left(b_L+a_Le(k'+p_F)+\eta \frac{\beta}{\pi}k_0\right)\right)\right]
\end{split}
\label{localization_W^d_DBC}
\end{equation}
where again
\begin{equation}
a_L=\frac{\pi/L}{\sin (\pi/L)}, \hspace{3mm} b_L=\cos p_F\frac{\left(\cos (\pi/L)-1\right)\pi/L}{\sin(\pi/L)}
\end{equation}
and $\mathcal T$ stays for {\it "Taylor expansion"}.

\item {\bf Case $2n=2$, kernel $\mathcal W^{(h)}_2=\mathcal R_\mathcal B W_2^{(h)}$} Because of the non-diagonality of the kernel $\mathcal W^{(h)}_2$, there is no advantage in defining a localization procedure in $\bm k$ space, so we work directly in the real space-time, being inspired by Appendix (\ref{appendix_real_space-time_localization}).
\begin{equation}
\begin{split}
\tilde {\mathcal L}_\mathcal T\left[\mathcal R_{\mathcal B}\left(\int _{[0,\beta)}dx_0 dy_0\sum_{ x, y \in \Lambda}\psi^{(\leq h)+}_{\bm x}\psi^{(\leq h)-}_{\bm y} W^{(h)}_2(\bm x,\bm y) \right)\right]=\\
=\tilde {\mathcal L}_\mathcal T\left(\int _{[0,\beta)}dx_0 dy_0\sum_{ x, y \in \Lambda}\psi^{(\leq h)+}_{\bm x}\psi^{(\leq h)-}_{\bm y}\mathcal W^{(h)}_2(\bm x,\bm y) \right)=\\=\int _{[0,\beta)}dx_0 dy_0\sum_{ x, y \in \Lambda}\left.\psi^{(\leq h)+}_{\bm x}\psi^{(\leq h)-}_{\bm y}\right|_{y_0=x_0}c_\beta(y_0-x_0) \mathcal W^{(h)}_2(\bm x,\bm y),
\end{split}
\label{localization_definition_mathcalWcd_DBC}
\end{equation}
where:
\begin{equation}
\begin{split}
c_\beta=\cos\left(\frac{\pi x_0}{\beta}\right),
\end{split}
\end{equation}
and we introduced the symbol $\tilde \cdot$ because the Teylor expansion is performed only in the {\it space-direction}.

\item {\bf Case $2n=4$, kernel $\bar W^{(h)}_4$}  In this case, we expand the Grassmann variables in the quasi-particles representation (\ref{quasi_particles_decomposition_scale_h_DBC}), and we define
\begin{equation}
\begin{split}
\mathcal L_\mathcal T\left[\mathcal L_{\mathcal B}\left(\int_{[0,\beta)}dx_{1_0}\dots dx_{4_0} \sum_{ \substack{x_1,\dots, x_4 \\ \in \Lambda}} \bar W^{(h)}_4(\bm x_1, \bm x_2, \bm x_3, \bm x_4)\cdot\right. \right.\\ \cdot e^{-ip_F(\omega_1x_1+\omega_2x_2-\omega_3 x_3-\omega_4 x_4)} \left.\left.\cdot\psi^{(\leq h)+}_{\sigma_1,\omega_1,\bm x_1} \psi^{(\leq h)+}_{\sigma_2,\omega_2,\bm x_2}\psi^{(\leq h)-}_{\sigma_3,\omega_3,\bm x_3}\psi^{(\leq h)-}_{\sigma_4,\omega_4,\bm x_4}\right)\right]=\\
=\mathcal L_\mathcal T \int_{[0,\beta)}dx_{1_0}\dots dx_{4_0} \sum_{ \substack{x_1,\dots, x_4 \\ \in \Lambda}} \bar W^{(h)}_4(\bm x_1, \bm x_2, \bm x_3, \bm x_4) \cdot \\ \cdot e^{-ip_F(\omega_1x_1+\omega_2x_2-\omega_3 x_3-\omega_4 x_4)} \psi^{(\leq h)+}_{\sigma_1,\omega_1,\bm x_1} \psi^{(\leq h)+}_{\sigma_2,\omega_2,\bm x_2}\psi^{(\leq h)-}_{\sigma_3,\omega_3,\bm x_3}\psi^{(\leq h)-}_{\sigma_4,\omega_4,\bm x_4}=\\
=  \int_{[0,\beta)}dx_{1_0}\dots dx_{4_0} \sum_{ \substack{x_1,\dots, x_4 \\ \in \Lambda}}\bar  W^{(h)}_4(\bm x_1, \bm x_2, \bm x_3, \bm x_4)\cdot \\ \cdot e^{-ip_F(\omega_1x_1+\omega_2x_2-\omega_3 x_3-\omega_4 x_4)}e^{i\bar{\bm k}_{+}\cdot(\bm x_1-\bm x_4)} \psi^{(\leq h)+}_{\sigma_1,\omega_1,\bm x_4} \\e^{i\bar{\bm k}_{+}\cdot(\bm x_2-\bm x_4)}\psi^{(\leq h)+}_{\sigma_2,\omega_2,\bm x_4}
e^{-i\bar{\bm k}_{+}\cdot(\bm x_3-\bm x_4)}\psi^{(\leq h)-}_{\sigma_3,\omega_3,\bm x_4}\psi^{(\leq h)-}_{\sigma_4,\omega_4,\bm x_4}=\\
=\int_{[0,\beta)}dx_{1_0}\dots dx_{4_0} \sum_{ \substack{x_1,\dots, x_4 \\ \in \Lambda}} \bar W^{(h)}_4(\bm x_1, \bm x_2, \bm x_3, \bm x_4)\cdot \\ \cdot e^{-ip_F(\omega_1x_1+\omega_2x_2-\omega_3 x_3-\omega_4 x_4)}e^{i\bar{\bm k}_{+}\cdot(\bm x_1+\bm x_2-\bm x_3-\bm x_4)} \cdot \\ \cdot\psi^{(\leq h)+}_{\sigma_1,\omega_1,\bm x_4} \psi^{(\leq h)+}_{\sigma_2,\omega_2,\bm x_4}\psi^{(\leq h)-}_{\sigma_3,\omega_3,\bm x_4}\psi^{(\leq h)-}_{\sigma_4,\omega_4,\bm x_4}
\end{split}
\label{localization_4el_real_spacetime_DBC}
\end{equation}
where, by recalling the definition (\ref{underline_k_eta})
$$\bar{\bm k}_{+}=\left(\frac{\pi}{L+1},\frac{\pi}{\beta}\right).$$

\item {\bf Other cases} 
The case with $2n=4$ and the kernel $\mathcal R_{\mathcal B}\mathcal W^{(h)}_4$ does not need to be renormalized:
\begin{equation}
\begin{split}
\mathcal L_\mathcal T \left[\mathcal R_{\mathcal B}\left(\sum_{\bm x_1,\dots,\bm x_4}   W^{(h)}_4(\bm x_1, \bm x_2, \bm x_3, \bm x_4)e^{-ip_F(\omega_1x_1+\omega_2x_2-\omega_3 x_3-\omega_4 x_4)}\right. \right.\cdot \\ \cdot\left. \left. \psi^{(\leq h)+}_{\sigma_1,\omega_1,\bm x_1} \psi^{(\leq h)+}_{\sigma_2,\omega_2,\bm x_2}\psi^{(\leq h)-}_{\sigma_3,\omega_3,\bm x_3}\psi^{(\leq h)-}_{\sigma_4,\omega_4,\bm x_4}\right)\right]=\\
=\mathcal L_\mathcal T \left(\sum_{\bm x_1,\dots,\bm x_4}  \mathcal W^{(h)}_4(\bm x_1, \bm x_2, \bm x_3, \bm x_4)e^{-ip_F(\omega_1x_1+\omega_2x_2-\omega_3 x_3-\omega_4 x_4)}\right. \cdot \\ \cdot \left. \psi^{(\leq h)+}_{\sigma_1,\omega_1,\bm x_1} \psi^{(\leq h)+}_{\sigma_2,\omega_2,\bm x_2}\psi^{(\leq h)-}_{\sigma_3,\omega_3,\bm x_3}\psi^{(\leq h)-}_{\sigma_4,\omega_4,\bm x_4}\right)=0,
\end{split}
\end{equation}
and this implies that $\mathcal R_\mathcal T\mathcal R_{\mathcal B}=\mathcal R_{\mathcal B}$ if it acts on a quartic term. The same holds for the operators $\tilde{\mathcal L}_\mathcal T$ and $\tilde{\mathcal R}_\mathcal T$.\\
Of course, as in the translation invariant case, if $2n\geq 6$
\begin{equation}
\mathcal L_\mathcal T\left(\sum_{\bm x_1,\dots,\bm x_{2n}}\left[\prod_{j=1}^{n}\psi^{(\leq h)+}_{\bm x_{2j-1}} \psi^{(\leq h)-}_{\bm x_{2j}} \right]W_{\bm 2n}^{(h)}(\bm x_1,\dots,\bm x_{2n})\right)=0.
\end{equation}
where $i=0,1$. The same holds for $\tilde{\mathcal L}_\mathcal T$.
\end{itemize}

\chapter{Weighted integration and dimensional gain}
\label{appendix_proof_lemma_effective_gain}
\section*{Proof of Lemma (\ref{lemma_effective_gain})}
\begin{proof}
Let us start by bounding
\begin{equation}
\left| \int d\bm y \rho_h(y) g^{(\bar h)}_{P,\omega}(\bm y-\bm x)\right|
\end{equation}
where $\bar h < h$, recalling that for any $N,M\in\mathbb N$
\begin{equation}
\begin{split}
\left| \rho_h(y)\right|\leq \frac{C_N}{1+\left(\gamma^h |y|\right)^N}\leq \sum_{k\leq h}\gamma^{N(k-h)}e^{-\gamma^{k|y|}},\\
\left| g^{(\bar h)}_{P,\omega}(\bm x-\bm y)\right|\leq \gamma^{\bar h}\frac{C_M}{1+\left(\gamma^{\bar h} |\bm x-\bm y|\right)^M}\leq C \sum_{\bar k\leq \bar h}\gamma^{M(\bar k-\bar h)}e^{-\gamma^{\bar k|\bm x-\bm y|}}.
\end{split}
\end{equation}
for a suitable $C>0$.\\
First of all, we can perform the time integration in order to be left with, up to a constant $C_N C_M$
\begin{equation}
\begin{split}
\int_0^L dy\frac{1}{1+\gamma^{N h}|y|^N}\frac{1}{1+\gamma^{M\bar h}|x-y|^M}\leq \sum_{\bar k\leq \bar h}\sum_{k\leq h}\gamma^{N(k-h)}\gamma^{M(\bar k-\bar h)}\int_0^L dy e^{-\gamma^k |y|}e^{-\gamma^{\bar k}|x-y|}\leq \\
\leq c_0 \sum_{\bar k\leq \bar h}\sum_{k\leq h}\gamma^{N(k-h)}\gamma^{M(\bar k-\bar h)}\gamma^{-\max\{k,\bar k\}} \left[e^{-\gamma^k |x|}+e^{-\gamma^{\bar k}|x|}\right].
\end{split}
\end{equation}
Now our strategy is to control the contribution coming from the different domains of the double sum we are considering, bounding the previous expression by
\begin{equation}
\begin{split}
c_1\left( \sum_{\bar k\leq \bar h}\sum_{\bar h< k\leq h}\gamma^{N(k-h)}\gamma^{M(\bar k-\bar h)}\gamma^{-\max\{k,\bar k\}} \left[e^{-\gamma^k |x|}+e^{-\gamma^{\bar k}|x|}\right]\right.+\\ \left. \sum_{\bar k\leq \bar h}\sum_{k\leq \bar h}\gamma^{N(k-h)}\gamma^{M(\bar k-\bar h)}\gamma^{-\max\{k,\bar k\}} \left[e^{-\gamma^k |x|}+e^{-\gamma^{\bar k}|x|}\right]\right)=\\
=c_1\left(\sum_{\bar k\leq \bar h}\gamma^{M(\bar k-\bar h)}e^{-\gamma^{\bar k} |x|}\sum_{\bar h< k\leq h}\gamma^{N(k-h)}\gamma^{-k}+ \sum_{\bar k\leq \bar h}\gamma^{M(\bar k-\bar h)}\sum_{\bar h< k\leq h}\gamma^{N(k-h)}\gamma^{-k}e^{-\gamma^k |x|}\right.+\\
+\left. \sum_{\bar k\leq \bar h}\sum_{k\leq \bar h}\gamma^{N(k-h)}\gamma^{M(\bar k-\bar h)}\gamma^{-\max\{k,\bar k\}} \left[e^{-\gamma^k |x|}+e^{-\gamma^{\bar k}|x|}\right]\right).
\end{split}
\end{equation}
Now we show that the dominant term is the first one:
\begin{equation}
\sum_{\bar k\leq \bar h}\gamma^{M(\bar k-\bar h)}e^{-\gamma^{\bar k} |x|}\sum_{\bar h< k\leq h}\gamma^{N(k-h)}\gamma^{-k}\leq c_2 \sum_{\bar k\leq \bar h}\gamma^{M(\bar k-\bar h)}e^{-\gamma^{\bar k} |x|} \gamma^{-h}\leq c_3 \gamma^{(\bar h-h)} \frac{\gamma^{-\bar h}}{1+\gamma^{M\bar h}|x|^{\bar h}},
\label{lemma_scale_gain_rho_dominant_term}
\end{equation}
indeed the second one is bounded by:
\begin{equation}
\begin{split}
\sum_{\bar k\leq \bar h}\gamma^{M(\bar k-\bar h)}\sum_{\bar h< k\leq h}\gamma^{N(k-h)}\gamma^{-k}e^{-\gamma^k |x|}\leq c_3 \gamma^{N(\bar h-h)}\gamma^{-\bar h}e^{-\gamma^{\bar h} |x|}
\end{split}
\end{equation}
while for the remaining ones we split again the domain of the double sum to check all the contributions:
\begin{equation}
\begin{split}
\sum_{\bar k\leq \bar h}\sum_{k\leq \bar h}\gamma^{N(k-h)}\gamma^{M(\bar k-\bar h)}\gamma^{-\max\{k,\bar k\}} \left[e^{-\gamma^k |x|}+e^{-\gamma^{\bar k}|x|}\right]\leq\\ \leq c_4\gamma^{N(\bar h-h)}\left(\sum_{k\leq\bar h}\gamma^{N(k-\bar h)}e^{-\gamma^k|x|}\sum_{k\leq \bar k\leq \bar h}\gamma^{M(\bar k-\bar h)}\gamma^{-\bar k}+\right.\\+\left.\sum_{k\leq\bar h}\gamma^{N(k-\bar h)}\sum_{k\leq \bar k\leq \bar h}\gamma^{M(\bar k-\bar h)}\gamma^{-\bar k}e^{-\gamma^{\bar k}|x|}+\right.\\
\left.\sum_{\bar k\leq\bar h}\gamma^{M(\bar k-\bar h)}e^{-\gamma^{\bar k|}x|}\sum_{\bar k\leq k\leq \bar h}\gamma^{N(k-\bar h)}\gamma^{- k}+\sum_{\bar k\leq\bar h}\gamma^{M(\bar k-\bar h)}\sum_{\bar k\leq k\leq \bar h}\gamma^{N( k-\bar h)}\gamma^{- k}e^{-\gamma^{ k}|x|}\right)\leq \\
\leq c_5 \gamma^{N(\bar h-h)}\left(\frac{\gamma^{-\bar h}}{1+\gamma^{N\bar h}|x|^N}+\frac{\gamma^{-\bar h}}{1+\gamma^{M\bar h}|x|^M}+ \frac{\gamma^{-\bar h}}{1+\gamma^{(M+N-1)\bar h}|x|^{N+M-1}}\right),
\end{split}
\end{equation}
which are of course smaller than the dominant term in (\ref{lemma_scale_gain_rho_dominant_term}).\\
Following the same strategy of splitting the domain of the double sum, we now bound 
\begin{equation}
\int d\bm y \varpi_h(y) g^{(\bar h)}_{P,\omega}(\bm y-\bm x),
\end{equation}
where $\bar h < h$, recalling that for any $\theta\in (0,1)$ and $M\in\mathbb N$,
\begin{equation}
\begin{split}
\left| \varpi_h(y)\right|\leq \frac{C_\theta}{1+\left(\gamma^h |y|\right)^\theta}\simeq \sum_{k\leq h}\gamma^{\theta(k-h)}e^{-\gamma^{k|}y|}\\
\left| g^{(\bar h)}_{P,\omega}(\bm x-\bm y)\right|\leq \gamma^{\bar h}\frac{C_M}{1+\left(\gamma^{\bar h} |\bm x-\bm y|\right)^M}\simeq \sum_{\bar k\leq \bar h}\gamma^{M(\bar k-\bar h)}e^{-\gamma^{\bar k}|\bm x-\bm y|}.
\end{split}
\end{equation}
First of all, we can perform the time integration in order to be left, up to a constant $C_\theta C_M$, with
\begin{equation}
\begin{split}
\int_0^L dy \frac{1}{1+\gamma^{\theta h}|y|^\theta}\frac{1}{1+\gamma^{M\bar h}|x-y|^M}\leq \sum_{\bar k\leq \bar h}\sum_{k\leq h}\gamma^{\theta(k-h)}\gamma^{M(\bar k-\bar h)}\int_0^L dy e^{-\gamma^k |y|}e^{-\gamma^{\bar k}|x-y|}\leq \\
\leq c_0^{(\theta)} \sum_{\bar k\leq \bar h}\sum_{k\leq h}\gamma^{\theta(k-h)}\gamma^{M(\bar k-\bar h)}\gamma^{-\max\{k,\bar k\}} \left[e^{-\gamma^k |x|}+e^{-\gamma^{\bar k}|x|}\right]
\end{split}
\end{equation}
Now let us split the double sum as 
\begin{equation}
\begin{split}
c_1^{(\theta)}\left( \sum_{\bar k\leq \bar h}\sum_{\bar h< k\leq h}\gamma^{\theta(k-h)}\gamma^{M(\bar k-\bar h)}\gamma^{-\max\{k,\bar k\}} \left[e^{-\gamma^k |x|}+e^{-\gamma^{\bar k}|x|}\right]\right.+\\ \left. \sum_{\bar k\leq \bar h}\sum_{k\leq \bar h}\gamma^{\theta(k-h)}\gamma^{M(\bar k-\bar h)}\gamma^{-\max\{k,\bar k\}} \left[e^{-\gamma^k |x|}+e^{-\gamma^{\bar k}|x|}\right]\right)=\\
\leq c_2^{(\theta)}\left(\sum_{\bar k\leq \bar h}\gamma^{M(\bar k-\bar h)}e^{-\gamma^{\bar k} |x|}\sum_{\bar h< k\leq h}\gamma^{\theta (k-h)}\gamma^{-k}+ \sum_{\bar k\leq \bar h}\gamma^{M(\bar k-\bar h)}\sum_{\bar h< k\leq h}\gamma^{\theta (k-h)}\gamma^{-k}e^{-\gamma^k |x|}\right.+\\
+\left. \sum_{\bar k\leq \bar h}\sum_{k\leq \bar h}\gamma^{\theta (k-h)}\gamma^{M(\bar k-\bar h)}\gamma^{-\max\{k,\bar k\}} \left[e^{-\gamma^k |x|}+e^{-\gamma^{\bar k}|x|}\right]\right)
\end{split}
\end{equation}
Let us study the first two terms of the latter expression:
\begin{equation}
\begin{split}
\sum_{\bar k\leq \bar h}\gamma^{M(\bar k-\bar h)}e^{-\gamma^{\bar k} |x|}\sum_{\bar h< k\leq h}\gamma^{\theta (k-h)}\gamma^{-k}+ \sum_{\bar k\leq \bar h}\gamma^{M(\bar k-\bar h)}\sum_{\bar h< k\leq h}\gamma^{\theta (k-h)}\gamma^{-k}e^{-\gamma^k |x|}\leq\\\leq c_3^{(\theta)}  \gamma^{-\bar h} \frac{\gamma^{\theta(\bar h - h )}}{1+\gamma^{M \bar h}|x|^M},
\end{split}
\end{equation}
Moreover,
\begin{equation}
\begin{split}
 \sum_{\bar k\leq \bar h}\sum_{k\leq \bar h}\gamma^{\theta (k-h)}\gamma^{M(\bar k-\bar h)}\gamma^{-\max\{k,\bar k\}} \left[e^{-\gamma^k |x|}+e^{-\gamma^{\bar k}|x|}\right]\leq\\
c_4^{(\theta)} \gamma^{\theta(\bar h- h)} \left(\left[\sum_{\bar k\leq \bar h} \gamma^{M(\bar k-\bar h)} e^{-\gamma^{\bar k} x} \sum_{\bar k\leq k \leq \bar h} \gamma^{\theta(k-\bar h)}\gamma^{-k} \right.\right. +
 \sum_{\bar k\leq \bar h} \gamma^{M(\bar k-\bar h)}  \sum_{\bar k\leq k \leq \bar h} \gamma^{\theta(k-\bar h)}\gamma^{-k} e^{-\gamma^{ k} x}+\\+\left. \left.\sum_{k\leq \bar h}\gamma^{\theta(k-\bar h)}    \sum_{ k\leq \bar k \leq \bar h} \gamma^{M(\bar k-\bar h)} \gamma^{-\bar k}e^{-\gamma^{\bar  k} x}\right]+\sum_{k\leq \bar h}\gamma^{\theta(k-\bar h)} e^{-\gamma^{ k} x}   \sum_{ k\leq \bar k \leq \bar h} \gamma^{M(\bar k-\bar h)} \gamma^{-\bar k}\right)\leq \\
 \leq c_5^{(\theta)}\left(\gamma^{\theta(\bar h-h)}\frac{\gamma^{-\bar h}}{1+(\gamma^{\bar h}|x|)^{M-(1-\theta)}}+ \gamma^{\theta(\bar h-h)}\frac{\gamma^{-\bar h}}{1+\gamma^{\theta \bar h}|x|^\theta}\right)\leq C_\theta \gamma^{\theta(\bar h-h)}\frac{\gamma^{-\bar h}}{1+\gamma^{\theta \bar h}|x|^\theta} ,
\end{split}
\end{equation}
where in the very last line, the first term in brackets bounds the three sums in square brackets, while the second one bounds the sum not included in square brackets.
\end{proof}

\end{document}